\documentclass[11pt]{article}

\usepackage[utf8]{inputenc} 
\usepackage[english]{babel}
\usepackage[T1]{fontenc}
\usepackage[bbgreekl]{mathbbol}
\usepackage[margin=2cm,vmargin=2cm]{geometry}
\usepackage{caption}
\usepackage{tikz-cd}
\usepackage{amsmath,amsfonts,amssymb,amsthm}
\usepackage{authblk}
\usepackage{url}
\usepackage{fullpage,setspace}
\usepackage{cancel,verbatim,slashed,todonotes}
\usepackage{xcolor,mdframed,graphicx,color}
\usepackage{dsfont}
\usepackage{mathtools}
\usepackage{multirow,stmaryrd}
\usepackage{array}
\usepackage{makecell}
\usepackage{upgreek}
\usepackage{mathrsfs}
\usepackage{braket}
\usepackage{pict2e}
\usepackage{stackengine}
\usepackage{dsfont}
\usepackage{todonotes}
\usepackage{multirow}
\usepackage{enumitem}

\DeclareSymbolFontAlphabet{\mathbbvar}{bbold}
\DeclareSymbolFontAlphabet{\mathbb}{AMSb}
\DeclareFontFamily{U}{BOONDOX-calo}{\skewchar\font=45 }
    \DeclareFontShape{U}{BOONDOX-calo}{m}{n}{
      <-> s*[1.05] BOONDOX-r-calo}{}
    \DeclareFontShape{U}{BOONDOX-calo}{b}{n}{
      <-> s*[1.05] BOONDOX-b-calo}{}
    \DeclareMathAlphabet{\mathcalboondox}{U}{BOONDOX-calo}{m}{n}
    \SetMathAlphabet{\mathcalboondox}{bold}{U}{BOONDOX-calo}{b}{n}
    \DeclareMathAlphabet{\mathbcalboondox}{U}{BOONDOX-calo}{b}{n}

\newcommand{\Coo}{\mathcal{C}^{\infty}}

\newcommand{\Mfd}{\mathsf{Mfd}}

\newcommand{\B}{\mathbf{B}}

\newcommand{\di}{\mathrm{d}}
\newcommand{\crit}[1]{\mathrm{Crit}({#1})}

\newcommand{\btimes}{\,\widehat{\otimes}\,}%{\,\widehat{\otimes}_\beta\,}
\newcommand{\CE}{\mathrm{CE}}
\newcommand{\Spec}{\mathrm{Spec}}

\newcommand{\op}{\mathrm{op}}
\newcommand{\Hom}{\mathrm{Hom}}

\newcommand{\sSet}{\mathsf{sSet}}
\newcommand{\St}{\mathbf{St}}

\newcommand{\Bun}{\mathbf{Bun}}
\newcommand{\Shoo}{\mathsf{St}}
\newcommand{\bbR}{\mathbb{R}}
\newcommand{\bbL}{\mathbb{L}}
\newcommand{\bfR}{\mathds{R}}
\newcommand{\bfL}{\mathds{L}}
\newcommand{\bbT}{\mathbb{T}}
\newcommand{\bbD}{\mathbb{D}}
 
\newcommand{\bfGamma}{\mathbf{\Gamma}}
\newcommand{\bfOmega}{\mathbf{\Omega}}
\newcommand{\dR}{\mathrm{dR}}

\newcommand{\BV}{\mathrm{BV}}

\newcommand{\itPhi}{\mathit{\Phi}}

\newcommand{\itDelta}{\mathit{\Delta}}

\newcommand{\jet}[1]{\mathrm{Jet}{#1}}

\newcommand{\quasicat}{\pmb{(\infty,}\mathbf{1}\pmb{)}\mathbf{Cat}}
\newcommand{\QCoh}[1]{\mathrm{QCoh}({#1})}

\usepackage{scalerel}

\theoremstyle{definition}
\newtheorem{theorem}{Theorem}[section]
\theoremstyle{definition}
\newtheorem{definition}[theorem]{Definition}
\theoremstyle{definition}
\newtheorem{lemma}[theorem]{Lemma}
\theoremstyle{definition}
\newtheorem{corollary}[theorem]{Corollary}
\theoremstyle{definition}
\newtheorem{remark}[theorem]{Remark}
\theoremstyle{definition}
\newtheorem{example}[theorem]{Example}
\theoremstyle{definition}
\newtheorem{proposition}[theorem]{Proposition}
\theoremstyle{definition}

\theoremstyle{definition}
\newtheorem{construction}[theorem]{Construction}
\numberwithin{equation}{subsection}

\DeclareRobustCommand\longtwoheadrightarrow
    {\relbar\joinrel\twoheadrightarrow}
\DeclareRobustCommand\longhookrightarrow
     {\lhook\joinrel\longrightarrow}
\newcommand{\xtwoheadrightarrow}[2][]{%
  \mathrel{\ooalign{$\xrightarrow[#1\mkern4mu]{#2\mkern4mu}$\cr%
  \hidewidth$\rightarrow\mkern4mu$}}
}
\tikzcdset{%
    triple line/.code={\tikzset{%
        double equal sign distance, % replace by double distance = 'measure' 
        double=\pgfkeysvalueof{/tikz/commutative diagrams/background color}}},
    quadruple line/.code={\tikzset{%
        double equal sign distance, % replace by double distance = 'measure'
        double=\pgfkeysvalueof{/tikz/commutative diagrams/background color}}},
    Rrightarrow/.code={\tikzcdset{triple line}\pgfsetarrows{tikzcd implies cap-tikzcd implies}},
    RRightarrow/.code={\tikzcdset{quadruple line}\pgfsetarrows{tikzcd implies cap-tikzcd implies}}
}

\makeatletter
\providecommand{\leftsquigarrow}{%
  \mathrel{\mathpalette\reflect@squig\relax}%
}
\newcommand{\reflect@squig}[2]{%
  \reflectbox{$\m@th#1\rightsquigarrow$}%
}
\makeatother

\usepackage{newunicodechar}%% YONEDA
\newunicodechar{よ}{\text{\usefont{U}{min}{m}{n}\symbol{'210}}} %Yoneda emdedding
\DeclareFontFamily{U}{min}{}
\DeclareFontShape{U}{min}{m}{n}{<-> udmj30}{}
\usepackage{CJKutf8}
\usepackage{ textcomp }

\makeatletter
\newcommand{\xRrightarrow}[2][]{\ext@arrow 0359\Rrightarrowfill@{#1}{#2}}
\newcommand{\Rrightarrowfill@}{\arrowfill@\equiv\equiv\Rrightarrow}
\newcommand{\xLleftarrow}[2][]{\ext@arrow 3095\Lleftarrowfill@{#1}{#2}}
\newcommand{\Lleftarrowfill@}{\arrowfill@\Lleftarrow\equiv\equiv}
\newcommand{\xLleftRrightarrow}[2][]{\ext@arrow 3399\LleftRrightarrowfill@{#1}{#2}}
\newcommand{\LleftRrightarrowfill@}{\arrowfill@\Lleftarrow\equiv\Rrightarrow}
\makeatother

\usepackage[style=alphabetic, sorting=nyt, backend=biber, doi=false, isbn=false, maxbibnames=9]{biblatex}
\usepackage[hypertexnames=true]{hyperref}
\newcommand*{\bibtitle}{References}

\renewcommand*{\bibfont}{\raggedright\small}

\usepackage{cleveref}

\title{
\begin{flushright}
\normalsize{ }
\end{flushright}
\vspace{1cm}
\bigskip
\bf
{Towards non-perturbative BV-theory\\via derived differential geometry}\vspace{0.3cm}}
\author{{\sc Luigi Alfonsi} and {\sc Charles Young}}
\affil{\em\normalsize Department of Physics, Astronomy and Mathematics,\\\em University of Hertfordshire, Hatfield AL10 9AB, UK\\\vspace{4mm}\tt \href{mailto:l.alfonsi@herts.ac.uk}{l.alfonsi@herts.ac.uk}, \href{mailto:c.young8@herts.ac.uk}{c.young8@herts.ac.uk}}
\date{\small July 27, 2023}

\begin{document}

\maketitle
\vspace{0.5cm}
\abstract{\noindent We propose a global geometric framework which allows one to encode a natural non-perturbative generalisation of usual Batalin–Vilkovisky (BV-)theory. 
Namely, we construct a concrete model of derived differential geometry, whose geometric objects are formal derived smooth stacks, i.e. stacks on formal derived smooth manifolds, together with a notion of differential geometry on them. This provides a working language to study generalised geometric spaces that are smooth, infinite-dimensional, higher and derived at the same time.
Such a formalism is obtained by combining Schreiber's differential cohesion with the machinery of T\"oen-Vezzosi's homotopical algebraic geometry applied to the theory of derived manifolds of Spivak and Carchedi-Steffens.
We investigate two classes of examples of non-perturbative classical BV-theories in the context of derived differential geometry: scalar field theory and Yang-Mills theory. 
\noindent

\vspace{0.3cm}

\vspace{0.3cm}
\noindent \textbf{MSC 2020}: \tt{81Txx}, \tt{14A30}, \tt{18N40}
}

\newpage
\tableofcontents
%\newpage

%%%%%%%%%%%%%%%%%%%%%%%%%%%%%%%%%%%%%%%%%%%%%%%%%%%%%%%%%%%%%%%%%%%%%%%%%%%%%%%%%%%%%%%%%%%%%
\section*{Introduction}
\addcontentsline{toc}{section}{Introduction}
\setlength{\parindent}{0pt}
\setlength{\parskip}{0.5em}

\paragraph{BV-theory.}
Batalin-Vilkovisky (BV-)theory \cite{BATALIN198127} is an extremely powerful and successful mathematical framework for perturbatively formalising and quantising classical field theories, including theories with gauge symmetries. 
BV-theory has been applied to a wide range of physical systems and has deep connections to various areas of mathematics, including homological algebra, Poisson geometry, and symplectic geometry. See \cite{Cattaneo:2023hxv} for an overview.

Essentially, classical BV-theory replaces the problem of determining the critical locus of the action functional -- i.e. the space of solutions of the field equations -- with the problem of constructing the derived critical locus of the action functional \cite{FactII}.

In the literature, various different approaches to BV-theory emerge in the settings of several broader programmes, including:
%there are, roughly speaking, at least three types of approach to BV-theory:
\begin{enumerate}[label=(\textit{\roman*})]
    \item \textit{$NQP$-manifolds} approach (see \cite{Pau14, Saem18bv, Saem19bv, Doubek_2019, Jurco:2019yfd, Jurco:2020yyu} and many others), where the algebra of classical observables is given by a Poisson dg-Lie algebra of functions on an $NQP$-manifold, i.e. a differential-graded manifold (dg-manifold) equipped with a $(-1)$-shifted symplectic form. 
    \item \textit{Factorisation Algebras} approach (see \cite{Costello2011RenormalizationAE, FactI, FactII}), where the algebra of classical observables of a theory is given by the $\mathbb{P}_0$-algebra of functions on a $(-1)$-shifted symplectic pointed formal moduli problem (i.e. a derived stack on Artinian dg-algebras), which is also sheaved on the spacetime manifold. Quantisation is then provided by a graded Heisenberg extension of such a algebra.
    \item \textit{Perturbative Algebraic Quantum Field Theory} -- pAQFT for short -- (see  \cite{Rejzner:2011jsa, Fredenhagen_2012, Fredenhagen:2011an, Rejzner:2016hdj, Benini:2018oeh, Benini:2019uge, Benini_2019,  hawkins2020star, Rejzner:2020bsc, Rejzner:2020xid}), where the algebra of observables is usually given by a net of locally convex topological Poisson $\ast$-algebras on spacetime equipped with Peierls bracket. The BV-complex here is interestingly related to such a bracket structure and BV-quantisation emerges by a time-ordering of its classical counterpart.
\end{enumerate}
Despite their different constructions, these approaches share a significant amount of common ground. 
In fact, in the $NQP$-geometric perspective, the central objects one studies are $\mathbb{Z}$-graded $NQP$-manifolds, which are nothing but symplectic $L_\infty$-algebroids (for instance, see \cite{Sev01}). Such geometric objects are evidently closely related to -- even if different prima facie from -- the sheaves of formal moduli problems on spacetime appearing in the factorisation algebras approach. 
In particular, they both give rise to an $L_\infty$-algebra structure on the space of classical observables, as seen respectively in \cite{Saem19bv} and \cite{FactII}.
Moreover, the factorisation algebra and pAQFT formalisms are understood to be intimately related, through a correspondence which was delineated and explored by \cite{Gwilliam:2017ses, Benini:2019ujs}.

Now, of course, a fully \emph{non-perturbative} formulation of quantum field theory (QFT) is a major goal of modern theoretical physics. 
Some of the most interesting and challenging features of gauge theories are intrinsically non-perturbative and, therefore, lie beyond the horizon of perturbation theory. These include the mass gap problem, the phenomenon of confinement, the Landau pole, instantons, solitons (e.g.$\,$'t Hooft–Polyakov monopoles), domain walls and flux tubes. Moreover, from a purely conceptual standpoint, the project of QFT cannot be considered fully accomplished until the framework is able to describe the totality of fundamental phenomena of quantum fields.

At this point we can phrase the objective of this paper as follows: we want to generalise the intrinsically perturbative geometric formalism underlying usual BV-theory to its global non-perturbative version.
To achieve such a goal we need to generalise the pointed formal moduli problems of BV-theory to geometric objects which are fully-fledged derived stacks.
In fact, pointed formal moduli problems can be understood as pointed spaces probed by formal derived disks: this way, they can geometrically encode the infinitesimal deformations of a field configuration. Thus, if we want to be able to formalise finite deformations of a field theory, we must generalise our probing spaces to ``finite'' geometric objects. This leads to derived stacks.

\paragraph{Stacks and derived stacks.}
In a sheaf-theoretic geometry, the geometric structure of the spaces of the theory is defined by probing them with a certain class of test spaces.
For example, in higher smooth geometry our test spaces are ordinary smooth manifolds and the smooth structure of our spaces - namely, smooth stacks - is determined by the simplicial set of ways every smooth manifold can probe them.
Similarly, formal smooth stacks are defined by using infinitesimally thickened manifolds as test spaces.

\begin{figure}[h!]
    \centering
% Gradient Info
\tikzset {_l9ezxcurk/.code = {\pgfsetadditionalshadetransform{ \pgftransformshift{\pgfpoint{83.16 bp } { -103.62 bp }  }  \pgftransformscale{1.32 }  }}}
\pgfdeclareradialshading{_fn3klwi7r}{\pgfpoint{-72bp}{88bp}}{rgb(0bp)=(1,1,1);
rgb(0bp)=(1,1,1);
rgb(25bp)=(0.55,0.55,0.55);
rgb(400bp)=(0.55,0.55,0.55)}
% Gradient Info
\tikzset {_fwibr2aj4/.code = {\pgfsetadditionalshadetransform{ \pgftransformshift{\pgfpoint{83.16 bp } { -103.62 bp }  }  \pgftransformscale{1.32 }  }}}
\pgfdeclareradialshading{_8usi6y6e0}{\pgfpoint{-72bp}{88bp}}{rgb(0bp)=(1,1,1);
rgb(0bp)=(1,1,1);
rgb(25bp)=(0.55,0.55,0.55);
rgb(400bp)=(0.55,0.55,0.55)}
% Gradient Info
\tikzset {_efzaulyaf/.code = {\pgfsetadditionalshadetransform{ \pgftransformshift{\pgfpoint{0 bp } { 0 bp }  }  \pgftransformscale{1.08 }  }}}
\pgfdeclareradialshading{_7d42issc8}{\pgfpoint{0bp}{0bp}}{rgb(0bp)=(0.29,0.56,0.89);
rgb(0bp)=(0.29,0.56,0.89);
rgb(25bp)=(0.29,0.56,0.89);
rgb(400bp)=(0.29,0.56,0.89)}
\tikzset{_p6n2trgor/.code = {\pgfsetadditionalshadetransform{\pgftransformshift{\pgfpoint{0 bp } { 0 bp }  }  \pgftransformscale{1.08 } }}}
\pgfdeclareradialshading{_qwzll125h} { \pgfpoint{0bp} {0bp}} {color(0bp)=(transparent!0);
color(0bp)=(transparent!0);
color(25bp)=(transparent!97);
color(400bp)=(transparent!97)} 
\pgfdeclarefading{_v7f0zyp1g}{\tikz \fill[shading=_qwzll125h,_p6n2trgor] (0,0) rectangle (50bp,50bp); } 
% Gradient Info
\tikzset {_4cq9aae63/.code = {\pgfsetadditionalshadetransform{ \pgftransformshift{\pgfpoint{0 bp } { 0 bp }  }  \pgftransformscale{1 }  }}}
\pgfdeclareradialshading{_znf8guanj}{\pgfpoint{0bp}{0bp}}{rgb(0bp)=(0.29,0.56,0.89);
rgb(0bp)=(0.29,0.56,0.89);
rgb(6.903833661760602bp)=(0.29,0.56,0.89);
rgb(25bp)=(1,1,1);
rgb(400bp)=(1,1,1)}
\tikzset{_qnyt6rxt6/.code = {\pgfsetadditionalshadetransform{\pgftransformshift{\pgfpoint{0 bp } { 0 bp }  }  \pgftransformscale{1 } }}}
\pgfdeclareradialshading{_p99urvpzc} { \pgfpoint{0bp} {0bp}} {color(0bp)=(transparent!38);
color(0bp)=(transparent!38);
color(6.903833661760602bp)=(transparent!72);
color(25bp)=(transparent!100);
color(400bp)=(transparent!100)} 
\pgfdeclarefading{_xb85gxzmw}{\tikz \fill[shading=_p99urvpzc,_qnyt6rxt6] (0,0) rectangle (50bp,50bp); } 
\tikzset{every picture/.style={line width=0.75pt}} %set default line width to 0.75pt        
\begin{tikzpicture}[x=0.75pt,y=0.75pt,yscale=-0.75,xscale=0.75]
%uncomment if require: \path (0,565); %set diagram left start at 0, and has height of 565
%Shape: Circle [id:dp3817065997550373] 
\path  [shading=_fn3klwi7r,_l9ezxcurk] (326.25,417.43) .. controls (326.25,354.63) and (377.16,303.72) .. (439.96,303.72) .. controls (502.76,303.72) and (553.67,354.63) .. (553.67,417.43) .. controls (553.67,480.23) and (502.76,531.14) .. (439.96,531.14) .. controls (377.16,531.14) and (326.25,480.23) .. (326.25,417.43) -- cycle ; % for fading 
 \draw  [color={rgb, 255:red, 0; green, 0; blue, 0 }  ,draw opacity=1 ] (326.25,417.43) .. controls (326.25,354.63) and (377.16,303.72) .. (439.96,303.72) .. controls (502.76,303.72) and (553.67,354.63) .. (553.67,417.43) .. controls (553.67,480.23) and (502.76,531.14) .. (439.96,531.14) .. controls (377.16,531.14) and (326.25,480.23) .. (326.25,417.43) -- cycle ; % for border 
%Shape: Polygon [id:ds4988637047978681] 
\draw  [draw opacity=0][fill={rgb, 255:red, 74; green, 144; blue, 226 }  ,fill opacity=0.34 ][line width=0.75]  (379.28,372.64) -- (384.11,358.81) -- (391.12,344.05) -- (403.71,330.44) -- (419.19,327.82) -- (438.33,326.33) -- (455.93,328.03) -- (472.16,331.05) -- (485.7,345.87) -- (494.11,362.14) -- (497.61,374.14) -- (500.44,389.14) -- (500.6,399.76) -- (466.45,393.79) -- (436.94,392.14) -- (406.78,393.31) -- (376.33,397.49) -- cycle ;
%Shape: Circle [id:dp23027282441372354] 
\path  [shading=_8usi6y6e0,_fwibr2aj4] (327.08,146.71) .. controls (327.08,83.91) and (377.99,33) .. (440.79,33) .. controls (503.59,33) and (554.5,83.91) .. (554.5,146.71) .. controls (554.5,209.51) and (503.59,260.42) .. (440.79,260.42) .. controls (377.99,260.42) and (327.08,209.51) .. (327.08,146.71) -- cycle ; % for fading 
 \draw  [color={rgb, 255:red, 0; green, 0; blue, 0 }  ,draw opacity=1 ] (327.08,146.71) .. controls (327.08,83.91) and (377.99,33) .. (440.79,33) .. controls (503.59,33) and (554.5,83.91) .. (554.5,146.71) .. controls (554.5,209.51) and (503.59,260.42) .. (440.79,260.42) .. controls (377.99,260.42) and (327.08,209.51) .. (327.08,146.71) -- cycle ; % for border 
%Shape: Ellipse [id:dp9545472601197613] 
\path  [shading=_7d42issc8,_efzaulyaf,path fading= _v7f0zyp1g ,fading transform={xshift=2}] (417.33,81.33) .. controls (417.33,69.76) and (428.68,60.39) .. (442.67,60.39) .. controls (456.66,60.39) and (468,69.76) .. (468,81.33) .. controls (468,92.9) and (456.66,102.28) .. (442.67,102.28) .. controls (428.68,102.28) and (417.33,92.9) .. (417.33,81.33) -- cycle ; % for fading 
 \draw  [color={rgb, 255:red, 74; green, 144; blue, 226 }  ,draw opacity=0.76 ][dash pattern={on 4.5pt off 4.5pt}][line width=0.75]  (417.33,81.33) .. controls (417.33,69.76) and (428.68,60.39) .. (442.67,60.39) .. controls (456.66,60.39) and (468,69.76) .. (468,81.33) .. controls (468,92.9) and (456.66,102.28) .. (442.67,102.28) .. controls (428.68,102.28) and (417.33,92.9) .. (417.33,81.33) -- cycle ; % for border 
%Shape: Circle [id:dp40056921878050433] 
\draw  [color={rgb, 255:red, 0; green, 67; blue, 148 }  ,draw opacity=1 ][fill={rgb, 255:red, 0; green, 80; blue, 173 }  ,fill opacity=1 ] (440.25,81.33) .. controls (440.25,80.67) and (440.79,80.13) .. (441.46,80.13) .. controls (442.13,80.13) and (442.67,80.67) .. (442.67,81.33) .. controls (442.67,82) and (442.13,82.54) .. (441.46,82.54) .. controls (440.79,82.54) and (440.25,82) .. (440.25,81.33) -- cycle ;
%Shape: Circle [id:dp9623222570372301] 
\path  [shading=_znf8guanj,_4cq9aae63,path fading= _xb85gxzmw ,fading transform={xshift=2}] (61.22,143.71) .. controls (61.22,129.31) and (72.89,117.64) .. (87.29,117.64) .. controls (101.69,117.64) and (113.36,129.31) .. (113.36,143.71) .. controls (113.36,158.11) and (101.69,169.78) .. (87.29,169.78) .. controls (72.89,169.78) and (61.22,158.11) .. (61.22,143.71) -- cycle ; % for fading 
 \draw  [color={rgb, 255:red, 74; green, 144; blue, 226 }  ,draw opacity=0.19 ] (61.22,143.71) .. controls (61.22,129.31) and (72.89,117.64) .. (87.29,117.64) .. controls (101.69,117.64) and (113.36,129.31) .. (113.36,143.71) .. controls (113.36,158.11) and (101.69,169.78) .. (87.29,169.78) .. controls (72.89,169.78) and (61.22,158.11) .. (61.22,143.71) -- cycle ; % for border 

%Shape: Circle [id:dp10718904208747926] 
\draw  [fill={rgb, 255:red, 0; green, 0; blue, 0 }  ,fill opacity=1 ] (84.58,143.71) .. controls (84.58,142.21) and (85.8,141) .. (87.29,141) .. controls (88.79,141) and (90,142.21) .. (90,143.71) .. controls (90,145.2) and (88.79,146.42) .. (87.29,146.42) .. controls (85.8,146.42) and (84.58,145.2) .. (84.58,143.71) -- cycle ;
%Straight Lines [id:da05951525775911071] 
\draw    (137,149) -- (296,148.89) ;
\draw [shift={(298,148.89)}, rotate = 179.96] [color={rgb, 255:red, 0; green, 0; blue, 0 }  ][line width=0.75]    (10.93,-3.29) .. controls (6.95,-1.4) and (3.31,-0.3) .. (0,0) .. controls (3.31,0.3) and (6.95,1.4) .. (10.93,3.29)   ;
%Straight Lines [id:da4505713161871294] 
\draw    (136,419.89) -- (295,419.78) ;
\draw [shift={(297,419.78)}, rotate = 179.96] [color={rgb, 255:red, 0; green, 0; blue, 0 }  ][line width=0.75]    (10.93,-3.29) .. controls (6.95,-1.4) and (3.31,-0.3) .. (0,0) .. controls (3.31,0.3) and (6.95,1.4) .. (10.93,3.29)   ;
%Shape: Arc [id:dp18949112305950022] 
\draw  [draw opacity=0] (376.33,397.49) .. controls (376.35,368.9) and (387.34,343.92) .. (403.71,330.44) -- (431.58,397.54) -- cycle ; \draw  [color={rgb, 255:red, 0; green, 80; blue, 173 }  ,draw opacity=1 ] (376.33,397.49) .. controls (376.35,368.9) and (387.34,343.92) .. (403.71,330.44) ;  
%Shape: Arc [id:dp5813995280210762] 
\draw  [draw opacity=0] (470.63,330.42) .. controls (488.43,343.33) and (500.6,369.43) .. (500.6,399.54) .. controls (500.6,399.61) and (500.6,399.68) .. (500.6,399.76) -- (445.35,399.54) -- cycle ; \draw  [color={rgb, 255:red, 0; green, 80; blue, 173 }  ,draw opacity=1 ] (470.63,330.42) .. controls (488.43,343.33) and (500.6,369.43) .. (500.6,399.54) .. controls (500.6,399.61) and (500.6,399.68) .. (500.6,399.76) ;  
%Shape: Arc [id:dp7603505561347381] 
\draw  [draw opacity=0] (376.33,397.49) .. controls (394.35,394.07) and (414.85,392.25) .. (436.56,392.5) .. controls (459.89,392.76) and (481.74,395.39) .. (500.6,399.76) -- (435.99,442.4) -- cycle ; \draw  [color={rgb, 255:red, 0; green, 80; blue, 173 }  ,draw opacity=1 ] (376.33,397.49) .. controls (394.35,394.07) and (414.85,392.25) .. (436.56,392.5) .. controls (459.89,392.76) and (481.74,395.39) .. (500.6,399.76) ;  
%Shape: Arc [id:dp15809426854277375] 
\draw  [draw opacity=0] (403.71,330.44) .. controls (414.48,327.91) and (426.16,326.59) .. (438.34,326.73) .. controls (450.26,326.87) and (461.65,328.39) .. (472.16,331.05) -- (437.63,389.05) -- cycle ; \draw  [color={rgb, 255:red, 0; green, 80; blue, 173 }  ,draw opacity=1 ] (403.71,330.44) .. controls (414.48,327.91) and (426.16,326.59) .. (438.34,326.73) .. controls (450.26,326.87) and (461.65,328.39) .. (472.16,331.05) ;  
%Shape: Arc [id:dp44106761982833476] 
\draw  [draw opacity=0] (376.99,380.46) .. controls (394.97,377.06) and (415.41,375.25) .. (437.06,375.5) .. controls (459.46,375.75) and (480.5,378.18) .. (498.83,382.24) -- (436.49,425.4) -- cycle ; \draw  [color={rgb, 255:red, 0; green, 80; blue, 173 }  ,draw opacity=1 ] (376.99,380.46) .. controls (394.97,377.06) and (415.41,375.25) .. (437.06,375.5) .. controls (459.46,375.75) and (480.5,378.18) .. (498.83,382.24) ;  
%Shape: Arc [id:dp9495056227346264] 
\draw  [draw opacity=0] (382.6,362.47) .. controls (399.18,359.71) and (417.63,358.27) .. (437.06,358.5) .. controls (457.63,358.73) and (477.06,360.8) .. (494.29,364.28) -- (436.49,408.4) -- cycle ; \draw  [color={rgb, 255:red, 0; green, 80; blue, 173 }  ,draw opacity=1 ] (382.6,362.47) .. controls (399.18,359.71) and (417.63,358.27) .. (437.06,358.5) .. controls (457.63,358.73) and (477.06,360.8) .. (494.29,364.28) ;  
%Shape: Arc [id:dp19624963677774776] 
\draw  [draw opacity=0] (391.12,344.05) .. controls (405.14,342.23) and (420.28,341.31) .. (436.06,341.5) .. controls (453.66,341.7) and (470.43,343.24) .. (485.7,345.87) -- (435.49,391.4) -- cycle ; \draw  [color={rgb, 255:red, 0; green, 80; blue, 173 }  ,draw opacity=1 ] (391.12,344.05) .. controls (405.14,342.23) and (420.28,341.31) .. (436.06,341.5) .. controls (453.66,341.7) and (470.43,343.24) .. (485.7,345.87) ;  
%Straight Lines [id:da5940879805974233] 
\draw [color={rgb, 255:red, 0; green, 80; blue, 173 }  ,draw opacity=1 ]   (438.33,326.33) -- (436.63,393.05) ;
%Shape: Arc [id:dp8739627273921831] 
\draw  [draw opacity=0] (406.54,394.71) .. controls (406.4,392.02) and (406.33,389.3) .. (406.33,386.54) .. controls (406.33,363.09) and (411.31,342.07) .. (419.19,327.82) -- (443.58,386.54) -- cycle ; \draw  [color={rgb, 255:red, 0; green, 80; blue, 173 }  ,draw opacity=1 ] (406.54,394.71) .. controls (406.4,392.02) and (406.33,389.3) .. (406.33,386.54) .. controls (406.33,363.09) and (411.31,342.07) .. (419.19,327.82) ;  
%Shape: Arc [id:dp29279040338192797] 
\draw  [draw opacity=0] (455.93,328.03) .. controls (462.48,342.26) and (466.6,362.97) .. (466.6,386.04) .. controls (466.6,388.65) and (466.55,391.24) .. (466.45,393.79) -- (434.72,386.04) -- cycle ; \draw  [color={rgb, 255:red, 0; green, 80; blue, 173 }  ,draw opacity=1 ] (455.93,328.03) .. controls (462.48,342.26) and (466.6,362.97) .. (466.6,386.04) .. controls (466.6,388.65) and (466.55,391.24) .. (466.45,393.79) ;  
%Shape: Grid [id:dp23530501747377564] 
\draw  [draw opacity=0][fill={rgb, 255:red, 74; green, 144; blue, 226 }  ,fill opacity=0.36 ] (52.5,382.83) -- (124.5,382.83) -- (124.5,454.83) -- (52.5,454.83) -- cycle ; \draw   (70.5,382.83) -- (70.5,454.83)(88.5,382.83) -- (88.5,454.83)(106.5,382.83) -- (106.5,454.83) ; \draw   (52.5,400.83) -- (124.5,400.83)(52.5,418.83) -- (124.5,418.83)(52.5,436.83) -- (124.5,436.83) ; \draw   (52.5,382.83) -- (124.5,382.83) -- (124.5,454.83) -- (52.5,454.83) -- cycle ;
% Text Node
\draw (185,125) node [anchor=north west][inner sep=0.75pt]   [align=left] {{\small probing}};
% Text Node
\draw (185,396) node [anchor=north west][inner sep=0.75pt]   [align=left] {{\small probing}};
% Text Node
\draw (-2,136) node [anchor=north west][inner sep=0.75pt]   [align=left] {a)};
% Text Node
\draw (-2,411) node [anchor=north west][inner sep=0.75pt]   [align=left] {b)};
\end{tikzpicture}
    \caption{Probing a formal smooth stack by (a) infinitesimally thickened points and (b) ordinary smooth manifolds.}
    \label{fig:intro_moduli}
\end{figure}
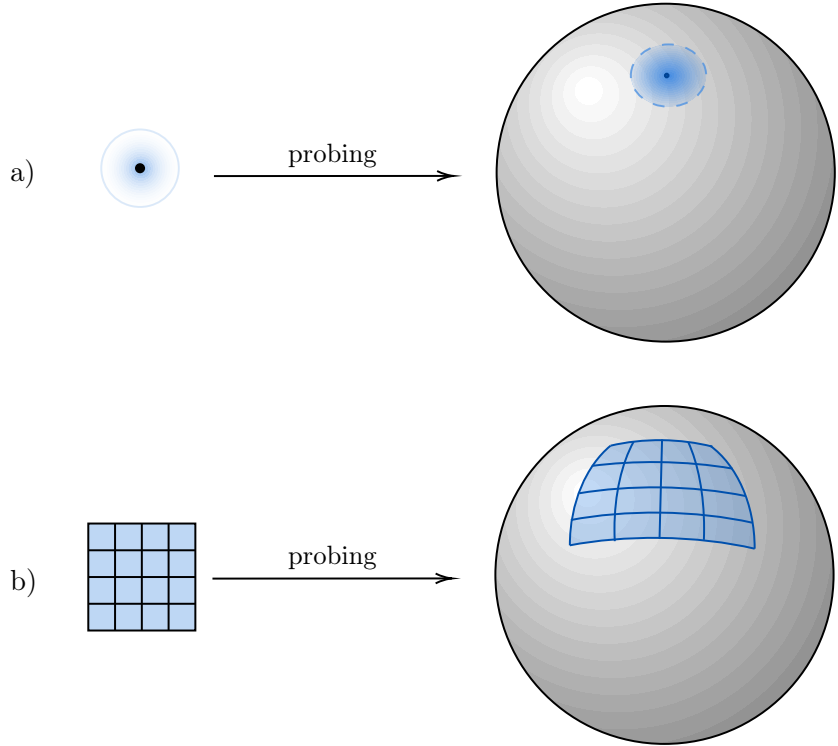

The success of smooth stacks is multifaceted.
First of all, just like smooth sheaves (also known as smooth sets) they generalise smooth manifolds by including infinite-dimensional smooth spaces. Secondly, they "categorify" smooth manifolds by relaxing the gluing conditions. The result is that spaces can be glued together by higher gauge transformations.
The archetypal example of a smooth stack is $\Bun_G(M)$, the stack of principal $G$-bundles on a fixed ordinary smooth manifold $M$. At any test manifold $U$, the space of sections $\Hom(U,\,\Bun_G(M))$ is a groupoid whose objects are $U$-parametrised families of $G$-bundles on $M$ and whose morphisms are $U$-parametrised families of gauge transformations.
The theory of smooth stacks has been systematised by the notion of differential cohesive $(\infty,1)$-topos developed by \cite{DCCTv2} (see also \cite{myers2022orbifolds}).

Most often, the intersection of two smooth sub-manifolds is not a smooth manifold. The only exception is when is when the two sub-manifolds are transverse.
As a reflection of this property of smooth manifolds, the limits in the category of smooth stacks (despite existing) do not behave well from an intersection theory point of view.
However, in mathematical physics it is of primary importance to construct a well-defined space of solutions of the equations of motion (also known as the phase space), which can be precisely understood as the intersection between the section induced by first variation of the action functional and a zero-section.

Derived manifolds were introduced by \cite{Spivak:2010} to solve the problem of arbitrary intersection of smooth manifolds. Therefore, it is reasonable to expect that, by replacing smooth manifolds with derived manifolds, we can construct a notion of derived stacks which behave nicely from an intersection theory standpoint.

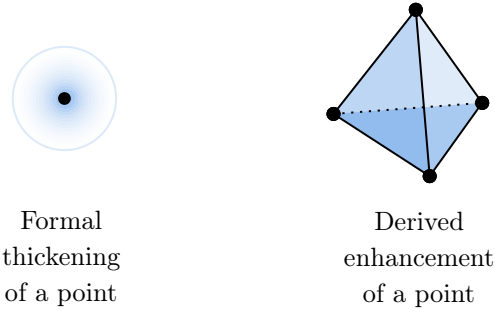
\begin{figure}[h]
    \centering
    % Gradient Info
    \tikzset {_t7y6k0nca/.code = {\pgfsetadditionalshadetransform{ \pgftransformshift{\pgfpoint{0 bp } { 0 bp }  }  \pgftransformscale{1 }  }}}
    \pgfdeclareradialshading{_h9j4xru2l}{\pgfpoint{0bp}{0bp}}{rgb(0bp)=(0.29,0.56,0.89);
    rgb(0bp)=(0.29,0.56,0.89);
    rgb(6.903833661760602bp)=(0.29,0.56,0.89);
    rgb(25bp)=(1,1,1);
    rgb(400bp)=(1,1,1)}
    \tikzset{_b9sf9lngv/.code = {\pgfsetadditionalshadetransform{\pgftransformshift{\pgfpoint{0 bp } { 0 bp }  }  \pgftransformscale{1 } }}}
    \pgfdeclareradialshading{_u6ru8twep} { \pgfpoint{0bp} {0bp}} {color(0bp)=(transparent!38);
    color(0bp)=(transparent!38);
    color(6.903833661760602bp)=(transparent!72);
    color(25bp)=(transparent!100);
    color(400bp)=(transparent!100)} 
    \pgfdeclarefading{_368fxbzlx}{\tikz \fill[shading=_u6ru8twep,_b9sf9lngv] (0,0) rectangle (50bp,50bp); } 
    \tikzset{every picture/.style={line width=0.75pt}} %set default line width to 0.75pt        
    \begin{tikzpicture}[x=0.75pt,y=0.75pt,yscale=-1,xscale=1]
    %Shape: Polygon [id:ds5056291234634489] 
    \draw  [draw opacity=0][fill={rgb, 255:red, 74; green, 144; blue, 226 }  ,fill opacity=0.18 ] (281,82.92) -- (254.5,119.92) -- (247.5,35.92) -- cycle ;
    %Shape: Polygon [id:ds07541505540294158] 
    \draw  [draw opacity=0][fill={rgb, 255:red, 74; green, 144; blue, 226 }  ,fill opacity=0.37 ] (254.5,119.92) -- (206,88.5) -- (247.5,35.92) -- cycle ;
    %Shape: Polygon [id:ds1799946619149786] 
    \draw  [draw opacity=0][fill={rgb, 255:red, 74; green, 144; blue, 226 }  ,fill opacity=0.37 ] (281,82.92) -- (254.5,119.92) -- (206,88.5) -- cycle ;
    %Shape: Circle [id:dp2467268351060885] 
    \path  [shading=_h9j4xru2l,_t7y6k0nca,path fading= _368fxbzlx ,fading transform={xshift=2}] (44.22,80.71) .. controls (44.22,66.31) and (55.89,54.64) .. (70.29,54.64) .. controls (84.69,54.64) and (96.36,66.31) .. (96.36,80.71) .. controls (96.36,95.11) and (84.69,106.78) .. (70.29,106.78) .. controls (55.89,106.78) and (44.22,95.11) .. (44.22,80.71) -- cycle ; % for fading 
    \draw  [color={rgb, 255:red, 74; green, 144; blue, 226 }  ,draw opacity=0.19 ] (44.22,80.71) .. controls (44.22,66.31) and (55.89,54.64) .. (70.29,54.64) .. controls (84.69,54.64) and (96.36,66.31) .. (96.36,80.71) .. controls (96.36,95.11) and (84.69,106.78) .. (70.29,106.78) .. controls (55.89,106.78) and (44.22,95.11) .. (44.22,80.71) -- cycle ; % for border 
    %Shape: Circle [id:dp7784113628901512] 
    \draw  [fill={rgb, 255:red, 0; green, 0; blue, 0 }  ,fill opacity=1 ] (67.58,80.71) .. controls (67.58,79.21) and (68.8,78) .. (70.29,78) .. controls (71.79,78) and (73,79.21) .. (73,80.71) .. controls (73,82.2) and (71.79,83.42) .. (70.29,83.42) .. controls (68.8,83.42) and (67.58,82.2) .. (67.58,80.71) -- cycle ;
    %Straight Lines [id:da5545231165777968] 
    \draw    (206,88.5) -- (254.5,119.92) ;
    \draw [shift={(254.5,119.92)}, rotate = 32.93] [color={rgb, 255:red, 0; green, 0; blue, 0 }  ][fill={rgb, 255:red, 0; green, 0; blue, 0 }  ][line width=0.75]      (0, 0) circle [x radius= 3.02, y radius= 3.02]   ;
    \draw [shift={(206,88.5)}, rotate = 32.93] [color={rgb, 255:red, 0; green, 0; blue, 0 }  ][fill={rgb, 255:red, 0; green, 0; blue, 0 }  ][line width=0.75]      (0, 0) circle [x radius= 3.02, y radius= 3.02]   ;
    %Straight Lines [id:da8824119131329295] 
    \draw    (254.5,119.92) -- (247.5,35.92) ;
    \draw [shift={(247.5,35.92)}, rotate = 265.24] [color={rgb, 255:red, 0; green, 0; blue, 0 }  ][fill={rgb, 255:red, 0; green, 0; blue, 0 }  ][line width=0.75]      (0, 0) circle [x radius= 3.02, y radius= 3.02]   ;
    \draw [shift={(254.5,119.92)}, rotate = 265.24] [color={rgb, 255:red, 0; green, 0; blue, 0 }  ][fill={rgb, 255:red, 0; green, 0; blue, 0 }  ][line width=0.75]      (0, 0) circle [x radius= 3.02, y radius= 3.02]   ;
    %Straight Lines [id:da6173275919470427] 
    \draw    (206,88.5) -- (247.5,35.92) ;
    \draw [shift={(247.5,35.92)}, rotate = 308.28] [color={rgb, 255:red, 0; green, 0; blue, 0 }  ][fill={rgb, 255:red, 0; green, 0; blue, 0 }  ][line width=0.75]      (0, 0) circle [x radius= 3.02, y radius= 3.02]   ;
    \draw [shift={(206,88.5)}, rotate = 308.28] [color={rgb, 255:red, 0; green, 0; blue, 0 }  ][fill={rgb, 255:red, 0; green, 0; blue, 0 }  ][line width=0.75]      (0, 0) circle [x radius= 3.02, y radius= 3.02]   ;
    %Straight Lines [id:da37433934266699453] 
    \draw  [dash pattern={on 0.84pt off 2.51pt}]  (206,88.5) -- (281,82.92) ;
    \draw [shift={(281,82.92)}, rotate = 355.74] [color={rgb, 255:red, 0; green, 0; blue, 0 }  ][fill={rgb, 255:red, 0; green, 0; blue, 0 }  ][line width=0.75]      (0, 0) circle [x radius= 3.02, y radius= 3.02]   ;
    \draw [shift={(206,88.5)}, rotate = 355.74] [color={rgb, 255:red, 0; green, 0; blue, 0 }  ][fill={rgb, 255:red, 0; green, 0; blue, 0 }  ][line width=0.75]      (0, 0) circle [x radius= 3.02, y radius= 3.02]   ;
    %Straight Lines [id:da8198953378409743] 
    \draw    (254.5,119.92) -- (281,82.92) ;
    \draw [shift={(281,82.92)}, rotate = 305.61] [color={rgb, 255:red, 0; green, 0; blue, 0 }  ][fill={rgb, 255:red, 0; green, 0; blue, 0 }  ][line width=0.75]      (0, 0) circle [x radius= 3.02, y radius= 3.02]   ;
    \draw [shift={(254.5,119.92)}, rotate = 305.61] [color={rgb, 255:red, 0; green, 0; blue, 0 }  ][fill={rgb, 255:red, 0; green, 0; blue, 0 }  ][line width=0.75]      (0, 0) circle [x radius= 3.02, y radius= 3.02]   ;
    %Straight Lines [id:da3022695719075825] 
    \draw    (247.5,35.92) -- (281,82.92) ;
    \draw [shift={(281,82.92)}, rotate = 54.52] [color={rgb, 255:red, 0; green, 0; blue, 0 }  ][fill={rgb, 255:red, 0; green, 0; blue, 0 }  ][line width=0.75]      (0, 0) circle [x radius= 3.02, y radius= 3.02]   ;
    \draw [shift={(247.5,35.92)}, rotate = 54.52] [color={rgb, 255:red, 0; green, 0; blue, 0 }  ][fill={rgb, 255:red, 0; green, 0; blue, 0 }  ][line width=0.75]      (0, 0) circle [x radius= 3.02, y radius= 3.02]   ;
    % Text Node
    \draw (20.5,136) node [anchor=north west][inner sep=0.75pt]  [font=\normalsize] [align=left] {\begin{minipage}[lt]{70.07pt}\setlength\topsep{0pt}
    \begin{center}
    {\small Formal thickening}\\{\small of a point}
    \end{center}
    \end{minipage}};
    % Text Node
    \draw (186,136.5) node [anchor=north west][inner sep=0.75pt]  [font=\normalsize] [align=left] {\begin{minipage}[lt]{92.54pt}\setlength\topsep{0pt}
    \begin{center}
    {\small Derived enhancement}\\{\small of a point}
    \end{center}
    \end{minipage}};
    \end{tikzpicture}
    \caption{Intuitive picture of the two main generalisations of smooth geometry: formal smooth geometry and derived smooth geometry. In the former, we allow points to be infinitesimally extended, i.e. formally thickened. In the latter, points can be enhanced to a geometric object whose algebra of functions is simplicial.}
    \label{fig:points}
\end{figure}

Usual BV-theory is perturbatively quantised by a certain deformation of the complex of functions on the formal moduli problem (see \cite{FactI,FactII}).
However, in the context of stacks, there exists a proposed quantisation procedure which is completely distinct from BV-theory: higher geometric quantisation.

\paragraph{Higher geometric quantisation.}
\textit{Higher geometric quantisation} \cite{Rog11,SaSza11x,Rog13,SaSza13, FRS18, Fiorenza:2013jz, FSS16, BSS16,BS16,Sza19} is a mathematical framework for constructing a quantum theory from a classical one which generalises ordinary geometric quantisation. See \cite{Bunk:2021quu} for an introduction to the field.
Recall that ordinary geometric quantization is a well-established method for constructing a global-geometric quantisation of the phase space of classical mechanical system, seen as a symplectic manifold $(M,\omega)$. 
This is achieved by the construction of the prequantum $U(1)$-bundle $P\twoheadrightarrow M$ on the symplectic manifold $(M,\omega)$, which is just a principal $U(1)$-bundle $P\twoheadrightarrow M$ whose curvature is $\mathrm{curv}(P)=\omega\in\Omega^2_{\mathrm{cl}}(M)$. The Hilbert space of the system is then constructed as the space of polarised sections of the associated bundle $P\times_{U(1)}\mathbb{C}$.
That being the case, higher geometric quantisation generalises ordinary geometric quantisation in two directions:
\begin{itemize}
    \item the ordinary prequantum $U(1)$-bundle can be generalised to a bundle $n$-gerbe;
    \item the ordinary phase space can be generalised to a symplectic higher stack, as firstly introduced by \cite{Sev01} and further developed by \cite{FRS18, Fiorenza:2013jz, FSS16}.
\end{itemize}
Higher geometric quantisation does, however, suffer from the difficulty that it is not clear, in general, how to polarise sections of the prequantum bundle and consequently how to obtain a fully fledged Hilbert space. In this sense, higher geometric \textit{pre}quantisation is quite successful, but the quantisation step itself is less understood. 

Nonetheless, higher geometric quantisation reminds us of the crucial lesson that quantisation is ultimately a global-geometric process. In contrast, BV-theory is perturbative, since the classical phase space is quantised in a series expansion around a fixed solution, but it has a good understanding of what the quantisation step should look like, at least locally. In this sense, one could argue that strengths and limitations of the two formalisms are complementary.

%%%%%%%%%%%%%%%%%%%%%%%%%%%%%%%%%%%%%%%%%
\subsection{Goals of this paper}

This paper is intended as a first step towards the following main two objectives.
The first one concerns the development of a global-geometric framework for BV-theory and the second concerns its non-perturbative quantisation.
This is closely related to the intriguing work by \cite{Benini:2021tyt} in the context of derived algebraic geometry. 

\paragraph{Goal I: global classical BV-theory.}
The usual approaches to BV-theory are intrinsically perturbative, even just at the classical level. 
As we argued, the reason is that the formalism of usual BV-theory studies a classical field theory in terms of its infinitesimal deformations around a fixed solution of its equations of motion. 
In other words, the formalism of usual BV-theory does not know anything about the global geometry of the configuration space of the field away from the fixed solution. 
However, quantisation is known to be a global process, which depends on the global geometry of the phase space of a field theory.

This fundamental issue is reified, in Yang-Mills field theory, as follows.
A Yang-Mills field configuration is the datum $(P,\nabla_{\!A})$ of a principal $G$-bundle $P\twoheadrightarrow M$ on the spacetime manifold with a connection $\nabla_{\!A}$.
However, pointed formal moduli problems can only encode infinitesimal deformations of some fixed $(P,\nabla_{\!A})$ and the Lie algebra of their infinitesimal gauge transformations. This makes usual BV-theory structurally blind to the global-geometric properties of gauge fields, as already observed by \cite{Benini:2021tyt}.
As an archetypal example, recall that the electromagnetic field has gauge group $U(1)$, so that its infinitesimal gauge transformations are indistinguishable from the ones of a theory with gauge group $\bbR$. However, the global geometry of the electromagnetic field is described by principal $U(1)$-bundles with connection, which come with fundamental global-geometric features -- such as magnetic charges, encoded by the Chern classes of the bundles, and Aharonov-Bohm effects -- that a gauge theory on $\bbR$ would not show.

The first goal is, then, to develop a framework which generalises the formal moduli problems of BV-theory beyond infinitesimal deformation theory. 
To do that, we want to apply To\"en-Vezzosi's derived geometry \cite{ToenVezzo05, ToenVezzo08} to Carchedi-Steffens' derived manifolds \cite{Carchedi2019OnTU} to construct \textit{formal derived smooth stacks}.
These geometric objects must generalise the traditional notion of manifold in the following ways:
\begin{itemize}
    \item[\textbf{formal}:] allows infinitesimally thickened geometric objects, e.g. formal disks;
    \item[\textbf{derived}:] allows a (categorified) generalisation of intersections, e.g. non-transversal intersections;
    \item[\textbf{smooth}:]  allows smooth geometric objects, e.g. smooth manifolds and diffeological spaces;
    \item[\textbf{stack}:] allows a (categorified) generalisation of gluing, e.g. gauge transformations.
\end{itemize}
Our proposed framework of formal derived smooth stacks will be rooted in the formalism of Schreiber's differential cohesion \cite{DCCTv2}, which has been applied to formalise many higher geometric structures underlying theoretical physics \cite{FSS12, FSS15x, BSS18, FSS19x, FSS19xxx, BSS19, HSS19, FSS19xx, FSS19coho, SS19, Fiorenza:2020iax, Sati:2020nob, Sati:2021uhj, myers2022orbifolds}.

\paragraph{Goal II: non-perturbative BV-quantisation.}
Once clarified the global geometry for classical BV-theory -- as in the previous point -- the next objective is to define a notion of \textit{non-perturbative BV-quantisation}.
Such a quantisation procedure is meant to turn a non-perturbative classical BV-theory, as constructed in the first goal, into a non-perturbative quantum BV-theory.
We will suggest that non-perturbative BV-quantisation should generalise at once usual perturbative BV-quantisation and higher geometric quantisation.
This may not be surprising since, as we argued, the limitations of the two quantisation procedures appear to be complementary.
In the outlook of this paper we will show how the next step towards non-perturbative BV-quantisation should look like.

\vspace{0.4cm}

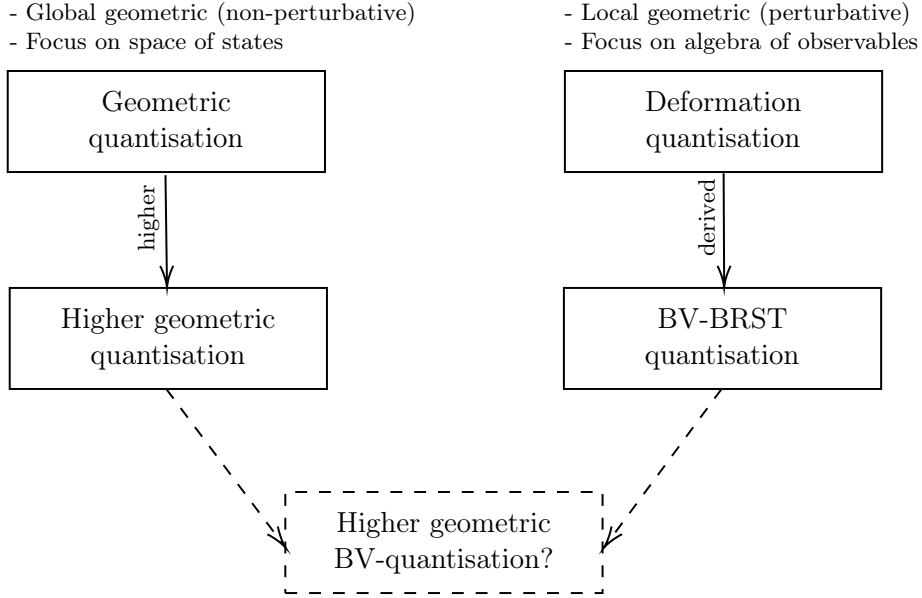
\begin{figure}[h!]
\centering
\tikzset{every picture/.style={line width=0.75pt}} %set default line width to 0.75pt        
\begin{tikzpicture}[x=0.75pt,y=0.75pt,yscale=-1,xscale=1]
%uncomment if require: \path (0,385); %set diagram left start at 0, and has height of 385
%Straight Lines [id:da19279664012478603] 
\draw    (90.31,146.75) -- (89.56,92.06) ;
\draw [shift={(90.33,148.75)}, rotate = 269.21] [color={rgb, 255:red, 0; green, 0; blue, 0 }  ][line width=0.75]    (10.93,-3.29) .. controls (6.95,-1.4) and (3.31,-0.3) .. (0,0) .. controls (3.31,0.3) and (6.95,1.4) .. (10.93,3.29)   ;
%Straight Lines [id:da8364177108090836] 
\draw    (371.32,146.75) -- (371.06,91.06) ;
\draw [shift={(371.33,148.75)}, rotate = 269.72] [color={rgb, 255:red, 0; green, 0; blue, 0 }  ][line width=0.75]    (10.93,-3.29) .. controls (6.95,-1.4) and (3.31,-0.3) .. (0,0) .. controls (3.31,0.3) and (6.95,1.4) .. (10.93,3.29)   ;
%Straight Lines [id:da17574257770428448] 
\draw  [dash pattern={on 4.5pt off 4.5pt}]  (90,200) -- (148.97,278.48) ;
\draw [shift={(150.17,280.08)}, rotate = 233.08] [color={rgb, 255:red, 0; green, 0; blue, 0 }  ][line width=0.75]    (10.93,-3.29) .. controls (6.95,-1.4) and (3.31,-0.3) .. (0,0) .. controls (3.31,0.3) and (6.95,1.4) .. (10.93,3.29)   ;
%Straight Lines [id:da6165500347938269] 
\draw  [dash pattern={on 4.5pt off 4.5pt}]  (370,200) -- (311.86,278.48) ;
\draw [shift={(310.67,280.08)}, rotate = 306.53] [color={rgb, 255:red, 0; green, 0; blue, 0 }  ][line width=0.75]    (10.93,-3.29) .. controls (6.95,-1.4) and (3.31,-0.3) .. (0,0) .. controls (3.31,0.3) and (6.95,1.4) .. (10.93,3.29)   ;
% Text Node
\draw    (10,39.46) -- (170,39.46) -- (170,90.46) -- (10,90.46) -- cycle  ;
\draw (90,64.96) node   [align=left] {\begin{minipage}[lt]{106.76pt}\setlength\topsep{0pt}
\begin{center}
Geometric quantisation
\end{center}
\end{minipage}};
% Text Node
\draw    (11,148.96) -- (171,148.96) -- (171,199.96) -- (11,199.96) -- cycle  ;
\draw (91,174.46) node   [align=left] {\begin{minipage}[lt]{106.76pt}\setlength\topsep{0pt}
\begin{center}
Higher geometric quantisation
\end{center}
\end{minipage}};
% Text Node
\draw    (290.5,148.96) -- (450.5,148.96) -- (450.5,199.96) -- (290.5,199.96) -- cycle  ;
\draw (370.5,174.46) node   [align=left] {\begin{minipage}[lt]{106.76pt}\setlength\topsep{0pt}
\begin{center}
BV-BRST\\quantisation
\end{center}
\end{minipage}};
% Text Node
\draw (74.33,132) node [anchor=north west][inner sep=0.75pt]  [font=\footnotesize,rotate=-270] [align=left] {higher};
% Text Node
\draw (356.83,134.5) node [anchor=north west][inner sep=0.75pt]  [font=\footnotesize,rotate=-270] [align=left] {derived};
% Text Node
\draw    (291,39.46) -- (451,39.46) -- (451,90.46) -- (291,90.46) -- cycle  ;
\draw (371,64.96) node   [align=left] {\begin{minipage}[lt]{106.76pt}\setlength\topsep{0pt}
\begin{center}
Deformation quantisation
\end{center}
\end{minipage}};
% Text Node
\draw  [dash pattern={on 4.5pt off 4.5pt}]  (150,251.96) -- (310,251.96) -- (310,302.96) -- (150,302.96) -- cycle  ;
\draw (230,277.46) node   [align=left] {\begin{minipage}[lt]{106.76pt}\setlength\topsep{0pt}
\begin{center}
Higher geometric \\BV-quantisation?
\end{center}
\end{minipage}};
% Text Node
\draw (9.5,3) node [anchor=north west][inner sep=0.75pt]  [font=\footnotesize] [align=left] {\mbox{-} Global geometric (non-perturbative)\\\mbox{-} Focus on space of states};
% Text Node
\draw (289.5,3) node [anchor=north west][inner sep=0.75pt]  [font=\footnotesize] [align=left] {\mbox{-} Local geometric (perturbative)\\\mbox{-} Focus on algebra of observables};
\end{tikzpicture}
\caption{The two main quantisation procedures and their potential relation.}
\label{fig:scheme}
\end{figure}

%%%%%%%%%%%%%%%%%%%%%%%%%%%%%%%%%%%%%%%%%%%%%%%%%%%%%
\subsection{Overview of main results}

Here, we will provide a brief overview of all the main results of this paper section by section.

\paragraph{Model of formal derived smooth stacks.} 
In section \ref{sec:fdsmoothstack}, we introduce the fundamental geometric object which we are going to consider in this paper: the formal derived smooth stack. 
To define formal derived smooth stacks, first we must introduce formal derived smooth manifolds, which will be our probing spaces.
In this respect, \cite{Carchedi2019OnTU} tells us that there is a canonical equivalence of $(\infty,1)$-categories $\mathbf{dMfd} \simeq \mathbf{sC^\infty Alg}^\op_{\mathrm{fp}}$ between the $(\infty,1)$-category $\mathbf{dMfd}$ of derived manifolds and the opposite $(\infty,1)$-category $\mathbf{sC^\infty Alg}_{\mathrm{fp}}$ of homotopically finitely presented $\Coo$-algebras.
However, to achieve our goals, we will need to slightly generalise the notion of derived manifold.
In analogy with the discussion of \cite{Calaque:2017} in the context of algebraic geometry, we define the $(\infty,1)$-category of formal derived smooth manifolds by
\begin{equation}
    \mathbf{dFMfd} \; \coloneqq\; \mathbf{sC^\infty Alg}^\op_{\mathrm{fg}},
\end{equation}
where $\mathbf{sC^\infty Alg}_{\mathrm{fg}}$ is the $(\infty,1)$-category of finitely generated $\Coo$-algebras.
We then define the notion of formally \'etale morphisms of formal derived smooth manifolds and, thus, we equip the $(\infty,1)$-category $\mathbf{dFMfd}$ with the structure of an \'etale $(\infty,1)$-site.

Finally, we define the $(\infty,1)$-category $\mathbf{dFSmoothStack}$ of formal derived smooth stacks as the $(\infty,1)$-category of stacks on the site $\mathbf{dFMfd}$. More technically, we will see that there is a certain simplicial model category $[\mathsf{dFMfd}^\op,\sSet]^\circ_{\mathrm{proj,loc}}$ whose homotopy coherent nerve presents the $(\infty,1)$-category of formal derived smooth stacks, i.e.
\begin{equation}
    \mathbf{dFSmoothStack} \; \coloneqq\; \mathbf{N}_{hc}([\mathsf{dFMfd}^\op,\sSet]^\circ_{\mathrm{proj,loc}}).
\end{equation}
The relation between formal derived smooth stacks and usual smooth stacks will be clarified by the following proposition.

\textbf{Proposition 3.21} (Relation with usual smooth stacks)\textbf{.}
There exists an adjunction $(i\dashv t_0)$ of $(\infty,1)$-functors between the $(\infty,1)$-category of smooth stacks into the $(\infty,1)$-category of formal derived smooth stacks
\begin{equation}
    \begin{tikzcd}[row sep=scriptsize, column sep=15.5ex, row sep=13.0ex]
     \mathbf{dFSmoothStack} \arrow[r, "t_0", shift left=-1.5ex] &  \mathbf{SmoothStack}, \arrow[l, "i "', shift left=-1.5ex, hook'] 
    \end{tikzcd}
\end{equation}
where $i$ is fully faithful and $t_0$ preserves finite products.

The relation of formal derived smooth stacks with smooth stacks and other relevant classes of smooth spaces is summed up in figure \ref{fig:my_label2}.
\begin{figure}[h]
    \centering
    \begin{equation*}
    \begin{tikzcd}[remember picture, row sep={23ex,between origins}, column sep={16ex,between origins}]
    \text{smooth manifolds}^{\,\op} \arrow[d, hook] \arrow[rrr, "\text{smooth sets}", "\!\!\!\!\!\!\!\!\!\!\!\!\text{\tiny(e.g.$\,$diffeological$\,$spaces)}"']\arrow[ddrrr, "\;\;\;\qquad\text{smooth stacks}",sloped, start anchor={[xshift=+2ex]}, start anchor={[yshift=+0.8ex]}] &&& \text{sets} \arrow[dd, hook] \\
    \shortstack{\text{formal}\\\text{smooth manifolds}}^{\!\!\begin{matrix}{\scriptsize}\\[-0.4cm]{\scriptsize\text{op}}\end{matrix}} \arrow[d, hook, dashed] \arrow[rrru, "\;\;\;\;\;\;\;\substack{\text{formal}\\\text{smooth sets}}", sloped, crossing over, start anchor={[xshift=-4.0ex]}] \arrow[rrrd, "\substack{\text{formal}\\\text{smooth stacks}}", sloped, crossing over] \\
    \shortstack{\textbf{formal derived}\\\textbf{smooth manifolds}}^{\!\!\op} \arrow[rrr, "\substack{\textbf{formal derived}\\\textbf{smooth stacks}}", start anchor={[xshift=+0.4ex]}, end anchor={[xshift=-0.4ex]}, " "', sloped, dashed]& && \infty\text{-groupoids}
    \end{tikzcd}
    \end{equation*}
    \caption{A summary family tree of stacks in formal derived smooth geometry. }
    \label{fig:my_label2}
\end{figure}
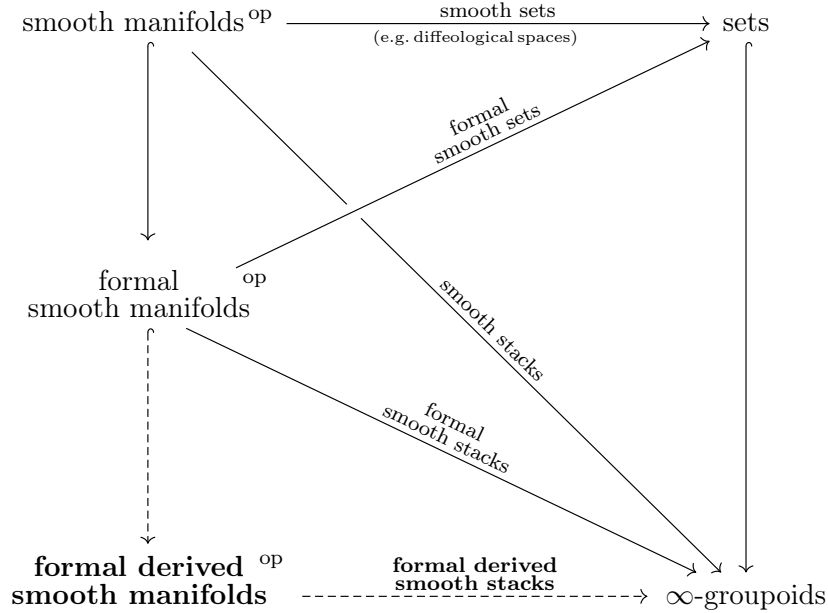

Since the functor $t_0$ preserves finite products, we have the following equivalence of smooth stacks:
\begin{equation}
    {t}_0\big(i(X)\times_{i(Z)}^h i(Y)\big)\,\xrightarrow{\;\;\,\simeq\,\;\;}\, X\times_Z Y,
\end{equation}
for any formal derived smooth stacks $X$ and $Y$.

\paragraph{Differential forms on formal derived smooth stacks.} 
In the last part of section \ref{sec:fdsmoothstack}, we define the $(\infty,1)$-category $\QCoh{X}$ of quasi-coherent sheaves of modules on a formal derived smooth stack $X\in\mathbf{dFSmoothStack}$. 
In particular, we provide the definition of cotangent complex $\bbL_X\in\QCoh{X}$ of a formal derived smooth stack $X$ in a sense which is compatible with its formal derived smooth structure.
Then, we construct the complex of $p$-forms on a formal derived smooth stack $X$ by 
\begin{equation}
    \mathrm{A}^p(X) \,\coloneqq\, \bfR\Gamma(X,\wedge^p_{\mathbb{O}_X}\bbL_X).
\end{equation}
Complex of closed $p$-forms on a formal derived smooth stack $X$ by 
\begin{equation}
    \mathrm{A}^p_{\mathrm{cl}}(X) \;\coloneqq\; \bigg(\prod_{n\geq p}\mathrm{A}^n(X)[-n]\bigg)[p].
\end{equation}
This implies that an $n$-cocycle in $\mathrm{A}_\mathrm{cl}^p(X)$ is given by a formal sum $(\omega_i)=(\omega_p + \omega_{p+1} + \dots)$, where each form $\omega_{i}\in \mathrm{A}^{i}(X)$ is an element of degree $n+p-i$, satisfying the equations
\begin{equation}
\begin{aligned}
    Q\omega_{p} \,&=\, 0,\\
    \di_\dR\omega_{i} + Q\omega_{i+1} \,&=\, 0,
\end{aligned}
\end{equation}
for every $i\geq p$.
Finally, we construct the formal derived smooth stack $\pmb{\mathcal{A}}^p(n)$ as moduli stack of $n$-shifted differential $p$-forms and $\pmb{\mathcal{A}}^p_{\mathrm{cl}}(n)$ as moduli stack of closed $n$-shifted differential $p$-forms.

\paragraph{Derived differential structure.} 
in sections \ref{sec:derdiffcohe} we show that the formalism of differential structures introduced by Schreiber in \cite{DCCTv2} extends very naturally to the derived smooth setting.

\textbf{Theorem 4.7} (Differential $(\infty,1)$-topos of formal derived smooth stacks)\textbf{.}
The $(\infty,1)$-topos $\mathbf{dFSmoothStack}$ of formal derived smooth stacks is naturally equipped with a differential structure.

Such a structure, which we will call derived differential structure, induces the following triplet of adjoint endofunctors:
\begin{equation}
   (\Re \;\dashv\;  \Im \;\dashv\; \&):\,\mathbf{dFSmoothStack}\,\longrightarrow\,\mathbf{dFSmoothStack},
\end{equation}
where we respectively have:
\begin{enumerate}[label=(\textit{\roman*})]
    \item \textit{infinitesimal reduction modality} $\Re$,
    \item \textit{infinitesimal shape modality} $\Im$,
    \item \textit{infinitesimal flat modality} $\&$.
\end{enumerate}
Differential topos geometry underpins the definition of the de Rham space $\Im(X)$ of any formal derived smooth stack $X$ by the infinitesimal shape modality. This could be interpreted as an infinitesimal version of the path $\infty$-groupoid of $X$ and its role will be pivotal.
In fact, we can define the formal disk $\bbD_{X,x}$ at the point $x:\ast\rightarrow X$ of a formal derived smooth stack $X\in\mathbf{dFSmoothStack}$ by the homotopy pullback of formal derived smooth stacks
\begin{equation}
    \begin{tikzcd}[row sep={12ex,between origins}, column sep={12.5ex,between origins}]
    {\bbD_{X,x}} \arrow[d] \arrow[r, hook] & X \arrow[d, "{\mathfrak{i}}_X"] \\
    \ast \arrow[r, hook, "{x}"] & \Im(X),
    \end{tikzcd}
\end{equation}
where $\mathfrak{i}_X:X\longrightarrow \Im(X)$ is a natural map.
The definition of formal disk entails the geometry of jets of formal derived smooth stacks.

\paragraph{Relation with formal moduli problems.} 
In the second half of section \ref{sec:derdiffcohe} we study the relation of formal derived smooth stacks with formal moduli problems.
We introduce the simplicial category $\mathsf{dgArt}^{\leq 0}_\bbR$ of dg-Artinian algebras, then
we construct the $(\infty,1)$-category of formal moduli problems by the $(\infty,1)$-category of pre-stacks
\begin{equation}
    \mathbf{FMP} \;\coloneqq\; \mathbf{N}_{hc}([\mathsf{dgArt}^{\leq 0}_\bbR, \sSet]_\mathrm{proj}^\circ),
\end{equation}
with its natural structure of $(\infty,1)$-topos of pre-stacks.

\vspace{1cm}

\begin{figure}[h]
    \centering
    % Gradient Info
    \tikzset {_7k28g3xcg/.code = {\pgfsetadditionalshadetransform{ \pgftransformshift{\pgfpoint{83.16 bp } { -103.62 bp }  }  \pgftransformscale{1.32 }  }}}
\pgfdeclareradialshading{_4mt4nvhz9}{\pgfpoint{-72bp}{88bp}}{rgb(0bp)=(1,1,1);
rgb(0bp)=(1,1,1);
rgb(25bp)=(0.55,0.55,0.55);
rgb(400bp)=(0.55,0.55,0.55)}
% Gradient Info
\tikzset {_emqq0w7mx/.code = {\pgfsetadditionalshadetransform{ \pgftransformshift{\pgfpoint{0 bp } { 0 bp }  }  \pgftransformscale{1.08 }  }}}
\pgfdeclareradialshading{_arphx4gqd}{\pgfpoint{0bp}{0bp}}{rgb(0bp)=(0.29,0.56,0.89);
rgb(0bp)=(0.29,0.56,0.89);
rgb(25bp)=(0.29,0.56,0.89);
rgb(400bp)=(0.29,0.56,0.89)}
\tikzset{_1h9pey1mj/.code = {\pgfsetadditionalshadetransform{\pgftransformshift{\pgfpoint{0 bp } { 0 bp }  }  \pgftransformscale{1.08 } }}}
\pgfdeclareradialshading{_dca4qijva} { \pgfpoint{0bp} {0bp}} {color(0bp)=(transparent!0);
color(0bp)=(transparent!0);
color(25bp)=(transparent!97);
color(400bp)=(transparent!97)} 
\pgfdeclarefading{_x253gi8zn}{\tikz \fill[shading=_dca4qijva,_1h9pey1mj] (0,0) rectangle (50bp,50bp); } 
\tikzset{every picture/.style={line width=0.75pt}} %set default line width to 0.75pt        
\begin{tikzpicture}[x=0.75pt,y=0.75pt,yscale=-1,xscale=1]
%Shape: Circle [id:dp23027282441372354] 
\path  [shading=_4mt4nvhz9,_7k28g3xcg] (161.08,159.71) .. controls (161.08,96.91) and (211.99,46) .. (274.79,46) .. controls (337.59,46) and (388.5,96.91) .. (388.5,159.71) .. controls (388.5,222.51) and (337.59,273.42) .. (274.79,273.42) .. controls (211.99,273.42) and (161.08,222.51) .. (161.08,159.71) -- cycle ; % for fading 
 \draw  [color={rgb, 255:red, 0; green, 0; blue, 0 }  ,draw opacity=1 ] (161.08,159.71) .. controls (161.08,96.91) and (211.99,46) .. (274.79,46) .. controls (337.59,46) and (388.5,96.91) .. (388.5,159.71) .. controls (388.5,222.51) and (337.59,273.42) .. (274.79,273.42) .. controls (211.99,273.42) and (161.08,222.51) .. (161.08,159.71) -- cycle ; % for border 
%Shape: Ellipse [id:dp9545472601197613] 
\path  [shading=_arphx4gqd,_emqq0w7mx,path fading= _x253gi8zn ,fading transform={xshift=2}] (239,60) .. controls (239,48.95) and (254.67,40) .. (274,40) .. controls (293.33,40) and (309,48.95) .. (309,60) .. controls (309,71.05) and (293.33,80) .. (274,80) .. controls (254.67,80) and (239,71.05) .. (239,60) -- cycle ; % for fading 
 \draw  [color={rgb, 255:red, 74; green, 144; blue, 226 }  ,draw opacity=0.76 ][dash pattern={on 4.5pt off 4.5pt}][line width=0.75]  (239,60) .. controls (239,48.95) and (254.67,40) .. (274,40) .. controls (293.33,40) and (309,48.95) .. (309,60) .. controls (309,71.05) and (293.33,80) .. (274,80) .. controls (254.67,80) and (239,71.05) .. (239,60) -- cycle ; % for border 
%Straight Lines [id:da4250454484312498] 
\draw    (274.5,23.42) -- (274.77,55.79) ;
\draw [shift={(274.79,57.79)}, rotate = 269.51] [color={rgb, 255:red, 0; green, 0; blue, 0 }  ][line width=0.75]    (6.56,-1.97) .. controls (4.17,-0.84) and (1.99,-0.18) .. (0,0) .. controls (1.99,0.18) and (4.17,0.84) .. (6.56,1.97)   ;
%Straight Lines [id:da8111776877877066] 
\draw    (351,35.92) -- (294.86,57.69) ;
\draw [shift={(293,58.42)}, rotate = 338.8] [color={rgb, 255:red, 0; green, 0; blue, 0 }  ][line width=0.75]    (6.56,-1.97) .. controls (4.17,-0.84) and (1.99,-0.18) .. (0,0) .. controls (1.99,0.18) and (4.17,0.84) .. (6.56,1.97)   ;
%Straight Lines [id:da690732046677645] 
\draw    (416,159.92) -- (358.5,159.43) ;
\draw [shift={(356.5,159.42)}, rotate = 0.48] [color={rgb, 255:red, 0; green, 0; blue, 0 }  ][line width=0.75]    (6.56,-1.97) .. controls (4.17,-0.84) and (1.99,-0.18) .. (0,0) .. controls (1.99,0.18) and (4.17,0.84) .. (6.56,1.97)   ;
%Shape: Circle [id:dp40056921878050433] 
\draw  [color={rgb, 255:red, 0; green, 67; blue, 148 }  ,draw opacity=1 ][fill={rgb, 255:red, 0; green, 80; blue, 173 }  ,fill opacity=1 ] (273.58,60) .. controls (273.58,59.33) and (274.12,58.79) .. (274.79,58.79) .. controls (275.46,58.79) and (276,59.33) .. (276,60) .. controls (276,60.67) and (275.46,61.21) .. (274.79,61.21) .. controls (274.12,61.21) and (273.58,60.67) .. (273.58,60) -- cycle ;
% Text Node
\draw (354.83,28.9) node [anchor=north west][inner sep=0.75pt]  [font=\small]  {$\mathbf{MC}(\mathfrak{L})$};
% Text Node
\draw (418.83,153.9) node [anchor=north west][inner sep=0.75pt]  [font=\small]  {$X$};
% Text Node
\draw (267.83,3.4) node [anchor=north west][inner sep=0.75pt]  [font=\small]  {$\itPhi $};
% Text Node
\draw (450.83,156) node [anchor=north west][inner sep=0.75pt]  [font=\scriptsize] [align=left] {Globally defined\\formal derived smooth stack\\of phases of a field theory};
% Text Node
\draw (414.33,29.5) node [anchor=north west][inner sep=0.75pt]  [font=\scriptsize] [align=left] {\begin{minipage}[lt]{121.38pt}\setlength\topsep{0pt}
\begin{center}
Infinitesimal neighborhood of $\displaystyle \itPhi $\\=
\end{center}
\textit{Pointed formal moduli problem}\\underlying the BV-complex
\end{minipage}};
% Text Node
\draw (124.83,7) node [anchor=north west][inner sep=0.75pt]  [font=\scriptsize] [align=left] {Fixed solution of the e.o.m};
\end{tikzpicture}
    \caption{The pointed formal moduli problem underlying the BV-complex can be seen as the infinitesimal neighborhood of a fixed solution in a formal derived smooth stack corresponding to a given classical field theory.}
    \label{fig:stackpicture}
\end{figure}
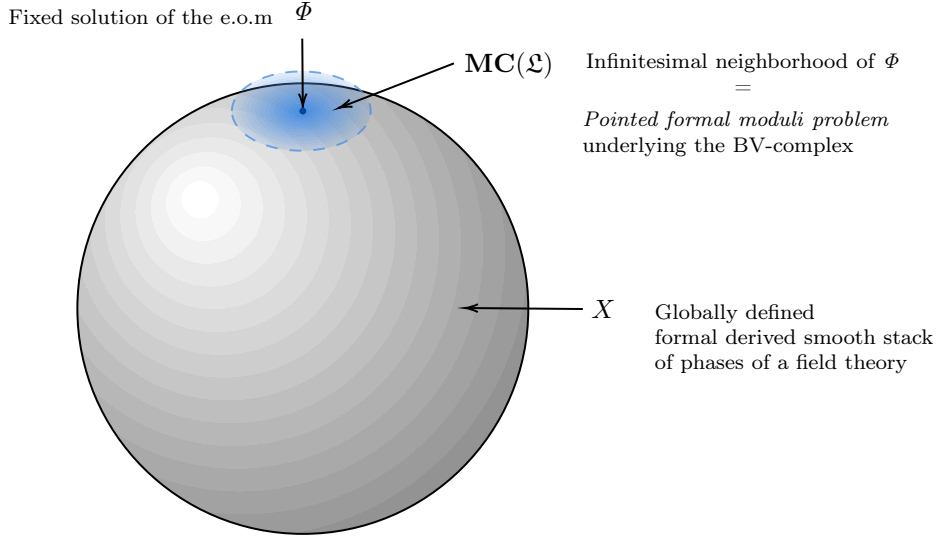

The following proposition characterises the $(\infty,1)$-category of formal moduli problems as a cohesive $(\infty,1)$-topos which is, in particular, infinitesimally cohesive in the sense of \cite[Definition 4.1.21]{DCCTv2}. 
This, roughly, means that the objects of $\mathbf{FMP}$ are infinitesimally thickened simplicial sets of points.

\textbf{Proposition 4.42} (Infinitesimal cohesive $(\infty,1)$-topos of formal moduli problems)\textbf{.} 
The $(\infty,1)$-topos $\mathbf{FMP}$ of formal moduli problems has a natural infinitesimally cohesive structure in the sense of \cite[Definition 4.1.21]{DCCTv2}.

Moreover, we will show that the $(\infty,1)$-topos of formal moduli problems is related to the one of formal derived smooth stacks by morphisms of $(\infty,1)$-topoi of the following form:
\begin{equation*}
    \displaystyle\text{Smooth stacks} \;\;\xhookrightarrow{\qquad\quad}\;\; \substack{ \displaystyle\text{Formal derived} \\ \displaystyle\text{smooth stacks}} \;\;\xtwoheadrightarrow{\qquad\quad}\;\; \substack{ \displaystyle\text{Formal} \\ \displaystyle\text{Moduli Problems,}}
\end{equation*}
presenting formal derived smooth stacks as a refinement of usual smooth stacks.
We will make this relation precise in terms of an adjunction,
inducing an endofunctor
\begin{equation}
   \flat^{\mathrm{rel}} \,:\,\mathbf{dFSmoothStack}\,\longrightarrow\,\mathbf{dFSmoothStack},
\end{equation}
which is strictly related to Lie differentiation in the formal smooth derived context.

\paragraph{Global BV-BRST formalism.} 
In section \ref{sec:global_aspects_of_bv_theory} we study some global aspects of BV-theory in the geometric context of derived differential cohesion.
Let $\mathrm{Bun}_G^\nabla(M)$ be the bare groupoid of principal $G$-bundles on $M$ with connection.
It is well-understood that this can be made into a smooth stack $\Bun_G^\nabla(M)$ such that, for any smooth manifold $U$, a section
\begin{equation}
    (P,\nabla_A)\,:U\,\rightarrow\,\Bun_G^\nabla(M),
\end{equation}
is a $U$-parametrised smooth family of principal $G$-bundles on the base manifold $M$ with connection. 
We can think of this smooth stack as the global configuration space of a gauge field with gauge group $G$ on a spacetime manifold $M$.

We will consider the Yang-Mills action functional $S$ as a smooth map of stacks.
The derived critical locus of the action functional is a derived formal smooth stack $\bfR\crit{S}(M)$ which is given by a homotopy pullback of the form
\begin{equation}
    \begin{tikzcd}[row sep={20ex,between origins}, column sep={22ex,between origins}]
    \bfR\crit{S}(M) \arrow[r]\arrow[d, ""'] & \Bun_G^\nabla(X) \arrow[d, "\delta S"] \\
    \Bun_G^\nabla(M) \arrow[r, "0"] & T^\vee_\mathrm{res}{\Bun_G^\nabla(M)} ,
\end{tikzcd}
\end{equation}
where $T^\vee_\mathrm{res}{\Bun_G^\nabla(M)}$ is the restricted cotangent bundle of the configuration space and $\delta S$ is the variational derivative of the action functional.
The derived critical locus $\bfR\crit{S}(M)$ is such that, for any formal derived smooth manifold $U$, a section
\begin{equation}
    (P,\nabla_A,A^+,c^+)\,:U\,\rightarrow\,\bfR\crit{S}(M),
\end{equation}
is given by a $U$-parametrised family $(P,\nabla_A)$ of principal $G$-bundles on $M$ with connection, together with global \textit{antifields} $A^+\in\Omega^{d-1}(M,\mathfrak{g}_P)$ and global \textit{antighosts} $c^+\in\Omega^{d}(M,\mathfrak{g}_P)$.

It is important to stress that a point $(P,\nabla_A)\in\bfR\crit{S}(M)$ in the derived critical locus is a globally defined principal $G$-bundle with connection which satisfies the Yang-Mills equations of motion, i.e. we have that it satisfies
\begin{equation*}
\begin{aligned}
    \nabla_{\!A}F_A \;&=\; 0 && \text{(Bianchi identity)}\\[0.5ex]
    \nabla_{\!A}\star\! F_A \;&=\; 0 && \text{(Equations of motion)}.
\end{aligned}
\end{equation*}

%%%%%%%%%%%%%%%%%%%%%%%%%%%%%%%%%%%%%%%%%%%%%%%%%%%%%%%%%%%%%%%%%%%%%%%%%%%%%%%%%%%%%%%%%%%%%
\section{Lightning review of smooth stacks}

In this section we will provide a brief review of the theory of smooth stacks -- which are sometimes known as differentiable stacks in the literature.

%%%%%%%%%%%%%%%%%%%%%%%%%%%%%%%%%%%%%%%%%%%%%%%%%%%%%%%%%%%%%%%%%%%%%%%%%%%%%%%%%%%%%%%%%%%%%
\subsection{Smooth sets}\label{sec:smooth_set}

Let $\Mfd$ be the ordinary category whose objects are smooth manifolds and whose morphisms are smooth maps between them. 
We stress that in all the rest of this paper sans serif will be used to denote ordinary categories.
Now, we can provide the category $\Mfd$ with the structure of a site by assigning to each smooth manifold $M\in\Mfd$ a collection of covering families, i.e. a collection of families of morphisms $\{ U_i \rightarrow M \}_{i\in I}$ satisfying some conditions.

\begin{definition}[Covering of a smooth manifold]
We define a \textit{covering family of a smooth manifold $M$} as a set of injective local diffeomorphisms
\begin{equation}
    \{ U_i \,\xhookrightarrow{\;\;\phi_i\;\;} M \}_{i\in I}
\end{equation}
such that they induce a surjective local diffeomorphism
\begin{equation}
    \coprod_{i\in I} U_i \,\xtwoheadrightarrow{\;\,(\phi_i)_{i\in I}\,\;} M.
\end{equation}
\end{definition}

The site structure on $\Mfd$ given by the choice of covering families above is known as \'etale site.

\begin{definition}[Smooth sets]\label{def:smooth_set}
\textit{Smooth sets} are defined  as sheaves on the site of smooth manifolds $\mathsf{Mfd}$. The category of smooth sets is, then, defined by
\begin{equation}
    \mathsf{SmoothSet} \;\coloneqq\; \mathsf{Sh}(\mathsf{Mfd}).
\end{equation}
\end{definition}

The usual gluing axiom of sheaves can be seen in the following light.
Let $\{ U_i \rightarrow M \}_{i\in I}$ be a covering family and notice that $M$ can be rewritten as the colimit of the diagram of manifolds
\begin{equation}
    M \;\simeq\; \mathrm{colim} \bigg( \begin{tikzcd}[row sep=scriptsize, column sep=3.5ex]
     \displaystyle\coprod_{i,j\in I}U_i\times_M U_j  \arrow[r, yshift=0.7ex] \arrow[r, yshift=-0.7ex] & \displaystyle\coprod_{i\in I}U_i 
    \end{tikzcd}    \bigg).
\end{equation}
Then, $X$ to be a sheaf, must have a set of sections on $M$ given by the limit of the diagram
\begin{equation}
    X(M) \;\simeq\; \mathrm{lim} \bigg( \begin{tikzcd}[row sep=scriptsize, column sep=3.5ex]
     \displaystyle\prod_{i,j\in I}X(U_i\times_M U_j)   & \displaystyle\prod_{i\in I}X(U_i) \arrow[l, yshift=0.7ex] \arrow[l, yshift=-0.7ex]
    \end{tikzcd} \bigg).
\end{equation}

\begin{example}[Yoneda embedding of smooth manifolds]
A smooth manifold is the simplest example of smooth set.
Let $M\in\mathsf{Mfd}$ be a smooth manifold, then it naturally Yoneda-embeds into a smooth set of the form
\begin{equation}
\begin{aligned}
    M\,:\,\mathsf{Mfd}^\op \,&\longrightarrow\, \mathsf{Set}  \\
    U \,&\longmapsto\, \Hom_\mathsf{Mfd}(U,M),
\end{aligned}
\end{equation}
where $\Hom_\mathsf{Mfd}(U,M)$. Thus, we have the full and faithful embedding of categories
\begin{equation}
    \mathsf{Mfd}  \longhookrightarrow \mathsf{SmoothSet}.
\end{equation}
\end{example}
(In what follows, we shall sometimes make use of this embedding without comment.) 

The notion of smooth set is a categorically well-behaved generalisation of smooth manifold which, crucially, allow us to encode finite-dimensional smooth spaces. A relevant example is the smooth space $[M,N]$ of functions from a smooth manifold $M$ to $N$. 

\begin{example}[Mapping space]
Let $M,N\in\mathsf{Mfd}$ be a pair of smooth manifolds. We can define the mapping space $[M,N]\in\mathsf{SmoothSet}$ by the smooth set
\begin{equation}
\begin{aligned}
    [M,N]\,:\,\mathsf{Mfd}^\op \,&\longrightarrow\, \mathsf{Set}  \\
    U \,&\longmapsto\, \Hom_\mathsf{FMfd}(U\times M,N),
\end{aligned}
\end{equation}
functorially, on elements $U\in\mathsf{Mfd}$ of the site.
\end{example}

\begin{example}[Moduli space of differential forms]
It is possible to define a smooth set $\Omega^1\in\mathsf{SmoothSet}$, which we can call moduli space of differential forms, by
\begin{equation}
\begin{aligned}
    \bfOmega^1\,:\,\mathsf{Mfd}^\op \,&\longrightarrow\, \mathsf{Set}  \\
    U \,&\longmapsto\, \Omega^1(U),
\end{aligned}
\end{equation}
and by sending morphisms $f:U\rightarrow U'$ to pullbacks $f^\ast:\Omega^1(U')\rightarrow\Omega^1(U)$.
\end{example}

This remarkably abstract moduli space of differential forms is very useful in practice, because it allows us to work with differential forms on general formal smooth sets, including mapping spaces.

\begin{definition}[Differential forms on a smooth set]
We define the \textit{set of differential $1$-forms} on a given smooth set $X\in\mathsf{SmoothSet}$ by the following hom-set of smooth sets:
\begin{equation}
    \Omega^1(X) \;\coloneqq\; \Hom(X,\bfOmega^1),
\end{equation}
where $\bfOmega^1\in\mathsf{SmoothSet}$ is the moduli space of differential forms.
\end{definition}

\begin{remark}[de Rham differential]
There exists a canonical morphism of smooth sets
\begin{equation}
    \di_\dR \,:\, \bbR \,\longrightarrow \, \bfOmega^1,
\end{equation}
which is given by the differential $\di:\Coo(U,\bbR)\rightarrow\Omega^1(U)$ of function on each smooth manifold $U$ in the site. This particularly exotic morphism of smooth sets $\di_\dR\in\Hom(\bbR,\bfOmega^1)$ is known as de Rham differential.
\end{remark}

\begin{remark}[Pullback of differential forms]
Given a morphism $f:X\longrightarrow Y$ of smooth sets $X,Y\in\mathsf{SmoothSet}$, we have a morphism of sets $f^\ast:\Omega^{p}(Y)\longrightarrow\Omega^p(X)$ such that the following square commutes
\begin{equation}
    \begin{tikzcd}[row sep=scriptsize, column sep=8.2ex, row sep=11.0ex]
     \Omega^p(Y) \arrow[r, "\di_\dR"]\arrow[d, "f^\ast"] & \Omega^{p+1}(Y)\arrow[d, "f^\ast"] \\
     \Omega^p(X) \arrow[r, "\di_\dR"] & \Omega^{p+1}(X)
    \end{tikzcd}
\end{equation}
\end{remark}

\begin{remark}[Variational calculus on smooth sets]
The power of smooth sets is their capacity to provide a well-defined formalism for variational calculus.
For example, we can consider the mapping space $[M,\bbR]$ for a given smooth manifold $M$. This can be thought as the infinite-dimensional smooth space of smooth functions on the manifold $M$, and there is no issue in working with differential forms on such a large space: differential $1$-forms are simply given by $\Omega^1([M,\bbR])\coloneqq\Hom([M,\bbR],\bfOmega^1)$, as above.
Similarly, a smooth functional on such a space will be given by a morphism of smooth sets
\begin{equation}
    S\,:\, [M,\bbR] \,\longrightarrow \, \bbR
\end{equation}
to the real line.
The so-called first variation of this functional is immediately given by the following composition:
\begin{equation}
    \di_\dR S\,:\, [M,\bbR] \,\xrightarrow{\;\;S\;\;} \, \bbR \,\xrightarrow{\;\;\di_\dR\;\;} \, \bfOmega^1,
\end{equation}
which means that we have obtained a perfectly legitimate $1$-form $\di_\dR S\in\Omega^1([M,\bbR])$ on the infinite-dimensional mapping space $[M,\bbR]$ of smooth functions on $M$.
\end{remark}

We can now define the functor which forgets the smooth structure of formal smooth sets, i.e. which sends any smooth set to its underlying bare set.

\begin{definition}[Global section functor]
We define the the \textit{global section functor} by 
\begin{equation}
    \mathit{\Gamma}(-)\coloneqq \Hom_\mathsf{SmoothSet}(\,\ast\,,-)\,:\,\mathsf{SmoothSet}\,\longrightarrow\, \mathsf{Set}.
\end{equation}
\end{definition}

The global section functor will allow us to define an important class of smooth sets: diffeological spaces.
Diffeological spaces were firstly introduced by \cite{Soriau79, Soriau84} and then reformulated by \cite{iglesiaszemmour}.
A diffeological space is a powerful generalisation of a smooth manifold which, in particular, provides a natural setting to study infinite-dimensional smooth spaces.
Useful examples of diffeological spaces will be the space of smooth sections of a fibre bundle and the infinite-jet bundle of a fibre bundle.
Diffeological spaces behave well under categorical properties and they embed into a sub-category, which is said concrete, of the topos of smooth sets \cite{khavkine2017synthetic, shulman_2021}.

\begin{definition}[Diffeological space]\label{def:diffeological_space}
A \textit{diffeological space} $X$ is defined as a concrete smooth set, i.e. such that for any smooth manifold $U\in\Mfd$ the natural map
\begin{equation}
    X(U)\,\hookrightarrow\, \Hom_\mathsf{Set}(\mathit{\Gamma}U,\, \mathit{\Gamma}X),
\end{equation}
is a monomorphism of sets.
\end{definition}

\begin{example}[Examples of diffeological spaces]
A smooth manifold $M\in\Mfd\hookrightarrow\mathsf{SmoothSet}$, Yoneda-embedded in smooth sets, is a diffeological space.
If we consider another smooth manifold $N\in\Mfd$, then the mapping space $[M,N]$ is also a diffeological space. This is because, given any section $f\in [M,N](U)\simeq \Hom_\Mfd(M\times U,N)$, we can embed it into a map $\mathit{\Gamma}U\rightarrow\mathit{\Gamma}[M,N] \simeq \Hom_\Mfd(M,N)$ which sends any point $u\in\mathit{\Gamma}U$ to $f(u)\in\Hom_\Mfd(M,N)$.
\end{example}

%%%%%%%%%%%%%%%%%%%%%%%%%%%%%%%%%%%%%%%%%%%%%%%%%%%%%%%%%%%%%%%%%%%%%%%%%%%%%%%%%%%%%%%%%%%%%
\subsection{Smooth stacks}

The category $\sSet$ of \textit{simplicial sets} can be seen as the functor category $[\Delta^\op,\mathsf{Set}]$, where $\Delta$ is the simplex category -- i.e. the category whose objects are non-empty finite ordinals and whose morphisms are order-preserving maps -- and $\mathsf{Set}$ is the category of sets.

The category $\sSet$ of simplicial sets is naturally a simplicial category, i.e. a category enriched over $\sSet$ itself. In the rest of the paper we will keep using sans serif to denote simplicial categories.
Moreover, we will denote by $\sSet_{\mathsf{Quillen}}$ the simplicial category of simplicial sets equipped with Quillen model structure \cite{Quillen:1967ha}, whose weak equivalences are weak homotopy equivalences of simplicial sets and whose fibrations are Kan fibrations. 

Let $\mathsf{W}$ be the set of weak homotopy equivalences of simplicial sets.
Then, by simplicial localisation, one can define the category of Kan complexes
\begin{equation}
    \mathsf{KanCplx} \;\coloneqq\; L_{\mathsf{W}}\sSet_{\mathsf{Quillen}}. 
\end{equation}
It can be shown that the full subcategory $\sSet_{\mathsf{Quillen}}^\circ$ of fibrant-cofibrant objects of $\sSet_{\mathsf{Quillen}}$ is equivalent to the simplicial-category of Kan complexes, i.e.
\begin{equation}
    \mathsf{KanCplx} \;\simeq\; \sSet_{\mathsf{Quillen}}^\circ.
\end{equation}
Moreover, we can make this simplicial category into a fully fledged $(\infty,1)$-category. Essentially, an $(\infty,1)$-category is a simplicial set which satisfies an extra condition, known as weak Kan condition (which requires all the inner horns of the simplicial set to have fillers).
It is a standard technique \cite[Section 1.1.5]{topos} that, by applying the homotopy-coherent nerve functor $\mathbf{N}_{hc}$ to our simplicial category, one obtains the $(\infty,1)$-category of $\infty$-groupoids, i.e.
\begin{equation}
    \mathbf{\infty Grpd} \;\coloneqq\; \mathbf{N}_{hc}(\sSet_{\mathsf{Quillen}}^\circ).
\end{equation}
In the rest of the paper, we will use bold roman font to denote $(\infty,1)$-categories. 
Now, given any category $\mathsf{C}$, consider the simplicial functor category $\mathsf{sPreSh}(\mathsf{C})\coloneqq[\mathsf{C}^\op,\sSet]$, known as the category of simplicial pre-sheaves on $\mathsf{C}$.
If $\mathsf{C}$ has the structure of a \textit{site} with \textit{enough points}, there exists a model structure $\mathsf{sPreSh}(\mathsf{C})_{\mathsf{proj,loc}}$ which is known as the \textit{projective local model structure} \cite{Blander2001LocalPM} and whose set of local weak equivalences $\mathsf{W}$ is the set of natural transformations which are stalk-wise weak homotopy equivalences of simplicial sets.
Then, we can define the simplicial category of \textit{stacks} on $\mathsf{C}$ by simplicial localisation
\begin{equation}
    \Shoo(\mathsf{C}) \,\coloneqq\, L_\mathsf{W}\mathsf{sPreSh}(\mathsf{C}).
\end{equation}
Moreover, the projective local model structure has the property that the full subcategory $\mathsf{sPreSh}(\mathsf{C})_{\mathsf{proj,loc}}^\circ$ of fibrant-cofibrant objects of the simplicial model category $\mathsf{sPreSh}(\mathsf{C})_{\mathsf{proj,loc}}$ is equivalent to the simplicial category of stacks, i.e. we have
\begin{equation}
    \Shoo(\mathsf{C}) \;\simeq\; \mathsf{sPreSh}(\mathsf{C})_{\mathsf{proj,loc}}^\circ.
\end{equation}
Thus, the $(\infty,1)$-category of stacks on the site $\mathsf{C}$ can be defined by the homotopy-coherent nerve of this simplicial category, i.e. by
\begin{equation}
    \St(\mathsf{C}) \,\coloneqq\, \mathbf{N}_{hc}(\mathsf{sPreSh}(\mathsf{C})_{\mathsf{proj,loc}}^\circ).
\end{equation}

Let us now specialize our discussion to smooth geometry.
The category $\Mfd$ of smooth manifolds, whose objects are smooth manifolds and whose morphisms are smooth maps between them, has a natural site structure where covering families $\{U_i\rightarrow M\}_{i\in I}$ are good open covers of smooth manifolds.
Then, \textit{smooth stacks} \cite{DCCTv2} -- also known as \textit{differentiable stacks} -- can be defined as stacks on the site of smooth manifolds $\Mfd$ and thus they live in the simplicial category 
\begin{equation}
    \mathsf{SmoothStack} \;\coloneqq\; \Shoo(\mathsf{Mfd}) \;\simeq\; \mathsf{sPreSh}(\mathsf{Mfd})_{\mathsf{proj,loc}}^\circ.
\end{equation}
Given a covering family $\{U_i\rightarrow U\}_{i\in I}$, it is possible to construct a simplicial object known as \v{C}ech nerve of the smooth manifold $U$ by
\begin{equation}
    \check{C}(U)_\bullet \;=\;  \bigg( \begin{tikzcd}[row sep=scriptsize, column sep=3.5ex]
    \; \cdots\; \arrow[r, yshift=1.8ex]\arrow[r, yshift=0.6ex]\arrow[r, yshift=-1.8ex]\arrow[r, yshift=-0.6ex]& \displaystyle\coprod_{i,j,k\in I}U_i\times_U U_j\times_U U_k 
    \arrow[r, yshift=1.4ex] \arrow[r] \arrow[r, yshift=-1.4ex] & \displaystyle\coprod_{i,j\in I}U_i\times_U U_j  \arrow[r, yshift=0.7ex] \arrow[r, yshift=-0.7ex] & \displaystyle\coprod_{i\in I}U_i 
    \end{tikzcd}    \bigg),
\end{equation}
whose colimit is the original smooth manifold $U \simeq \mathrm{co}\!\lim_{[n]\in\Delta}\check{C}(U)_n$.
By unravelling the definition of a smooth stack, more concretely, one has that a smooth stack is a simplicially enriched functor $X:\Mfd \longrightarrow \sSet$ satisfying the following properties:
\begin{enumerate}[label=(\textit{\roman*})]
    \item \textit{object-wise fibrancy}: for any $U\in\mathsf{Mfd}$, the simplicial set $X(U)$ is Kan-fibrant;
    \item \textit{pre-stack condition}: for any diffeomorphism $U\xrightarrow{\,\simeq\,} U'$ in $\mathsf{Mfd}$, the induced morphism $X(U')\longrightarrow X(U)$ is an equivalence of simplicial sets;
    \item \textit{descent condition}: for any \v{C}ech nerve $\check{C}(U)_\bullet\rightarrow U$, the natural morphism
    \begin{equation}
    X(U) \;\longrightarrow\; \lim_{[n]\in\Delta} \bigg(\prod_{i_1,\dots,i_n\in I} \! X(U_{i_1\!}\times_{U\!}\cdots_{\!}\times_{U\!}U_{i_n})\,\bigg)
    \end{equation}
    is an equivalence of simplicial sets.
\end{enumerate}

\begin{example}[Quotient stack]
Let $M$ be a smooth manifold and $G$ a Lie group.
A typical example of smooth stack is given by the quotient stack $[M/G]\in\mathbf{SmoothStack}$, which is constructed as follows.
The $\infty$-groupoid $[M/G](U)$ of sections on a smooth manifold $U$ is such that $0$-simplices are couples $(p:P\rightarrow U, f:P\rightarrow M)$, where $p$ is a $G$-bundle and $f$ is a $G$-equivariant map, and higher simplices are given by automorphisms and composition of those.
On a Cartesian space $U\simeq \bbR^n$, its simplicial set of sections takes the simpler form
\begin{equation*}
    [M/G](U) \;\simeq\; \mathrm{cosk}_{2\!}\left(\! \begin{tikzcd}[row sep={22.5ex,between origins}, column sep={4.5ex}]
    \Hom(U,G^{\times 2}\!\times\!M) \, \arrow[r, yshift=2.0ex , "{ }"] \arrow[r, description] \arrow[r, yshift=-2.0ex, "{ }"'] & \Hom(U,G\!\times\!M) \arrow[r, yshift=1.0ex , "{\partial_0}"] \arrow[r, yshift=-1.0ex, "{\partial_1}"'] & \Hom(U,M)
    \end{tikzcd}\!\right)\!,
\end{equation*}
where the face maps on $1$-simplices are $\partial_0(g,f)\mapsto f$ and $\partial_1(g,f)\mapsto g\cdot f$ for $f\in\Hom(U,M)$ and $g\in\Hom(U,G)$, which means that $1$-simplices are from $f$ to $g\cdot f$. Moreover, the $2$-simplices encode group multiplication. 
\end{example}

%%%%%%%%%%%%%%%%%%%%%%%%%%%%%%%%%%%%%%%%%%%%%%%%%%%%%%%%%%%%%%%%%%%%%%%%%%%%%%%%%%%%%%%%%%%%%
\section{Zoology of formal smooth stacks}\label{sec:Coo}

The concept of smooth stack can be generalised to the notion of {formal smooth stacks}, which can be intuitively thought of as infinitesimally thickened smooth stacks.
These are defined as stacks on the site of formal smooth manifolds, which can be thought of as smooth manifolds, but whose points are infinitesimally thickened. 
In this section we introduce $\Coo$-algebras, $\Coo$-varieties and formal smooth stacks.

\vspace{0.2cm}
\begin{table}[h!]
\begin{center}
\begin{tabular}{||c | c | c | c | c ||} 
 \hline
  & \multicolumn{2}{c|}{\textbf{Algebraic geometry}} & \multicolumn{2}{c||}{\textbf{Formal smooth geometry}} \\[0.8ex]
 \hline
   & Ordinary & Derived & Ordinary & Derived  \\[0.8ex] 
 \hline\hline
 Lawvere & \multicolumn{2}{c|}{Affine $\mathbbvar{k}$-spaces $\{\mathbb{A}^n_\mathbbvar{k}\}_{n\in\mathbb{N}}$} & \multicolumn{2}{c||}{Cartesian spaces $\{\bbR^n\}_{n\in\mathbb{N}}$} \\ 
  theory &  \multicolumn{2}{c|}{with polynomial maps} & \multicolumn{2}{c||}{with smooth maps}  \\[1ex]
  $\mathsf{T}$ & \multicolumn{2}{c|}{{\small$\mathsf{PolySp}_\mathbbvar{k}$}} & \multicolumn{2}{c||}{{\small$\mathsf{CartSp}$}} \\[1ex]
 \hline
 \multirow{2}{*}{Algebras} & Commutative & Simplicial comm. & \multirow{2}{*}{$\Coo$-algebras} & Simplicial \\ 
  &  $\mathbbvar{k}$-algebras & $\mathbbvar{k}$-algebras &  & $\Coo$-algebras  \\[1ex]
  $\mathsf{TAlg}$ & {\small$\mathsf{cAlg}_\mathbbvar{k}$} & {\small$\mathbf{scAlg}_\mathbbvar{k}$} & {\small$\mathsf{C^\infty Alg}$} & {\small$\mathbf{sC^\infty Alg}$} \\[1ex]
 \hline
 Affine & \multirow{2}{*}{Affine $\mathbbvar{k}$-schemes} & Affine derived & \multirow{2}{*}{Affine $\Coo$-schemes} & Affine derived  \\
  schemes & & $\mathbbvar{k}$-schemes & & $\Coo$-schemes \\[1ex]
  $\mathsf{TAff}$ & {\small$\mathsf{Aff}_\mathbbvar{k} \!\coloneqq\! \mathsf{cAlg}_\mathbbvar{k}^\op$}  & {\small$\mathbf{dAff}_\mathbbvar{k} \!\coloneqq\! \mathbf{scAlg}_\mathbbvar{k}^{\!\op}$} & {\small$\mathsf{C^{\infty\!} Aff} \!\coloneqq\! \mathsf{C^\infty Alg}^\op$} & {\small$\mathbf{dC^{\infty\!\!}Aff} \!\coloneqq\! \mathbf{sC^{\infty\!\!} Alg}^{\!\op}$} \\[1ex]
 \hline
\end{tabular}
\end{center}
    \caption{Comparison between algebraic geometry and formal smooth geometry.}
    \label{tab1}
\end{table}

%%%%%%%%%%%%%%%%%%%%%%%%%%%%%%%%%%%%%%%%%%%%%%%%%%%%%%%%%%%%%%%%%%%%%%%%%%%%%%%%%%%%%%%%%%%%
\subsection{$\Coo$-algebras as a Lawvere theory}

In this subsection we will introduce the notion of $\Coo$-algebra, in the context of Lawvere theories.
First, we will provide a brief review of the notion of a Lawvere theory and of algebra over a given Lawvere theory. 
An algebra over some Lawvere theory is, fundamentally, a generalisation of a ring, given by a set equipped with a set of $n$-ary operations. %These are required satisfy some essential axioms.

\begin{definition}[Lawvere theory]
A \textit{Lawvere theory} (or \textit{algebraic theory}) is a category $\mathsf{T}$ with finite products, whose set of objects is $\{T^{ n}\}_{n\in\mathbb{N}}$ for a fixed object $T\in\mathsf{T}$.
\end{definition}

One can interpret the hom-set $\Hom_\mathsf{T}(T^n,T)$ as the set of abstract $n$-ary operations of the of the Lawvere theory $\mathsf{T}$.

\begin{definition}[$\mathsf{T}$-algebra]
An \textit{algebra over a Lawvere theory} is a product-preserving functor
\begin{equation}
    A\,:\;\mathsf{T} \,\longrightarrow\, \mathsf{Set}.
\end{equation}
\end{definition}

\begin{definition}[Category of $\mathsf{T}$-algebras]
We call $\mathsf{TAlg}$ the category whose objects are all the algebras over the Lawvere theory $\mathsf{T}$, i.e. product-preserving functors $A:\mathsf{T} \longrightarrow \mathsf{Set}$, and whose morphisms are natural transformations between these. 
\end{definition}

\begin{definition}[Forgetful functor of a $\mathsf{T}$-algebra]
We call $U_\mathsf{T}:\mathsf{TAlg}\rightarrow\mathsf{Set}$ the functor which sends a any $\mathsf{T}$-algebra $A$ to its underlying set, i.e.
\begin{equation}
    U_\mathsf{T}(A) \,\coloneqq\, A(T).
\end{equation}
\end{definition}

Notice that, since a $\mathsf{T}$-algebra $A$ is a product preserving functor, any abstract $n$-ary operation $\alpha_n\in\Hom_\mathsf{T}(T^n,T)$ will give rise to a morphism of sets
\begin{equation}
    A(\alpha_n) \,:\, A(T)^{\times n} \,\longrightarrow\, A(T),
\end{equation}
which can be interpreted as an $n$-ary bracket on our particolar $\mathsf{T}$-algebra.

For any Lawvere theory $\mathsf{T}$, it is possible to show that there exists a left adjoint $F_\mathsf{T}\dashv U_\mathsf{T}$ to the forgetful functor. In other words, we have an adjunction
\begin{equation}
    (F_\mathsf{T}\dashv U_\mathsf{T}):\begin{tikzcd}[row sep={10ex,between origins}, column sep={18ex,between origins}]
    \mathsf{Set}
    \arrow[r, "F_\mathsf{T}"{name=L}, bend left=35] &
   \mathsf{TAlg}.
    \arrow[l, "U_\mathsf{T}"{name=R}, bend left=35] \arrow[phantom, from=L, to=R, "\dashv" rotate=-90]
    \end{tikzcd}
\end{equation}
Such a functor works as follows. If given a finite set $S_n\cong\{1,\dots,n\}$ with $n$ elements, one has simply $F_\mathsf{T}(S)=\mathrm{Hom}_\mathsf{T}(T^n,-)$. On the other hand, if given a generic set $S$, one has the filtered colimit $F_\mathsf{T}(S)=\mathrm{colim}_{S_n\in\mathrm{Sub}(S)}\mathrm{Hom}_\mathsf{T}(T^n,-)$, where $\mathrm{Sub}(S)$ is the poset of finite subsets of $S$. The $\mathsf{T}$-algebras lying in the image of the functor $F_\mathsf{T}$ are also known as free $\mathsf{T}$-algebras.

The archetypal example of Lawvere theory is the one of usual rings.

\begin{example}[Rings]
Let $\mathsf{T}$ be the category whose objects are affine schemes $\{\mathbb{A}^n_{\mathbb{Z}}\}_{n\in\mathbb{N}}$, where $\mathbb{A}^n_{\mathbb{Z}}=\Spec\,\mathbb{Z}[x_1,\dots,x_n]$, and whose morphisms are polynomial maps between these. Then $\mathsf{TAlg}=\mathsf{Ring}$ is the category of rings.
\end{example}

By directly generalising the example right above, we have the following class of examples.

\begin{example}[$S$-algebras]\label{ex:srings}
Let now $S$ be any commutative ring and let $\mathsf{T}$ be the category whose objects are the affine schemes $\{\mathbb{A}^n_{S}\}_{n\in\mathbb{N}}$, where $\mathbb{A}^n_{S}=\Spec\,S[x_1,\dots,x_n]$, and whose morphisms are polynomial maps between these. Then $\mathsf{TAlg}=\mathsf{Alg}_S$ is the category of $S$-algebras.
\end{example}

Now we have all the ingredients to introduce the notion of $\Coo$-algebra in the context of Lawvere theories. The Lawvere theory underlying $\Coo$-algebras will be a natural generalisation of the Lawvere theory underlying the $S$-rings from example \ref{ex:srings}.

\begin{definition}[Lawvere theory of smooth Cartesian spaces]
We define $\mathsf{T}=\mathsf{CartSp}$ as the category whose objects are Cartesian spaces $\{\bbR^n\}_{n\in\mathbb{N}}$ and whose morphisms are smooth maps between these. 
\end{definition}

We can now provide the definition of $\Coo$-algebra as an algebra over the Lawvere theory of smooth Cartesian spaces.

\begin{definition}[$\Coo$-algebra]
Let $\mathsf{T}=\mathsf{CartSp}$. Then, we call $\mathsf{C^\infty Alg}\coloneqq\mathsf{TAlg}$ the \textit{category of $\Coo$-algebras} and an object $A\in\mathsf{C^\infty Alg}$ a \textit{$\Coo$-algebra}.
\end{definition}

Notice that, given a $\Coo$-algebra $A$, its underlying set $U_\mathsf{CartSp}(A)=A(\bbR)$ has a natural ring structure. In fact, addiction and multiplication $+,\,\cdot\,:\bbR\times\bbR\rightarrow \bbR$, opposite $-:\bbR\rightarrow \bbR$, zero element $0:\bbR^0\hookrightarrow \bbR$ and unit $1:\bbR^0\hookrightarrow \bbR$ are all smooth maps in the category $\mathsf{CartSp}$ of Cartesian spaces. Since $A$ is a functor which preserves products, then the functions $A(+),A(\,\cdot\,),A(0),A(1),A(-)$ satisfy the axioms of a ring structure on the set $A(\bbR)$.

\begin{remark}[Limits and filtered colimits]
The category $\mathsf{C^\infty Alg}$ has all limits and all filtered colimits. They can be computed object-wise in $\mathsf{CartSp}$ by taking the corresponding limits and filtered colimits in $\mathsf{Set}$
\end{remark}

\begin{definition}[$\Coo$-tensor product]
The \textit{$\Coo$-tensor product} in the category $\mathsf{C^\infty Alg}$ is defined to be the pushout
\begin{equation}
    A\,\widehat{\otimes}_B\,C \,\coloneqq\, A \sqcup_{B} C,
\end{equation}
for any $\Coo$-algebras $A,B,C\in\mathsf{C^\infty Alg}$.
\end{definition}

The following is the archetypal example of $\Coo$-algebras. Given smooth manifold $M\in\Mfd$, we can construct a $\Coo$-algebra of functions on $M$ by the functor
\begin{equation}
    \Coo(M)\,:\, \bbR^n\,\mapsto\, \Coo(M,\bbR^n).
\end{equation}
We can construct a contravariant functor by sending any smooth manifold $M$ to its $\Coo$-algebra of functions $\Coo(M)$ and any smooth map $f:M\rightarrow N$ to its pullback $f^\ast:\Coo(N)\rightarrow\Coo(M)$.

\begin{lemma}[Smooth manifolds as $\Coo$-algebras \cite{Moerdijk:1991}]
The contravariant functor
\begin{equation}
\begin{aligned}
    \mathsf{Mfd}^\op \,&\longhookrightarrow\, \mathsf{C^\infty Alg} \\
    M \,&\longmapsto\, \Coo(M),
\end{aligned}
\end{equation}
is full and faithful.
\end{lemma}

\begin{definition}[Transverse maps]
Two smooth maps of smooth manifolds $f:\Sigma\rightarrow M$ and $g:\Sigma'\rightarrow M$ are called transverse if the map $f\sqcup g :\Sigma\sqcup \Sigma'\rightarrow M$ is a submersion, i.e. if its differential $(f\sqcup g)_\ast:T\Sigma\sqcup T\Sigma'\rightarrow TM$ is a surjective bundle map. 
\end{definition}

The following lemma makes the crucial point that two smooth maps are transverse, then their fibre product exists in the category of smooth manifolds.

\begin{lemma}[$\Coo$-algebra of functions on intersection of smooth manifolds \cite{Moerdijk:1991}]
Let $f:\Sigma\rightarrow M$ and $g:\Sigma'\rightarrow M$ be transverse maps of smooth manifolds and let the square
\begin{equation}
    \begin{tikzcd}[row sep={10.5ex,between origins}, column sep={12ex,between origins}]
    \Sigma\times_M \Sigma'\arrow[d] \arrow[r] & \Sigma \arrow[d, "f"] \\
    \Sigma' \arrow[r, "g"] & M
    \end{tikzcd}
\end{equation}
be a pullback in $\Mfd$. 
Then, the square
\begin{equation}
    \begin{tikzcd}[row sep={18ex,between origins}, column sep={20ex,between origins}]
   \Coo(M) \arrow[d] \arrow[r] & \Coo(\Sigma) \arrow[d, "f^\ast"] \\
    \Coo(\Sigma') \arrow[r, "g^\ast"] & \Coo(\Sigma\times_M \Sigma')
    \end{tikzcd}
\end{equation}
is a pushout in $\Coo\mathsf{Alg}$. In other words, we have an isomorphism of $\Coo$-algebras
\begin{equation}
    \Coo(\Sigma\times_M \Sigma') \;=\; \Coo(\Sigma)\,\widehat{\otimes}_{\Coo(M)}\,\Coo(\Sigma').
\end{equation}
\end{lemma}

If we choose $M=\ast\,$ to be the point and the smooth maps $f,g$ to be the terminal maps to the point in the category of smooth manifolds, we immediately have the following proposition.

\begin{corollary}[$\Coo$-algebra of functions on product manifolds]
For any pair of manifolds $M,N\in\Mfd$, we have an isomorphism of $\Coo$-algebras
\begin{equation}
    \Coo(M)\,\widehat{\otimes}_{\bbR}\,\Coo(N) \;=\; \Coo(M\times N).
\end{equation}
\end{corollary}

Notice that the $\Coo$-tensor product $A\btimes_\bbR B$ is much smaller than the usual tensor product $A(\bbR)\otimes_{\bbR} B(\bbR)$ of the underlying $\bbR$-algebras.

\begin{definition}[Ideal of a $\Coo$-algebra]
An \textit{ideal $\mathcal{I}$ of a $\Coo$-algebra} $A$ is defined as an ideal of the underlying ring $A(\bbR)$.
\end{definition}

As shown in \cite{Moerdijk:1991, Joyce:2009}, given an ideal $\mathcal{I}$ of a $\Coo$-algebra $A$, there is a canonical $\Coo$-algebra $A/\mathcal{I}$ whose underlying ring is precisely the quotient ring $A(\bbR)/\mathcal{I}$. 

\begin{definition}[Finitely generated and finitely presented $\Coo$-algebras]
By following \cite[Chapter I]{Moerdijk:1991}, we define:
\begin{itemize}[topsep=-25pt]
    \item  a \textit{finitely generated $\Coo$-algebra} as a $\Coo$-algebra of the form $A\cong \Coo(\bbR^n)/\mathcal{I}$, for some Cartesian space $\bbR^n$ and an ideal $\mathcal{I}\subset \Coo(\bbR^n)$;
    \item a \textit{finitely presented $\Coo$-algebra} as a $\Coo$-algebra of the form $A\cong \Coo(\bbR^n)/\mathcal{I}$, for some Cartesian space $\bbR^n$ and a finitely generated ideal $\mathcal{I}\subset \Coo(\bbR^n)$.
\end{itemize}
\end{definition}
We denote by $\mathsf{C^\infty Alg}_\mathrm{fg}$ and $\mathsf{C^\infty Alg}_\mathrm{fp}$ the full subcategories of $\mathsf{C^\infty Alg}$ on those objects which are respectively finitely generated and finitely presented $\Coo$-algebras.

The archetypal example of finitely presented $\Coo$-algebra is again the the $\Coo$-algebra $\Coo(M)$ of functions on any smooth manifold $M\in\Mfd$. This is because any smooth manifold can be embedded in $\bbR^N$ for $N$ large enough.

\begin{example}[Smooth manifold as finitely presented $\Coo$-algebra]
Consider a circle $S^1$. Its $\Coo$-algebra of functions is $\Coo(S^1) = \Coo(\bbR^2)/(x^2+y^2-1)$, which is finitely presented.
\end{example}

\begin{example}[Local Artinian $\bbR$-algebra]\label{ex:weil_algebra}
Another crucial example is provided by local Artinian $\bbR$-algebras, also known as Weil algebras in the context of differential geometry. Recall that a local Artinian algebra is a finite-dimensional commutative $\bbR$-algebra $W$ with a maximal differential ideal $\mathfrak{m}_W\subset W$ such that $W/\mathfrak{m}_W\cong \bbR$ and $\mathfrak{m}_W^N=0$ for some $N$ large enough. By \cite[Proposition 1.5]{Dubuc1979}, any local Artinian $\bbR$-algebra can be uniquely lifted to $\Coo$-algebra, which is always finitely presented.
\end{example}

\begin{example}[Algebra of truncated Taylor series as finitely presented $\Coo$-algebra]
The local Artinian algebra $W^n_k = \Coo(\bbR^n)/(x_1,\dots,x_n)^k$ of $k$-truncated Taylor series in $n$ variables comes with canonical $\Coo$-algebra structure.
\end{example}

\begin{remark}[Reduced $\Coo$-algebras]\label{rem:reducedalgebras}
Let $\mathsf{C^\infty Alg}^\mathrm{red}$ be the full sub-category of $\mathsf{C^\infty Alg}$ on those $\Coo$-algebras whose underlying $\bbR$-algebra is reduced in the usual sense, i.e. it has no non-zero nilpotent elements. 
Then, we have an adjunction
\begin{equation}
    \begin{tikzcd}[row sep=scriptsize, column sep=8.2ex, row sep=18.0ex]
     \mathsf{C^\infty Alg}^\mathrm{red} \arrow[r, "\iota^\mathrm{red}"',""{name=0, anchor=center, inner sep=0}, shift left=-1.1ex, hook] &  \mathsf{C^\infty Alg}. \arrow[l, "(-)^\mathrm{red}"', ""{name=1, anchor=center, inner sep=0}, shift left=-1.1ex] \arrow["\dashv"{anchor=center, rotate=-90}, draw=none, from=1, to=0]
    \end{tikzcd}
\end{equation}
where $\iota^\mathrm{red}$ is the natural embedding and $(-)^\mathrm{red}$ is the functor which sends a $\Coo$-algebra $A$ to the reduced $\Coo$-algebra $A^\mathrm{red}\coloneqq A/\mathfrak{m}_A$, where we called $\mathfrak{m}_A$ the nilradical of the the underlying $\bbR$-algebra. 
\end{remark}

\begin{example}[Examples of reduction]
Consider a local Artinian algebra $W$, then we have $W^\mathrm{red} = \bbR$. If $M$ is a smooth manifold, then we have $\Coo(M)^\mathrm{red}= \Coo(M)$.
Moreover, for a $\Coo$-tensor product of the form $\Coo(M)\,\widehat{\otimes}\,W$, then we have $(\Coo(M)\,\widehat{\otimes}\,W)^\mathrm{red}= \Coo(M)$.
\end{example}

\begin{remark}[Smooth manifolds embed into reduced $\Coo$-algebras]\label{rem:embedding_manifolds}
Notice from the previous example that the $\Coo$-algebra $\Coo(M)$ of functions on an ordinary smooth manifold $M$ lies always in $\mathsf{C^\infty Alg}^\mathrm{red}$.
More precisely, the embedding of smooth manifolds into $\Coo$-algebras factors by $\Coo(-):\Mfd^\op \longhookrightarrow \mathsf{C^\infty Alg}^\mathrm{red}_\mathrm{fp} \longhookrightarrow \mathsf{C^\infty Alg}_\mathrm{fp}\longhookrightarrow  \mathsf{C^\infty Alg}_\mathrm{fg} \longhookrightarrow \mathsf{C^\infty Alg}$, where we called $\mathsf{C^\infty Alg}_\mathrm{fp}^\mathrm{red}$ the category of reduced finitely presented $\Coo$-algebras.
\end{remark}

%%%%%%%%%%%%%%%%%%%%%%%%%%%%%%%%%%%%%%%%%%%%%%%%%%%%%%%%%%%%%%%%%%%%%%%%%%%%%%%%%%%%%%%%%%%%%
\subsection{$\Coo$-varieties and formal smooth manifolds}

As we have seen in the previous subsection, we have a fully faithful embedding $\Mfd \hookrightarrow \mathsf{C^\infty Alg}_\mathrm{fg}^\op$ of smooth manifolds into the opposite category of finitely generated $\Coo$-algebras.
Thus, in a certain sense, we may interpret the category $\mathsf{C^\infty Alg}_\mathrm{fg}^\op$ as a category of generalised smooth spaces of some sort. Such an intuition, for instance, underlies the formalisation by \cite{Moerdijk:1991} of analysis.

\begin{definition}[$\Coo$-variety \cite{Moerdijk:1991}]
We define a \textit{$\Coo$-variety} as an element of the opposite category of finitely generated $\Coo$-algebras, i.e. of the category
\begin{equation}
    \mathsf{C^\infty Var} \;\coloneqq\; \mathsf{C^\infty Alg}^\op_\mathrm{fg}.
\end{equation}
We use the notation $X=\Spec(A)$ for the $\Coo$-variety corresponding to the finitely generated $\Coo$-algebra $A\in\mathsf{C^\infty Alg}$. Conversely, we may use the notation $\mathcal{O}(X)$ for the finitely generated $\Coo$-algebra corresponding to the $\Coo$-variety $X\in\mathsf{C^\infty Var}$.
\end{definition}

Let us look at a few simple examples of such a geometric object which go beyond the notion of smooth manifolds. First, we can consider infinitesimally thickened points, i.e. formal disks. 

\begin{example}[Thickened point]
Consider the local Artinian algebra of $k$-truncated Taylor series $W^n_k = \Coo(\bbR^n)/(x_1,\dots,x_n)^k$ with its canonical $\Coo$-algebra structure.
Then we have an infinitesimally thickened point given by $D^n_k = \Spec(W^n_k)$.
\end{example}

This example can be directly generalised to construct an example of infinitesimally thickened smooth manifolds.

\begin{example}[Thickened circle]
Consider the thickened circle given by $S^1\times \Spec W$, where $S^1$ is a circle and $W = \Coo(\bbR)/(z^2)$. Dually, this can be constructed by $\Coo$-tensor product of the corresponding $\Coo$-algebras
\begin{equation}
    \frac{\Coo(\bbR^2)}{(x^2+y^2-1)} \,\widehat{\otimes}\, \frac{\Coo(\bbR)}{(z^2)} \;=\; \frac{\Coo(\bbR^3)}{(x^2+y^2-1, z^2)}
\end{equation}
Thus, it can be expressed as $S^1\times \Spec W = \mathrm{Spec}(\Coo(\bbR^3)/(x^2+y^2-1, z^2))$. 
\end{example}

Now, the category $\mathsf{C^\infty Var}$ of $\Coo$-varieties that we have presented here does not have an internal hom-functor, in general. However, we have the following stricter statement.

\begin{lemma}[Exponential by a thickened point]
Let $D=\Spec W$ where $W$ is a local Artinian algebra and $Y$ any $\Coo$-variety. Then there exists a endofunctor of $\Coo$-varieties
\begin{equation}
    (-)^D \,: \,Y \,\longmapsto\, Y^D,
\end{equation}
which is the right adjoint of the functor $(-)\times D$ given by taking the product with $D$.
In other words, $Y^D$ is a $\Coo$-variety which satisfies the property
\begin{equation}\label{eq:adjunction_varieties}
    \Hom_{\mathsf{C^\infty Var}}(X, \,Y^D) \;\simeq\; \Hom_{\mathsf{C^\infty Var}}(X\times D,\, Y)
\end{equation}
for any $\Coo$-variety $X\in\mathsf{C^\infty Var}$.
\end{lemma}

\begin{proof}
We deploy an argument similar to \cite[Theorem 1.13]{Moerdijk:1991}. 
First we have to verify that $\bbR^D$ exists. So, for any $\Coo$-variety $X\in\mathsf{C^\infty Var}$ we have the equivalences
\begin{equation}
 \begin{aligned}
     \Hom_{\mathsf{C^\infty Var}}(X\times D,\bbR)   \;  &\simeq\; (\mathcal{O}(X)\,\widehat{\otimes}\, W)(\bbR) \\
     &\simeq\; \mathcal{O}(X)(\bbR^{\mathrm{dim}(W)}) \\
     &\simeq\; \Hom_{\mathsf{C^\infty Var}}(X,\bbR^{\mathrm{dim}(W)}),
 \end{aligned}   
\end{equation}
where in the penultimate step we used the fact that any smooth function $g\in\mathcal{O}(X)\,\widehat{\otimes}\, W$ can be expanded as $(g_1,\dots,g_{\mathrm{dim}(W)})$ with each $g_i\in\mathcal{O}(X)$. Thus we have $\bbR^D \simeq \bbR^{\mathrm{dim}(W)}$, which exists.
By the same argument, we have an equivalence $\Hom_{\mathsf{C^\infty Var}}(X\times D,\bbR^k)\simeq \Hom_{\mathsf{C^\infty Var}}(X,\bbR^{k\,\mathrm{dim}(W)})$ for any natural number $k$ and $\Coo$-variety $X$. This implies that $(\bbR^0)^D \simeq \bbR^0$ exist and  that $(\bbR^k)^D \simeq (\bbR^D)^k$ exist for any $k>0$.
Now, given a smooth map $f:\bbR^n\rightarrow \bbR^m$, the new map $f^D:(\bbR^n)^D\rightarrow (\bbR^m)^D$ is given by the equivalence $\Hom_{\mathsf{C^\infty Var}}(X,f^D)\simeq \Hom_{\mathsf{C^\infty Var}}(X\times D,f)$ for any $\Coo$-variety $X$.
Now, let us fix a generic $\Coo$-variety $Y = \Spec(A)$, where $A \cong \Coo(\bbR^n)/(f_1,\dots,f_m)$ is a finitely generated $\Coo$-algebra with $f_i\in\Coo(\bbR^n)$. We must show that there exists a $\Coo$-variety $Y^D$ such that the equivalence \ref{eq:adjunction_varieties} holds.
Since $A$ is a quotient, $Y= \Spec(A)$ is equivalently defined by the pullback square
\begin{equation}
    \begin{tikzcd}[row sep=scriptsize, column sep=10.0ex, row sep=10.5ex]
     Y \arrow[r, ""]\arrow[d, ""] & \bbR^0 \arrow[d, "0"] \\
     \bbR^n \arrow[r, "{(f_1,\dots,f_m)}"] & \bbR^m.
    \end{tikzcd}
\end{equation}
On the one hand, since the functor $\Hom_{\mathsf{C^\infty Var}}(X\times D,-)$ preserves pullbacks for any $\Coo$-variety $X$, we have a pullback square of sets
\begin{equation*}
    \begin{tikzcd}[row sep=scriptsize, column sep=19.0ex, row sep=16.0ex]
     \Hom_{\mathsf{C^\infty Var}}(X\times D,Y) \arrow[r, ""]\arrow[d, ""] & \Hom_{\mathsf{C^\infty Var}}(X,(\bbR^0)^D) \arrow[d, "{\Hom_{\mathsf{C^\infty Var}}(X,0^D)}"] \\
     \Hom_{\mathsf{C^\infty Var}}(X,(\bbR^n)^D) \arrow[r, "{\Hom_{\mathsf{C^\infty Var}}(X,f_1^D\!,\dots,f_m^D)}"] & \Hom_{\mathsf{C^\infty Var}}(X,(\bbR^m)^D).
    \end{tikzcd}
\end{equation*}
for any $\Coo$-variety $X$.
On the other hand, we have the pullback square of $\Coo$-varieties
\begin{equation}
    \begin{tikzcd}[row sep=scriptsize, column sep=7.0ex, row sep=15.0ex]
     (\bbR^n)^D \!\times_{(\bbR^m)^D}\!(\bbR^0)^D \arrow[r, ""]\arrow[d, ""] & (\bbR^0)^D \arrow[d, "0^D"] \\
     (\bbR^n)^D \arrow[r, "{(f_1^D\!,\dots,f_m^D)}"] & (\bbR^m)^D.
    \end{tikzcd}
\end{equation}
Thus, the $\Coo$-variety $Y^D$ exists and it is indeed given by $Y^D \simeq (\bbR^n)^D \!\times_{(\bbR^m)^D}\!(\bbR^0)^D$.
\end{proof}

Notice that, for any $\Coo$-variety $Y$ and $D=\Spec W$ where $W$ is a local Artinian algebra, there is a natural morphism $\mathrm{ev}_0:Y^D\rightarrow Y$ from the $D$-exponential to the original $Y$. This is induced by the canonical inclusion $\ast \rightarrow D$ of the point into the canonical point of $D$.

\begin{definition}[$\mathcal{D}$-\'etale map]
We say that a morphism $f:X\rightarrow Y$ of $\Coo$-varieties is \textit{$\mathcal{D}$-\'etale} if we have a pullback diagram
\begin{equation}\label{eq:def_of_formallyetale}
    \begin{tikzcd}[row sep=scriptsize, column sep=8.2ex, row sep=9.0ex]
     X^{D} \arrow[r, "f^{D}"]\arrow[d, "\mathrm{ev}_0"] & Y^{D} \arrow[d, "\mathrm{ev}_0"] \\
     X \arrow[r, "f"] & Y
    \end{tikzcd}
\end{equation}
for any thickened point $D=\Spec W$, where $W$ is a local Artinian algebra.
\end{definition}

\begin{corollary}[$\mathcal{D}$-\'etale maps generalise local diffeomorphisms]
Let $M$ and $N$ be ordinary smooth manifolds, seen as $\Coo$-varieties. Then, we have that any $\mathcal{D}$-\'etale map $f:M\rightarrow N$ is equivalently a local diffeomorphism in the ordinary differential geometry sense.
\end{corollary}
\begin{proof}
To see this, notice that by setting $D = \Spec(\Coo(\bbR)/(x^2))$ to be the local Artinian algebra of dual numbers, then the pullback square \eqref{eq:def_of_formallyetale} becomes precisely
\begin{equation}
    \begin{tikzcd}[row sep=scriptsize, column sep=7.5ex, row sep=9.0ex]
     TM \arrow[r, "f_\ast"]\arrow[d, "\pi_M"] & TN \arrow[d, "\pi_N"] \\
     M \arrow[r, "f"] & N,
    \end{tikzcd}
\end{equation}
making $f$ into a local diffeomorphism. 
Conversely, a local diffeomorphism $f$ induces a diffeomorphism $U_x \xrightarrow{\;\simeq\;} V_{f(x)}$ of open neighborhoods respectively of $x$ and of its image for any point $x\in M$. Thus we have the diagram
\begin{equation}
\begin{tikzcd}[row sep=scriptsize, column sep={8.0ex,between origins}, row sep={7.5ex,between origins}]
	U_x^D && V_{f(x)}^D \\
	& M^D && N^D \\
	{U_x} && {V_{f(x)}} \\
	& M && N,
	\arrow["", from=2-2, to=4-2]
	\arrow[from=2-2, to=2-4]
	\arrow[""', from=2-4, to=4-4]
	\arrow["", from=4-2, to=4-4]
	\arrow[""', from=1-1, to=2-2]
	\arrow["", from=1-1, to=3-1]
	\arrow[" ", from=3-1, to=4-2]
	\arrow[" ", from=1-3, to=2-4]
	\arrow["\simeq", from=1-1, to=1-3]
	\arrow["", from=1-3, to=3-3]
	\arrow["\simeq\;\;\;\;", from=3-1, to=3-3]
	\arrow[" ", from=3-3, to=4-4]
\end{tikzcd}
\end{equation}
which implies that the square on the front is a pullback.
\end{proof}

In the spirit of interpreting $\Coo$-varieties as formal generalisations of ordinary smooth manifolds, we can equip their category $\mathsf{C^\infty Var}$ with a coverage which is compatible with the coverage of $\mathsf{Mfd}$ from previous section.
Thus we define a coverage as follows.

\begin{lemma}[$\mathcal{D}$-\'etale covering family of a $\Coo$-variety]\label{lem:covering_var}
We may declare a covering family of a $\Coo$-variety $X$ to be a set of $\mathcal{D}$-\'etale monomorphisms
\begin{equation}
    \{ U_i \,\xhookrightarrow{\;\;\phi_i\;\;} X \}_{i\in I}
\end{equation}
such that they induce the $\mathcal{D}$-\'etale epimorphism
\begin{equation}
    \coprod_{i\in I}U_i \,\xtwoheadrightarrow{\;\;(\phi_i)_{i\in I}\;\;} X .
\end{equation}
\end{lemma}

\begin{proof}
First, we show that $\mathcal{D}$-\'etale morphisms are stable under pullback.
Consider a pullback diagram of $\Coo$-varieties of the form
\begin{equation}
    \begin{tikzcd}[row sep=scriptsize, column sep=7.3ex, row sep=9.0ex]
     X\times_ZY \arrow[r, " "]\arrow[d, "\psi"] & Y \arrow[d, "\phi"] \\
     X \arrow[r, ""] & Z,
    \end{tikzcd}
\end{equation}
where we assume that $\phi$ is a $\mathcal{D}$-\'etale map.
As previously noticed, $(-)^D$ preserves pullbacks, thus we have a bigger diagram
\begin{equation}
\begin{tikzcd}[row sep=scriptsize, column sep={10.0ex,between origins}, row sep={9.5ex,between origins}]
	{(X\times_ZY)^D} && {Y^D} \\
	& {X\times_ZY} && Y \\
	{X^D} && {Z^D} \\
	& X && Z,
	\arrow["\psi", near start, from=2-2, to=4-2]
	\arrow[from=2-2, to=2-4]
	\arrow["\phi", from=2-4, to=4-4]
	\arrow[from=4-2, to=4-4]
	\arrow["{\mathrm{ev}_0}"', from=1-1, to=2-2]
	\arrow["{\psi^D}"', from=1-1, to=3-1]
	\arrow["{\mathrm{ev}_0}", from=3-1, to=4-2]
	\arrow["{\mathrm{ev}_0}", from=1-3, to=2-4]
	\arrow[from=1-1, to=1-3]
	\arrow["{\phi^D}", near start, from=1-3, to=3-3]
	\arrow[from=3-1, to=3-3]
	\arrow["{\mathrm{ev}_0}", from=3-3, to=4-4]
\end{tikzcd}
\end{equation}
where both the front and the back square are pullbacks.
Moreover, $\phi$ being $\mathcal{D}$-\'etale implies that the right square is a pullback too. Then, by applying the pasting law for pullbacks we obtain that the left square is a pullback and thus $\psi$ is $\mathcal{D}$-\'etale. Therefore, $\mathcal{D}$-\'etale maps are stable under pullbacks.
Now, consider a covering family $\{ U_i \,\xhookrightarrow{\;\;\phi_i\;\;} X \}_{i\in I}$ as above and a morphism $f:Y\rightarrow X$. We can form the pullback square
\begin{equation}
    \begin{tikzcd}[row sep=scriptsize, column sep=6.8ex, row sep=9.0ex]
     Y\times_XU_i \arrow[r, ""]\arrow[d, "\psi_i"] & U_i \arrow[d, "\phi_i", hook] \\
     Y \arrow[r, "f"] & X,
    \end{tikzcd}
\end{equation}
Since monomorphisms and $\mathcal{D}$-\'etale monomorphisms are stable under pullbacks, then $\psi_i$ is a $\mathcal{D}$-\'etale monomorphism.
Moreover, we have that the morphism $\coprod_{i\in I}Y\times_X U_i \xrightarrow{\;(\psi_i)_{i\in I}\;} Y$ is a $\mathcal{D}$-\'etale epimorphism.
\end{proof}

The following definition is a specialization of the general one provided by \cite{kock_2006}.

\begin{definition}[Formal smooth manifolds]
We define a \textit{formal smooth manifold} $M$ as a $\Coo$-variety such that there exist a family $\{\bbR^{n\!}\times \Spec W \,\xhookrightarrow{\;\,\phi_i\,\;}\, M\}_{i\in I}$ of $\mathcal{D}$-\'etale monomorphisms, where $W$ is Artinian, with the property that the induced map
\begin{equation}
    \bigsqcup_{i\in I}\bbR^{n\!}\times \Spec W \,\xtwoheadrightarrow{\;\;(\phi_i)_{i\in I}\;\;}\, M
\end{equation}
is an \'etale epimorphism.
We denote by $\mathsf{FMfd}$ the category of formal smooth manifolds, i.e. the full and faithful subcategory of $\mathsf{C^\infty Var}$ whose objects are all the formal smooth manifolds and we denote its embedding into the latter by
\begin{equation}\label{eq:embedding_formalmanifolds}
    \iota^{\mathsf{FMfd}}:\,\,\mathsf{FMfd} \,\longhookrightarrow\, \mathsf{C^\infty Var}.
\end{equation}
\end{definition}

In other words, a $\Coo$-variety is a formal smooth manifold if it admits a covering of thickened charts of the form $\bbR^n\times \Spec W$ for some $n\in\mathbb{N}$ and local Artinian algebra $W\in\mathrm{Art}_\bbR$.

\begin{remark}[Covering family of a formal smooth manifold]
Notice that we can naturally make the category $\mathsf{FMfd}$ of formal smooth manifold into a site by restricting the covering families of the site $\mathsf{C^\infty Var}$ of $\Coo$-varieties we constructed in theorem \ref{lem:covering_var}.  
\end{remark}

\begin{example}[Thickened circle]
Consider the thickened circle from the previous subsection
\begin{equation}
    S^1\times \Spec W \;=\; \mathrm{Spec}_{\!}\big(\Coo(\bbR^3)/(x^2+y^2-1, z^2)\big)
\end{equation}
Notice that it can be covered by a covering $\{\bbR^{\!}\times \Spec W \,\xhookrightarrow{\;\,(\psi_i,\mathrm{id})\,\;}\, S^1\times \Spec W\}_{i=0,1}$ where the set $\{\bbR\,\xhookrightarrow{\;\,\psi_i\,\;}\, S^1\}_{i=0,1}$ is just a covering of the underlying circle as a smooth manifold.
\end{example}

\begin{construction}[Reduction of formal smooth manifolds]
By reduction and co-reduction of adjunction \ref{rem:reducedalgebras}, we can obtain the adjunction of categories
\begin{equation}\label{rem:reducedalgebras2}
    \begin{tikzcd}[row sep=scriptsize, column sep=12.2ex, row sep=15.0ex]
     \mathsf{Mfd} \arrow[r, " "', ""{name=0, anchor=center, inner sep=0}, shift left=1.1ex, hook] &  \mathsf{FMfd}. \arrow[l, " "', ""{name=1, anchor=center, inner sep=0}, shift left=1.1ex] \arrow["\dashv"{anchor=center, rotate=-90}, draw=none, from=1, to=0]
    \end{tikzcd}
\end{equation}
In particular, this is an adjunction of ordinary sites, since by construction of formal smooth manifolds both functors send covering families to covering families on the nose.
\end{construction}

\begin{definition}[Formally \'etale map]
\begin{equation}
    \begin{tikzcd}[row sep=scriptsize, column sep=3.5ex, row sep=11.0ex]
     X \arrow[r, "f"] & Y \\
     \Spec(R) \arrow[r, ""]\arrow[u, ""] & \Spec(R/\mathfrak{m}_R) \arrow[u, ""]\arrow[ul, "\exists!"', dashed]
    \end{tikzcd}
\end{equation}
\end{definition}

\begin{lemma}[Formally \'etale $\Rightarrow$ $\mathcal{D}$-\'etale]
Let $f:X\rightarrow Y$ be a formally \'etale morphism of $\Coo$-varieties. Then $f$ is a $\mathcal{D}$-\'etale morphism.
\end{lemma}

\begin{proof}
By element chasing, we have the pullback square of sets for any $\Coo$-algebra $R$
\begin{equation}
    \begin{tikzcd}[row sep=scriptsize, column sep=2.0ex, row sep=11.0ex]
     \Hom(A,R)\arrow[d, ""] \arrow[r, ""] & \Hom(B,R)\arrow[d, ""] \\
     \Hom(A,R/\mathfrak{m}_R) \arrow[r, ""] & \Hom(A,R/\mathfrak{m}_R)
    \end{tikzcd}
\end{equation}
By pasting law for pullbacks, 
\begin{equation}
    \begin{tikzcd}[row sep=scriptsize, column sep=2.0ex, row sep=11.0ex]
     \Hom(A,R'\,\widehat{\otimes}\,W)\arrow[d, ""] \arrow[r, ""] & \Hom(B,R'\,\widehat{\otimes}\,W)\arrow[d, ""] \\
     \Hom(A,R') \arrow[r, ""] & \Hom(A,R')
    \end{tikzcd}
\end{equation}
Dually, in the category of $\Coo$-varieties we have the pullback square of sets
\begin{equation}
    \begin{tikzcd}[row sep=scriptsize, column sep=2.0ex, row sep=11.0ex]
     \Hom(U\times D, X)\arrow[d, ""] \arrow[r, ""] & \Hom(U\times D,Y)\arrow[d, ""] \\
     \Hom(U,X) \arrow[r, ""] & \Hom(U,Y)
    \end{tikzcd}
\end{equation}
which implies that the morphism $f:X\rightarrow Y$ is $\mathcal{D}$-\'etale.
\end{proof}

\begin{lemma}[Formally \'etale covering family of a $\Coo$-variety]
We may declare a covering family of a $\Coo$-variety $X$ to be a set of formally \'etale monomorphisms
\begin{equation}
    \{ U_i \,\xhookrightarrow{\;\;\phi_i\;\;} X \}_{i\in I}
\end{equation}
such that they induce the formally \'etale epimorphism
\begin{equation}
    \coprod_{i\in I}U_i \,\xtwoheadrightarrow{\;\;(\phi_i)_{i\in I}\;\;} X .
\end{equation}
\end{lemma}

\begin{proof}
We have the diagram
\begin{equation}
\begin{tikzcd}[row sep=scriptsize, column sep={15.0ex,between origins}, row sep={11ex,between origins}]
	{\Hom(\Spec{R},X_{\!}\times_{Z\!}Y)} && \Hom(\Spec{R},Y) \\
	& {\Hom(\Spec{R}/\mathfrak{m}_R,X_{\!}\times_{Z\!}Y)} && \Hom(\Spec{R}/\mathfrak{m}_R,Y) \\
	\Hom(\Spec{R},X) && \Hom(\Spec{R},Z) \\
	& \Hom(\Spec{R}/\mathfrak{m}_R,X) && \Hom(\Spec{R}/\mathfrak{m}_R,Z),
	\arrow["", near start, from=2-2, to=4-2]
	\arrow[from=2-2, to=2-4]
	\arrow["", from=2-4, to=4-4]
	\arrow[from=4-2, to=4-4]
	\arrow[""', from=1-1, to=2-2]
	\arrow[""', from=1-1, to=3-1]
	\arrow["", from=3-1, to=4-2]
	\arrow["", from=1-3, to=2-4]
	\arrow[from=1-1, to=1-3]
	\arrow["", near start, from=1-3, to=3-3]
	\arrow[from=3-1, to=3-3]
	\arrow["", from=3-3, to=4-4],
\end{tikzcd}
\end{equation}
where the front, the back and the right squares are pullbacks. 
Then, by pasting, we have that the left square is a pullback, which implies that formally \'etale morphisms are stable under pullbacks.
\end{proof}

%%%%%%%%%%%%%%%%%%%%%%%%%%%%%%%%%%%%%%%%%%%%%%%%%%%%%%%%%%%%%%%%%%%%%%%%%%%%%%%%%%%%%%%%%%%%%
\subsection{Definition of formal smooth stacks}\label{subsec:formalsmoothset}

In this subsection, we will generalise smooth sets and smooth stacks, respectively, to formal smooth set and formal smooth stacks.

Let us start from the definition of formal smooth sets, which are roughly ordinary sheaves on formal smooth manifolds.
Geometrically, they provide a rich class of generalisations of smooth manifolds.
In particular, they allow us to formalise a large variety of infinite-dimensional smooth spaces, such as smooth mapping spaces and smooth spaces of sections of a bundle.

\begin{definition}[Formal smooth sets]\label{def:formal_smooth_set}
\textit{Formal smooth sets} are defined  as sheaves on the site of formal smooth manifolds $\mathsf{FMfd}$. The category of formal smooth sets is, then, defined by
\begin{equation}
    \mathsf{FSmoothSet} \;\coloneqq\; \mathsf{Sh}(\mathsf{FMfd}).
\end{equation}
\end{definition}

This definition is equivalent\footnote{
In \cite{Dubuc1979} formal smooth sets were defined as sheaves on the site $\mathsf{FCartSp}$ of formal Cartesian spaces, i.e. spaces of the form $\bbR^n\times \Spec W$, where $\bbR^n$ is a Cartesian space and $W$ is a local Artinian algebra. However, $\mathsf{FCartSp}$ is by construction a dense sub-site of $\mathsf{FMfd}$. This implies a natural equivalence $\mathsf{Sh}(\mathsf{FCartSp})\simeq \mathsf{Sh}(\mathsf{FMfd})$, which makes the definition in the reference equivalent to definition \ref{def:formal_smooth_set}.
} to the original one provided by \cite{Dubuc1979}.
Since this is a category of sheaves on a site, it is naturally a topos, which is known as \textit{Cahiers topos} after the reference. 

\begin{definition}[Formal smooth stacks]
\textit{Formal smooth stacks} are defined  as stacks on the site of formal smooth manifolds $\mathsf{FMfd}$. The $(\infty,1)$-category of formal smooth stacks is, then, defined by
\begin{equation}
\begin{aligned}
        \mathbf{FSmoothStack} \;&\coloneqq\; \mathbf{St}(\mathsf{FMfd}) \\[0.2ex]
        \;&\,=\; \mathbf{N}_{hc}(\mathsf{sPreSh}(\mathsf{FMfd})_{\mathsf{proj,loc}}^\circ).
\end{aligned}
\end{equation}
\end{definition}

\begin{construction}[Diagram of sites]\label{rem:diagram_of_sites}
By combining adjunctions \ref{rem:reducedalgebras} and \eqref{rem:reducedalgebras2} with functors \ref{rem:embedding_manifolds} and \eqref{eq:embedding_formalmanifolds},
we have the following commuting diagram of ordinary sites:
\begin{equation}
    \begin{tikzcd}[column sep=14.2ex, row sep=8.0ex]
     \mathsf{Mfd} \arrow[dd, ""{name=0, anchor=center, inner sep=0}, shift left=-1ex, hook] \arrow[r, "\iota^{\mathsf{Mfd}}", hook] & \mathsf{C^\infty Var}^\mathrm{red} \arrow[dd, ""{name=2, anchor=center, inner sep=0}, shift left=-1ex, hook]       \\
     & \\
     \mathsf{FMfd} \arrow[uu, ""{name=1, anchor=center, inner sep=0}, shift left=-1ex]\arrow[r, "\iota^{\mathsf{FMfd}}", hook] & \mathsf{C^\infty Var}\arrow[uu, ""{name=3, anchor=center, inner sep=0}, shift left=-1ex]  
     \arrow["\dashv"{anchor=center}, draw=none, from=0, to=1]\arrow["\dashv"{anchor=center}, draw=none, from=2, to=3]
    \end{tikzcd}
\end{equation}
\end{construction}

Given the diagram of sites presented in construction \ref{rem:diagram_of_sites}, we could be tempted to extend the notions of formal smooth sets and formal smooth stacks, which we defined above.

\begin{definition}[Extended smooth sets and stacks]\label{def:extended_formal_smooth}
Let us give the following definitions:
\begin{itemize}
    \item We define the $1$-category \textit{extended smooth sets} as the $1$-category of sheaves on the site of reduced $\Coo$-varieties, i.e.
    \begin{equation}
    \mathsf{SmoothSet}^{+} \;\coloneqq\; \mathsf{Sh}(\mathsf{C^\infty Var}^\mathrm{red}).
    \end{equation}
    \item We define the $1$-category \textit{extended formal smooth sets} as the $1$-category of sheaves on the site of $\Coo$-varieties, i.e.
    \begin{equation}
    \mathsf{FSmoothSet}^{+} \;\coloneqq\; \mathsf{Sh}(\mathsf{C^\infty Var}).
    \end{equation}
    \item We define the $(\infty,1)$-category of \textit{extended smooth stacks} as the $(\infty,1)$-category of stacks on the site of reduced $\Coo$-varieties, i.e.
    \begin{equation}
    \mathbf{SmoothStack}^{\pmb{+}} \;\coloneqq\; \mathbf{St}(\mathsf{C^\infty Var}^\mathrm{red}).
    \end{equation}
    \item We define the $(\infty,1)$-category of \textit{extended formal smooth stacks} as the $(\infty,1)$-category of stacks on the site of $\Coo$-varieties, i.e.
    \begin{equation}
    \mathbf{FSmoothStack}^{\pmb{+}} \;\coloneqq\; \mathbf{St}(\mathsf{C^\infty Var}).
    \end{equation}
\end{itemize}
\end{definition}

\begin{remark}[Embeddings]\label{rem: embeddings}
Since the definitions \ref{def:extended_formal_smooth} are given on the sites of diagram \ref{rem:diagram_of_sites}, we can obtain a diagram of $(\infty,1)$-categories
\begin{equation}
\begin{tikzcd}[column sep=-4.0ex, row sep=8.0ex]
	& \mathbf{N}({\mathsf{SmoothSet}}) && \mathbf{N}({\mathsf{SmoothSet}^+}) \\
	{\mathbf{SmoothStack}} && {\mathbf{SmoothStack}^{\pmb{+}}} \\
	& \mathbf{N}({\mathsf{FSmoothSet}}) && \mathbf{N}({\mathsf{FSmoothSet}^+}) \\
	{\mathbf{FSmoothStack}} && {\mathbf{FSmoothStack}^{\pmb{+}}}
	\arrow[hook', from=1-2, to=2-1]
	\arrow[hook, from=2-1, to=4-1]
	\arrow[hook', from=1-2, to=3-2]
	\arrow[hook', from=3-2, to=4-1]
	\arrow[hook, from=1-2, to=1-4]
	\arrow[hook, from=3-2, to=3-4]
	\arrow[hook', from=1-4, to=3-4]
	\arrow[hook, from=4-1, to=4-3]
	\arrow[hook', from=3-4, to=4-3]
	\arrow[hook', from=1-4, to=2-3]
	\arrow[hook, from=2-1, to=2-3]
	\arrow[hook, from=2-3, to=4-3]
\end{tikzcd}
\end{equation}
\end{remark}

%%%%%%%%%%%%%%%%%%%%%%%%%%%%%%%%%%%%%%%%%%%%%%%%%%%%%%%%%%%%%%%%%%%%%%%%%%%%%%%%%%%%%%%%%%%%%
\section{Formal derived smooth stacks}\label{sec:fdsmoothstack}

In this section we will propose a definition for the notion of formal derived smooth stack. 
Our construction of formal derived smooth stacks is related to \cite{Wallbridge:2016} and to the research program by \cite{Grady:2014oqa, Grady:2016LIEAA, Grady:2020}.

In the previous two sections we considered at most stacks on ordinary sites, such as smooth stacks on the site of smooth manifolds.
In principle, it is possible to generalise the construction of stacks to the case where the site $\mathsf{C}$ itself is a simplicial category -- usually, presenting some $(\infty,1)$-category. Consider a simplicially-enriched category $\mathsf{C}$ equipped with the structure of a simplicial-site, i.e. such that its homotopy category $\mathrm{Ho}(\mathsf{C})$ has the structure of a site. Recall that, given two simplicially-enriched categories $\mathsf{C}$ and $\mathsf{D}$, the functor category $[\mathsf{C}^\op,\mathsf{D}]$ is naturally a simplicially-enriched category. 
In particular, we can define the simplicial-category of presheaves $[\mathsf{C}^\op,\sSet]$ on the simplicial-site $\mathsf{C}$. 
By \cite[Theorem 3.4.1]{ToenVezzo05}, for suitable simplicial-sites, there is still a notion of local projective simplicial model category structure $[\mathsf{C}^\op,\sSet]_{\mathsf{proj,loc}}$ that allows us to define the simplicial-category of \textit{derived stacks} on $\mathsf{C}$ by
\begin{equation}
    \Shoo(\mathsf{C}) \;\simeq\; [\mathsf{C}^\op,\sSet]_{\mathsf{proj,loc}}^\circ.
\end{equation}
Finally, by applying the homotopy coherent nerve functor on such a simplicial category, it is possible to obtain the $(\infty,1)$-category of derived stacks on $\mathsf{C}$, i.e. 
\begin{equation}\label{eq:gentheory_derived}
    \St(\mathsf{C}) \;\coloneqq\; \mathbf{N}_{hc}([\mathsf{C}^\op,\sSet]_{\mathsf{proj,loc}}^\circ).
\end{equation}
In this section, we will introduce the $(\infty,1)$-site of formal derived smooth manifolds, we will equip it with the structure of a site and we we will construct derived stacks on it: these will be the $(\infty,1)$-category of formal derived smooth stacks.

%%%%%%%%%%%%%%%%%%%%%%%%%%%%%%%%%%%%%%%%%%%%%%%%%%%%%%%%%%%%%%%%%%%%%%%%%%%%%%%%%%%%%%%%%%%%%
\subsection{Homotopy $\Coo$-algebras}

Let $\mathsf{T}$ be a generic Lawvere theory, as we reviewed at the beginning of section \ref{sec:Coo}.
As suggested first by \cite{Quillen:1967ha}, we can consider the simplicial category $[\Delta^{\op},\mathsf{TAlg}]$ of simplicial $\mathsf{T}$-algebras, where $\Delta$ is the simplex category.
By \cite[Section 2.4]{Quillen:1967ha} this can be equipped with a natural model structure, known as projective model structure.
The following model category can be called category of strict simplicial $\mathsf{T}$-algebras:
\begin{equation}
    \mathsf{sTAlg} \;\coloneqq\; [\Delta^{\op},\mathsf{TAlg}]_{\mathrm{proj}},
\end{equation}
where weak equivalences and fibrations are given object-wise.
In fact, the fibrant-cofibrant simplicial $\mathsf{T}$-algebras according to this model structure are known as strict simplicial $\mathsf{T}$-algebras in the literature.
By following \cite{badzioch2002algebraic}, there is a Quillen equivalence $\mathsf{sTAlg}\simeq_{\mathrm{Qu}}[\mathsf{T},\sSet]_{\mathrm{proj,loc}}$ between the model category above and the local projective model structure on the simplicial category of pre-cosheaves on $\mathsf{T}$. Fibrant-cofibrant objects in the latter model category are known in the literature as homotopy $\mathsf{T}$-algebras and they are given as follows.

\begin{definition}[Homotopy $\mathsf{T}$-algebra]
A \textit{homotopy algebra over a Lawvere theory} $\mathsf{T}$ is a functor
\begin{equation}
    A\,:\;\mathsf{T} \,\longrightarrow\, \sSet
\end{equation}
valued in Kan complexes, such that for any $\bbR^n\in\mathrm{CartSp}$ the canonical morphism 
\begin{equation}
    \bigsqcup_{i=1}^n A(\mathrm{prod}_i):\, A(\bbR^n)\,\xrightarrow{\;\;\simeq\;\;}\,A(\bbR)^n
\end{equation}
is a weak equivalence of simplicial sets.
\end{definition}

By the Quillen equivalence above, any homotopy $\mathsf{T}$-algebra is equivalent to a strict simplicial $\mathsf{T}$-algebra and both the model categories provide a model for the same $(\infty,1)$-category, which we will denote by $\mathbf{sT Alg}$.
This $(\infty,1)$-category $\mathbf{sT Alg}$ of homotopy $\mathsf{T}$-algebras can be constructed by applying the homotopy-coherent nerve to the simplicial category of fibrant-cofibrant objects, namely by
\begin{equation}
    \mathbf{sT Alg} \;=\; \mathbf{N}_{hc}([\Delta^{\op},\mathsf{TAlg}]_{\mathrm{proj}}^\circ).
\end{equation}

Now, we can specify $\mathsf{T}=\mathsf{CartSp}$ to be the Lawvere theory of $\Coo$-algebras, as in section \ref{sec:Coo}. 
Thus, a \textit{simplicial $\Coo$-algebra} is going to be defined as simplicial algebra over the Lawvere theory $\mathsf{CartSp}$.
Accordingly, we can define the model category of {simplicial $\Coo$-algebras} by
\begin{equation}
    \mathsf{sC^\infty Alg} \;=\; [\Delta^{\op},\mathsf{C^\infty Alg}]_{\mathrm{proj}}.
\end{equation}
A fibrant-cofibrant element of the model category $\mathsf{sC^\infty Alg}$ is precisely a homotopy $\Coo$-algebra.
The corresponding $(\infty,1)$-category of homotopy $\Coo$-algebras is given by the homotopy coherent nerve of the simplicial category $[\Delta^{\op},\mathsf{C^\infty Alg}]_{\mathrm{proj}}^\circ$ of fibrant-cofibrant objects.

\begin{definition}[$(\infty,1)$-category of homotopy $\Coo$-algebras]
The \textit{$(\infty,1)$-category of homotopy $\Coo$-algebras} is defined by
\begin{equation}
    \mathbf{sC^\infty Alg} \;=\; \mathbf{N}_{hc}([\Delta^{\op},\mathsf{C^\infty Alg}]_{\mathrm{proj}}^\circ).
\end{equation}
\end{definition}

Crucially, the $(\infty,1)$-category of homotopy $\Coo$-algebras can be naturally equipped with a $\Coo$-version of a derived tensor product which is going to be very relevant for geometric reasons. Recall that homotopy pushouts exist; see e.g. \cite{Joyce:2009}.

\begin{definition}[Derived $\Coo$-tensor product]
We define the \textit{derived $\Coo$-tensor product} in the category $\mathbf{sC^\infty Alg}$ by the homotopy pushout
\begin{equation}
    A \,\widehat{\otimes}_{C}^\bfL\, B \;\simeq \; A \sqcup_{C}^h B
\end{equation}
for any homotopy $\Coo$-algebras $A,B,C\in\mathbf{sC^\infty Alg}$.
\end{definition}

It is known that an ordinary $\Coo$-algebra $A$ is finitely presented precisely if its co-Yoneda embedding $\Hom(A,-):\mathsf{C^\infty Alg}\longrightarrow \mathsf{Set}$ preserves filtered colimits (see e.g. \cite{adámek_rosický_vitale_lawvere_2010}).
In \cite{Carchedi2019OnTU}, homotopically finitely presented $\Coo$-algebras are defined by generalising this statement to homotopy $\Coo$-algebras as follows.

\begin{definition}[Homotopically finitely presented $\Coo$-algebra]
A \textit{homotopically finitely presented $\Coo$-algebra} is defined as a homotopy $\Coo$-algebra $A\in\mathbf{sC^\infty Alg}$ such that it is a compact object in the $(\infty,1)$-category $\mathbf{sC^\infty Alg}$, i.e. such that the co-Yoneda $(\infty,1)$-functor $\Hom(A,-):\mathbf{sC^\infty Alg}\longrightarrow \mathbf{\infty Grpd}$ preserves filtered $(\infty,1)$-colimits.
The $(\infty,1)$-category of homotopically finitely presented $\Coo$-algebras $\mathbf{sC^\infty Alg}_\mathrm{fp} \longhookrightarrow \mathbf{sC^\infty Alg}$ is defined as the full subcategory on those objects which are homotopically finitely presented $\Coo$-algebras.
\end{definition}

In analogy with \cite{ToenVezzo05}, in the rest of the paper we will denote by $\mathsf{sC^\infty Alg}_\mathrm{fp}\hookrightarrow\mathsf{sC^\infty Alg}$ the model sub-category on those objects whose derived co-Yoneda functor preserves filtered homotopy colimits, so that we have $\mathbf{sC^\infty Alg}_\mathrm{fp} \simeq \mathbf{N}_{hc}(\mathsf{sC^\infty Alg}_\mathrm{fp}^\circ)$.

Now, as stressed by \cite{toen2007moduli}, being finitely presented is quite a stringent condition on a homotopy $\Coo$-algebra.
In analogy with the non-derived case, we can introduce the notion of homotopically finitely generated $\Coo$-algebra\footnote{In a previous version of this paper we proposed a faulty definition of almost finitely presented $\Coo$-algebra. We would like to thank Pelle Steffens for pointing out this issue.}.

\begin{definition}[Finitely generated $\Coo$-algebra]
A \textit{finitely generated $\Coo$-algebra} is defined as a homotopy $\Coo$-algebra $A\in\mathbf{sC^\infty Alg}$ such that $\pi_0A$ is finitely generated as an ordinary $\Coo$-algebra.
The $(\infty,1)$-category of finitely generated $\Coo$-algebras $\mathbf{sC^\infty Alg}_\mathrm{fg} \longhookrightarrow \mathbf{sC^\infty Alg}$ is defined as the full sub-category on those objects which are finitely generated $\Coo$-algebras.
\end{definition}

\begin{remark}[Finitely presented $\Coo$-algebras are finitely generated]
We have the following full sub-$(\infty,1)$-categories of homotopy $\Coo$-algebras:
\begin{equation}
    \mathbf{sC^\infty Alg}_\mathrm{fp} \,\longhookrightarrow\, \mathbf{sC^\infty Alg}_\mathrm{fg} \,\longhookrightarrow\, \mathbf{sC^\infty Alg}.
\end{equation}
In fact, by \cite[Proposition 3.27]{carchedi2023derived} any homotopically finitely presented simplicial $\Coo$-algebra $A\in\mathbf{sC^\infty Alg}_\mathrm{fp}$ has a $0$-th truncation $\pi_0A$ which is, in particular, a finitely presented $\Coo$-algebra in the ordinary sense.
\end{remark}

%%%%%%%%%%%%%%%%%%%%%%%%%%%%%%%%%%%%%%%%%%%%%%%%%%%%%%%%%%%%%%%%%%%%%%%%%%%%%%%%%%%%%%%%%%%%%
\subsection{Formal derived smooth manifolds}

In this subsection, we will introduce the $(\infty,1)$-category of formal derived smooth manifolds and we explore some of its entailments.
A formal derived smooth manifold is a slight generalisation of the notion of derived manifold \'a la Spivak \cite{Spivak:2010} and Carchedi-Steffens \cite{Carchedi2019OnTU}. 
Other relevant references on derived manifolds include \cite{Borisov:2011, Borisov:2012derived, Joyce:2012dmanifolds, vogler2013derived, Joyce:2014, joyce2017kuranishi, zeng2022derived}. Moreover, during the final stage of the preparation of this paper, the systematic foundational work of \cite{steffens2023derived} for the geometry of derived $\Coo$-schemes appeared.
Derived manifolds are a categorifications of smooth manifolds which are designed to crucially generalise the ordinary concept of intersection of smooth manifolds. In contrast to its ordinary counterpart, this \textit{derived} intersection always comes with a natural smooth structure.

Let us investigate more the core issue with intersections of smooth manifolds.
Let $M$ be an ordinary smooth manifold and $\Sigma,\Sigma'\subset M$ two smooth submanifolds of $M$. One would be tempted to categorically define the intersection of these submanifolds by the pullback $\Sigma\cap \Sigma' = \Sigma\times_M\Sigma'$. However, this definition generally fails, since the intersection may not a smooth manifold. More precisely, this happens if the embeddings $\Sigma,\Sigma'\hookrightarrow M$ are not transversal embeddings. 
To have a concrete example in mind, the reader can look at figure \ref{fig:intersection_example}.

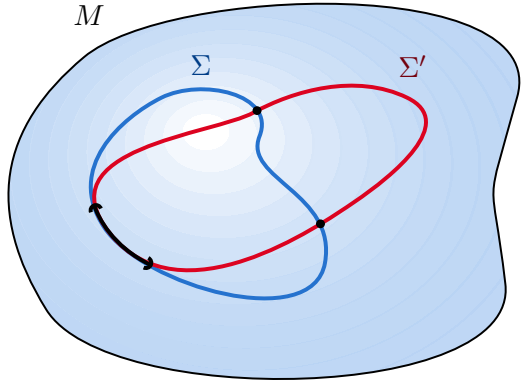
\begin{figure}[h!]
    \centering
% Gradient Info
\tikzset {_5qzrs5jsg/.code = {\pgfsetadditionalshadetransform{ \pgftransformshift{\pgfpoint{89.1 bp } { -108.9 bp }  }  \pgftransformscale{1.32 }  }}}
\pgfdeclareradialshading{_q83t3mpr0}{\pgfpoint{-72bp}{88bp}}{rgb(0bp)=(1,1,1);
rgb(0bp)=(1,1,1);
rgb(25bp)=(0.29,0.56,0.89);
rgb(400bp)=(0.29,0.56,0.89)}
\tikzset{_xluw9z1ml/.code = {\pgfsetadditionalshadetransform{\pgftransformshift{\pgfpoint{89.1 bp } { -108.9 bp }  }  \pgftransformscale{1.32 } }}}
\pgfdeclareradialshading{_2h5ctahj3} { \pgfpoint{-72bp} {88bp}} {color(0bp)=(transparent!0);
color(0bp)=(transparent!0);
color(25bp)=(transparent!71);
color(400bp)=(transparent!71)} 
\pgfdeclarefading{_1cdquizzy}{\tikz \fill[shading=_2h5ctahj3,_xluw9z1ml] (0,0) rectangle (50bp,50bp); } 
\tikzset{every picture/.style={line width=0.75pt}} %set default line width to 0.75pt        
\begin{tikzpicture}[x=0.75pt,y=0.75pt,yscale=-1,xscale=1]
%Shape: Polygon Curved [id:ds5031398446608062] 
\path  [shading=_q83t3mpr0,_5qzrs5jsg,path fading= _1cdquizzy ,fading transform={xshift=2}] (62.83,58.08) .. controls (105.83,22.08) and (298.33,16.08) .. (281.33,78.08) .. controls (264.33,140.08) and (268.33,112.58) .. (275.33,172.58) .. controls (282.33,232.58) and (74.83,235.58) .. (43.83,185.58) .. controls (12.83,135.58) and (19.83,94.08) .. (62.83,58.08) -- cycle ; % for fading 
 \draw   (62.83,58.08) .. controls (105.83,22.08) and (298.33,16.08) .. (281.33,78.08) .. controls (264.33,140.08) and (268.33,112.58) .. (275.33,172.58) .. controls (282.33,232.58) and (74.83,235.58) .. (43.83,185.58) .. controls (12.83,135.58) and (19.83,94.08) .. (62.83,58.08) -- cycle ; % for border 
%Shape: Polygon Curved [id:ds43255923869876134] 
\draw  [color={rgb, 255:red, 37; green, 114; blue, 206 }  ,draw opacity=1 ][line width=1.5]  (101.83,78.58) .. controls (121.83,68.58) and (159.83,76.58) .. (151.33,98.08) .. controls (142.83,119.58) and (190.67,126.67) .. (184.33,162.08) .. controls (178,197.5) and (93,174) .. (73,144) .. controls (53,114) and (81.83,88.58) .. (101.83,78.58) -- cycle ;
%Shape: Polygon Curved [id:ds3660180267503006] 
\draw  [color={rgb, 255:red, 220; green, 3; blue, 32 }  ,draw opacity=1 ][line width=1.5]  (148.33,86.08) .. controls (168.33,75.08) and (200,66) .. (225.83,78.58) .. controls (251.67,91.17) and (218.5,118) .. (199.33,131.08) .. controls (180.17,144.17) and (111.83,191.08) .. (76.33,147.58) .. controls (40.83,104.08) and (128.33,97.08) .. (148.33,86.08) -- cycle ;
%Shape: Ellipse [id:dp46357055826672333] 
\draw  [color={rgb, 255:red, 0; green, 0; blue, 0 }  ,draw opacity=1 ][fill={rgb, 255:red, 0; green, 0; blue, 0 }  ,fill opacity=1 ] (148.33,85.08) .. controls (148.33,84.14) and (149.12,83.38) .. (150.09,83.38) .. controls (151.06,83.38) and (151.85,84.14) .. (151.85,85.08) .. controls (151.85,86.02) and (151.06,86.79) .. (150.09,86.79) .. controls (149.12,86.79) and (148.33,86.02) .. (148.33,85.08) -- cycle ;
%Shape: Ellipse [id:dp02651034325501156] 
\draw  [color={rgb, 255:red, 0; green, 0; blue, 0 }  ,draw opacity=1 ][fill={rgb, 255:red, 0; green, 0; blue, 0 }  ,fill opacity=1 ] (180.33,142.23) .. controls (180.33,141.29) and (181.12,140.52) .. (182.09,140.52) .. controls (183.06,140.52) and (183.85,141.29) .. (183.85,142.23) .. controls (183.85,143.17) and (183.06,143.93) .. (182.09,143.93) .. controls (181.12,143.93) and (180.33,143.17) .. (180.33,142.23) -- cycle ;
%Curve Lines [id:da019671310812396925] 
\draw [color={rgb, 255:red, 0; green, 0; blue, 0 }  ,draw opacity=1 ][line width=1.5]    (67.9,132.95) .. controls (70.19,139.05) and (76.76,153.33) .. (96.48,163.05) ;
\draw [shift={(96.48,163.05)}, rotate = 206.23] [color={rgb, 255:red, 0; green, 0; blue, 0 }  ,draw opacity=1 ][line width=1.5]      (2.68,-2.68) .. controls (1.2,-2.68) and (0,-1.48) .. (0,0) .. controls (0,1.48) and (1.2,2.68) .. (2.68,2.68) ;
\draw [shift={(67.9,132.95)}, rotate = 69.44] [color={rgb, 255:red, 0; green, 0; blue, 0 }  ,draw opacity=1 ][line width=1.5]      (2.68,-2.68) .. controls (1.2,-2.68) and (0,-1.48) .. (0,0) .. controls (0,1.48) and (1.2,2.68) .. (2.68,2.68) ;
% Text Node
\draw (56,30.4) node [anchor=north west][inner sep=0.75pt]    {$M$};
% Text Node
\draw (115,55.4) node [anchor=north west][inner sep=0.75pt]    {$\textcolor[rgb]{0,0.24,0.53}{\Sigma }$};
% Text Node
\draw (220,55.4) node [anchor=north west][inner sep=0.75pt]    {$\textcolor[rgb]{0.45,0,0.05}{\Sigma '}$};
\end{tikzpicture}
    \caption{Example of non-transverse intersection of smooth submanifolds $\Sigma,\Sigma'\subset M$.}
    \label{fig:intersection_example}
\end{figure}

Let us now explore an interesting example more in detail.

\begin{example}[Intersection is not locally homeomorphic to a Cartesian space]\label{ex:intersection}
Consider the ordinary smooth manifolds $\Sigma,\Sigma' \,=\, \bbR^2$, and  $M \,=\, \bbR^3$, together with embeddings $e_\Sigma:\Sigma\hookrightarrow \bbR^3$ and $e_{\Sigma'}:\Sigma'\hookrightarrow \bbR^3$ given by the maps
\begin{equation*}
    e_\Sigma \,:\,  (x,y) \mapsto (x,y, x^2y^2), \quad\;\; e_{\Sigma'} \,:\,  (x,y) \mapsto (x,y, 0).
\end{equation*}
As a set, the intersection of these two submanifolds is $\{(x,y,0)\in\bbR^3\,|\,x^2y^2=0\}$, which is precisely the union of the line $\{(x,0,0)\in\bbR^3\,|\,x\in\bbR\}$ and the line $\{(0,y,0)\in\bbR^3\,|\,y\in\bbR\}$. This cross-shaped subset of $\bbR^3$ is clearly not locally homeomorphic to $\bbR$ and, therefore, it does not allow the structure of an ordinary smooth manifold.
\end{example}

To make sense of arbitrary intersections of smooth manifolds, we need to introduce the concept of a \textit{derived manifold}. We will exploit the following proposition by \cite[Corollary 5.4]{Carchedi2019OnTU}.

\begin{proposition}[Derived manifolds \cite{Carchedi2019OnTU}]
There is a natural equivalence of $(\infty,1)$-categories
\begin{equation}\label{eq:Carcheditheorem}
    \mathbf{dMfd} \; \simeq\; \mathbf{sC^{\infty\!}Alg}^\op_{\mathrm{fp}}
\end{equation}
between the $(\infty,1)$-category $\mathbf{dMfd}$ of derived manifolds, and the opposite of the $(\infty,1)$-category $\mathbf{sC^{\infty\!}Alg}_{\mathrm{fp}}$ of homotopically finitely presented homotopy $\Coo$-algebras.
\end{proposition}

In this paper we will regard the equivalence \eqref{eq:Carcheditheorem} as an effective definition of derived manifolds. 
However, we will need a slight generalisation of the notion of derived manifold. 
In fact, as stressed by \cite{toen2007moduli}, being homotopically finitely presented is much more a stringent notion than being finitely presented in the ordinary sense.
In fact, in general, ordinary finitely presented $\Coo$-algebras $A\in\mathsf{C^\infty Alg}_{\mathrm{fp}}$ such as Weil algebras do not embed into $\mathbf{C^{\infty\!} Alg}_{\mathrm{fp}}$.
For this reason, in analogy with the discussion of \cite[Section 2]{Calaque:2017} in the context of algebraic geometry, we give the following definition.

\begin{definition}[Formal derived smooth manifolds]
We define the $(\infty,1)$-category of \textit{formal derived smooth manifolds} by
\begin{equation}
    \mathbf{dFMfd} \; \coloneqq\; \mathbf{sC^{\infty\!}Alg}^\op_{\mathrm{fg}},
\end{equation}
where $\mathbf{sC^{\infty\!}Alg}^\op_{\mathrm{fg}}$ is the $(\infty,1)$-category of finitely generated $\Coo$-algebras.
\end{definition}

In analogy with the ordinary case from the previous section, these may also be thought of as derived $\Coo$-varieties.

\begin{remark}[Intuitive picture of formal derived smooth manifolds]
At an intuitive level, a formal derived smooth manifold $U\in \mathbf{dFMfd}$ is a geometric object whose algebra of smooth function is, by definition, a homotopically finitely presented homotopy $\Coo$-algebra modelled by some simplicial object
\begin{equation}
    \mathcal{O}(U)\;=\; \left( \begin{tikzcd}[row sep=scriptsize, column sep=5ex]
    \; \cdots\; \arrow[r, yshift=1.4ex] \arrow[r, yshift=2.8ex] \arrow[r] \arrow[r, yshift=-1.4ex]\arrow[r, yshift=-2.8ex] & \mathcal{O}(U)_3 \arrow[r, yshift=1.8ex]\arrow[r, yshift=0.6ex]\arrow[r, yshift=-1.8ex]\arrow[r, yshift=-0.6ex]& \mathcal{O}(U)_2
    \arrow[r, yshift=1.4ex] \arrow[r] \arrow[r, yshift=-1.4ex] & \mathcal{O}(U)_1  \arrow[r, yshift=0.7ex] \arrow[r, yshift=-0.7ex] & \mathcal{O}(U)_0 
    \end{tikzcd}    \right),
\end{equation}
where each $\mathcal{O}(U)_k$ is an ordinary $\Coo$-algebra.   
\end{remark}

\begin{remark}[Derived-extension of ordinary smooth manifolds]
According to the construction by \cite{Carchedi2019OnTU}, we have a natural fully faithful functor $\mathbf{N}(\mathsf{Mfd}) \longrightarrow\mathbf{dMfd}$ which embeds ordinary smooth manifolds into derived manifolds which preserves transverse pullbacks and the terminal object. Thus, by composition with the embedding $\mathbf{dMfd}\hookrightarrow \mathbf{dFMfd}$, we naturally have a fully faithful functor
\begin{equation}
    i\,:\; \mathbf{N}(\mathsf{Mfd}) \,\longrightarrow\, \mathbf{dFMfd}
\end{equation}
preserving transverse pullbacks and the terminal object.
From now on, we will call this functor \textit{derived-extension} of ordinary smooth manifolds. 
\end{remark}

\begin{remark}[Homotopy pullback of ordinary smooth manifolds]
Given a pair of smooth maps $f:M\rightarrow B$ and $g:N\rightarrow B$ of smooth manifolds $M,N,B\in\mathsf{Mfd}$, we can consider the formal derived smooth manifold given by the $(\infty,1)$-pullback
\begin{equation}
    \begin{tikzcd}[row sep={14.5ex,between origins}, column sep={16ex,between origins}]
    i(M)\times_{i(B)\!}^hi(N) \arrow[d]\arrow[r] & i(M) \arrow[d, "i(f)"] \\
    i(N) \arrow[r, "i(g)"] & i(B).
    \end{tikzcd}
\end{equation}
Only if  $f$ and $g$ are transverse smooth maps in $\Mfd$, the ordinary pullback $M\times_B N$ exists in $\Mfd$ and so there is a natural morphism of formal derived smooth manifolds
\begin{equation}
   i(M\times_B N)    \;\xrightarrow{\;\;\simeq\;\;}\;      i(M)\times^h_{i(B)\!}i(N), 
\end{equation}
which is, in particular, an equivalence. In other words, the derived-extension functor preserves transverse pullbacks.
As a corollary, we can notice that for any pair of smooth manifolds $M$ and $N$, we have the equivalence of formal derived smooth manifolds
\begin{equation}
   i(M \times N) \;\xrightarrow{\;\;\simeq\;\;}\;  i(M)\times^h i(N), 
\end{equation}
which means that the functor $i$ preserves finite products.
On the other hand, if $f$ and $g$ are not transverse smooth maps in $\Mfd$, then the ordinary pullback $M\times_B N$ does not exists in $\Mfd$. The great power of formal derived smooth manifolds comes from the fact that the homotopy pullback $i(M)\times^h_{i(B)\!}i(N)$ in $\mathbf{dFMfd}$ always exists.
\end{remark}

\begin{remark}[Derived intersection of smooth manifolds]
Let $\Sigma,\Sigma'\subset M$ be two smooth submanifolds of $M$. As we have seen, if we try to define their intersection by the pullback $\Sigma\cap \Sigma' = \Sigma\times_M\Sigma'$, this definition fails whenever the embeddings $\Sigma,\Sigma'\hookrightarrow M$ are not transversal. 
However, we have seen that in the $(\infty,1)$-category of formal derived smooth manifolds, homotopy pullbacks always exist.
Thus, we can embed our diagram of smooth manifolds $\Sigma\rightarrow M\leftarrow \Sigma'$ into a diagram of formal derived smooth manifolds $i(\Sigma)\rightarrow i(M)\leftarrow i(\Sigma)$. 
Then, we can call \textit{derived intersection} of the smooth manifolds $\Sigma$ and $\Sigma'$ in $M$ their homotopy pullback $i(\Sigma)\times_{i(M)\!}^hi(\Sigma')$. Crucially, the derived intersection is always a well-defined formal derived smooth manifold.
\end{remark}

A notational warning: whenever it is clear from the context, we will tend to omit the symbol of the embedding $i$ and simply write $\Sigma\times_{M}^h\Sigma'$ to mean the derived intersection $i(\Sigma)\times_{i(M)\!}^hi(\Sigma')$ of ordinary smooth manifolds.

Now, since the $(\infty,1)$-category of formal derived smooth manifolds satisfies the equivalence $\mathbf{dFMfd} \simeq \mathbf{sC^\infty Alg}_{\mathrm{fp}}^\op$, we have that the homotopy $\Coo$-algebra $\mathcal{O}_{{}_{\!}}\big( \Sigma\times^h_{M\!}\Sigma' \big)$ of smooth functions on a homotopy pullback $\Sigma\times^h_{M\!}\Sigma'$ is given by the derived $\Coo$-tensor product of the corresponding ordinary $\Coo$-algebras, i.e.
\begin{equation}\label{eq:functions_derived_tensor_prod}
\begin{aligned}
     \mathcal{O}_{{}_{\!}}\big( \Sigma\times^h_{M\!}\Sigma' \big) \;&\simeq\; \Coo(\Sigma)\,\widehat{\otimes}_{\Coo(M)}^\bfL\,\Coo(\Sigma').
     \end{aligned}
\end{equation}

\begin{construction}[Computing the derived intersection of smooth manifolds]
Equivalence \eqref{eq:functions_derived_tensor_prod} suggests a practical way to compute the derived intersection of given smooth manifolds.
In fact, we can consider a cofibrant replacement $Q\Coo(\Sigma) \longrightarrow \Coo(\Sigma)$ in the co-slice category $\mathsf{sC^\infty Alg}_{\Coo(M)/}$ of homotopy $\Coo$-algebras over $\Coo(M)$ with respect to its model structure.
By replacing $\Coo(\Sigma)$ with a cofibrant replacement $Q\Coo(\Sigma)$ in equation \eqref{eq:functions_derived_tensor_prod}, we can compute the derived $\Coo$-tensor product as an ordinary $\Coo$-tensor product, namely we have
\begin{equation}
    \mathcal{O}_{{}_{\!}}\big( \Sigma\times^h_{M\!}\Sigma' \big) \;\simeq\; Q\Coo(\Sigma)\,\widehat{\otimes}_{\Coo(M)}\,\Coo(\Sigma').
\end{equation}
In principle, we may exploit the Bar construction $\mathrm{Bar}(\Coo(M),\Coo(\Sigma))\longrightarrow\Coo(\Sigma)$ to produce a suitable cofibrant replacement, but other methods may be available depending on the amount of structure.
The simplicial $\Coo$-algebra obtained by this $\Coo$-tensor product will be an explicit model of the wanted homotopy $\Coo$-algebra.
\end{construction}

\begin{example}[Back to previous example]
We look back at example \ref{ex:intersection}. Let us exploit the fact that $e_\Sigma$ and $e_{\Sigma'}$ are sections of the vector bundle $\pi:\bbR^3\rightarrow \bbR^2$, given by the obvious projection $(x,y,z)\mapsto(x,y)$.
We want to compute the derived $\Coo$-tensor product 
\begin{equation}
\begin{aligned}
     \mathcal{O}_{{}_{\!}}\big( \bbR^{2\!}\times^h_{\bbR^3\!}\bbR^2 \big) \;&\simeq\; \Coo(\bbR^2)\,\widehat{\otimes}_{\Coo(\bbR^3)}^\bfL\,\Coo(\bbR^2) \\[0.5ex]
     \;&\simeq\; Q\Coo(\bbR^2)\,\widehat{\otimes}_{\Coo(\bbR^3)}\,\Coo(\bbR^2),
     \end{aligned}
\end{equation}
by using some cofibrant replacement $Q\Coo(\bbR^2)\longrightarrow \Coo(\bbR^2)$ in the simplicial co-slice model category $\mathsf{sC^{\infty}Alg}_{\Coo(\bbR^3)/}$ of homotopy $\Coo$-algebras over $\Coo(\bbR^3)$. Such a homotopy $\Coo$-algebra must be a simplicial resolution of the ordinary $\Coo$-algebra $\Coo(\bbR^2)$. 
Now, let us consider the $\bbR$-algebra $B \coloneqq \Coo(\bbR^2,\bbR) \oplus \Gamma(\bbR^2,\bbR^3)$, where $\Coo(\bbR^2,\bbR)$ and $\Gamma(\bbR^2,\bbR^3)$ are respectively the vector space of functions on $\bbR^2$ and of sections of the bundle $\pi:\bbR^3\rightarrow\bbR^2$, and where the product given by $(f,\phi)\cdot(f',\phi') = (ff',f\phi'+f'\phi)$ for any $f,f'\in\Coo(\bbR^2)$ and $\phi,\phi'\in\Gamma(\bbR^2,\bbR^3)$. We can canonically equip the $\bbR$-algebra $b$ with the structure of a $\Coo$-algebra by the pre-cosheaf $\widehat{B}: \bbR^k\mapsto \Hom_{\mathsf{Alg}_\bbR\!}(\Coo(\bbR^k)^\mathrm{alg},B)$ on Cartesian spaces.
Let us then try with the following:
\begin{equation}
        Q\Coo( \bbR^{2} )_n \;=\; \begin{cases}
      \Coo(\bbR^3), & n=0 \\
      \Coo(\bbR^3)\,\widehat{\otimes}_{\Coo(\bbR^2)}\,\widehat{B}, & n>0,
    \end{cases}
\end{equation}
where we used the the fact that there is a pullback map $\pi^\ast:\Coo(\bbR^2)\rightarrow \Coo(\bbR^3)$.
So, the simplicial $\Coo$-algebra $Q\Coo( \bbR^{2} )$ will be truncated at $n=1$, which, more precisely, means that it is $1$-skeletal.
In fact, we construct a simplicial $\Coo(\bbR^3)$-algebra
\begin{equation*}
    Q\Coo( \bbR^{2} ) \;\simeq\; \mathrm{sk}_1\bigg( \begin{tikzcd}[row sep=scriptsize, column sep=3.8ex]
    \Coo(\bbR^3)\,\widehat{\otimes}_{\Coo(\bbR^2)}\,\widehat{B}\, \arrow[r, yshift=0.7ex, "\partial_0"] \arrow[r, yshift=-0.7ex, "\partial_1"'] & \Coo(\bbR^3)
    \end{tikzcd}    \bigg),
\end{equation*}
with face maps given by the morphisms
\begin{equation}
    \partial_0(1\otimes (f,\phi))= f, \qquad \partial_1(1\otimes (f,\phi))= f+ (z - x^2y^2)\phi,
\end{equation}
for any pair $f\in\Coo(\bbR^2)$ and $\phi\in\Gamma(\bbR^2,\bbR^3)$.
To see that this is indeed a cofibrant replacement of $\Coo( \bbR^{2} )$, notice that we have $\pi_0Q\Coo( \bbR^{2} ) = \Coo( \bbR^{3} )/(z-x^2y^2) \cong \Coo( \bbR^{2} )$ and $\pi_iQ\Coo( \bbR^{2} ) \cong 0$ for $i>0$.
Now we can compute the ordinary $\Coo$-tensor product
\begin{equation}
\begin{aligned}
     \mathcal{O}_{{}_{\!}}\big( \bbR^{2\!}\times^h_{\bbR^3\!}\bbR^2 \big) \;&\simeq\; Q\Coo(\bbR^2)\,\widehat{\otimes}_{\Coo(\bbR^3)}\,\Coo(\bbR^2),
     \end{aligned}
\end{equation}
which is given by the $1$-skeletal simplicial $\Coo$-algebra
\begin{equation}
    \mathcal{O}_{{}_{\!}}\big( \bbR^{2\!}\times^h_{\bbR^3\!}\bbR^2 \big) \;\simeq\;  \mathrm{sk}_1\bigg( \begin{tikzcd}[row sep=scriptsize, column sep=3.8ex] \widehat{B}  \arrow[r, yshift=0.7ex, "\partial_0"] \arrow[r, yshift=-0.7ex, "\partial_1"'] & \Coo(\bbR^2)
    \end{tikzcd}    \bigg),
\end{equation}
with face maps given by the morphisms
\begin{equation}
    \partial_0(f,\phi)= f, \qquad \partial_1(f,\phi)= f+ x^2y^2\phi.
\end{equation}
This provides a model for the homotopy $\Coo$-algebra of functions on the derived intersection $\bbR^{2\!}\times^h_{\bbR^3\!}\bbR^2$ of the smooth manifold in the example.

\begin{figure}[h!]
    \centering
\tikzset{every picture/.style={line width=0.75pt}} %set default line width to 0.75pt        
\begin{tikzpicture}[x=0.75pt,y=0.75pt,yscale=-1,xscale=1]
%uncomment if require: \path (0,367); %set diagram left start at 0, and has height of 367
%Straight Lines [id:da749673741055084] 
\draw [color={rgb, 255:red, 74; green, 144; blue, 226 }  ,draw opacity=1 ][fill={rgb, 255:red, 74; green, 144; blue, 226 }  ,fill opacity=0.3 ][line width=0.75]  [dash pattern={on 0.75pt off 0.75pt}]  (76.07,89.11) .. controls (78.24,88.2) and (79.79,88.83) .. (80.7,91) .. controls (81.61,93.17) and (83.15,93.8) .. (85.32,92.89) .. controls (87.49,91.98) and (89.04,92.61) .. (89.95,94.78) .. controls (90.86,96.95) and (92.41,97.59) .. (94.58,96.68) .. controls (96.75,95.77) and (98.3,96.4) .. (99.21,98.57) .. controls (100.12,100.74) and (101.67,101.37) .. (103.84,100.46) .. controls (106.01,99.55) and (107.56,100.18) .. (108.47,102.35) .. controls (109.38,104.52) and (110.93,105.15) .. (113.1,104.24) .. controls (115.27,103.33) and (116.81,103.96) .. (117.72,106.13) .. controls (118.63,108.3) and (120.18,108.93) .. (122.35,108.02) .. controls (124.52,107.11) and (126.07,107.74) .. (126.98,109.91) .. controls (127.89,112.08) and (129.44,112.71) .. (131.61,111.8) .. controls (133.78,110.89) and (135.33,111.53) .. (136.24,113.7) .. controls (137.15,115.87) and (138.7,116.5) .. (140.87,115.59) .. controls (143.04,114.68) and (144.59,115.31) .. (145.5,117.48) .. controls (146.41,119.65) and (147.95,120.28) .. (150.12,119.37) .. controls (152.29,118.46) and (153.84,119.09) .. (154.75,121.26) .. controls (155.66,123.43) and (157.21,124.06) .. (159.38,123.15) .. controls (161.55,122.24) and (163.1,122.87) .. (164.01,125.04) .. controls (164.92,127.21) and (166.47,127.84) .. (168.64,126.93) .. controls (170.81,126.02) and (172.36,126.65) .. (173.27,128.82) .. controls (174.18,130.99) and (175.73,131.63) .. (177.9,130.72) .. controls (180.07,129.81) and (181.61,130.44) .. (182.52,132.61) .. controls (183.43,134.78) and (184.98,135.41) .. (187.15,134.5) .. controls (189.32,133.59) and (190.87,134.22) .. (191.78,136.39) .. controls (192.69,138.56) and (194.24,139.19) .. (196.41,138.28) .. controls (198.58,137.37) and (200.13,138) .. (201.04,140.17) .. controls (201.95,142.34) and (203.5,142.97) .. (205.67,142.06) .. controls (207.84,141.15) and (209.39,141.78) .. (210.3,143.95) .. controls (211.21,146.12) and (212.75,146.75) .. (214.92,145.84) .. controls (217.09,144.93) and (218.64,145.56) .. (219.55,147.73) .. controls (220.46,149.9) and (222.01,150.54) .. (224.18,149.63) .. controls (226.35,148.72) and (227.9,149.35) .. (228.81,151.52) .. controls (229.72,153.69) and (231.27,154.32) .. (233.44,153.41) .. controls (235.61,152.5) and (237.16,153.13) .. (238.07,155.3) .. controls (238.98,157.47) and (240.53,158.1) .. (242.7,157.19) .. controls (244.87,156.28) and (246.42,156.91) .. (247.33,159.08) -- (249.23,159.86) -- (249.23,159.86)(74.93,91.89) .. controls (77.1,90.98) and (78.65,91.61) .. (79.56,93.78) .. controls (80.47,95.95) and (82.02,96.58) .. (84.19,95.67) .. controls (86.36,94.76) and (87.91,95.39) .. (88.82,97.56) .. controls (89.73,99.73) and (91.28,100.36) .. (93.45,99.45) .. controls (95.62,98.54) and (97.17,99.17) .. (98.08,101.34) .. controls (98.99,103.51) and (100.53,104.15) .. (102.7,103.24) .. controls (104.87,102.33) and (106.42,102.96) .. (107.33,105.13) .. controls (108.24,107.3) and (109.79,107.93) .. (111.96,107.02) .. controls (114.13,106.11) and (115.68,106.74) .. (116.59,108.91) .. controls (117.5,111.08) and (119.05,111.71) .. (121.22,110.8) .. controls (123.39,109.89) and (124.94,110.52) .. (125.85,112.69) .. controls (126.76,114.86) and (128.31,115.49) .. (130.48,114.58) .. controls (132.65,113.67) and (134.19,114.3) .. (135.1,116.47) .. controls (136.01,118.64) and (137.56,119.27) .. (139.73,118.36) .. controls (141.9,117.45) and (143.45,118.08) .. (144.36,120.25) .. controls (145.27,122.42) and (146.82,123.06) .. (148.99,122.15) .. controls (151.16,121.24) and (152.71,121.87) .. (153.62,124.04) .. controls (154.53,126.21) and (156.08,126.84) .. (158.25,125.93) .. controls (160.42,125.02) and (161.97,125.65) .. (162.88,127.82) .. controls (163.79,129.99) and (165.33,130.62) .. (167.5,129.71) .. controls (169.67,128.8) and (171.22,129.43) .. (172.13,131.6) .. controls (173.04,133.77) and (174.59,134.4) .. (176.76,133.49) .. controls (178.93,132.58) and (180.48,133.21) .. (181.39,135.38) .. controls (182.3,137.55) and (183.85,138.18) .. (186.02,137.27) .. controls (188.19,136.36) and (189.74,137) .. (190.65,139.17) .. controls (191.56,141.34) and (193.11,141.97) .. (195.28,141.06) .. controls (197.45,140.15) and (198.99,140.78) .. (199.9,142.95) .. controls (200.81,145.12) and (202.36,145.75) .. (204.53,144.84) .. controls (206.7,143.93) and (208.25,144.56) .. (209.16,146.73) .. controls (210.07,148.9) and (211.62,149.53) .. (213.79,148.62) .. controls (215.96,147.71) and (217.51,148.34) .. (218.42,150.51) .. controls (219.33,152.68) and (220.88,153.31) .. (223.05,152.4) .. controls (225.22,151.49) and (226.77,152.12) .. (227.68,154.29) .. controls (228.59,156.46) and (230.13,157.1) .. (232.3,156.19) .. controls (234.47,155.28) and (236.02,155.91) .. (236.93,158.08) .. controls (237.84,160.25) and (239.39,160.88) .. (241.56,159.97) .. controls (243.73,159.06) and (245.28,159.69) .. (246.19,161.86) -- (248.1,162.64) -- (248.1,162.64) ;
%Straight Lines [id:da19605405971360335] 
\draw [color={rgb, 255:red, 74; green, 144; blue, 226 }  ,draw opacity=1 ][fill={rgb, 255:red, 74; green, 144; blue, 226 }  ,fill opacity=0.31 ][line width=0.75]  [dash pattern={on 0.75pt off 0.75pt}]  (244.61,76.16) .. controls (244.04,78.45) and (242.62,79.31) .. (240.33,78.75) .. controls (238.04,78.18) and (236.62,79.04) .. (236.05,81.33) .. controls (235.49,83.62) and (234.07,84.48) .. (231.78,83.92) .. controls (229.49,83.36) and (228.07,84.22) .. (227.5,86.51) .. controls (226.93,88.8) and (225.51,89.66) .. (223.22,89.1) .. controls (220.93,88.54) and (219.51,89.4) .. (218.94,91.69) .. controls (218.37,93.98) and (216.95,94.84) .. (214.66,94.28) .. controls (212.37,93.71) and (210.95,94.57) .. (210.39,96.86) .. controls (209.82,99.15) and (208.4,100.01) .. (206.11,99.45) .. controls (203.82,98.89) and (202.4,99.75) .. (201.83,102.04) .. controls (201.26,104.33) and (199.84,105.19) .. (197.55,104.63) .. controls (195.26,104.07) and (193.84,104.93) .. (193.27,107.22) .. controls (192.71,109.51) and (191.29,110.37) .. (189,109.81) .. controls (186.71,109.24) and (185.29,110.1) .. (184.72,112.39) .. controls (184.15,114.68) and (182.73,115.54) .. (180.44,114.98) .. controls (178.15,114.42) and (176.73,115.28) .. (176.16,117.57) .. controls (175.59,119.86) and (174.17,120.72) .. (171.88,120.16) .. controls (169.59,119.6) and (168.17,120.46) .. (167.61,122.75) .. controls (167.04,125.04) and (165.62,125.9) .. (163.33,125.34) .. controls (161.04,124.77) and (159.62,125.63) .. (159.05,127.92) .. controls (158.48,130.21) and (157.06,131.07) .. (154.77,130.51) .. controls (152.48,129.95) and (151.06,130.81) .. (150.5,133.1) .. controls (149.93,135.39) and (148.51,136.25) .. (146.22,135.69) .. controls (143.93,135.13) and (142.51,135.99) .. (141.94,138.28) .. controls (141.37,140.57) and (139.95,141.43) .. (137.66,140.87) .. controls (135.37,140.3) and (133.95,141.16) .. (133.38,143.45) .. controls (132.82,145.74) and (131.4,146.6) .. (129.11,146.04) .. controls (126.82,145.48) and (125.4,146.34) .. (124.83,148.63) .. controls (124.26,150.92) and (122.84,151.78) .. (120.55,151.22) .. controls (118.26,150.66) and (116.84,151.52) .. (116.27,153.81) .. controls (115.7,156.1) and (114.28,156.96) .. (111.99,156.4) .. controls (109.7,155.83) and (108.28,156.69) .. (107.72,158.98) .. controls (107.15,161.27) and (105.73,162.13) .. (103.44,161.57) .. controls (101.15,161.01) and (99.73,161.87) .. (99.16,164.16) .. controls (98.59,166.45) and (97.17,167.31) .. (94.88,166.75) .. controls (92.59,166.19) and (91.17,167.05) .. (90.6,169.34) .. controls (90.04,171.63) and (88.62,172.49) .. (86.33,171.92) .. controls (84.04,171.36) and (82.62,172.22) .. (82.05,174.51) .. controls (81.48,176.8) and (80.06,177.66) .. (77.77,177.1) -- (75.61,178.41) -- (75.61,178.41)(243.06,73.59) .. controls (242.49,75.88) and (241.07,76.74) .. (238.78,76.18) .. controls (236.49,75.62) and (235.07,76.48) .. (234.5,78.77) .. controls (233.93,81.06) and (232.51,81.92) .. (230.22,81.36) .. controls (227.93,80.79) and (226.51,81.65) .. (225.95,83.94) .. controls (225.38,86.23) and (223.96,87.09) .. (221.67,86.53) .. controls (219.38,85.97) and (217.96,86.83) .. (217.39,89.12) .. controls (216.82,91.41) and (215.4,92.27) .. (213.11,91.71) .. controls (210.82,91.15) and (209.4,92.01) .. (208.83,94.3) .. controls (208.27,96.59) and (206.85,97.45) .. (204.56,96.89) .. controls (202.27,96.32) and (200.85,97.18) .. (200.28,99.47) .. controls (199.71,101.76) and (198.29,102.62) .. (196,102.06) .. controls (193.71,101.5) and (192.29,102.36) .. (191.72,104.65) .. controls (191.15,106.94) and (189.73,107.8) .. (187.44,107.24) .. controls (185.15,106.68) and (183.73,107.54) .. (183.17,109.83) .. controls (182.6,112.12) and (181.18,112.98) .. (178.89,112.42) .. controls (176.6,111.85) and (175.18,112.71) .. (174.61,115) .. controls (174.04,117.29) and (172.62,118.15) .. (170.33,117.59) .. controls (168.04,117.03) and (166.62,117.89) .. (166.05,120.18) .. controls (165.49,122.47) and (164.07,123.33) .. (161.78,122.77) .. controls (159.49,122.21) and (158.07,123.07) .. (157.5,125.36) .. controls (156.93,127.65) and (155.51,128.51) .. (153.22,127.95) .. controls (150.93,127.38) and (149.51,128.24) .. (148.94,130.53) .. controls (148.37,132.82) and (146.95,133.68) .. (144.66,133.12) .. controls (142.37,132.56) and (140.95,133.42) .. (140.39,135.71) .. controls (139.82,138) and (138.4,138.86) .. (136.11,138.3) .. controls (133.82,137.74) and (132.4,138.6) .. (131.83,140.89) .. controls (131.26,143.18) and (129.84,144.04) .. (127.55,143.48) .. controls (125.26,142.91) and (123.84,143.77) .. (123.27,146.06) .. controls (122.71,148.35) and (121.29,149.21) .. (119,148.65) .. controls (116.71,148.09) and (115.29,148.95) .. (114.72,151.24) .. controls (114.15,153.53) and (112.73,154.39) .. (110.44,153.83) .. controls (108.15,153.27) and (106.73,154.13) .. (106.16,156.42) .. controls (105.59,158.71) and (104.17,159.57) .. (101.88,159) .. controls (99.59,158.44) and (98.17,159.3) .. (97.61,161.59) .. controls (97.04,163.88) and (95.62,164.74) .. (93.33,164.18) .. controls (91.04,163.62) and (89.62,164.48) .. (89.05,166.77) .. controls (88.48,169.06) and (87.06,169.92) .. (84.77,169.36) .. controls (82.48,168.8) and (81.06,169.66) .. (80.5,171.95) .. controls (79.93,174.24) and (78.51,175.1) .. (76.22,174.53) -- (74.06,175.84) -- (74.06,175.84) ;
%Straight Lines [id:da33086774723559675] 
\draw    (160.33,209.25) -- (160.33,31.75) ;
\draw [shift={(160.33,29.75)}, rotate = 90] [color={rgb, 255:red, 0; green, 0; blue, 0 }  ][line width=0.75]    (10.93,-3.29) .. controls (6.95,-1.4) and (3.31,-0.3) .. (0,0) .. controls (3.31,0.3) and (6.95,1.4) .. (10.93,3.29)   ;
%Straight Lines [id:da7947419816825838] 
\draw    (74.83,177.13) -- (253.62,69.28) ;
\draw [shift={(255.33,68.25)}, rotate = 148.9] [color={rgb, 255:red, 0; green, 0; blue, 0 }  ][line width=0.75]    (10.93,-3.29) .. controls (6.95,-1.4) and (3.31,-0.3) .. (0,0) .. controls (3.31,0.3) and (6.95,1.4) .. (10.93,3.29)   ;
%Straight Lines [id:da4346571973121225] 
\draw    (75.5,90.5) -- (260.98,166) ;
\draw [shift={(262.83,166.75)}, rotate = 202.15] [color={rgb, 255:red, 0; green, 0; blue, 0 }  ][line width=0.75]    (10.93,-3.29) .. controls (6.95,-1.4) and (3.31,-0.3) .. (0,0) .. controls (3.31,0.3) and (6.95,1.4) .. (10.93,3.29)   ;
%Straight Lines [id:da9692057087624275] 
\draw [color={rgb, 255:red, 52; green, 101; blue, 155 }  ,draw opacity=1 ][line width=1.5]    (75.5,90.5) -- (248.67,161.25) ;
%Straight Lines [id:da537494953676549] 
\draw [color={rgb, 255:red, 52; green, 101; blue, 155 }  ,draw opacity=1 ][line width=1.5]    (74.83,177.13) -- (243.83,74.88) ;
% Text Node
\draw (264.33,167.65) node [anchor=north west][inner sep=0.75pt]  [font=\footnotesize]  {$x$};
% Text Node
\draw (257.83,67.15) node [anchor=north west][inner sep=0.75pt]  [font=\footnotesize]  {$y$};
% Text Node
\draw (167.5,30.9) node [anchor=north west][inner sep=0.75pt]  [font=\footnotesize]  {$z$};
\end{tikzpicture}
    \caption{Morally speaking, we can picture the formal derived smooth manifold in the example above as a smooth "cloud" around the bare set of the intersection.}
    \label{fig:derived_manifold_example}
\end{figure}
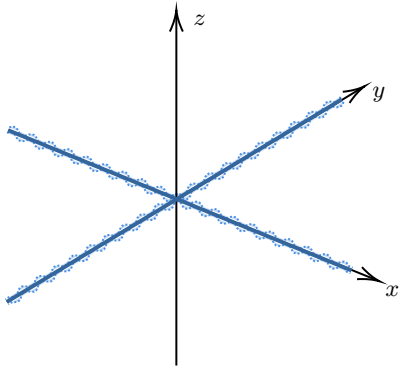

\end{example}

Let us see how this simple example can be generalised to a relevant class of examples: the derived intersection of the graph of a section of a vector bundle with the one of the zero-section.

\begin{example}[Derived zero locus of a section of a vector bundle]
Let $\Sigma,M$ be again ordinary smooth manifolds and let $\pi_\Sigma:M\rightarrow \Sigma$ be an ordinary vector bundle. Let us also fix a section $e_\Sigma:\Sigma\hookrightarrow M$ of such a vector bundle.
The derived intersection $\Sigma\times^h_{M\!}\Sigma$ of $e_\Sigma:\Sigma\hookrightarrow M$ with the zero-section $0:\Sigma\hookrightarrow M$ is also known as derived zero locus of $e_\Sigma$.
To explicitly find such a derived intersection it is convenient to deploy the notion of dg-$\Coo$-algebra, see e.g. \cite{Pridham:2018, carchedi2023derived}.
A dg-$\Coo$-algebra $K_\bullet$ is a dg-commutative $\bbR$-algebra where $K_0$ is equipped with a compatible $\Coo$-algebra structure. Maps of dg-$\Coo$-algebra are maps of dg-commutative $\bbR$-algebras which respect the $\Coo$-structure in degree $0$ and, 
similarly to $\bbR$-algebras, the category of non-positively graded dg-$\Coo$-algebras is naturally simplicially enriched.
Now, there exists a non-positively graded dg-$\Coo$-algebra $K_{-n} = \wedge^n_{\mathcal{C}^{\infty\!}(\Sigma)}\Gamma(\Sigma,M^\vee)$ with differential $\di_K = \langle e_\Sigma,- \rangle$ given by contraction with the section $e_\Sigma$.
By the construction in \cite{carchedi2023derived}, we can consider the following simplicial $\Coo$-algebra:
\begin{equation}
        \mathcal{O}_{\!}\big( \Sigma\times^h_{M\!}\Sigma \big)\,:\,\bbR^k \,\longmapsto\, \Hom_{\mathsf{dgC^\infty Alg}^{\leq 0\!}}(\Coo(\bbR^k), K_\bullet),
\end{equation}
which provides a model for the homotopy $\Coo$-algebra of functions on the derived zero locus.
In fact, by forgetting the $\Coo$-structure, the underlying simplicial set of such a homotopy $\Coo$-algebra is the $(\mathrm{dim} M- \mathrm{dim} \Sigma)$-skeletal simplicial set given by the usual Dold-Kan correspondence
\begin{equation*}
    \hspace{-0.25cm}\mathcal{O}_{\!}\big( \Sigma\times^h_{M\!}\Sigma \big)^{\mathrm{alg}\!} \,\simeq\, \left(\!\!\! \begin{tikzcd}[row sep=scriptsize, column sep=1.6ex]
    \, \cdots  \arrow[r, yshift=1.8ex]\arrow[r, yshift=0.6ex]\arrow[r, yshift=-1.8ex]\arrow[r, yshift=-0.6ex]&  \Coo(\Sigma)_{\!}\oplus_{\!}\Gamma(\Sigma,M^\vee)^{\oplus 2\!}_{\!}\oplus_{\!} \wedge^2_{\mathcal{C}^{\infty\!}(\Sigma)}\Gamma(\Sigma,M^\vee) 
    \arrow[r, yshift=1.4ex] \arrow[r] \arrow[r, yshift=-1.4ex] & \Coo(\Sigma)_{\!}\oplus_{\!} \Gamma(\Sigma,M^\vee)  \arrow[r, yshift=0.7ex] \arrow[r, yshift=-0.7ex] & \Coo(\Sigma)
    \end{tikzcd}    \!\!\!\right)
\end{equation*}
where the face maps of the $1$-simplices are given by $\partial_0(f,\phi)= f$ and $\partial_1(f,\phi)= f+ \langle \phi, e_\Sigma\rangle$.
\end{example}

%%%%%%%%%%%%%%%%%%%%%%%%%%%%%%%%%%%%%%%%%%%%%%%%%%%%%%%%%%%%%%%%%%%%%%%%%%%%%%%%%%%%%%%%%%%%%
\subsection{Definition of formal derived smooth stacks}

Now that we have the $(\infty,1)$-category of formal derived manifolds, the next step will be to equip it with the structure of an $(\infty,1)$-site. To do that, we must first introduce a suitable notion of \'etale map, which is going to generalise local diffeomorphisms of smooth manifolds. Once we have constructed an $(\infty,1)$-site of formal derived manifolds, we will be able to define derived stacks on such a site. These will be called formal derived smooth stacks and will provide a language to make precise the idea of formal derived smooth spaces which are infinite-dimensional and which have a notion of gauge-transformations.

To provide the category of formal derived smooth manifolds with the structure of an \'etale site, we must first understand more about  their truncations.

\begin{construction}[Relation between formal derived smooth manifolds and $\Coo$-varieties]\label{construction: relation between formal derived smooth manifolds and c-infinity varieties}
The first step is to notice that by \cite[Proposition 3.27]{Carchedi2019OnTU} the ordinary category of $\Coo$-algebras is a coreflective sub-$(\infty,1)$-category of the $(\infty,1)$-category of homotopy $\Coo$-algebras. In fact, we have an $(\infty,1)$-adjunction $\pi_0\dashv \iota$ of the form
\begin{equation}\label{eq:coref}
    \begin{tikzcd}[row sep=scriptsize, column sep=8.2ex, row sep=13.0ex]
     \mathbf{N}(\mathsf{C^\infty Alg}) \arrow[r, "\iota"', ""{name=0, anchor=center, inner sep=0}, shift left=-1.1ex, hook] &  \mathbf{sC^\infty Alg}, \arrow[l, "\pi_0"', ""{name=1, anchor=center, inner sep=0}, shift left=-1.1ex] \arrow["\dashv"{anchor=center, rotate=-90}, draw=none, from=0, to=1]
    \end{tikzcd}
\end{equation}
where $\iota$ is the natural inclusion sending an ordinary $\Coo$-algebra to the
corresponding constant simplicial $\Coo$-algebra and $\pi_0$ is the $(\infty,1)$-functor sending a homotopy $\Coo$-algebra $R$ to the coequaliser $\pi_0R\coloneqq\mathrm{coeq}(\!\!\begin{tikzcd}[row sep=scriptsize, column sep=2.5ex] R_1\! \arrow[r, shift left=0.8ex]\arrow[r, shift left=-0.8ex] &  \!R_0  \!\!\end{tikzcd})$ of the face maps of the $1$-simplices.
If we restrict the functor $\pi_0$ to finitely generated $\Coo$-algebras, we obtain $\pi_0:\mathbf{sC^\infty Alg}_{\mathrm{fg}}\rightarrow \mathbf{N}(\mathsf{C^\infty Alg}_\mathrm{fg})$, where $\mathsf{C^\infty Alg}_\mathrm{fg}$ is the ordinary category of finitely generated $\Coo$-algebras.
If we go to the opposite categories, we immediately get the $(\infty,1)$-functor
\begin{equation}
t_0\coloneqq \pi_0^\op\,:\;\mathbf{dFMfd} \,\longrightarrow\,  \mathbf{N}(\mathsf{C^{\infty} Var}).
\end{equation}
Notice that formal derived smooth manifolds provide a derived enhancement not only of usual formal smooth manifold, but, more generally, of $\Coo$-varieties.
\end{construction}

Now, we must introduce a notion of \'etale maps between formal derived smooth manifolds.

\begin{definition}[Formally \'etale map of formal derived smooth manifolds]
We say that a morphism $f:M \longrightarrow N$ of formal derived smooth manifolds is a \textit{formally \'etale map} if
\begin{enumerate}[label=(\textit{\roman*})]
    \item its underived-truncation $t_0f:t_0M \longrightarrow t_0N$ is a formally \'etale map of $\Coo$-varieties,
    \item for each $i\in\mathbb{N}$ the canonical morphism
    \begin{equation}
    \pi_i\mathcal{O}(N)\,\widehat{\otimes}_{\pi_0\mathcal{O}(N)}\,\pi_0\mathcal{O}(M) \;\xrightarrow{\;\;\,\simeq\,\;\;}\; \pi_i\mathcal{O}(M)
\end{equation}
is an isomorphism of ordinary $\Coo$-algebras.
\end{enumerate}
\end{definition}

In the definition above, the $\pi_i$ are the categorical homotopy groups in the $(\infty,1)$-category of homotopy $\Coo$-algebras as constructed in \cite[section 6.5.1]{topos}.

\begin{construction}[\'Etale $(\infty,1)$-site of formal derived smooth manifolds]
Now, by following \cite{ToenVezzo05, ToenVezzo08}, the $(\infty,1)$-category $\mathbf{dFMfd}$ of formal derived smooth manifolds can be naturally equipped with the structure of an \'etale $(\infty,1)$-site, whose coverage is provided by the assignment of \'etale covers $\{U_i\xrightarrow{\;\,\phi_i\;\,} M\}_{i\in I}$ to any formal derived smooth manifold $M$. Such \'etale covers are collections of morphisms such that: 
\begin{enumerate}[label=(\textit{\roman*})]
    \item each $U_i\xrightarrow{\;\,\phi_i\;\,} M$ is a formally \'etale map,
    \item there exist a finite subset $I'\subset I$ such that the truncation $\{t_0U_i\xrightarrow{\;\,t_0\phi_i\;\,} t_0M\}_{i\in I'}$ is a covering in the ordinary site of $\Coo$-varieties.
\end{enumerate}
\end{construction}

\begin{construction}[Simplicial category of formal derived smooth stacks]
Now, given the definition of the $(\infty,1)$-site of formal derived smooth manifolds, we can apply the general discussion above about derived stacks to our case of interest. By \cite[Theorem 3.4.1]{ToenVezzo05} there exists a local projective model structure on the simplicial-category of pre-stacks $[\mathsf{sC^\infty Alg}_\mathrm{fg},\mathsf{sSet}]$ induced by the definition of formally \'etale maps of formal derived smooth manifolds. Thus, by localisation of such a simplicial model structure, one can obtain the simplicial category of formal derived smooth stacks, i.e.
\begin{equation}
    \mathsf{dFSmoothStack} \;\coloneqq\; [\mathsf{sC^\infty Alg}_\mathrm{fg},\mathsf{sSet}]_{\mathrm{proj,loc}}^\circ.
\end{equation}
\end{construction}

To write concretely a formal derived smooth stack, we need to introduce a certain refinement of an \'etale cover, namely we need the definition of an \'etale hypercover.

\begin{definition}[\'Etale hypercover of a formal derived smooth manifold]\label{etale hypercover}
An \'etale hypercover $H(U)_\bullet\rightarrow U$ of a formal derived smooth manifold $U$ is a simplicial object $H(U)_\bullet$ in the \'etale $(\infty,1)$-site $\mathbf{dFMfd}$ such that 
$H(U)_0\rightarrow U$ is an \'etale cover and all natural morphisms 
\begin{equation}
    H(U)_n \;\longrightarrow\; (\mathrm{cosk}_{n-1\!}\circ\mathrm{tr}_{n-1}H(U)_\bullet)_n
\end{equation}
for $n>0$ are \'etale covers.
\end{definition}

In the definition above, $\mathrm{tr}_{n}$ and $\mathrm{cosk}_{n}$ are respectively the $n$-truncation functor and the $n$-coskeleton functor on simplicial objects.
Thus, $H(U)_\bullet\rightarrow U$ being an \'etale hypercover means that, for each $n\geq 0$, one has the equivalence of the form
\begin{equation}
    H(U)_n \;\simeq\; \coprod_{i\in I_n} \!U_{i}^n
\end{equation}
where $U_i^n$ are formal derived smooth manifolds such that the following are all \'etale covers:
\begin{equation}
\begin{aligned}
    \{U_{i}^0 \,\rightarrow\, U\}_{i\in I_0} \\[0.6ex]
    \Big\{U_{i}^1 \,\rightarrow \!\!\!\coprod_{j_1,j_2\in I_0}\!\!\! U_{j_1\!}^0\times_{U\!}U_{j_2}^0\Big\}_{i\in I_1} \\[0.5ex]
    \Big\{U_{i}^2 \,\rightarrow (\mathrm{cosk}_{1\!}\circ\mathrm{tr}_{1}H(U)_\bullet)_2\Big\}_{i\in I_2} \\
    \vdots \quad
\end{aligned}
\end{equation}

Now, we have all the ingredients to unravel the definition of formal derived smooth stacks in concrete terms.

\begin{remark}[Formal derived smooth stack in concrete terms]
A formal derived smooth stack $X\in \mathsf{dFSmoothStack}$ is modelled by a fibrant object in the simplicial model category $[\mathsf{dFMfd}^\op,\mathsf{sSet}]_{\mathrm{proj,loc}}$. Thus, by the general argument in \cite{ToenVezzo05, Moerdijk2010SimplicialMF}, we have that a formal derived smooth stack $X$ is concretely given by a simplicial functor $X:\mathsf{dFMfd}^\op\longrightarrow\mathsf{sSet}$ such that the following conditions are satisfied:
\begin{enumerate}[label=(\textit{\roman*})]
    \item \textit{object-wise fibrancy}: for any $U\in\mathsf{dFMfd}$, the simplicial set $X(U)$ is Kan-fibrant;
    \item \textit{pre-stack condition}: for any equivalence $U\xrightarrow{\,\simeq\,} U'$ in $\mathsf{dFMfd}$, the induced morphism $X(U')\longrightarrow X(U)$ is an equivalence of simplicial sets;
    \item \textit{descent condition}: for any \'etale hypercover $H(U)_\bullet\rightarrow U$ in $\mathsf{dFMfd}$, the natural morphism
    \begin{equation}
        X(U) \;\longrightarrow\; \bfR\!\!\lim_{\!\!\!\![n]\in\Delta}\bigg(\prod_{i\in I_n}X(U_i^n)\bigg)
    \end{equation}
    is an equivalence of simplicial sets.
\end{enumerate}
Notice that this last condition provides an interesting generalisation of the gluing conditions of ordinary sheaves. 
Moreover, from the perspective of applications, it provides a recipe to construct a formal derived smooth stack by gluing together simpler spaces of sections. 
\end{remark}

Finally, we can take the homotopy-coherent nerve of the simplicial-category of formal derived smooth stacks to obtain its $(\infty,1)$-categorical version, as previously discussed at the beginning of this section at equality \eqref{eq:gentheory_derived}.
\begin{definition}[$(\infty,1)$-category of formal derived smooth stacks]
We define the \textit{$(\infty,1)$-category of formal derived smooth stacks} by
\begin{equation}
    \mathbf{dFSmoothStack} \;\coloneqq\; \mathbf{N}_{hc}([\mathsf{dFMfd}^\op,\mathsf{sSet}]_{\mathrm{proj,loc}}^\circ),
\end{equation}
i.e. by the $(\infty,1)$-category of stacks on the \'etale $(\infty,1)$-site presented by $\mathsf{dFMfd}= \mathsf{sC^\infty Alg}^\op_{\mathrm{fg}}$ of formal derived smooth manifolds.
\end{definition}

As we will see in \cref{sec:derdiffcohe} below, the $(\infty,1)$-category $\mathbf{dFSmoothStack}$ comes equipped with a very rich structure: it is a differential cohesive $(\infty,1)$-topos in the sense of \cite{DCCTv2}.

\begin{proposition}[Relation with usual smooth stacks]\label{prop:ext-trunc}
There exists an adjunction $(i\dashv t_0)$ of $(\infty,1)$-functors between the $(\infty,1)$-category of smooth stacks into the $(\infty,1)$-category of formal derived smooth stacks
\begin{equation}\label{eq:ref_emb}
    \begin{tikzcd}[row sep=scriptsize, column sep=15.5ex, row sep=13.0ex]
     \mathbf{dFSmoothStack} \arrow[r, "t_0", shift left=-1.5ex] &  \mathbf{SmoothStack}, \arrow[l, "i "', shift left=-1.5ex, hook'] 
    \end{tikzcd}
\end{equation}
where $i$ is fully faithful and $t_0$ preserves finite products.
\end{proposition}

\begin{proof}
The logic of the proof is the following: first, we must show that we have an adjunction between the corresponding $(\infty,1)$-categories of pre-stacks and, then, that this restricts to an adjunction of the the $(\infty,1)$-categories of stacks.
A simplicial functor $f:\mathsf{C}\rightarrow\mathsf{D}$ gives rise to an adjunction $(f_!\dashv f^\ast)$ between the corresponding simplicial-functor categories $[\mathsf{C}^\op,\sSet]$ and $[\mathsf{D}^\op,\sSet]$, where the pullback functor $f^\ast = (-)\circ f$ is just the pre-composition with $f$ and $f_!$ is the left Kan extension of $f$. (See e.g. \cite{dubuc2006kan}.) 
In our case of interest, the embedding $\upiota^{\mathsf{Mfd}}:\mathsf{Mfd}\hookrightarrow\mathsf{dFMfd}$ induces an adjunction of functors between $[\mathsf{dFMfd}^\op,\mathsf{sSet}]$ and $[\Mfd^\op,\mathsf{sSet}]$.
Thus, by \cite[Section 2.3.1]{ToenVezzo05} we have a Quillen adjunction of simplicial-functors
\begin{equation}
    \begin{tikzcd}[row sep=scriptsize, column sep=15.5ex, row sep=13.0ex]
     {[\mathsf{dFMfd}^\op,\mathsf{sSet}]}_{\mathrm{proj}} \arrow[r, "\upiota^{\mathsf{Mfd}\ast}", shift left=-1.5ex] &  {[\Mfd^\op,\mathsf{sSet}]}_{\mathrm{proj}} \arrow[l, "\upiota_!^{\mathsf{Mfd}} "', shift left=-1.5ex, hook'] .
    \end{tikzcd}
\end{equation}
(Recall that, in the projective model structure, fibrations and weak equivalences are computed object-wise).
This simplicial Quillen adjunction provides a model of an $(\infty,1)$-adjunction of pre-stacks.
Now, to see that this restricts to stacks, we need to show that these simplicial-functors send locally fibrant/cofibrant objects (i.e. fibrant/cofibrant objects in the local projective model structure) to other locally fibrant/cofibrant objects.
However, by the properties of Quillen adjunctions, it is sufficient to check this for the right adjoint functor.
So, given any $X\in{[\mathsf{dFMfd}^\op,\mathsf{sSet}]}_{\mathrm{proj,loc}}^\circ$, its image is $\upiota^{\mathsf{Mfd}\ast\!}X = X\circ \upiota^{\mathsf{Mfd}}$. For any manifold $U\in\Mfd$, a \v{C}ech nerve $\check{C}(U)_\bullet\rightarrow U$ precisely embeds into an \'etale hypercover, thus $\upiota^{\mathsf{Mfd}\ast\!}X$ satisfying descent on ordinary smooth manifold is an immediate consequence of $X$ satisfying descent on formal derived smooth manifolds.
Therefore, there is a Quillen adjunction of simplicial-functors
\begin{equation}
    \begin{tikzcd}[row sep=scriptsize, column sep=15.5ex, row sep=13.0ex]
     {[\mathsf{dFMfd}^\op,\mathsf{sSet}]}_{\mathrm{proj,loc}} \arrow[r, "\upiota^{\mathsf{Mfd}\ast}", shift left=-1.5ex] &  {[\Mfd^\op,\mathsf{sSet}]}_{\mathrm{proj,loc}} \arrow[l, "\upiota_!^{\mathsf{Mfd}} "', shift left=-1.5ex, hook'].
    \end{tikzcd}
\end{equation}
This simplicial Quillen adjunction provides a model of an $(\infty,1)$-adjunction of stacks.
Now, since the functor $\upiota^{\mathsf{Mfd}}$ is fully faithful, we have that $\upiota_!^{\mathsf{Mfd}}$ is also fully faithful.
Finally, $\upiota^{\mathsf{Mfd}\ast}$ preserves finite products, since finite limits are computed object-wise, so we have $(X\times^hY)(\upiota^{\mathsf{Mfd}}U)\simeq X(\upiota^{\mathsf{Mfd}}U)\times Y(\upiota^{\mathsf{Mfd}}U)$ for any smooth manifold $U$ and formal derived smooth stacks $X,Y$.
\end{proof}

\begin{definition}[Derived-extension and underived-truncation functor]\label{def:der_ext_fun}
In the diagram right above we defined the following functors:
\vspace{-0.2cm}\begin{itemize}
    \item the \textit{derived-extension functor} $i\coloneqq\upiota_!^{\mathsf{Mfd}}$ in the diagram above,
    \item the \textit{underived-truncation functor} $t_0 \coloneqq \upiota^{\mathsf{Mfd}\ast}$ in the diagram above.
\end{itemize}
\end{definition}

More concretely, the underived-truncation functor $t_0$ sends any formal derived smooth stack $X\in\mathbf{dFSmoothStack}$ to the smooth stack $t_0X\in \mathbf{SmoothStack}$ given by the composition
\begin{equation}
\begin{aligned}
    t_0X \,:\; \Mfd^\op\;\xhookrightarrow{\;\;\;\upiota^{\mathsf{Mfd}}\;\;\;}\;\mathsf{dFMfd}^\op \;\xrightarrow{\;\;\;X\;\;\;}\; \sSet.
\end{aligned}
\end{equation}

\begin{remark}[Derived-extension functor does not preserve limits]
As we noticed, the derived-extension functor $i$ preserves finite products. However, crucially, it does not generally preserve pullbacks or other limits.
\end{remark}

\begin{remark}[Homotopy pullback of non-derived stacks]
Let $f:X\rightarrow Z$ and $g:Y\rightarrow Z$ be morphisms of smooth stacks. We can consider the formal derived smooth stack given by the $(\infty,1)$-pullback
\begin{equation}
    \begin{tikzcd}[row sep={14.0ex,between origins}, column sep={16ex,between origins}]
    i(X)_{\!}\times_{i(Z)\!}^hi(Y) \arrow[d]\arrow[r] & i(Y) \arrow[d, "i(f)"] \\
    i(X) \arrow[r, "i(g)"] & i(Z).
    \end{tikzcd}
\end{equation}
Since, as we just remarked, the $(\infty,1)$-functor $i$ does not generally preserve limits, there is therefore a natural morphism of formal derived smooth stacks
\begin{equation}\label{eq:natural_map_fiber_product}
    i(X)\times^h_{i(Z)}i(Y)\;\longrightarrow\; i(X\times_Z Y),
\end{equation}
which is generally not an equivalence.
However, the underived-truncation of such a morphism
\begin{equation}
    {t}_0\big(i(X)\times_{i(Z)}^h i(Y)\big)\,\xrightarrow{\;\;\simeq\;\;}\, {t}_0i(X\times_Z Y) \,\simeq\, X\times_ZY
\end{equation}
is an equivalence of smooth stacks.
\end{remark}

\begin{example}[Derived-extension of a quotient smooth stack]\label{derived stack warm up}
Let us consider a simple smooth stack: a quotient stack $[M\!/G]\in\mathbf{SmoothStack}$, where $M$ is an ordinary smooth manifold and $G$ a Lie group.
Recall that, on a smooth manifold $U\simeq \bbR^n$ diffeomorphic to a Cartesian space, its simplicial set of sections is given by
\begin{equation*}
    [M/G](U) \;\simeq\; \mathrm{cosk}_{2\!}\left(\! \begin{tikzcd}[row sep={22.5ex,between origins}, column sep={4.5ex}]
    \Hom(U,G^{\times 2}\!\times\!M) \, \arrow[r, yshift=2.0ex , "{ }"] \arrow[r, description] \arrow[r, yshift=-2.0ex, "{ }"'] & \Hom(U,G\!\times\!M) \arrow[r, yshift=1.0ex , "{\partial_0}"] \arrow[r, yshift=-1.0ex, "{\partial_1}"'] & \Hom(U,M)
    \end{tikzcd}\!\right)\!,
\end{equation*}
where the face maps, on $1$-simplices, are $\partial_0(g,f)\mapsto f$ and $\partial_1(g,f)\mapsto g\cdot f$ and, on $2$-simplices, are given respectively by the group multiplication and by bare projections, as usual. This means that $1$-simplices from a $0$-simplex $f\in \Hom(U,M)$ to a $0$-simplex $f'\in \Hom(U,M)$ are of the form $f'=g\cdot f$ for some $g\in\Hom(U,G)$.
How does this picture of $1$-simplices generalise when we consider the space of sections of the derived-extension $i[M/G]\in\mathbf{dFSmoothStack}$ of our quotient stack?
Let now $U$ be a formal derived smooth manifold. By unravelling its definition, the simplicial set of sections of our formal derived smooth stack is of the form\footnote{From now on we will denote by $\bfR\mathrm{Hom}(X,Y)$ the $(\infty,1)$-categorical hom-space between formal derived smooth stacks $X$ and $Y$. Notice that such a notation is different from the one deployed so far for non-derived stacks.}
\begin{equation*}
    i[M/G](U) \,\simeq\,\left(\!\! \begin{tikzcd}[row sep={22.5ex,between origins}, column sep={3.0ex}] 
    \cdots\arrow[r, yshift=1.0ex , "{ }"] \arrow[r, yshift=-1.0ex, "{ }"']\arrow[r, yshift=3.0ex , "{ }"] \arrow[r, yshift=-3.0ex, "{ }"'] &
    \!\!\begin{array}{c}\bfR\Hom(U,G^{\times 2}\!\times\!M)_0\\ \times\bfR\Hom(U,G\!\times\!M)_1\\ \times \bfR\Hom(U,M)_2 \end{array} \!\! \arrow[r, yshift=2.0ex , "{ }"] \arrow[r, description] \arrow[r, yshift=-2.0ex, "{ }"'] & \!\!\begin{array}{c}\bfR\Hom(U,G\!\times\!M)_0 \\ \times\bfR\Hom(U,M)_1\end{array}\!\! \arrow[r, yshift=1.0ex , " "] \arrow[r, yshift=-1.0ex, " "'] &  \bfR\Hom(U,M)_0
    \end{tikzcd}\!\!\right)\!.
\end{equation*}
So, a $1$-simplex is a triplet $(g,f,f_1)$, where  $(g,f)\in\bfR\Hom(U,G\times M)_0$ and $f_1\in\bfR\Hom(U,M)_1$ is a homotopy of the form  $f'\xleftarrow{\;f_1\;}g\cdot f$, where this compact notation means that the homotopy $f_1$ has boundaries $\partial_0f_1 = g\cdot f$ and $\partial_1f_1 =f'$. This means that a $1$-simplex $(g,f,f_1)$ goes from $f$ to $f'$, where the latter is not anymore equal on the nose to $g\cdot f$, but only homotopic to it by $f_1$.
An analogous story holds for $2$-simplices, where homotopies of homotopies will appear and so on for higher simplices. 
This example will be propaedeutic to the study of more complicated stacks in section \ref{sec:global_aspects_of_bv_theory}.
\end{example}

%%%%%%%%%%%%%%%%%%%%%%%%%%%%%%%%%%%%%%%%%%%%%%%%%%%%%%%%%%%%%%%%%%%%%%%%%%%%%%%%%%%%%%%%%%%%%
\subsection{Discussion of formal derived smooth sets}

In the previous subsection we constructed formal derived smooth stacks.
In analogy with non-derived smooth stacks, we may wonder if there is any possible notion of formal derived smooth set.
We should remark, however, that there is no meaningful notion of sheaf on formal derived smooth manifolds, so that the idea of defining formal derived smooth sets this way seems hopeless.
Having said that, in this subsection we will propose a working definition of formal derived smooth sets based on a different principle: a formal derived smooth set will be defined as a formal derived smooth stack which is the derived enhancement of an ordinary smooth set.

Recall from Remark \ref{rem: embeddings} that there is a natural embedding $\mathbf{N}(\mathsf{SmoothSet})\longhookrightarrow \mathbf{SmoothStack}$ of smooth sets into smooth stacks. Moreover, by \cite[Section 5.6]{topos}, such an embedding has a left adjoint functor $\tau_0$, which is known as $0$\textit{-truncation} of smooth stacks.
So, by putting everything together, we have the following diagram of coreflective and reflective embeddings of $(\infty,1)$-categories:
\begin{equation}
    \begin{tikzcd}[row sep=scriptsize, column sep=6.2ex, row sep=15.0ex]
    \mathbf{dFSmoothStack} \arrow[r, "t_0"', ""{name=0, anchor=center, inner sep=0}, shift left=-1.0ex] & \mathbf{SmoothStack} \arrow[d, "\tau_0", ""{name=2, anchor=center, inner sep=0}, shift left=1.0ex] \arrow[l, hook, shift left=-1.0ex, "i"', ""{name=1, anchor=center, inner sep=0}] \\
    & \mathbf{N}(\mathsf{SmoothSet}) \arrow[u, hook, shift left=1.0ex, ""{name=3, anchor=center, inner sep=0}] . \arrow["\dashv"{anchor=center, rotate=-90}, draw=none, from=0, to=1]\arrow["\dashv"{anchor=center}, draw=none, from=2, to=3]
    \end{tikzcd}
\end{equation}

We have now all the ingredients to provide the definition of formal derived smooth sets.

\begin{definition}[Formal derived smooth set]
A \textit{formal derived smooth set} $X$ is a formal derived smooth stack $X\in\mathbf{dFSmoothStack}$ such that its underived-truncation $t_0X$ is in the essential image of the natural embedding $\mathbf{N}(\mathsf{FSmoothSet})\longhookrightarrow\mathbf{FSmoothStack}$.
Thus, we define the $(\infty,1)$-category of formal derived smooth sets by the pullback
\begin{equation}
    \mathbf{dFSmoothSet} \;\coloneqq\; \mathbf{dFSmoothStack}\times^h_{\mathbf{SmoothStack}}\mathbf{N}(\mathsf{SmoothSet})
\end{equation}
in the $(\infty,1)$-category of $(\infty,1)$-categories.
\end{definition}

In other words, we construct a formal derived smooth set to be a formal derived smooth stack $X$ whose underived-truncation $t_0X$ is, in particular, a $0$-truncated smooth stack, or equivalently just an ordinary smooth set.

Now, we have the following square of reflective embeddings:
\begin{equation}
    \begin{tikzcd}[row sep=scriptsize, column sep=6.2ex, row sep=15.0ex]
    \mathbf{dFSmoothStack} \arrow[r, "t_0"',""{name=0, anchor=center, inner sep=0},  shift left=-1.0ex] \arrow[d, "\tau_0", shift left=1.0ex, ""{name=6, anchor=center, inner sep=0}] & \mathbf{SmoothStack} \arrow[d, "\tau_0", shift left=1.0ex, ""{name=2, anchor=center, inner sep=0}] \arrow[l, hook, shift left=-1.0ex] \\
    \mathbf{dFSmoothSet} \arrow[r, "t_0"', shift left=-1.0ex, ""{name=4, anchor=center, inner sep=0}] \arrow[u, hook, shift left=1.0ex, ""{name=7, anchor=center, inner sep=0}, ] & \mathsf{SmoothSet} \arrow[l, hook, ""{name=5, anchor=center, inner sep=0},  shift left=-1.0ex] \arrow[u, hook, shift left=1.0ex, ""{name=3, anchor=center, inner sep=0}] \arrow["\dashv"{anchor=center, rotate=-90}, draw=none, from=0, to=1]\arrow["\dashv"{anchor=center}, draw=none, from=2, to=3] \arrow["\dashv"{anchor=center, rotate=-90}, draw=none, from=4, to=5]\arrow["\dashv"{anchor=center}, draw=none, from=6, to=7].
    \end{tikzcd}
\end{equation}
Reflective sub-categories are stable under pullback along cocartesian fibrations, as shown for example in \cite[Proposition 6.2.2.17]{Kerodon}. But any left fibration is a cocartesian fibration, as seen in \cite[Example 3.3]{Barwick:2016}, so $\tau_0$ on the right a cocartesian fibration. This implies that $\mathbf{dFSmoothSet}\hookrightarrow\mathbf{dFSmoothStack}$ is itself a reflective sub-category.

Let us now look at a few relevant examples of formal derived smooth sets, which will be useful later in dealing with physics.  

\begin{example}[Formal derived smooth manifold]
The simplest, but also the archetypal, class of examples of formal derived smooth set is provided by the formal derived smooth manifolds themselves.
Let $M\in\mathsf{dFMfd}$ be a formal derived smooth manifold. It naturally Yoneda-embeds into a formal derived smooth set of the form
\begin{equation}
\begin{aligned}
    M\,:\,\mathsf{dFMfd} \,&\longrightarrow\, \sSet  \\
    U \,&\longmapsto\,  \bfR\Hom_\mathsf{dFMfd}(U,M),
\end{aligned}
\end{equation}
where $\bfR\Hom_\mathsf{dFMfd}(U,M)= \bfR\Hom_\mathsf{sC^\infty Alg_\mathrm{fp}}\big(\mathcal{O}(M),\mathcal{O}(U)\big)$ and $\mathcal{O}(M),\mathcal{O}(U)$ are respectively the homotopy $\Coo$-algebras of functions on $M,U$.
Thus, we have an embedding of $(\infty,1)$-categories $\mathbf{dFMfd} \,\longhookrightarrow\, \mathbf{dFSmoothSet}$.
\end{example}
We can now explicitly show that the natural embedding of smooth manifolds into derived smooth manifolds is compatible with the embedding of smooth sets into formal derived smooth sets, i.e. with the derived-extension functor. 

\begin{example}[Ordinary smooth manifolds]\label{ex:explanation_i}
Recall that, given a smooth manifold $M\in\mathsf{Mfd}$, we have that it Yoneda embeds into smooth sets to the functor $M:U\mapsto \Hom_{\mathsf{Mfd}}(U,M)$ on the site of smooth manifolds $U\in\mathsf{Mfd}$.
The derived-extension functor embeds this smooth set into the following formal derived smooth set
\begin{equation}
\begin{aligned}
    i(M)\,:\,\mathsf{dFMfd} \,&\longrightarrow\, \sSet  \\
    U \,&\longmapsto\,  \bfR\Hom_\mathsf{dFMfd}(U,\upiota^{\mathsf{Mfd}}(M)),
\end{aligned}
\end{equation}
where $\upiota^{\mathsf{Mfd}}:\mathsf{Mfd}\hookrightarrow\mathsf{dFMfd}$ is the natural embedding of smooth manifolds into formal derived smooth manifolds.
\end{example}

The following is the first non-obvious class of examples which we can study in the context of formal derived smooth sets.

\begin{example}[Formal derived mapping space]
A more interesting class of examples of formal derived smooth sets is provided by mapping spaces.
Let $M,N\in\mathsf{Mfd}$ be a pair of ordinary smooth manifolds. We can define a formal derived smooth set $[iM,iN]\in\mathsf{dFSmoothSet}$ by
\begin{equation}
\begin{aligned}
    [iM,iN]\,:\,\mathsf{dFMfd} \,&\longrightarrow\, \sSet  \\
    U \,&\longmapsto\,  \bfR\Hom_\mathsf{dFMfd}(U\times \upiota^{\mathsf{Mfd}}(M),\,\upiota^{\mathsf{Mfd}}(N)),
\end{aligned}
\end{equation}
functorially on elements $U\in\mathsf{dFMfd}$ of the site. 
This is the natural derived enhancement of the ordinary mapping space of two ordinary smooth manifolds.
To see that this is indeed a formal derived smooth set, it is enough to notice that we have the equivalences of simplicial sets
$[iM,iN](\ast) \simeq \bfR\Hom_\mathsf{dFMfd}(\upiota^{\mathsf{Mfd}}(M),\,\upiota^{\mathsf{Mfd}}(N))\simeq \Hom_{\Mfd}(M,N)$.
\end{example}

Let, more generally, $M,N\in\mathbf{dFMfd}$ be a pair of formal derived smooth manifolds. Then we can construct their formal derived mapping stack
$[M,N]\,:\, U \,\mapsto\,  \bfR\Hom_\mathsf{dFMfd}(U\times M,N)$.
However, notice that this is not a derived formal smooth set, contrarily to what one may have expected. To see this, one can pick $U=\ast$, so that $[M,N](\ast) \simeq \bfR\Hom_\mathsf{dFMfd}(M,N)$ is generally not a constant simplicial set.

%%%%%%%%%%%%%%%%%%%%%%%%%%%%%%%%%%%%%%%%%%%%%%%%%%%%%%%%%%%%%%%%%%%%%%%%%%%%%%%%%%%%%%%%%%%%%
\subsubsection{Derived affine $\Coo$-schemes}

We will now introduce a fundamental and very concrete class of examples of formal derived smooth sets: derived affine $\Coo$-schemes. 
These geometric objects are defined similarly to the derived affine schemes of derived algebraic geometry, but instead of derived commutative algebras, they correspond to homotopy $\Coo$-algebras.

\begin{remark}[Searching for formal derived smooth pro-manifolds]\label{foot:pro}
The ind-category of an $(\infty,1)$-category $\mathbf{C}$ is defined by $\mathrm{Ind}(\mathbf{C}) \simeq [\mathbf{C}^\op,\mathbf{\infty Grpd}]_{\mathrm{acc,lex}}$, where we called $[-,-]_{\mathrm{acc,lex}}$ the $(\infty,1)$-category of functors which are accessible and left-exact (see for instance \cite[Section 5.3]{topos}).
In the case of formal derived smooth manifolds, \cite[Theorem 3.10]{Carchedi2019OnTU} tells us that the $(\infty,1)$-category $\mathbf{sC^\infty Alg}$ of homotopy $\Coo$-algebras is compactly generated and, in particular, there is an equivalence
\begin{equation}\label{eq:int}
    \mathrm{Ind}\big(\mathbf{sC^\infty Alg}_\mathrm{fp}\big) \;\simeq\; \mathbf{sC^\infty Alg}
\end{equation}
between the ind-$(\infty,1)$-category of finitely presented homotopy $\Coo$-algebras and the $(\infty,1)$-category of homotopy $\Coo$-algebras.
The pro-$(\infty,1)$-category $\mathrm{Pro}(\mathbf{C})$ of any given $(\infty,1)$-category $\mathbf{C}$ is defined by the equivalence $\mathrm{Pro}(\mathbf{C}) \simeq \mathrm{Ind}(\mathbf{C}^\op)^\op$.
Thus, from the equivalence \eqref{eq:int}, we can be immediately obtain the following equivalences:
\begin{equation}\label{corollary:pro}
    \mathrm{Pro}(\mathbf{dMfd}) \;\simeq\; \mathrm{Ind}\big(\mathbf{sC^\infty Alg}_\mathrm{fp}\big)^\op \;\simeq\; \mathbf{sC^\infty Alg}^\op,
\end{equation}
where $\mathbf{dMfd}\simeq \mathbf{sC^\infty Alg}_\mathrm{fp}^\op$ is the $(\infty,1)$-category of derived manifolds in the sense of \cite{Carchedi2019OnTU}.
Thus, there is a natural notion of pro-object for the $(\infty,1)$-category of derived manifolds, which can be seen as the opposite of a general homotopy $\Coo$-algebra.
This provides a motivation for the definition of derived affine $\Coo$-schemes: they can be seen as derived pro-manifolds.
\end{remark}

\begin{definition}[Derived affine $\Coo$-scheme]
We define the $(\infty,1)$-category of \textit{derived affine $\Coo$-schemes} by the opposite $(\infty,1)$-category of homotopy $\Coo$-algebras, i.e. by
\begin{equation}
    \mathbf{dC^{\infty\!\!} Aff} \;\coloneqq\; \mathbf{sC^\infty Alg}^\op.
\end{equation}
An alternative nomenclature for such spaces would be \textit{derived pro-manifolds}, in the light of the discussion at Remark \ref{foot:pro} above.
\end{definition}

\begin{lemma}[Derived affine $C^\infty$-schemes are formal derived smooth stacks]\label{lem:FdMfd_are_dCooAff}
There is a natural embedding of derived affine $C^\infty$-schemes into formal derived smooth stacks.
If we denote by $\bfR\Spec(A)\in \mathbf{dC^{\infty\!\!} Aff}$ the derived affine $\Coo$-scheme whose homotopy $\Coo$-algebra is $A\in\mathbf{sC^\infty Alg}$, its embedding into formal derived smooth stacks is given by
\begin{equation}
\begin{aligned}
    \bfR\Spec(A)\,:\,\mathsf{dFMfd} \,&\longrightarrow\, \sSet  \\
    U \,&\longmapsto\,  \bfR\Hom_\mathsf{sC^\infty Alg}(A,\,\mathcal{O}(U)).
\end{aligned}
\end{equation}
\end{lemma}

\begin{proof}
Recall that we have an embedding $\mathbf{dFMfd}\simeq\mathbf{sC^{\infty\!}Alg}_\mathrm{fg}^\op\longhookrightarrow \mathbf{sC^{\infty\!}Alg}^\op$.
We can construct a functor $\mathcal{O}:\mathbf{dFSmoothStack}\longrightarrow \mathbf{sC^{\infty\!}Alg}^\op$ by Yoneda extension of such an embedding.
More concretely, we can write any formal derived smooth stack $X\in\mathbf{dFSmoothStack}$ as the colimit of representables and construct the limit of homotopy $\Coo$-algebras
\begin{equation}
    \mathcal{O}(X) \;\simeq\; \bfR\!\!\lim_{\!\!U\rightarrow X}\mathcal{O}(U) \quad \text{for}\quad   X \;\simeq\; \bfL\mathrm{co}\!\!\lim_{\!\!U\rightarrow X} U,
\end{equation}
where $\mathcal{O}(U)$ is the usual homotopically finitely presented $\Coo$-algebra of functions on the formal derived smooth manifold $U$.
Since limits become colimits in the opposite category, by construction, the $(\infty,1)$-functor $\mathcal{O}$ preserves colimit.
Notice that both $\mathbf{dFSmoothStack}$ and $\mathbf{sC^{\infty\!}Alg}$ are presentable $(\infty,1)$-categories, the former since it is an $(\infty,1)$-topos and the latter by \cite[Proposition 3.6]{Carchedi2019OnTU}.
Therefore, by the adjoint $(\infty,1)$-functor theorem, the $(\infty,1)$-functor $\mathcal{O}$ has a right adjoint $\bfR\Spec:\mathbf{sC^{\infty\!}Alg}^\op\longrightarrow\mathbf{dFSmoothStack}$.
In fact, for any $X\in\mathbf{dFSmoothStack}$ and $A\in\mathbf{sC^{\infty\!}Alg}$ we have the following chain of equivalences:
\begin{equation}
    \begin{aligned}
    \bfR\Hom(X,\bfR\Spec A) \;&\simeq\; \bfR\Hom\big(\bfL\mathrm{co}\!\!\lim_{\!\!U\rightarrow X} U, \,\bfR\Spec A\big)  \\
    &\simeq\; \bfR\!\!\lim_{\!\!U\rightarrow X} \bfR\Hom(U,\,\bfR\Spec A) \\
    &\simeq\; \bfR\!\!\lim_{\!\!U\rightarrow X} \bfR\Hom_{\mathbf{sC^{\infty\!}Alg}}(A,\,\mathcal{O}(U)) \\
    &\simeq\;  \bfR\Hom_{\mathbf{sC^{\infty\!}Alg}}\big(A,\,\bfR\!\!\lim_{\!\!U\rightarrow X}\mathcal{O}(U)\big) \\
    &\simeq\;  \bfR\Hom_{\mathbf{sC^{\infty\!}Alg}^\op}(\mathcal{O}(X),\,A) 
    \end{aligned}
\end{equation}
Now, a sufficient and necessary condition for $\bfR\Spec$ being a fully faithful $(\infty,1)$-functor is that the counit is an equivalence, which means that the morphism $\mathcal{O}(\bfR\Spec A) \,\xrightarrow{\;\,\simeq\,\;}\, A$.
must be an equivalence for any homotopy $\Coo$-algebra $A$.
Notice that we have the equivalences
\begin{equation}
    \begin{aligned}
    \mathcal{O}(\bfR\Spec A) \;&\simeq\; \;\;\bfR\!\!\!\!\!\!\!\!\lim_{\!\!U\rightarrow \bfR\Spec A}\! \mathcal{O}(U)  \;\,\simeq\; \;\bfR\!\!\!\!\!\!\lim_{\!\!A\rightarrow\mathcal{O}(U)}\! \mathcal{O}(U)   \;\simeq\;  A,
    \end{aligned}
\end{equation}
where in the second line we used the fact that $\bfR\Spec$ is the right adjoint to $\mathcal{O}$. This then shows that the $(\infty,1)$-functor $\bfR\Spec$ is indeed full and faithful.
\end{proof}

The relevance of derived affine $\Coo$-schemes will be mostly a consequence of the fact that they constitute a particularly tractable example of formal derived smooth sets which generalise formal derived smooth manifolds.

\begin{remark}[Formal derived smooth manifolds are derived affine $C^\infty$-schemes]
There is a natural (coreflective) embedding $\mathbf{dFMfd}\simeq\mathbf{sC^{\infty\!} Alg}_\mathrm{fg}^\op\hookrightarrow \mathbf{sC^{\infty\!}Alg}^\op$, since any derived smooth manifold $M$ is immediately equivalent to the spectrum of its homotopy $\Coo$-algebra of functions, i.e. $M\simeq\bfR\Spec\mathcal{O}(M)$.
This embedding allows us to naturally embed formal derived smooth manifolds into derived $\Coo$-schemes.
Thus, by combining this fact with proposition \ref{lem:FdMfd_are_dCooAff}, we obtain the following inclusions of $(\infty,1)$-categories:
\begin{equation}
    \mathbf{dFMfd} \;\longhookrightarrow\; \mathbf{dC^{\infty\!\!} Aff} \;\longhookrightarrow\; \mathbf{dFSmoothSet} \;\longhookrightarrow\; \mathbf{dFSmoothStack},
\end{equation}
where, as before, $\mathbf{dFMfd}$ is the $(\infty,1)$-category of formal derived smooth manifolds, $\mathbf{dC^{\infty\!\!} Aff}$ is the $(\infty,1)$-category of derived $\Coo$-schemes and $\mathbf{dFSmoothSet}$ is the $(\infty,1)$-category of formal derived smooth sets.
\end{remark}

By construction above, the $(\infty,1)$-functor $\bfR\Spec:\mathbf{sC^{\infty\!}Alg}^\op\longrightarrow\mathbf{dFSmoothStack}$ preserves limits. Thus, we have the following corollary.

\begin{corollary}[Pullbacks of affine derived $\Coo$-schemes]
We have the following equivalence of formal derived smooth stacks:
\begin{equation}
    \bfR\Spec A \times^h_{\bfR\Spec C}\bfR\Spec B \;\,\simeq\,\; \bfR\Spec (A\,\widehat{\otimes}^\bfL_C\,B),
\end{equation}
for any given homotopy $\Coo$-algebras $A,B,C\in\mathbf{sC^{\infty\!}Alg}$.
\end{corollary}

\begin{remark}[Underived-truncation of derived affine $\Coo$-schemes]
Notice that the underived-truncation functor sends a derived affine $\Coo$-scheme $\bfR\Spec(R)\in\mathbf{dC^{\infty\!\!}Aff}$ corresponding to a simplicial $\Coo$-algebra $R\in\mathsf{sC^\infty Alg}$ to an ordinary affine $\Coo$-scheme 
\begin{equation}
    t_0\bfR\Spec(R) \,\simeq\, \Spec(\pi_0R),
\end{equation}
corresponding to the ordinary $\Coo$-algebra $\pi_0R\in\mathsf{C^\infty Alg}$.
\end{remark}

\begin{remark}[Derived-extension of affine $\Coo$-schemes]
Notice that the derived-extension functor $i$ sends an ordinary affine $\Coo$-scheme $\Spec(R)$ corresponding to the ordinary $\Coo$-algebra $R\in\mathsf{C^\infty Alg}$ to a derived affine $\Coo$-scheme 
\begin{equation}
    i\Spec(R) \,\simeq\, \bfR\Spec(\iota(R))
\end{equation}
in $\mathbf{dC^{\infty\!\!}Aff}$, which corresponds to the homotopy $\Coo$-algebra $\iota(R)\in\mathbf{sC^\infty Alg}$.
\end{remark}

More generally, these last remarks provide a good intuition for the role played by the underived-truncation and derived-extension of formal derived smooth stacks.

%%%%%%%%%%%%%%%%%%%%%%%%%%%%%%%%%%%%%%%%%%%%%%%%%%%%%%%%%%%%%%%%%%%%%%%%%%%%%%%%%%%%%%%%%%%%%
\subsubsection{Formal derived diffeological spaces}

In this subsection we will define and explore the derived version of a diffeological space, which we will call formal derived diffeological space. 
Recall from Definition \ref{def:diffeological_space} that an ordinary diffeological space is a concrete smooth set, i.e. a concrete sheaf on the site of ordinary smooth manifolds.

\begin{definition}[Formal derived diffeological space]
The \textit{$(\infty,1)$-category of formal derived diffeological spaces} is defined by the pullback of $(\infty,1)$-categories
\begin{equation}
    \mathbf{dFDiffSp} \;\coloneqq\; \mathbf{dFSmoothStack}\times_{\mathbf{SmoothStack}}^h\mathbf{N}(\mathsf{DiffSp}),
\end{equation}
An element of such an $(\infty,1)$-category will be called \textit{formal derived diffeological space}.
\end{definition}

In other words, we have a pullback diagram
\begin{equation}
    \begin{tikzcd}[row sep=scriptsize, column sep=5.5ex, row sep=16.5ex]
    \mathbf{dFDiffSp}  \arrow[d, "", hook]  \arrow[r, "t_0"] & \mathbf{N}(\mathsf{DiffSp}) \arrow[d, "", hook] \\
    \mathbf{dFSmoothStack} \arrow[r, "t_0"] &  \mathbf{SmoothStack},
    \end{tikzcd}
\end{equation}
which, since monomorphisms are stable under pullback by \cite[Proposition 6.5.1.16]{topos}, makes $\mathbf{dFDiffSp} \hookrightarrow \mathbf{dFSmoothSet}$ a full and faithful reflective sub-$(\infty,1)$-category.

\begin{lemma}[Derived affine $\Coo$-schemes are formal derived diffeological spaces]
The $(\infty,1)$-category $\mathbf{dC^{\infty\!\!}Aff}$ of derived affine $\Coo$-schemes is a full and faithful sub-$(\infty,1)$-category of the $(\infty,1)$-category $\mathbf{dFDiffSp}$ of formal derived diffeological spaces.
\end{lemma}

\begin{proof}
Derived affine $\Coo$-schemes form a full and faithful sub-$(\infty,1)$-category of formal derived smooth stacks. Therefore, it is enough to show that every object of $\mathbf{dC^{\infty\!\!}Aff}$ is an object of $\mathbf{dFDiffSp}$.
Consider a derived affine $\Coo$-scheme $\bfR\Spec(R)\in\mathbf{dC^{\infty\!\!}Aff}$, for any given homotopy $\Coo$-algebra $R\in\mathsf{sC^\infty Alg}$. Its underived-truncation is the ordinary $\Coo$-scheme $t_0\bfR\Spec(R)\simeq \Spec (R^{\mathrm{red}})$ with $R^{\mathrm{red}}=\pi_0(R)/\mathfrak{m}_{\pi_0(R)}$.
Thus, it is enough to show that $\Spec(R^{\mathrm{red}})$ is an ordinary diffeological space, i.e that it is a concrete sheaf on the site of smooth manifolds: namely, that for any ordinary smooth manifold $U\in\Mfd$ there is an injective map of sets 
\begin{equation}
    \Hom_{\mathsf{C^\infty Alg}}(R^\mathrm{red}\!,\,\Coo(U)) \;\longhookrightarrow\; \Hom_\mathsf{Set}\big(\mathit{\Gamma}(U), \,\mathit{\Gamma}(\Spec R^\mathrm{red})\big),
\end{equation}
where $\mathit{\Gamma}(\Spec(R^\mathrm{red})) = \Hom_{\mathsf{C^\infty Alg}}(R^\mathrm{red}\!,\, \bbR)$ is the underlying set of points of the reduced scheme and where $\mathit{\Gamma}(U)=\Hom_{\mathsf{C^\infty Alg}}(\Coo(U), \bbR)$ is the underlying set of point of the smooth manifold.
Such a function is given by mapping every element
$f\in\Hom_{\mathsf{C^\infty Alg}}(R^\mathrm{red}\!,\,\Coo(U))$ 
to the precomposition function
$(-)\circ f: \mathit{\Gamma}(U) \rightarrow \mathit{\Gamma}(\Spec(R^\mathrm{red}))$ which sends points of the smooth manifold $U$ to their image in the underlying set of points of the smooth set $\Spec(R^\mathrm{red})$. This function is, in fact, injective, since both $\Coo(U)$ and $R^\mathrm{red}$ are reduced $\Coo$-algebras.
\end{proof}

\begin{remark}[Embeddings of $(\infty,1)$-categories of derived spaces]
To sum up, we have the following full and faithful inclusions of $(\infty,1)$-categories:
\begin{equation*}
    \begin{tikzcd}[row sep=scriptsize, column sep=4.6ex, row sep=9.5ex]
    \mathbf{dFMfd} \arrow[r, hook]  & \mathbf{dC^{\infty\!\!}Aff} \arrow[r, hook] & \mathbf{dFDiffSp} \arrow[r, hook]  &  \mathbf{dFSmoothSet}\arrow[r, hook]  &  \mathbf{dFSmoothStack}.
    \end{tikzcd}
\end{equation*}
\end{remark}

%%%%%%%%%%%%%%%%%%%%%%%%%%%%%%%%%%%%%%%%%%%%%%%%%%%%%%%%%%%%%%%%%%%%%%%%%%%%%%%%%%%%%%%%%
\subsection{Derived mapping stacks and bundles}

In this subsection we focus briefly on the definition of mapping stacks and fibre bundles in the $(\infty,1)$-category of formal derived smooth stacks.

\begin{example}[Formal derived mapping stack]
An interesting and motivational class of examples of formal derived smooth set is provided by mapping spaces of formal derived smooth manifolds.
Let $X,Y\in\mathsf{dFSmoothStack}$ be a pair of ordinary smooth manifolds, then we can define a formal derived smooth set $[X,Y]\in\mathbf{dFSmoothStack}$ by
\begin{equation}
\begin{aligned}
    [X,Y]\,:\,\mathsf{dFMfd} \,&\longrightarrow\, \sSet  \\
    U \,&\longmapsto\,  \bfR\Hom(U\times X,Y),
\end{aligned}
\end{equation}
functorially on elements $U\in\mathsf{dFMfd}$ of the site, where $\bfR\Hom(-,-)$ is the Hom-$\infty$-groupoid of formal derived smooth stacks.
\end{example}

Now, we will introduce the notion of fibre bundle of formal derived smooth sets.
The following two definitions are specific cases of the general definitions appearing in \cite[Section 4]{Principal1}.

\begin{definition}[Fiber bundle]\label{def:fibre_bundle}
A \textit{bundle} is a morphism $E\xrightarrow{\;p\;} X$.
A \textit{fiber bundle} is a morphism $E\xrightarrow{\;p\;} X$ such that there is an effective epimorphism $Y\twoheadrightarrow X$ and, for some formal derived smooth stack $F$, a pullback of the form
\begin{equation}
    \begin{tikzcd}[row sep=scriptsize, column sep=6ex, row sep=7.5ex]
    Y\times F \arrow[d] \arrow[r] & E \arrow[d] \\
    Y \arrow[r] & X,
    \end{tikzcd}
\end{equation}
in the $(\infty,1)$-category $\mathbf{dFSmoothStack}$ of formal derived smooth stacks.
We say that the fiber bundle $E\rightarrow X$ locally trivialises with respect to $Y$ and we call $F$ the fiber of the bundle.
\end{definition}

\begin{definition}[$\infty$-groupoid of sections]\label{def:sec_fibre_bundle}
The $\infty$-\textit{groupoid of sections} of a bundle $E\xrightarrow{\;p\;} X$ is defined as the homotopy fiber
\begin{equation}
    \Gamma(X,E) \;\coloneqq\; \bfR\Hom(X,E) \times_{\bfR\Hom(X,X)\!} \{\mathrm{id}_X\}
\end{equation}
of the $\infty$-groupoid of all morphisms $X\rightarrow E$ on those who cover the identity on $X$.
\end{definition}

Notice that, if $E\rightarrow X$ is a fibre bundle of ordinary smooth manifolds, then by Yoneda embedding $\Gamma(X,E)$ as defined above reduces to the usual notion of set of smooth sections. 

\begin{remark}[On the slice category]
Notice that the $\infty$-groupoid of sections of a bundle $E\xrightarrow{\;p\;} X$ can be equivalently expressed as the $\infty$-groupoid
\begin{equation}
    \Gamma(X,E) \;\simeq\; \bfR\Hom_{{/X}}(\mathrm{id}_X,p)
\end{equation}
where $\bfR\Hom_{{/X}}(-,-)$ is the hom-$\infty$-groupoid of the slice $(\infty,1)$-category $\mathbf{dFSmoothStack}_{/X}$.
\end{remark}

%%%%%%%%%%%%%%%%%%%%%%%%%%%%%%%%%%%%%%%%%%%%%%%%%%%%%%%%%%%%%%%%%%%%%%%%%%%%%%%%%%%%%%%%%%%%%
\subsection{Derived de Rham cohomology}

In this section we will define a notion of quasi-coherent $(\infty,1)$-sheaves of modules on formal derived smooth stacks. In particular, we will introduce the notion of tangent and cotangent complex of formal derived smooth stacks, which will be instrumental to the construction of derived differential forms. 
Moreover, this discussion will be a crucial premise for \cite{AlfYouFuture}, in preparation.

%%%%%%%%%%%%%%%%%%%%%%%%%%%%%%%%%%%%%%%%%%%%%%%%%%%%%%%%%%%%%%%%%%%%%%%%%%%%%%%%%%%%%%%%%%%%%
\subsubsection{Quasi-coherent $(\infty,1)$-sheaves of modules}\label{sec:Qcoh}

Our strategy in this subsection will be to use the notion of homotopy $\Coo$-algebra $\mathcal{O}(X)$ of functions on a formal derived smooth stack $X$ to construct the $(\infty,1)$-category of quasi-coherent sheaves of modules $\QCoh{X}$ on $X$.
First, recall that the definition of module for a homotopy $\Coo$-algebra appears in \cite{Carchedi2019OnTU} and it is exactly the following.

\begin{definition}[Module for a homotopy $\Coo$-algebra]
A \textit{module for a homotopy $\Coo$-algebra} $R\in\mathbf{sC^\infty Alg}$ is a module for the underlying derived commutative algebra $R^\mathrm{alg}\in\mathbf{scAlg}_\bbR$.
\end{definition}

Here $\mathbf{scAlg}_\bbR$ is the $(\infty,1)$-category of derived commutative $\bbR$-algebras, i.e. simplicial commutative $\bbR$-algebras with the classical simplicial algebra model structure.
In the following, let $\quasicat$ be the $(\infty,1)$-category of $(\infty,1)$-categories.

For any given simplicial commutative $\bbR$-algebra $A\in\mathsf{scAlg}_\bbR$, let $\mathrm{N}A\in\mathsf{dgAlg_\bbR}$ be the dg-commutative algebra given by the normalized chains complex functor $\mathrm{N}:\mathsf{scAlg}_\bbR\longrightarrow\mathsf{dgAlg_\bbR}$ and let $\mathrm{N}A\text{-}\mathsf{Mod}$ be the category of $\mathrm{N}A$-dg-modules, which is naturally simplicially-enriched.
Moreover, let $\mathsf{W}_{\mathrm{qi}}$ be the set of quasi-isomorphisms in the category $\mathrm{N}A\text{-}\mathsf{Mod}$.
Thus we can define the $(\infty,1)$-functor
\begin{equation}
\begin{aligned}
    \mathrm{QCoh}\, :\;  \mathbf{dFMfd} \,&\longrightarrow\, \quasicat \\
    M \,&\longmapsto\, L_{\mathsf{W}_{\mathrm{qi}}}\mathrm{N}\mathcal{O}(M)^\mathrm{alg}\text{-}\mathsf{Mod},
\end{aligned}
\end{equation}
which sends any $\Coo$-algebra $A\in\mathsf{sC^\infty Alg}$ to the $(\infty,1)$-category obtained by simplicial localisation of the simplicial category of dg-modules of its underlying algebra.

Let us now provide a definition of quasi-coherent sheaves on a general derived smooth stack.
First, we must recall that any stack can be canonically written as a colimit of representables (see for instance \cite{Dugger2009SHEAVESAH}) by
\begin{equation}
    X \;\simeq\; \bfL\mathrm{co}\!\lim_{\!\!\!\!\!{\substack{U\rightarrow X}}} U.
\end{equation}

\begin{definition}[Quasi-coherent sheaves of modules]
Given any formal derived smooth stack $X\in\mathbf{dFSmoothStack}$, the $(\infty,1)$-category of \textit{quasi-coherent $(\infty,1)$-sheaves} on $X$ is given by the homotopy limit
\begin{equation}
    \QCoh{X} \;\simeq\; \bfR\!\lim_{\!\!\!\!\!{\substack{U\rightarrow X}}} \,\QCoh{U} \;\;\;\in\;\quasicat,
\end{equation}
where $U\in\mathbf{dFMfd}$ runs over all formal derived smooth manifolds.
\end{definition}

\begin{definition}[Complex of sections of a quasi-coherent $(\infty,1)$-sheaf]
The \textit{dg-vector space of global sections of a quasi-coherent $(\infty,1)$-sheaf of modules}  $\mathbb{M}_X\in\QCoh{X}$ is given by the functor
\begin{equation}
\begin{aligned}
        \bfR\Gamma(X,-) \,:\; \QCoh{X} &\longrightarrow\, \mathbf{dgVec}_\bbR\\[0.1ex]
        \mathbb{M}_X &\longmapsto\,\bfR\Gamma(X,\mathbb{M}_X),
\end{aligned}
\end{equation}
which is defined as the base change morphism $\bbR^0_\ast :\QCoh{X}\rightarrow\QCoh{\bbR^0}\simeq \mathbf{dgVec}_\bbR$ along the unique terminal morphism $\bbR^0:X\rightarrow \bbR^0$ to the point, where $\mathbf{dgVec}_\bbR$ is the $(\infty,1)$-category of dg-vector spaces.
\end{definition}

\begin{definition}[Quasi-coherent sheaf cohomology]
We define the \textit{quasi-coherent $(\infty,1)$-sheaf cohomology} $\mathrm{H}^n(X,\mathbb{M}_X)$ of any $\mathbb{M}_X\in\QCoh{X}$ on a given formal derived smooth stack $X\in\mathbf{dFSmoothStack}$ by the cohomology of the dg-vector space of its sections, i.e. by
\begin{equation}
    \mathrm{H}^n(X,\mathbb{M}_X) \;\coloneqq\; \mathrm{H}^{n\!}\big(\bfR\Gamma(X,\mathbb{M}_X)\big).
\end{equation}
\end{definition}\vspace{-0.3cm}

Recall that Dold-Kan correspondence gives us a Quillen equivalence $|-|:\mathsf{sSet}\;\begin{matrix}\leftarrow \\[-1.6ex] \rightarrow\end{matrix}\; \mathsf{dgVec}_\bbR^{\leq 0}:\mathrm{N}$ between simplicial sets and non-positively graded dg-vector spaces.
Then, if $A$ is a general dg-algebra, we will denote by $|A|$ the $\infty$-groupoid obtained by Dold-Kan correspondence applied to the dg-vector space given by the non-positive truncation of $A$.

\begin{definition}[$\infty$-groupoid of sections of a quasi-coherent sheaf]
The \textit{$\infty$-groupoid of $n$-shifted sections of a quasi-coherent $(\infty,1)$-sheaf} $\mathbb{M}_X\in\QCoh{X}$ on a formal derived smooth stack $X\in\mathbf{dFSmoothStack}$ is defined by the $\infty$-groupoid
\begin{equation}
    \mathcal{M}(X,n) \;\coloneqq\; \big| \bfR\Gamma(X,\mathbb{M}_X)[n] \big|.
\end{equation}
\end{definition}\vspace{-0.2cm}

Notice that the set of connected components of such a groupoid is related to the $n$-th quasi-coherent sheaf cohomology by the isomorphism $\pi_0\mathcal{M}(X,n) \cong \mathrm{H}^n(X,\mathbb{M}_X)$.

Let us now look at the most important examples of quasi-coherent $(\infty,1)$-sheaves on formal derived smooth stacks, which we will need in the rest of the paper.

\begin{example}[Structure sheaf]
Given a formal derived smooth stack $X\in\mathbf{dFSmoothStack}$, its \textit{structure sheaf} $\mathbb{O}_X\in\QCoh{X}$ is defined by the homotopy limit
\begin{equation}
    \mathbb{O}_X \;\coloneqq\;  \bfR\!\lim_{\!\!\!\!\!{\substack{U\rightarrow X}}}  \mathrm{N}\mathcal{O}(U)^\mathrm{alg},
\end{equation}
where, clearly, $\mathrm{N}\mathcal{O}(U)^\mathrm{alg}$ is in $\mathrm{N}\mathcal{O}(U)^\mathrm{alg}\text{-}\mathsf{Mod}$.
\end{example}

In analogy with \cite{Joyce:2009}, we want to define a cotangent complex for formal derived smooth stacks which is compatible with their smooth structure. In fact, even if in our definition a module of a $\Coo$-algebra is just a module of the underlying $\bbR$-algebra, we will introduce a cotangent module, whose definition is non-trivially reliant on the smooth structure of $\Coo$-algebras. 
We remark that, in the spirit of \cite{Joyce:2009}, such a cotangent module is not the usual one given by the usual K\"ahler differentials which one can find in algebraic geometry.

\begin{definition}[Cotangent module of a formal derived smooth manifold]\label{def:cotmod}
Let $U\in\mathbf{dFMfd}$ be a formal derived smooth manifold. The \textit{cotangent module} $\Omega^1_{\mathcal{O}(U)}\in \mathcal{O}(U)\text{-}\mathsf{Mod}$ is defined as the $\mathrm{N}\mathcal{O}(U)^\mathrm{alg}$-dg-module generated by elements of the form $\di_\dR f$, where $f\in \mathrm{N}\mathcal{O}(U)^\mathrm{alg}$ is any homogeneous element, such that the following conditions hold
\begin{enumerate}[label=(\textit{\roman*})]
    \item the degree of $\di_\dR f$ is the same as the degree of $f$,
    \item Leibniz's rule holds, i.e. $\di_\dR(f_1f_2)=(\di_\dR f_1)f_2 + (-1)^{|f_1|} f_1(\di_\dR f_2)$,
    \item for any $f_1,\cdots,f_n\in  \mathrm{N}\mathcal{O}(U)^\mathrm{alg}$ and any smooth map $\phi:\bbR^n\rightarrow \bbR$, we have
\begin{equation}
    \di_\dR \big(\mathrm{N}\mathcal{O}(U,\phi)(f_1,\dots,f_n)\big) \,=\, \sum_{i=1}^n \mathrm{N}\mathcal{O}\Big(U,\frac{\partial \phi}{\partial x^i}\Big)(f_1,\dots,f_n)\cdot\di_\dR f_i,
\end{equation}
where $\mathcal{O}(U,\phi):\mathcal{O}(U,\bbR)^n\rightarrow \mathcal{O}(U,\bbR)$ is the image of the smooth map $\phi$ on $\mathcal{O}(U)$.
\end{enumerate}
\end{definition}

By following \cite{Pridham:2021}, we can define the cotangent complex $\bbL_M\in\QCoh{M}$ of a formal derived smooth manifold $M\in\mathbf{dFMfd}$ by deriving the functor on the slice category $\mathsf{sC^{\infty}Alg}_{/\mathcal{O}(M)}$
\begin{equation}
    \Omega^1_{(-)\!}\,\widehat{\otimes}_{(-)\!}\,\mathcal{O}(M)\,: \;U\, \longmapsto \, \Omega^1_{U}\,\widehat{\otimes}_{\mathcal{O}(U)}\,\mathcal{O}(M),
\end{equation}
where $\widehat{\otimes}$ is the $\Coo$-tensor product of homotopy $\Coo$-algebras, and evaluating it at $M$.
More precisely, we can define the cotangent complex $\bbL_M \coloneqq \bfL\big(\Omega^1_{(-)}\widehat{\otimes}_{(-)}\mathcal{O}(M)\big)(M)$. In other words, we have $\bbL_M = \Omega^1_{Q\mathcal{O}(M)}\widehat{\otimes}_{Q\mathcal{O}(M)}\mathcal{O}(M)$, where $Q\mathcal{O}(M)$ is a cofibrant replacement of the original homotopy $\Coo$-algebra $\mathcal{O}(M)$. 

\begin{definition}[Cotangent complex]
The \textit{cotangent complex} $\bbL_X\in\QCoh{X}$ is defined by the homotopy limit
\begin{equation}
   \bbL_X \;\coloneqq\; \bfR\!\lim_{\!\!\!\!\!{\substack{U\rightarrow X}}} \, \bbL_{U},
\end{equation}
where $\bbL_{U}$ is the cotangent complex of the formal derived smooth manifold $U\in\mathbf{dFMfd}$ we introduced right above.
\end{definition}

\begin{definition}[Relative cotangent complex]
Any morphism $f:X\rightarrow Y$ of stacks induces a morphism $f_!:f^\ast\bbL_Y  \rightarrow\bbL_X$ of quasi-coherent $(\infty,1)$-sheaves. The \textit{relative cotangent complex} $\bbL_f\in\QCoh{X}$ is defined by the homotopy cofibre of such a map, i.e.
\begin{equation}
    \begin{tikzcd}[row sep=6.0ex, column sep=4.5ex]
    f^\ast\bbL_Y  \arrow[d] \arrow[r, "f_!"] & \bbL_X  \arrow[d, "\mathrm{hocofib}(f_!)"]\\
    0 \arrow[r] & \bbL_f .
    \end{tikzcd}
\end{equation}
\end{definition}

\begin{definition}[Tangent complex]
Whenever the cotangent complex $\bbL_X$ of a formal derived smooth stack $X\in\mathbf{dFSmoothStack}$ is a perfect complex, we can define the \textit{tangent complex} of $X$ by
\begin{equation}
    \bbT_X \;\coloneqq\; \bbL^\vee_X,
\end{equation}
where $\bbL_X^\vee\coloneqq[\bbL_X,\mathbb{O}_X]\in \QCoh{X}$ is the dual quasi-coherent sheaf of the cotangent complex.
\end{definition}

The $\infty$-groupoid of $n$-shifted vectors on $X\in\mathbf{dFSmoothStack}$ is given by the $\infty$-groupoid of $n$-shifted sections of $\bbT_X$, i.e. by $\mathfrak{X}(X,n) \;\coloneqq\; \big| \bfR\Gamma(X,\mathbb{T}_X)[n] \big|$.

%%%%%%%%%%%%%%%%%%%%%%%%%%%%%%%%%%%%%%%%%%%%%%%%%%%%%%%%%%%%%%%%%%%%%%%%%%%%%%%%%%%%%%%%%%%%%
\subsubsection{Derived de Rham algebra}

In this subsection we will provide a definition of differential forms on a formal derived smooth stack.
By using the fact that a module for a homotopy $\Coo$-algebra is defined as a module for the underlying derived commutative algebra, we will translate the formulation by \cite{ToenVezzo08, Toen14} in our framework.
Moreover, we will introduce the notion of formal derived smooth stack of differential forms on a formal derived smooth stack.

\begin{definition}[Complex of $p$-forms]
We define the \textit{complex of $p$-forms} on the derived stack $X\in\mathbf{dFSmoothStack}$ by the dg-vector space of sections
\begin{equation}
    \mathrm{A}^p(X) \,\coloneqq\, \bfR\Gamma(X,\wedge^p_{\mathbb{O}_X}\bbL_X).
\end{equation}
We denote by the symbol $\mathrm{A}^{\!p}(X)_n$ the degree $n\in\mathbb{Z}$ component of the dg-vector space $\mathrm{A}^{\!p}(X)$ and by $Q:\mathrm{A}^p(X)_n\rightarrow\mathrm{A}^p(X)_{n+1}$ its differential.
\end{definition}

\begin{remark}[Homotopy between $p$-forms]
A homotopy from an element $\alpha$ to an element $\beta$ of $\mathrm{A}^p(X)_n$ is given by an element $\gamma\in \mathrm{A}^{p}(X)_{n-1}$ such that
\begin{equation}
    \beta - \alpha \,=\, Q \gamma.
\end{equation}
\end{remark}

\begin{definition}[$n$-degree differential $p$-form]
An \textit{$n$-degree differential $p$-form} on a formal derived smooth stack $X\in\mathbf{dFSmoothStack}$ is defined as a cohomology class in
\begin{equation}
    \Omega^p(X)_n \,\coloneqq\, \mathrm{H}^n(\mathrm{A}^p(X)).
\end{equation}
\end{definition}

Notice that, in general, we obtain a bi-complex $\mathrm{A}^p(X)_n$ with $(p,n)\in\mathbb{N}\times\mathbb{Z}$ of the form
\begin{equation}
    \begin{tikzcd}[row sep=6.8ex, column sep=5.0ex]
    \vdots \arrow[d, "Q"] & \vdots \arrow[d, "Q"] & \vdots \arrow[d, "Q"] & \\
    \mathrm{A}^0(X)_{-2} \arrow[r, "\di_\dR"]\arrow[d, "Q"] & \mathrm{A}^1(X)_{-2} \arrow[r, "\di_\dR"]\arrow[d, "Q"] & \mathrm{A}^2(X)_{-2} \arrow[r, "\di_\dR"]\arrow[d, "Q"] & \,\cdots \\
    \mathrm{A}^0(X)_{-1} \arrow[r, "\di_\dR"]\arrow[d, "Q"] & \mathrm{A}^1(X)_{-1} \arrow[r, "\di_\dR"]\arrow[d, "Q"] & \mathrm{A}^2(X)_{-1} \arrow[r, "\di_\dR"]\arrow[d, "Q"] & \,\cdots \\
    \mathrm{A}^0(X)_{0} \arrow[r, "\di_\dR"] & \mathrm{A}^1(X)_{0} \arrow[r, "\di_\dR"] & \mathrm{A}^2(X)_{0} \arrow[r, "\di_\dR"] & \,\cdots  \\
    \vdots \arrow[u, "Q", leftarrow] & \vdots \arrow[u, "Q", leftarrow] & \vdots \arrow[u, "Q", leftarrow] &
    \end{tikzcd}
\end{equation}
where the following relations between de Rham and internal differentials are satisfied:
\begin{equation}
\begin{aligned}
    \di_\dR^2 \;=\; Q^2 \;=\; \di_\dR \circ Q + Q \circ \di_\dR \;=\; 0.
\end{aligned}
\end{equation}

We will now start introducing the technology which will allow us to deal with closed differential forms on derived formal smooth stacks.

\begin{definition}[Total de Rham dg-algebra]
The \textit{total de Rham algebra} is the dg-algebra whose underlying dg-vector space is defined by the totalisation
\begin{equation}
    \mathrm{DR}(X) \;\coloneqq\; \prod_{n\in\mathbb{N}}\mathrm{A}^n(X)[-n],
\end{equation}
with total differential $\di_\dR + Q$, where $\di_\dR$ is the de Rham differential and $Q$ is the internal differential of each dg-vector space $\mathrm{A}^p(X)$.
\end{definition}

\begin{definition}[Complex of closed $p$-forms]
Consider the following filtration of the total de Rham algebra:
\begin{equation}
    F^p\mathrm{DR}(X) \;=\; \prod_{n\geq p}\mathrm{A}^n(X)[-n] \;\subset\; \mathrm{DR}(X).
\end{equation}
The \textit{complex of closed $p$-forms} is defined for any $p\in\mathbb{N}$ by the following dg-vector space:
\begin{equation}
    \mathrm{A}^p_{\mathrm{cl}}(X) \;\coloneqq\; F^p\mathrm{DR}(X)[p].
\end{equation}
\end{definition}

\begin{remark}[Homotopy between closed $p$-forms]
A homotopy from an element $(\alpha_i)$ to $(\beta_i)$ in $\mathrm{A}^p_\mathrm{cl}(X)_n$ is given by an element $(\gamma_i)\in \mathrm{A}^{p}_\mathrm{cl}(X)_{n-1}$ such that
\begin{equation}
    \beta_i - \alpha_i \,=\, \di_\dR\gamma_{i-1} + Q \gamma_i.
\end{equation}
\end{remark}

\begin{definition}[Closed form]
An \textit{$n$-shifted closed $p$-form} on a derived formal smooth stacks $X$ is defined as an $n$-cocycle $(\omega_i)\in\mathrm{Z}^n\mathrm{A}^p_{\mathrm{cl}}(X)$ of the dg-vector space of closed $p$-forms on $X$, i.e. as an element $(\omega_i)\in\mathrm{A}^p_{\mathrm{cl}}(X)$ such that $(\di_\dR + Q)(\omega_i)=0$.
\end{definition}

In other words, an $n$-cocycle in $\mathrm{A}_\mathrm{cl}^p(X)$ is given by a formal sum $(\omega_i)=(\omega_p + \omega_{p+1} + \dots)$, where each form $\omega_{i}\in \mathrm{A}^{i}(X)$ is an element of degree $n+p-i$, satisfying the equations
\begin{equation}
\begin{aligned}
    Q\omega_{p} \,&=\, 0,\\
    \di_\dR\omega_{i} + Q\omega_{i+1} \,&=\, 0,
\end{aligned}
\end{equation}
for every $i\geq p$.

This embodies the idea that the underlying $p$-form $\omega_p\in \mathcal{A}^{p}(X)$ is de Rham-closed up to homotopy, which is given by a choice of higher forms $\omega_i$ with $i>p$.

\begin{definition}[$n$-degree closed differential $p$-form]
An \textit{$n$-degree closed $p$-form} is defined as a cohomology class in
\begin{equation}
    \Omega^p_\mathrm{cl}(X)_n \,\coloneqq\, \mathrm{H}^n(\mathrm{A}_\mathrm{cl}^p(X)).
\end{equation}
\end{definition}

\begin{definition}[$\infty$-groupoid of differential forms]
We define the \textit{$\infty$-groupoid of differential $p$-forms} $\mathcal{A}^p(X,n)$ and \textit{of closed differential $p$-forms} $\mathcal{A}^p_{\mathrm{cl}}(X,n)$ by
\begin{equation}
\begin{aligned}
    \mathcal{A}^p(X,n) \;&\simeq\; \big|\mathrm{A}^p(X)[n]\big|, \\[0.7ex]
    \mathcal{A}^p_{\mathrm{cl}}(X,n) \;&\simeq\; \big|\mathrm{A}^p_{\mathrm{cl}}(X)[n]\big|,
    \end{aligned}
\end{equation}
where the functor $|-|:\mathsf{dgcAlg}_\bbR\rightarrow\sSet$, as before, is the Dold-Kan correspondence functor applied on the non-positive truncation of the argument. 
\end{definition}

\begin{remark}[Differential forms from $\infty$-groupoid of differential forms]\label{rem:homcat}
Notice that the $\infty$-groupoid of differential $p$-forms $\mathcal{A}^p(X,n)$ and of closed differential $p$-forms $\mathcal{A}^p_{\mathrm{cl}}(X,n)$ have the following sets of connected components
\begin{equation}
\begin{aligned}
        \pi_0\mathcal{A}^p(X,n) \,\;&\simeq\;\,   \mathrm{H}^n(\mathrm{A}^p(X))_{\phantom{\mathrm{cl}}} \,\eqqcolon\, \Omega^p(X)_n,\\[0.5ex]
        \pi_0\mathcal{A}^p_\mathrm{cl}(X,n) \,\,&\simeq\,\,  \mathrm{H}^n(\mathrm{A}^p(X)_\mathrm{cl}) \;\eqqcolon\; \Omega^p_\mathrm{cl}(X)_n.
\end{aligned}
\end{equation}
\end{remark}

As we discussed in \cref{subsec:formalsmoothset}, in ordinary smooth geometry it is possible to construct a smooth set $\Omega^p$ such that the hom-set $\Hom(M,\Omega^p)$ in the category of smooth sets from a smooth manifold $M$ to $\Omega^p$ is exactly the set of differential forms $\Omega^p(M)\in\mathsf{Set}$. This (formal) smooth set $\Omega^p$ also known as moduli space of differential $p$-forms.
We will now construct something analogous for formal derived smooth stacks.

\begin{proposition}[Derived stack of differential forms]\label{def:stacksiffform}
There exist formal derived smooth stacks $\pmb{\mathcal{A}}^p(n)$ and $\pmb{\mathcal{A}}^p_\mathrm{cl}(n)$ satisfying respectively the universal properties
\begin{equation}
\begin{aligned}
    \bfR\Hom\big(X,\,\pmb{\mathcal{A}}^p(n)\big) \;&\simeq\; \mathcal{A}^p(X,n), \\[0.7ex]
    \bfR\Hom\big(X,\,\pmb{\mathcal{A}}^p_{\mathrm{cl}}(n)\big) \;&\simeq\; \mathcal{A}^p_{\mathrm{cl}}(X,n),
    \end{aligned}
\end{equation}
where $X$ is any formal derived smooth stack and $\bfR\Hom(-,-)$ is the hom-$\infty$-groupoid of the $(\infty,1)$-category $\mathbf{dFSmoothStack}$.
\end{proposition}

\begin{proof}
First, notice that we can immediately define a pre-stack $\pmb{\mathcal{A}}^p(n):U\mapsto\mathcal{A}^p(U,n)$ on the $(\infty,1)$-category $\mathbf{dFMfd}$ of formal derived smooth manifolds.
The fact that this satisfies the descent respect to the $(\infty,1)$-\'etale site structure of $\mathbf{dFMfd}$ is a consequence of the fact that the functor $U\mapsto \wedge_{\mathbb{O}_U}^p\!\bbL_U$ with $U\in\mathbf{dFMfd}$ satisfies descent, as $\wedge_{\mathbb{O}_U}^p\!\bbL_U\in\QCoh{U}$ is a quasi-coherent $(\infty,1)$-sheaf on any $U$. 
We have the following chain of equivalences:
\begin{equation}
    \begin{aligned}
    \bfR\Hom\big(X,\pmb{\mathcal{A}}^p(n)\big) \;&\simeq\; \bfR\Hom\big(\bfL\mathrm{co}\!\!\lim_{\!\!U\rightarrow X} U, \,\pmb{\mathcal{A}}^p(n) \big)  \\
    &\simeq\; \bfR\!\!\lim_{\!\!U\rightarrow X} \bfR\Hom\big(U,\, \pmb{\mathcal{A}}^p(n)\big) \\
    &\simeq\; \bfR\!\!\lim_{\!\!U\rightarrow X} \mathcal{A}^p(U,n) \\
    &\simeq\;  \mathcal{A}^p(X,n).
    \end{aligned}
\end{equation}
Moreover, by a completely analogous argument, also the pre-stack $\pmb{\mathcal{A}}^p_\mathrm{cl}(n)$ satisfies descent.
\end{proof}

\begin{definition}[Derived stack of differential forms]
We call $\pmb{\mathcal{A}}^p(n)$ the \textit{formal derived smooth stacks of differential $p$-forms} and $\pmb{\mathcal{A}}^p_{\mathrm{cl}}(n)$ the one \textit{of closed differential $p$-forms}.
Moreover, we will write $\pmb{\mathcal{A}}^p\coloneqq\pmb{\mathcal{A}}^p(0)$ and $\pmb{\mathcal{A}}^p_\mathrm{cl}\coloneqq\pmb{\mathcal{A}}^p_\mathrm{cl}(0)$ for the $0$-shifted cases.
\end{definition}

\begin{corollary}[Differential forms from the homotopy category]
By putting together remark \ref{rem:homcat} and proposition \ref{def:stacksiffform}, we have the following equivalences of sets
\begin{equation}
\begin{aligned}
    \Hom_{\mathrm{Ho}}\big(X,\,\pmb{\mathcal{A}}^p(n)\big) \;&\simeq\; \pi_0\mathcal{A}^p(X,n) \, \;\simeq\; \Omega^p(X)_n, \\[0.7ex]
    \Hom_{\mathrm{Ho}}\big(X,\,\pmb{\mathcal{A}}^p_{\mathrm{cl}}(n)\big) \;&\simeq\; \pi_0\mathcal{A}^p_{\mathrm{cl}}(X,n) \;\simeq\; \Omega^p_\mathrm{cl}(X)_n,
    \end{aligned}
\end{equation}
where $\Hom_{\mathrm{Ho}}(-,-)$ is the hom-set of the homotopy category $\mathrm{Ho}(\mathbf{dFSmoothStack})$ of formal derived smooth stacks.
Therefore, a morphism $\xi:X\rightarrow \pmb{\mathcal{A}}^p(n)$ in the homotopy category $\mathrm{Ho}(\mathbf{dFSmoothStack})$ is equivalently an $n$-shifted $p$-form $\xi\in\Omega^p(X)_n$. Similarly for $\pmb{\mathcal{A}}^p_\mathrm{cl}(n)$.
\end{corollary}
\newpage

\begin{example}[Derived zero locus]
The affine derived zero locus $\bfR f^{-1}(0)\in\mathbf{dFMfd}$ of a smooth function $f:\bbR^n\rightarrow\bbR^k$ is a formal derived smooth manifold defined by a homotopy pullback of the following form
\begin{equation}
\begin{tikzcd}[column sep={6em,between origins}, row sep={5.5em,between origins}]
   \bfR f^{-1}(0) \arrow[r]\arrow[d] & \bbR^n\arrow[d, "{(0,\mathrm{id})}"]\\
   \bbR^n \arrow[r, "{(f,\mathrm{id})}"] & \bbR^{k+n},
\end{tikzcd}
\end{equation}
where $\mathrm{id}:\bbR^n\rightarrow \bbR^n$ is the identity, in the $(\infty,1)$-category of derived manifolds. For more details about its algebraic geometric version see \cite{Vezzosi:2011}.
The tangent complex will be given by $\bbT_{\bfR f^{-1}(0)} = \Big(T_{\bbR^n}[0] \xrightarrow{\;\;f_\ast\;\;} f^\ast T_{\bbR^k}[-1]\Big )$, concentrated in cohomological degree $0$ and $1$.
In degree $1$ we have the sheaf $f^\ast T_{\bbR^k}\simeq\Coo_{\bbR^n}(-,\bbR^k)$.
Analogously, the cotangent complex will be $\bbL_{\bfR f^{-1}(0)} = \Big(f^\ast\Omega_{\bbR^k}^1[1] \xrightarrow{\;\;f_\ast\;\;} \Omega_{\bbR^n}^1[0] \Big)$, concentrated in  cohomological degree $-1$ and $0$. In degree $-1$ we have the sheaf $f^\ast\Omega_{\bbR^k} \simeq \Coo_{\bbR^n}(-,(\bbR^{k})^\vee)$.
Thus, by unravelling its definition, the complex of $0$-forms is the following:
\begin{equation}
\begin{aligned}
    \mathrm{A}^0(\bfR f^{-1}(0)) \;&=\; \bfR\Gamma(\bfR f^{-1}(0),\mathbb{O}_{\bfR f^{-1}(0)}) \\
    &=\; \Coo(\bbR^n) \otimes_\bbR \wedge^{\!\ast} (\bbR^{k})^\vee,
\end{aligned}
\end{equation}
where the differential is given by $Qx^i = 0$ and $Qx^+_j = f_j(x)$, on $\{x^i\}_{i=1,\dots,n}$ global coordinates of $\bbR^n$ in degree $0$ and $\{x_j^+\}_{j=1,\dots,k}$ the generators of the exterior algebra $\wedge^\ast (\bbR^{k})^\vee$ in degree $-1$.
By unravelling its definition, we can explicitly see that the complex of $1$-forms is the following:
\begin{equation}
\begin{aligned}
    \mathrm{A}^1(\bfR f^{-1}(0)) \;&=\; \bfR\Gamma(\bfR f^{-1}(0),\mathbb{L}_{\bfR f^{-1}(0)}) \\
    &=\; \bigoplus_{i=1}^{n}\mathcal{A}^0(\bfR f^{-1}(0))[\di x^i] \oplus \bigoplus_{j=1}^{k}\mathcal{A}^0(\bfR f^{-1}(0))[\di x_j^+], 
\end{aligned}
\end{equation}
with the graded-commutation relations given by the equations
\begin{equation}
    \di x^i\wedge \di x^j = - \di x^j\wedge \di x^i, \,\quad \di x^i\wedge \di x_j^+ =  \di x_j^+\wedge \di x^i, \,\quad \di x_i^+\wedge \di x_j^+ =  \di x_j^+\wedge \di x_i^+.
\end{equation}
Similarly, one obtains all the differential $p$-forms.
\end{example}

%%%%%%%%%%%%%%%%%%%%%%%%%%%%%%%%%%%%%%%%%%%%%%%%%%%%%%%%%%%%%%%%%%%%%%%%%%%%%%%%%%%%%%%%%%%%%
\section{Derived differential geometry}\label{sec:derdiffcohe}

In the previous section, we constructed the $(\infty,1)$-category $\mathbf{dFSmoothStack}$ of formal derived smooth stacks.

In this section, we show that the formalism of differential structures, introduced by Schreiber \cite{DCCTv2} in the setting of formal smooth stacks, extends very naturally to our present setting of formal derived smooth stacks.
Many statements and constructions follow through very naturally.

Since it is known that an $(\infty,1)$-category of stacks is an $(\infty,1)$-topos (see e.g. \cite{ToenVezzo05,topos}), the $(\infty,1)$-category $\mathbf{dFSmoothStack}$ is, in particular, an $(\infty,1)$-topos.
In subsection \ref{subsec:ddiffcohe}, we show that the $(\infty,1)$-topos of formal derived smooth stacks comes naturally equipped with a differential structure.
Roughly speaking, a differential structure provides an $(\infty,1)$-topos with the properties required for differential geometry to take place in it and for its objects to be fully-fledged formal spaces.
In subsection \ref{subsec:FMP}, we will show that the formal moduli problems appearing in BV-theory naturally arise in the context of derived differential structures.
In the last two subsections, we explore some entailments of such a structure, including generalisations of the notions of $L_\infty$-algebroids and jet bundles.

%%%%%%%%%%%%%%%%%%%%%%%%%%%%%%%%%%%%%%%%%%%%%%%%%%%%%%%%%%%%%%%%%%%%%%%%%%%%%%%%%%%%%%%%%%%%%
\subsection{Derived differential $(\infty,1)$-topos}\label{subsec:ddiffcohe}

After showing that formal derived smooth stacks constitute an $(\infty,1)$-topos, we will investigate its natural differential structure.

In the previous section, we stressed the fact that we are working not on the site of derived smooth manifolds, but, slightly more generally, on the site of formal derived smooth manifolds.
Now we will directly exploit the formal aspects of our formal derived smooth stacks.
The $(\infty,1)$-topos $\mathbf{dFSmoothStack}$ is naturally equipped with a differential structure, as defined in \cite[Section 4.2.1]{DCCTv2}. Such a differential structure includes a functor $\Im$ sending a formal derived smooth stack to its de Rham space, which can be thought as its infinitesimal path groupoid (see the reference for more details).

The notion of differential topos can be traced back to the seminal work of \cite{Simpson:1997, kontsevich1998noncommutative}. The concept of a differential topos provides a unifying framework for studying a range of structures, including formal smooth manifolds and, more generally, spaces that admit some notion of local chart and infinitesimal extension. For a detailed and comprehensive discussion of differential structures, we point at the main reference.
At its core, a differential topos is a category of sheaves over a site that satisfies certain axioms, which ensure that it has enough structure to capture formal geometry. 

First, let us look at the global sections functor for formal derived smooth stacks.
Every ordinary topos of sheaves $\mathsf{Sh(C)}$ on some site $\mathsf{C}$ comes naturally equipped with a global section functor $\mathit{\Gamma}:\mathsf{Sh(C)}\rightarrow\mathsf{Set}$ which sends a sheaf $X$ to the section $\mathit{\Gamma}(X)\coloneqq \Hom(\ast,X)$ at the point (i.e. at the terminal object, which exists).
The global sections functor $\mathit{\Gamma}$ naturally fits into a geometric morphism, which is given by the adjunction $\mathrm{Disc}\dashv\mathit{\Gamma}$, where the functor $\mathrm{Disc}:\mathsf{Set}\rightarrow \mathsf{Sh(C)}$ embeds sets into the corresponding locally constant sheaves.
As explained by \cite{DCCTv2}, this construction can be generalised to a $(\infty,1)$-topos of stacks if we replace the ordinary category of sets with the $(\infty,1)$-category of $\infty$-groupoids.

\begin{remark}[Terminal geometric morphism]
The terminal geometric morphism on an $(\infty,1)$-topos $\mathbf{H}$ is the datum of a pair of adjoint $(\infty,1)$-functors of the following form \footnote{In a previous version of this paper was wrongly claimed that the $(\infty,1)$-topos of formal derived smooth stacks is cohesive. We would like to thank David Carchedi for pointing out the issue.}:
\begin{equation}\label{eq:tergeom_strcuture}
    \begin{tikzcd}[row sep={12ex,between origins}, column sep={12ex}]
    \mathbf{H}\,\arrow[r, "\mathit{\Gamma}"',"\text{\rotatebox[origin=c]{-90}{$\dashv$}}", shift right=1.1ex] &  \mathbf{\infty Grpd} , \arrow[l, "\mathrm{Disc}"', shift right=1.1ex, hook']   
    \end{tikzcd}
\end{equation}
such that:
\begin{enumerate}[label=(\textit{\roman*})]
    \item the $(\infty,1)$-functor $\mathit{\Gamma}$ is the global section functor.
    \item the $(\infty,1)$-functors $\mathrm{Disc}$.
\end{enumerate} 
\end{remark}

\begin{remark}[Global section functor factors through $t_0$]
Notice that the point $\ast\simeq \bbR^0\in\mathbf{dFSmoothStack}$ lies in the essential image of $\mathsf{Mfd}$. This immediately implies that the global section functor $\mathit{\Gamma}(-)=\bfR\Hom(\bbR^0,-)$ will factor through the underived-truncation $t_0$.
\end{remark}

\begin{example}[Global sections of a formal derived smooth set]
Because of the remark right above, the global sections $\mathit{\Gamma}(X)$ of a formal derived smooth set $X\in\mathbf{dFSmoothSet}$ will be nothing but a set $\mathit{\Gamma}(X)=\bfR\Hom(\bbR^0,X)\simeq \Hom(\bbR^0,t_0X)$.
\end{example}

\begin{definition}[Flat modality]
We define the \textit{flat modality} as the following endofunctor:
\begin{equation}
   \flat\coloneqq \mathrm{Disc}\circ \mathit{\Gamma}\;:\,\mathbf{dFSmoothStack}\,\longrightarrow\,\mathbf{dFSmoothStack}.
\end{equation}
\end{definition}

The term "modality" was imported by \cite{DCCTv2} from type theory.
The flat modality sends a formal derived smooth stack $X$ to the formal derived smooth stack $\flat X$ with the same underlying simplicial set, but which has forgotten all the formal derived smooth structure, i.e.
\begin{equation}\label{eq:flat_points}
    \flat X \;\simeq\, \coprod_{x:\ast\rightarrow X}\! \ast.
\end{equation}
To see this, we use the fact that any $\infty$-groupoid $S\in\mathbf{\infty Grpd}$ is equivalent to the colimit $S \simeq \coprod_{S} \ast$. Moreover, since it is a left adjoint, the functor $\mathrm{Disc}$ preserves colimits, but we can see that it also preserves the terminal object. Thus, we have the equivalences of formal derived smooth stacks $\mathrm{Disc}(S) \simeq \coprod_S \mathrm{Disc}(\ast) \simeq \coprod_S \ast$. Equivalence \eqref{eq:flat_points} is then obtained by choosing the $\infty$-groupoid to be the one of global sections $S=\mathit{\Gamma}(X)$ of some formal derived smooth stack.

%%%%%%%%%%%%%%%%%%%%%%%%%%%%%%%%%%%%%%%%%%%%%%%%%%%%%%%%%%%%%%%%%%%%%%%%%%%%%%%%%%%%%%%%%%%%%

A differential topos is a category of sheaves over a site that satisfies certain axioms, which ensure that it has enough structure to capture the basic features of smooth and topological spaces. A differential topos includes a natural notion of differentiation and integration that allows us to define differential forms and cohomology on the sheaves.
The following definition appears in \cite[Section 4.2.1]{DCCTv2}.

\begin{definition}[Differential structure]
A \textit{differential structure} on an $(\infty,1)$-topos $\mathbf{H}$ is the datum of a sub-$(\infty,1)$-topos $\mathbf{H}^\mathrm{red}$ which is embedded via a quadruple of adjoint functors
\begin{equation}\label{eq:diff_strcuture}
    \begin{tikzcd}[row sep={12ex,between origins}, column sep={16ex}]
    \mathbf{H}\, \arrow[r, "\mathit{\Gamma}^\mathrm{dif}", shift left=-4.5ex]\arrow[r, "\mathit{\Pi}^\mathrm{dif}", shift right=-1.5ex] &  \,\mathbf{H}^\mathrm{red} , \arrow[l, "\hat{\imath}"', shift right=4.5ex, hook']  \arrow[l, "\mathrm{Disc}^\mathrm{dif}"', shift left=1.5ex, hook']  
    \end{tikzcd}
\end{equation}
such that the functor $\hat{\imath}$ is fully faithful and preserves finite products.
\end{definition}

Let $\mathsf{C^\infty Alg}^\mathrm{red}\hookrightarrow\mathsf{C^\infty Alg}$ be the full and faithful sub-category of reduced $\Coo$-algebras, i.e. of those $\Coo$-algebras whose underlying $\bbR$-algebra is reduced in the usual sense.
We introduce the reduction functor by the following assignment:
\begin{equation}
\begin{aligned}
    (-)^\mathrm{red}\,:\; \mathsf{sC^\infty Alg} \,&\longrightarrow\, \mathsf{C^\infty Alg}^\mathrm{red} \\
    R\,&\longmapsto\, R^\mathrm{red}\,\coloneqq\,\pi_0R/\mathfrak{m}_{\pi_{0\!}R}
\end{aligned}
\end{equation}
where $\pi_0R$ is the ordinary $\Coo$-algebra given by the coequaliser $\pi_0R = \mathrm{coeq}(\!\!\begin{tikzcd}[row sep=scriptsize, column sep=2.5ex] R_1\! \arrow[r, shift left=0.8ex]\arrow[r, shift left=-0.8ex] &  \!R_0  \!\!\end{tikzcd})$ of the face maps of the $1$-simplices of $R$ (which exists, since all finite colimits exist in $\mathsf{C^\infty  Alg}$, see e.g. \cite{Joyce:2009}) and where $\mathfrak{m}_{\pi_{0\!}R}\subset \pi_0R$ is the nilradical of $\pi_0R$, i.e. the ideal consisting of the nilpotent elements of $\pi_0R$ regarded as an $\bbR$-algebra. Recall from example \ref{ex:weil_algebra} that the quotient $R^\mathrm{red}=\pi_0R/\mathfrak{m}_{\pi_{0\!}R}$ of a $\Coo$-algebra by any of its ideals is canonically a $\Coo$-algebra. Recall adjunction \eqref{eq:coref}.
Now, we can see that we have a simplicial Quillen adjunction $(-)^\mathrm{red}\dashv\iota^\mathrm{red}$, where $\iota^\mathrm{red}:\mathsf{C^\infty Alg}^\mathrm{red} \hookrightarrow \mathsf{sC^\infty Alg}$ is the natural embedding (in fact, $(-)^\mathrm{red}$ automatically preserves cofibrant objects and $\iota^\mathrm{red}$ fibrant objects).
Now, we can restrict everything to finitely generated algebras and obtain the following simplicial Quillen adjunction:
\begin{equation}\label{eq:coref_inf}
    \begin{tikzcd}[row sep=scriptsize, column sep=8.2ex, row sep=18.0ex]
     \mathsf{C^\infty Alg}_\mathrm{fg}^\mathrm{red} \arrow[r, "\iota^\mathrm{red}"', shift left=-0.8ex, hook] &  \mathsf{sC^\infty Alg}_\mathrm{fg}. \arrow[l, "(-)^\mathrm{red}"', shift left=-0.8ex]
    \end{tikzcd}
\end{equation}
Since it is a simplicial Quillen adjunction, it gives naturally rise to a reflective embedding of $(\infty,1)$-categories $\mathbf{N}(\mathsf{C^\infty Alg}_\mathrm{fg}^\mathrm{red})\;\begin{matrix}\leftarrow \\[-1.6ex] \hookrightarrow\end{matrix}\; \mathbf{sC^{\infty\!}Alg}_\mathrm{fg}$.

\begin{construction}[Diagram of sites]
Let us now denote by $\mathsf{C^\infty Var}^\mathrm{red}\coloneqq (\mathsf{sC^\infty Alg}_\mathrm{fg}^\mathrm{red})^\op$ the category of reduced $\Coo$-varieties.
We can extend the diagram of ordinary sites from remark \ref{rem:diagram_of_sites} to include the $(\infty,1)$-category of formal derived smooth manifolds.
Thus, by putting all together, we have the following diagram of $(\infty,1)$-sites:
\begin{equation}
    \begin{tikzcd}[row sep=scriptsize, column sep=16.2ex, row sep=8.0ex]
     \mathsf{Mfd} \arrow[dd, "", shift left=-1.0ex, hook]\arrow[rrd, "\upiota^\mathsf{Mfd}", hook, bend left=55] \arrow[r, "\iota^{\mathsf{Mfd}}", hook] & \mathsf{C^\infty Var}^\mathrm{red} \arrow[dd, "", shift left=-1.0ex, hook] \arrow[rd, "\iota^\mathrm{red}", shift left=1.0ex, hook] &     \\
     && \,\mathsf{dFMfd} \arrow[ld, "", shift left=+1.0ex] \arrow[lu, "(-)^\mathrm{red}", shift left=+1.0ex]\\
     \mathsf{FMfd} \arrow[rru, "\upiota^\mathsf{FMfd}"', hook, bend right=55] \arrow[uu, "", shift left=-1.0ex]\arrow[r, "\iota^{\mathsf{FMfd}}", hook] & \mathsf{C^\infty Var}\arrow[uu, "", shift left=-1.0ex] \arrow[ru, "\iota", shift left=1.0ex, hook] &  
    \end{tikzcd}
\end{equation}
\end{construction}\vspace{-0.4cm}

The diagram of $(\infty,1)$-sites we constructed above encodes all the relations between the relevant sites in the context of derived smooth geometry and it is going to be the main ingredient to show the following theorems of this subsection.

\begin{theorem}[Differential $(\infty,1)$-topos of formal derived smooth stacks]
The $(\infty,1)$-topos $\mathbf{dFSmoothStack}$ of formal derived smooth stacks is naturally equipped with a differential structure, i.e. with a quadruplet of adjoint $(\infty,1)$-functors
\begin{equation}\label{eq:diffcohesiondiag}
    \begin{tikzcd}[row sep={12ex,between origins}, column sep={16ex}]
    \mathbf{dFSmoothStack}\, \arrow[r, "\mathit{\Gamma}^\mathrm{dif}", shift left=-4.5ex]\arrow[r, "\mathit{\Pi}^\mathrm{dif}", shift right=-1.5ex] &  \,\mathbf{SmoothStack}^{\pmb{+}} , \arrow[l, "\hat{\imath}"', shift right=4.5ex, hook']  \arrow[l, "\mathrm{Disc}^\mathrm{dif}"', shift left=1.5ex, hook']  
    \end{tikzcd}
\end{equation}
such that the functor $\hat{\imath}$ is fully faithful and preserves finite products.
\end{theorem}

\begin{proof}
Recall that we have an equivalence $\mathsf{sC^\infty Alg}_\mathrm{fg}^\op \simeq \mathsf{dFMfd}$ of the opposite category of finitely generated $\Coo$-algebras to the category of formal derived smooth manifolds.
By left and right Kan extension, the reflective embedding \eqref{eq:coref_inf} of simplicial sites induces the following Quillen adjunctions:
\begin{equation}
\begin{aligned}
    &\begin{tikzcd}[row sep={18ex,between origins}, column sep={24ex}]
     {[\mathsf{sC^\infty Alg}_\mathrm{fg},\mathsf{sSet}]_\mathrm{proj}} \arrow[r, "\,\iota^\mathrm{red\ast\!} \,\simeq\,(-)^{\mathrm{red}}_!", shift right=+1.5ex] &  {[\mathsf{C^\infty Alg}^\mathrm{red}_\mathrm{fg},\mathsf{sSet}]_\mathrm{proj}} \arrow[l, "\iota^\mathrm{red}_!\phantom{\,\,\simeq\,(-)^{\mathrm{red}}_!}"', shift left=-1.5ex],
    \end{tikzcd} \\[1.1ex]
    &\begin{tikzcd}[row sep={18ex,between origins}, column sep={24ex}]
     {[\mathsf{sC^\infty Alg}_\mathrm{fg},\mathsf{sSet}]_\mathrm{proj}} \arrow[r, "\,\iota^\mathrm{red\ast\!} \,\simeq\,(-)^{\mathrm{red}}_!", shift right=-1.5ex] &  {[\mathsf{C^\infty Alg}^\mathrm{red}_\mathrm{fg},\mathsf{sSet}]_\mathrm{proj}}  \arrow[l, "\,\,\iota^\mathrm{red}_\ast\,\,\,\simeq\,(-)^{\mathrm{red}\ast\!}"', shift right=-1.5ex],
    \end{tikzcd} \\[1.1ex]
    &\begin{tikzcd}[row sep={18ex,between origins}, column sep={25ex}]
     {[\mathsf{sC^\infty Alg}_\mathrm{fg},\mathsf{sSet}]_\mathrm{inj}} \arrow[r, "\phantom{\,\iota^\mathrm{red}_!\,\,\simeq\,}(-)^\mathrm{red}_\ast", shift left=-1.5ex] &  {[\mathsf{C^\infty Alg}^\mathrm{red}_\mathrm{fg},\mathsf{sSet}]_\mathrm{inj}}  \arrow[l, "\,\,\iota^\mathrm{red}_\ast\,\,\,\simeq\,(-)^{\mathrm{red}\ast\!}"', shift right=+1.5ex] ,
    \end{tikzcd}
\end{aligned}
\end{equation}  
which encodes a quadruple of adjoint $(\infty,1)$-functors between the corresponding $(\infty,1)$-categories of pre-stacks.
Now we must show that the adjunctions above give rise to a quadruplet of adjoint $(\infty,1)$-functors from stacks to stacks.
First, we notice that the functors $(\iota^{\mathrm{red}})^\op$ and $((-)^{\mathrm{red}})^\op$ preserve \'etale maps. Moreover, $((-)^{\mathrm{red}})^\op$ preserves limits, as it is a right adjoint.
This is enough to show that $\iota^{\mathrm{red}}_!$ preserves \'etale hypercovers and, therefore, that $\iota^{\mathrm{red}\ast}$ preserves locally fibrant objects.
Similarly, we have that $(-)^{\mathrm{red}}_!$ preserves \'etale hypercovers and, thus, that $(-)^{\mathrm{red}\ast}$ preserves locally fibrant objects.
Thus, by \cite[Section 4.8]{ToenVezzo05}, we have that the functor $\iota^\mathrm{red}$ is continuous and cocontinuous, which means that the triplet of adjoint functors $(\iota^{\mathrm{red}}_!\dashv\iota^{\mathrm{red}\ast}\dashv\iota^{\mathrm{red}}_\ast)$ restricts and corestricts to stacks.

Now, we have left to show that the last Quillen adjunction above restricts and corestricts to stacks or, in other words, that the pre-stack $(-)^\mathrm{red}_\ast X\in[\mathsf{C^\infty Alg}^\mathrm{red}_\mathrm{fg},\mathsf{sSet}]_\mathrm{inj}$ satisfies descent for any locally fibrant object $X\in[\mathsf{sC^\infty Alg}_\mathrm{fg},\mathsf{sSet}]_\mathrm{loc,inj}$.
Since the the adjunction exists already at the level of global model structure, we have the $(\infty,1)$-adjunctions
\begin{equation*}
\begin{tikzcd}[row sep={18ex,between origins}, column sep={9ex}]
     \mathbf{dFSmoothStack}\arrow[r, "", shift left=-1.5ex, hook] & {\mathbf{PredFSmoothStack}} \arrow[l, "L"', shift right=+1.5ex] \arrow[r, "(-)^\mathrm{red}_\ast", shift left=-1.5ex] &  {\mathbf{PreSmoothStack}^+}  \arrow[l, "\iota^\mathrm{red}_\ast"', shift right=+1.5ex] ,
    \end{tikzcd}
\end{equation*} 
where $\mathbf{PredFSmoothStack}$ and $\mathbf{PreSmoothStack}^+$ are the $(\infty,1)$-categories of pre-stacks respectively on $\mathsf{dFMfd}$ and $\mathsf{C^\infty Var}^\mathrm{red}$.
Let us denote by $H(U) \coloneqq \bfL\mathrm{colim}_n H(U)_n$ the geometric realisation of a hypercover.
By \cite[Proposition 3.5.2]{ToenVezzo05}, it is sufficient to check that the map $\iota^{\mathrm{red}}_\ast H(U)\rightarrow \iota^{\mathrm{red}}_\ast U$ is a local equivalence of formal derived smooth pre-stacks for any \'etale hypercover $H(U)_\bullet$ of any representable $U\in\mathsf{C^\infty Var}^\mathrm{red}$.
To check this fact, we construct the composite $\varphi:\iota^{\mathrm{red}}_!\rightarrow  \iota^{\mathrm{red}}_\ast\iota^{\mathrm{red}\ast}\iota^{\mathrm{red}}_! \simeq \iota^{\mathrm{red}}_\ast$.
Notice that $\varphi_U:L\iota^{\mathrm{red}}_! U \rightarrow \iota^{\mathrm{red}}_\ast U$ is a $\tau$-covering for the \'etale topology in the sense of \cite{ToenVezzo05}.
Therefore, by \cite[Corollary 3.3.4]{ToenVezzo05}, if the pullback
\begin{equation}\label{eq:map_loc_equiva}
    \iota^{\mathrm{red}}_! U \times^h_{\iota^{\mathrm{red}}_\ast U}\iota^{\mathrm{red}}_\ast H(U) \,\longrightarrow\, \iota^{\mathrm{red}}_! U
\end{equation}
is a local equivalence, we have that the morphism $\iota^{\mathrm{red}}_\ast H(U) \rightarrow \iota^{\mathrm{red}}_\ast U$ is a local equivalence too.
For this reason, it is enough to show that \eqref{eq:map_loc_equiva} is a local equivalence.
Since the morphism $H(U)\rightarrow U$ is formally \'etale, there is a pullback square
\begin{equation}
    \begin{tikzcd}[row sep=scriptsize, column sep=6.0ex, row sep=9.5ex]
    \iota^{\mathrm{red}}_! H(U) \arrow[d, "\varphi_{H(U)}"] \arrow[r, ""] & \iota^{\mathrm{red}}_! U \arrow[d, "\varphi_U"] \\
    \iota^{\mathrm{red}}_\ast H(U)  \arrow[r, ""] & \iota^{\mathrm{red}}_\ast U,
    \end{tikzcd} 
\end{equation}
implying that $\iota^{\mathrm{red}}_! H(U) \simeq \iota^{\mathrm{red}}_! U \times^h_{\iota^{\mathrm{red}}_\ast U}\iota^{\mathrm{red}}_\ast H(U)$.
At this point, since the functor $\iota^{\mathrm{red}}$ is continuous (as we have seen above), the map $\iota^{\mathrm{red}}_! H(U) \rightarrow \iota^{\mathrm{red}}_! U$ is already a local equivalence. 
Thus, the morphism \eqref{eq:map_loc_equiva} is a local equivalence, and so is the morphism $\iota^{\mathrm{red}}_\ast H(U) \rightarrow \iota^{\mathrm{red}}_\ast U$. In particular, this means that there is an equivalence of formal derived smooth stacks $L\iota^{\mathrm{red}}_\ast H(U) \simeq L\iota^{\mathrm{red}}_\ast U$.
This is enough to shows that the adjunction $(\iota^\mathrm{red}_\ast\dashv (-)_\ast^\mathrm{red})$ restricts to stacks.

Thus, so far, we have constructed a quadruple of adjoint $(\infty,1)$-functors of $(\infty,1)$-categories, which we will denote by
\begin{equation}
    \begin{tikzcd}[row sep={18ex,between origins}, column sep={24ex}]
     \mathbf{dFSmoothStack} \arrow[r, "\phantom{\,\iota^\mathrm{red}_!\,\,\simeq\,}(-)^\mathrm{red}_\ast", shift left=-4.5ex]\arrow[r, "\,\iota^\mathrm{red\ast\!} \,\simeq\,(-)^{\mathrm{red}}_!", shift right=-1.5ex] &  \mathbf{SmoothStack}^{\pmb{+}}  \arrow[l, "\,\,\iota^\mathrm{red}_\ast\,\,\,\simeq\,(-)^{\mathrm{red}\ast\!}"', shift right=-1.5ex]  \arrow[l, "\iota^\mathrm{red}_!\phantom{\,\,\simeq\,(-)^{\mathrm{red}}_!}"', shift left=-4.5ex],
    \end{tikzcd}
\end{equation}
where $\mathbf{SmoothStack}^{\pmb{+}}=\mathbf{N}_{hc}([\mathsf{C^\infty Alg}^\mathrm{red}_\mathrm{fg},\mathsf{sSet}]_\mathrm{proj,loc}^\circ)$ is by definition the $(\infty,1)$-topos of stacks on the ordinary \'etale site of reduced $\Coo$-varieties $\mathsf{C^\infty Var}^\mathrm{red}=(\mathsf{C^\infty Alg}_\mathrm{fg}^\mathrm{red})^\op$, which we constructed in definition \ref{def:extended_formal_smooth}. 
Now, we have left to show that $\iota^\mathrm{red}_!$ is fully faithful and preserves finite products. As for the first property, $\iota^\mathrm{red}_!,\iota^\mathrm{red}_\ast$ are both fully faithful, since $\iota^\mathrm{red}$ fully faithful implies that $\mathrm{id}\rightarrow\iota^\mathrm{red\ast}\iota^\mathrm{red}_!$ and $\iota^\mathrm{red\ast}\iota^\mathrm{red}_\ast \rightarrow \mathrm{id}$ are object-wise equivalences.
As for the second one, it is sufficient to show that for any formal derived smooth stack $X$ and formal derived smooth manifold $U$ the functor
$$X \;\longmapsto\; \bfL\mathrm{colim}\big(\iota^\mathrm{red}\!\!\downarrow\!\mathcal{O}(U) \,\rightarrow\, \mathsf{C^\infty Alg}^\mathrm{red}_\mathrm{fg} \,\xrightarrow{\,\;X\;\,}\, \sSet\big)$$
preserves finite products, which is the case if the comma category $\iota^\mathrm{red}\!\!\downarrow\!\mathcal{O}(U)$ has finite coproducts. This is equivalent to $U\!\!\downarrow\!\!(\iota^\mathrm{red})^\op$ having finite products, which, since $(\iota^\mathrm{red})^\op$ preserves finite products, is true.
Therefore, if we redefine the functors by $\hat{\imath}\coloneqq \iota^\mathrm{red}_!$, $\mathit{\Pi}^\mathrm{dif}\coloneqq \iota^\mathrm{red\ast} \simeq (-)^{\mathrm{red}}_!$, $\mathrm{Disc}^{\mathrm{dif}}\coloneqq \iota^\mathrm{red}_\ast \simeq (-)^{\mathrm{red}\ast}$ and $\mathit{\Gamma}^\mathrm{dif}\coloneqq (-)^{\mathrm{red}}_\ast$, we have the conclusion.
\end{proof}

In the terminology of \cite[Definition 4.2.1]{DCCTv2}, the quadruple of adjoint functors in diagram \eqref{eq:diffcohesiondiag} characterises the $(\infty,1)$-topos $\mathbf{dFSmoothStack}$ as an infinitesimal neighbourhood of the $(\infty,1)$-topos $\mathbf{SmoothStack}^{\pmb{+}}$.
Intuitively speaking, this tells us that, in a certain sense, any stack in $\mathbf{dFSmoothStack}$ can be thought of as an infinitesimal extension of some stack in $\mathbf{SmoothStack}^{\pmb{+}}$. 

\begin{remark}[Interpretation of reduced and co-reduced objects]
In analogy with non-derived differential structures, we could call the functor $\hat{\imath}$ inclusion of \textit{reduced objects} and $\mathrm{Disc}^{\mathrm{dif}}$ inclusion of \textit{co-reduced objects}.
\begin{itemize}
    \item The reduced objects are, intuitively, the ones whose infinitesimal and derived behaviour is determined by their non-infinitesimal ordinary behavior;  
    \item on the other hand, the co-reduced objects are the ones who are lacking of any infinitesimal and derived behaviour.
\end{itemize}
Finally, the functor $\mathit{\Pi}^\mathrm{dif}$ can be thought of as the functor which contracts away the infinitesimal and derived extension of a formal derived smooth stack.
\end{remark}

\begin{remark}[Extending smooth stacks into formal derived smooth stacks]
Notice that we have a diagram of $(\infty,1)$-categories
\begin{equation}
    \begin{tikzcd}[row sep=scriptsize, column sep=9.2ex]
     \mathbf{N}(\mathsf{Mfd}) \arrow[r, "\iota^{\mathsf{Mfd}}", hook] & \mathbf{N}(\mathsf{C^\infty Var}^\mathrm{red}) \arrow[r, "\iota^\mathrm{red}", shift left=+0.8ex, hook] &  \mathbf{dFMfd} \arrow[l, "(-)^\mathrm{red}", shift left=+0.8ex] 
    \end{tikzcd}
\end{equation}
where $\iota^{\mathsf{Mfd}}$ is the full and faithful embedding of smooth manifolds into reduced finitely generated $\Coo$-algebras. Notice that such an embedding does not come with a natural adjoint. In fact, we can always see a smooth manifold as a $\Coo$-variety, but there is no standard way to make a $\Coo$-variety into a smooth manifold.
Thus, the diagram above gives rise to a diagram of $(\infty,1)$-categories of the form
\begin{equation*}
    \begin{tikzcd}[row sep={18ex,between origins}, column sep={17ex}]
     \mathbf{dFSmoothStack} \arrow[r, "\phantom{\,\iota^\mathrm{red}_!\,\,\simeq\,}(-)^\mathrm{red}_\ast", shift left=-4.5ex]\arrow[r, "\,\iota^\mathrm{red\ast\!} \,\simeq\,(-)^{\mathrm{red}}_!", shift right=-1.5ex] &  \mathbf{SmoothStack}^{\pmb{+}}  \arrow[l, "\,\,\iota^\mathrm{red}_\ast\,\,\,\simeq\,(-)^{\mathrm{red}\ast\!}"', shift right=-1.5ex, hook]  \arrow[l, "\iota^\mathrm{red}_!\phantom{\,\,\simeq\,(-)^{\mathrm{red}}_!}"', shift left=-4.5ex, hook]                                                                                    \arrow[r, "\phantom{\ast}\iota^{\mathsf{Mfd}\ast}", shift right=-1.5ex] &  \mathbf{SmoothStack}   \arrow[l, "\iota^\mathsf{Mfd}_!"', shift left=-4.5ex, hook].
    \end{tikzcd}
\end{equation*}
We have that the derived-extension functor is equivalently the composition 
$i = \iota^{\mathrm{red}}_! \circ \iota^{\mathsf{Mfd}}_!$ and the underived-truncation functor is
$t_0 = \iota^{\mathsf{Mfd}\ast}\circ\iota^{\mathrm{red}\ast}$ constructed in proposition \ref{prop:ext-trunc}.
\end{remark}

\begin{lemma}[Underived-truncation and derived-extension of affine $\Coo$-schemes]\label{lem:reduction}
We have the following results about affine $\Coo$-schemes.
\vspace{-0.2cm}\begin{itemize}
    \item For any given ordinary reduced affine $\Coo$-scheme $\Spec(R)\in\mathsf{C^\infty Aff}$ corresponding to the ordinary reduced $\Coo$-algebra $R\in\mathsf{C^\infty Alg}^\mathrm{red}$ we have an equivalence 
\begin{equation}
    \hat{\imath}_{\,}\Spec(R) \;\simeq\; \bfR\Spec(\iota^\mathrm{red}(R))
\end{equation}
in $\mathbf{dC^{\infty\!\!}Aff}$, which corresponds to the homotopy $\Coo$-algebra $\iota^\mathrm{red}(R)\in\mathsf{sC^\infty Alg}$.
    \item For any given homotopy $\Coo$-algebra $R\in\mathsf{sC^\infty Alg}$, we have the equivalence of ordinary smooth sets
    \begin{equation}
        \mathit{\Pi}^\mathrm{dif}\bfR\Spec (R) \;\simeq\; \Spec (R^\mathrm{red}),
    \end{equation}
where $\Spec (R^\mathrm{red})$ is the ordinary affine $\Coo$-scheme corresponding to the reduced ordinary $\Coo$-algebra $R^\mathrm{red}\in\mathbf{C^\infty Alg}^\mathrm{red}$.
\end{itemize}
\end{lemma}

\begin{proof}
For any finitely generated $\Coo$-algebra $A\in\mathsf{sC^\infty Alg}_\mathrm{fg}$, we have the equivalences
\begin{equation}
    \begin{aligned}
        (\mathit{\Pi}^\mathrm{dif}\bfR\Spec R)(A) \;&\simeq\; (\bfR\Spec R)(\iota^\mathrm{red} A) \\
        \;&\simeq\; \Hom_{\mathsf{sC^{\infty\!}Alg}}(R,\iota^\mathrm{red} A) \\
        \;&\simeq\; \Hom_{\mathsf{C^{\infty\!}Alg}}(R^{\mathrm{red}},A) \\
        \;&\simeq\; (\Spec R^\mathrm{red})(A)
    \end{aligned}
\end{equation}
where in the penultimate line we used the adjunction $(-)^\mathrm{red}\dashv \iota^\mathrm{red}$. Thus, the conclusion.
\end{proof}

Just like any quadruple of adjoint $(\infty,1)$-functors, the derived differential structure presented by diagram \eqref{eq:diffcohesiondiag} gives naturally rise to a triplet of adjoint $(\infty,1)$-endofunctors.

\begin{definition}[Modalities of derived differential structure]\label{def:infinitesimalshape}
We define the following endofunctors:
\begin{equation}
   (\Re \;\dashv\;  \Im \;\dashv\; \&):\,\mathbf{dFSmoothStack}\,\longrightarrow\,\mathbf{dFSmoothStack},
\end{equation}
where we respectively called
\begin{enumerate}[label=(\textit{\roman*})]
    \item \textit{infinitesimal reduction modality} $\Re \coloneqq \hat{\imath}\circ \mathit{\Pi}^\mathrm{dif}$,
    \item \textit{infinitesimal shape modality} $\Im \coloneqq \mathrm{Disc}^\mathrm{dif} \circ \mathit{\Pi}^\mathrm{dif}$,
    \item \textit{infinitesimal flat modality} $\& \coloneqq \mathrm{Disc}^\mathrm{dif} \circ \mathit{\Gamma}^\mathrm{dif}$.
\end{enumerate}
\end{definition}

The modalities of our derived differential structure will constitute our fundamental toolbox in dealing with the geometry of formal derived smooth stacks.

\begin{remark}[Infinitesimal reduction counit]
Since there is an adjunction $(\hat{\imath}\dashv\mathit{\Pi}^\mathrm{dif})$, there will be an adjunction counit $\mathfrak{r}:\Re\rightarrow \mathrm{id}$, which, at any $X\in\mathbf{dFSmoothStack}$, will give rise to the canonical morphism
\begin{equation}
\begin{aligned}
    \mathfrak{r}_X:\,\Re(X)\,\longrightarrow\,X.
\end{aligned}
\end{equation}
We will call this \textit{infinitesimal reduction counit}, for short.
\end{remark}

Since by construction we have $\mathit{\Pi}^\mathrm{dif}\circ\hat{\imath}\simeq \mathrm{id}$, it is possible to see that the infinitesimal reduction modality is an idempotent comonad, i.e. we have an equivalence
\begin{equation}
    \Re \,\xrightarrow{\;\simeq\;}\, \Re\Re.
\end{equation}

Let us show the geometric meaning of the infinitesimal reduction counit more concretely.
The following corollary will provide a concrete characterisation of the infinitesimal reduction on the relevant class of examples of derived affine $\Coo$-schemes.

\begin{corollary}[Infinitesimal reduction of derived affine $\Coo$-schemes]
For any given homotopy $\Coo$-algebra $R\in\mathsf{sC^\infty Alg}$, by lemma \ref{lem:reduction} we directly obtain the equivalence
\begin{equation}
    \Re(\bfR\Spec R) \;\simeq\; \bfR\Spec (R^\mathrm{red}).
\end{equation}
\end{corollary}

Roughly speaking, we can see that the infinitesimal reduction modality is the functor whcih contracts away the formal derived directions of a formal derived smooth stack.

\begin{definition}[Infinitesimal object]
We say that $X$ is an \textit{infinitesimal object} if $\Re(X)\simeq \ast$.
\end{definition}

Notice that the infinitesimal reduction counit of an infinitesimal object $X\in \mathbf{dFSmoothStack}$ becomes the embedding of the canonical point $\mathfrak{r}_X: \ast \rightarrow X$.

\begin{definition}[Reduced object]
We say that $X$ is a \textit{reduced object} if $\Re(X)\simeq X$.
\end{definition}

Notice that the infinitesimal reduction counit of a reduced object $X\in \mathbf{dFSmoothStack}$ becomes the identity $\mathfrak{r}_X: X \xrightarrow{\;\simeq\;} X$.

\begin{remark}[Infinitesimal shape unit]
Since there is an adjunction $(\mathit{\Pi}^\mathrm{dif}\dashv\mathrm{Disc}^\mathrm{dif})$, there will be an adjunction unit $\mathfrak{i}:\mathrm{id}\rightarrow \Im$, which, at any $X\in\mathbf{dFSmoothStack}$, will give rise to the canonical morphism
\begin{equation}
\begin{aligned}
    \mathfrak{i}_X:\,X\,\longrightarrow\,\Im(X),
\end{aligned}
\end{equation}
We will call this \textit{infinitesimal shape unit}, for short.
\end{remark}

Similarly to $\Re$, the infinitesimal shape modality is an idempotent monad, i.e. we have an equivalence
\begin{equation}
    \Im\Im  \,\xrightarrow{\;\simeq\;}\,  \Im.
\end{equation}

Let us show the geometric meaning of the infinitesimal shape unit more concretely. Let us consider a derived formal smooth stack $X\in\mathbf{dFSmoothStack}$. Then we can see that the infinitesimal shape modality will send it to the formal derived smooth stack
\begin{equation}
    \begin{aligned}
        \Im(X)\,:\; \mathsf{dFMfd}^\op \;&\longmapsto\; \sSet  \\
        U \;&\longmapsto\; X(t_0U),
    \end{aligned}
\end{equation}
where $t_0 U$ is the underived-truncation of the formal derived smooth manifold $U$.
Moreover, the infinitesimal shape unit $\mathfrak{i}_X:X\rightarrow\Im(X)$ of $X$ will be concretely given by the natural map of simplicial sets
\begin{equation}
    \mathfrak{i}_X(U)\,:\; X(U) \;\longrightarrow\; X(t_0U)
\end{equation}
on each formal derived smooth manifold $U\in\mathbf{dFMfd}$ in our site.

\begin{definition}[de Rham space]
The \textit{de Rham space} of a formal derived smooth stack $X\in\mathbf{dFSmoothStack}$ is defined by the formal derived smooth stack $\Im(X)\in\mathbf{dFSmoothStack}$.
\end{definition}

\begin{remark}[$\mathscr{D}$-modules]
The $(\infty,1)$-category of $\mathscr{D}$-modules on a formal derived smooth stack $X\in\mathbf{dFSmoothStack}$ can be defined precisely by $\mathscr{D}(X) \coloneqq \QCoh{\Im(X)}$, i.e. by the $(\infty,1)$-category of quasi-coherent sheaves on its de Rham space $\Im(X)$. 
\end{remark}

\begin{remark}[Infinitesimal flat unit]
Since there is an adjunction $(\mathrm{Disc}^\mathrm{dif}\dashv \mathit{\Gamma}^\mathrm{dif})$, there will be an adjunction unit $\mathfrak{e}:\mathrm{id}\rightarrow\&$, which, at any $X\in\mathbf{dFSmoothStack}$, will give rise to the canonical morphism
\begin{equation}
\begin{aligned}
    \mathfrak{e}_X:\,X\,\longrightarrow\,\&(X).
\end{aligned}
\end{equation}
We will call this \textit{infinitesimal flat unit}, for short.
\end{remark}

Similarly to $\Re$ and $\Im$, the infinitesimal flat modality is an idempotent comonad, i.e. we have an equivalence
\begin{equation}
    \&  \,\xrightarrow{\;\simeq\;}\,  \&\&.
\end{equation}

\begin{remark}[Analogy with derived algebraic geometry]
The adjoint $(\infty,1)$-functors $(\Re\dashv\Im)$ from the derived differential structure described above can be thought of as a smooth version of the adjunction constructed by \cite[Section 2]{Calaque:2017} in the context of derived algebraic geometry.
\end{remark}

In the rest of this subsection, we will provide generalisations of the formal geometric objects constructed in \cite{khavkine2017synthetic} to formal derived smooth stacks.

\begin{remark}[Points on a formal derived smooth stack]
Notice that the point $\ast \simeq \bfR\Spec(\bbR)$ is the terminal object in $\mathbf{dFSmoothStack}$. Thus, the Hom-space of morphisms  $\ast \rightarrow X$ from the point into any formal derived smooth stack $X\in\mathbf{dFSmoothStack}$ is nothing but the simplicial set $\mathit{\Gamma}(X)=X(\ast)\in\mathsf{sSet}_{\mathrm{Quillen}}^\circ$. Therefore, we can equivalently give a point $x:\ast \rightarrow X$ on the formal derived smooth stack $X\in\mathbf{dFSmoothStack}$ as an element $x\in \mathit{\Gamma}(X)=X(\ast)$.
\end{remark}

We can now provide a well-defined notion of formal neighborhood of a formal derived smooth stack at any of its points.

\begin{definition}[Formal disk]\label{def:formal_disk}
The \textit{formal disk} $\bbD_{X,x}$ at the point $x:\ast\rightarrow X$ of the formal derived smooth stack $X\in\mathbf{dFSmoothStack}$ is defined by the homotopy pullback
\begin{equation}
    \begin{tikzcd}[row sep=scriptsize, column sep=6.5ex, row sep=8ex]
    \bbD_{X,x}\arrow[d] \arrow[r, hook, "\iota_x"] & X \arrow[d, "\mathfrak{i}_X"] \\
    \ast \arrow[r, hook, "x"] & \Im(X)
    \end{tikzcd}
\end{equation}
in the $(\infty,1)$-category $\mathbf{dFSmoothStack}$ of formal derived smooth stacks.
\end{definition}

\begin{figure}[h]
    \centering
\tikzset {_9126sosej/.code = {\pgfsetadditionalshadetransform{ \pgftransformshift{\pgfpoint{0 bp } { 0 bp }  }  \pgftransformscale{1 }  }}}
\pgfdeclareradialshading{_61jtso174}{\pgfpoint{0bp}{0bp}}{rgb(0bp)=(0.29,0.56,0.89);
rgb(0bp)=(0.29,0.56,0.89);
rgb(14.939547947474887bp)=(0.29,0.56,0.89);
rgb(25bp)=(1,1,1);
rgb(400bp)=(1,1,1)}
\tikzset{_8m362wbff/.code = {\pgfsetadditionalshadetransform{\pgftransformshift{\pgfpoint{0 bp } { 0 bp }  }  \pgftransformscale{1 } }}}
\pgfdeclareradialshading{_9gk41blp1} { \pgfpoint{0bp} {0bp}} {color(0bp)=(transparent!44.99999999999999);
color(0bp)=(transparent!44.99999999999999);
color(14.939547947474887bp)=(transparent!75);
color(25bp)=(transparent!100);
color(400bp)=(transparent!100)} 
\pgfdeclarefading{_ez8r125a6}{\tikz \fill[shading=_9gk41blp1,_8m362wbff] (0,0) rectangle (50bp,50bp); } 
\tikzset{every picture/.style={line width=0.75pt}} %set default line width to 0.75pt        
\begin{tikzpicture}[x=0.75pt,y=0.75pt,yscale=-1,xscale=1]
%Shape: Circle [id:dp8849563442501762] 
\path  [shading=_61jtso174,_9126sosej,path fading= _ez8r125a6 ,fading transform={xshift=2}] (62.38,76.12) .. controls (62.38,69.8) and (67.51,64.67) .. (73.83,64.67) .. controls (80.16,64.67) and (85.29,69.8) .. (85.29,76.12) .. controls (85.29,82.45) and (80.16,87.58) .. (73.83,87.58) .. controls (67.51,87.58) and (62.38,82.45) .. (62.38,76.12) -- cycle ; % for fading 
 \draw  [color={rgb, 255:red, 74; green, 144; blue, 226 }  ,draw opacity=0.59 ] (62.38,76.12) .. controls (62.38,69.8) and (67.51,64.67) .. (73.83,64.67) .. controls (80.16,64.67) and (85.29,69.8) .. (85.29,76.12) .. controls (85.29,82.45) and (80.16,87.58) .. (73.83,87.58) .. controls (67.51,87.58) and (62.38,82.45) .. (62.38,76.12) -- cycle ; % for border 
%Shape: Square [id:dp7468487220743263] 
\draw   (9.92,10.5) -- (141.17,10.5) -- (141.17,141.75) -- (9.92,141.75) -- cycle ;
%Shape: Circle [id:dp22656634080903415] 
\draw  [fill={rgb, 255:red, 0; green, 0; blue, 0 }  ,fill opacity=1 ] (72.13,76.12) .. controls (72.13,75.18) and (72.89,74.42) .. (73.83,74.42) .. controls (74.78,74.42) and (75.54,75.18) .. (75.54,76.12) .. controls (75.54,77.07) and (74.78,77.83) .. (73.83,77.83) .. controls (72.89,77.83) and (72.13,77.07) .. (72.13,76.12) -- cycle ;
% Text Node
\draw (143.17,13.9) node [anchor=north west][inner sep=0.75pt]  [font=\footnotesize]  {$X$};
% Text Node
\draw (79,55.9) node [anchor=north west][inner sep=0.75pt]  [font=\footnotesize]  {$\bbD_{X,x}$};
% Text Node
\draw (69.5,79.4) node [anchor=north west][inner sep=0.75pt]  [font=\scriptsize]  {$x$};
\end{tikzpicture}
    \caption{The formal disk $\bbD_{X,x}$ at a point $x:\ast\rightarrow X$ of a formal derived smooth stack.}
    \label{fig:disk}
\end{figure}
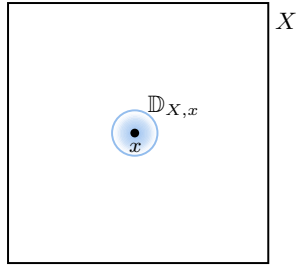

In other words, the formal disk $\bbD_{X,x}$ is the fibre at the point $x:\ast\rightarrow X\rightarrow\Im(X)$ of the bundle provided by the infinitesimal shape unit $\mathfrak{i}_X:X\longrightarrow\Im(X)$, which is the canonical morphism from the stack $X$ to its de Rham space.

\begin{remark}[Formal disk is infinitesimal]
Notice that the formal disk is an infinitesimal object, since we have the natural equivalence $\Re(\bbD_{X,x})\simeq\ast$.
This can be seen by unpacking the definition $\Re \coloneqq \hat{\imath}\circ \mathit{\Pi}^\mathrm{dif}$. In fact, object-wise, one has $(\mathit{\Pi}^\mathrm{dif}\bbD_{X,x})(A)\simeq \bbD_{X,x}(A^\mathrm{red})\simeq \ast$ at any $A$. The, by embedding back trough the functor $\hat{\imath}$, the final object is preserved.
\end{remark}

Before we proceed further, let us provide an extremely simple example of formal disk, namely a formal disk on the real line.

\begin{example}[Formal disk $\mathbb{D}_{\bbR,0}$ on $\bbR$]
Recall that the real line can be Yoneda-embedded into a formal derived smooth stack by the functor $\bbR(U) \simeq \mathcal{O}(U,\,\bbR)$, on any $U\in\mathsf{dFMfd}$.
Let us consider the formal disk $\bbD_{\bbR,0}\xhookrightarrow{\;\iota_{0}}\bbR$ defined as above at the zero point $0\in\bbR$ on the real line. Thus sections are given by pullback of simplicial sets 
\begin{equation}
    \bbD_{\bbR,0}(U) \;\simeq\; \mathcal{O}(U,\,\bbR) \times_{\mathcal{O}(\mathit{\Pi}^\mathrm{dif}U,\,\bbR)}^h \{0\}
\end{equation}
on any formal derived smooth manifold $U\in\mathsf{dFMfd}$.
This means that, for example: 
\begin{itemize}
    \item if $U$ is in the essential image of an ordinary smooth manifold, i.e. if $\Re(U)\simeq U$, then we have that the space of sections is just $\bbD_{\bbR,0}(U)\simeq\{0\}$;
    \item on the other hand, if $U$ is a derived thickened point, i.e. if $\Re(U)\simeq \ast$, then we have that the space of sections $\bbD_{\bbR,0}(U)$ is given by the simplicial set of nilpotent elements of $\mathcal{O}(U)$.
\end{itemize}
Notice that the infinitesimal derived behaviour of $\bbD_{\bbR,0}$ is seen only by formal derived probing spaces $U$ and if we try to probe $\bbD_{\bbR,0}$ with ordinary smooth manifolds we do not see anything but a point. This shows why we can think at $\bbD_{\bbR,0}$ as a derived thickened point.
\end{example}

Now, we can introduce the notion of formal disk bundle, i.e. a fibre bundle of formal disks.

\begin{definition}[Formal disk bundle]\label{def:formal_disk_bundle}
The \textit{formal disk bundle} $T^\infty X$ of a formal derived smooth stack $X\in\mathbf{dFSmoothStack}$ is defined by the homotopy pullback
\begin{equation}
    \begin{tikzcd}[row sep=scriptsize, column sep=6ex, row sep=8.5ex]
    T^\infty X\arrow[d, "\pi"] \arrow[r, "\mathrm{ev}"] & X \arrow[d, "\mathfrak{i}_X"] \\
    X \arrow[r, "\mathfrak{i}_X"] & \Im(X),
    \end{tikzcd}
\end{equation}
in the $(\infty,1)$-category $\mathbf{dFSmoothStack}$ of formal derived smooth stacks.
\end{definition}

\begin{remark}[Formal disk as fibre of the formal disk bundle]
Notice that the fibre at any point $x:\ast\rightarrow X$ of the bundle $T^\infty X \rightarrow X$ is an infinitesimal disk $\bbD_{X,x}\xhookrightarrow{\;\iota_x\;}X$ at such point.
\end{remark}

If we are given a formal derived smooth sub-set of our original formal derived smooth set, we can consider a natural notion of infinitesimal normal bundle. This is given as follows.

\begin{definition}[\'Etalification]
The \textit{\'etalification} $\Im_XY $ of a formal derived smooth stack $Y$ respect to a map $f:Y\rightarrow X$ in $\mathbf{dFSmoothStack}$ is defined by the homotopy pullback
\begin{equation}
    \begin{tikzcd}[row sep=scriptsize, column sep=6ex, row sep=7.5ex]
   \Im_XY \arrow[d] \arrow[r, "\iota_Y"] & X \arrow[d, "\mathfrak{i}_X"] \\
    \Im(Y) \arrow[r, "\Im(f)"] & \Im(X)
    \end{tikzcd}
\end{equation}
in the $(\infty,1)$-category $\mathbf{dFSmoothStack}$ of formal derived smooth stacks.
\end{definition}

\begin{definition}[Normal formal disk bundle]
The \textit{normal formal disk bundle} $N^\infty_X Y$ of a monomorphism $Y\xhookrightarrow{\;e\;}X$ of formal derived smooth stacks in $\mathbf{dFSmoothStack}$ is defined by the homotopy pullback
\begin{equation}
    \begin{tikzcd}[row sep=scriptsize, column sep=6ex, row sep=7.5ex]
    N^\infty_X Y\arrow[d] \arrow[r] & \Im_XY \arrow[d] \\
    Y \arrow[r, "e", hook] & X
    \end{tikzcd}
\end{equation}
in the $(\infty,1)$-category $\mathbf{dFSmoothStack}$ of formal derived smooth stacks.
\end{definition}

\begin{figure}[h]
    \centering
\tikzset {_b3yxkmzjy/.code = {\pgfsetadditionalshadetransform{ \pgftransformshift{\pgfpoint{0 bp } { 0 bp }  }  \pgftransformrotate{-90 }  \pgftransformscale{2 }  }}}
\pgfdeclarehorizontalshading{_p06qgihmx}{150bp}{rgb(0bp)=(1,1,1);
rgb(37.5bp)=(1,1,1);
rgb(49.732142857142854bp)=(0.29,0.56,0.89);
rgb(62.5bp)=(1,1,1);
rgb(100bp)=(1,1,1)}
\tikzset{_ja06e2yp7/.code = {\pgfsetadditionalshadetransform{\pgftransformshift{\pgfpoint{0 bp } { 0 bp }  }  \pgftransformrotate{-90 }  \pgftransformscale{2 } }}}
\pgfdeclarehorizontalshading{_w1snrc4xm} {150bp} {color(0bp)=(transparent!100);
color(37.5bp)=(transparent!100);
color(49.732142857142854bp)=(transparent!72);
color(62.5bp)=(transparent!100);
color(100bp)=(transparent!100) } 
\pgfdeclarefading{_tjkcxzjz1}{\tikz \fill[shading=_w1snrc4xm,_ja06e2yp7] (0,0) rectangle (50bp,50bp); } 
\tikzset{every picture/.style={line width=0.75pt}} %set default line width to 0.75pt        
\begin{tikzpicture}[x=0.75pt,y=0.75pt,yscale=-1,xscale=1]
%Shape: Path Data [id:dp1749026653238681] 
\draw  [draw opacity=0][shading=_p06qgihmx,_b3yxkmzjy,path fading= _tjkcxzjz1 ,fading transform={xshift=2}] (140.17,74.58) .. controls (140.5,74.51) and (140.84,74.44) .. (141.17,74.37) -- (141.17,91.32) .. controls (140.84,91.41) and (140.5,91.49) .. (140.17,91.58) .. controls (80.86,107.23) and (74.76,58.67) .. (9.92,80.36) -- (9.92,63.35) .. controls (71.2,40.8) and (77.81,88.05) .. (140.17,74.58) -- cycle ;
%Curve Lines [id:da5128444985082148] 
\draw [line width=0.75]    (10.5,72) .. controls (70.67,50.08) and (80.67,96.08) .. (141.17,83.08) ;
%Curve Lines [id:da18652551107800575] 
\draw [color={rgb, 255:red, 74; green, 144; blue, 226 }  ,draw opacity=0.54 ][line width=0.75]    (9.5,80.5) .. controls (69.67,58.58) and (79.67,104.58) .. (140.17,91.58) ;
%Curve Lines [id:da5363794137771529] 
\draw [color={rgb, 255:red, 74; green, 144; blue, 226 }  ,draw opacity=0.54 ][line width=0.75]    (9.5,63.5) .. controls (69.67,41.58) and (79.67,87.58) .. (140.17,74.58) ;
%Shape: Square [id:dp7468487220743263] 
\draw   (9.92,10.5) -- (141.17,10.5) -- (141.17,141.75) -- (9.92,141.75) -- cycle ;
% Text Node
\draw (143.17,13.9) node [anchor=north west][inner sep=0.75pt]  [font=\footnotesize]  {$X$};
% Text Node
\draw (143,74.9) node [anchor=north west][inner sep=0.75pt]  [font=\footnotesize]  {$Y$};
% Text Node
\draw (59.6,45.2) node [anchor=north west][inner sep=0.75pt]  [font=\footnotesize]  {$N_{X}^{\infty } Y$};
\end{tikzpicture}
    \caption{The normal formal disk bundle of a formal derived smooth stack $Y\hookrightarrow X$.}
    \label{fig:normaldisk}
\end{figure}
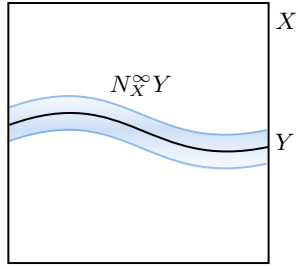

\begin{example}[Trivial embedding]
Notice that, if we consider the trivial formal embedding $e:X\xrightarrow{\;\mathrm{id}\;}X$, then we immediately have the identification $N^\infty_XX\simeq X$, i.e. the bundle with trivial fibre.
\end{example}

\begin{example}[Case of formal disk]
Notice that, if we consider the embedding $e:\ast\xrightarrow{\;x\;}X$ of a point, then we immediately have the identification $N^\infty_X\ast\simeq \bbD_{X,x}$, i.e. the formal disk at $x$.
\end{example}

%The next example will be important in \cite{AlfYouFuture}.

\begin{example}[Thickened hypersurface]\label{ex:thicksurface}
Let $M\simeq\Sigma\times \bbR$ be a smooth manifold and let $\Sigma_{0}=\Sigma\times \{0\}$ be a submanifold for a fixed element $t_0\in\bbR$. Thus, we have the normal formal disk bundle $N^\infty_M\Sigma_{0} = \Sigma\times \bbD_{\bbR,0}$, where $\bbD_{\bbR,0}$ is the formal disk of $\bbR$ at $0$. 
Let us look at the formal embedding map of the normal formal disk bundle $N^\infty_M\Sigma_{0}$ into $M$ in detail, i.e. at the map
\begin{equation}
    N^\infty_M\Sigma_{0} \,\simeq\, \Sigma\times \bbD_{\bbR,0} \;\xhookrightarrow{\;\;\;\;\iota_{\Sigma}\;\;\;\;}\; \Sigma\times \bbR \,\simeq\, M.
\end{equation}
This can be understood dually by the map 
\begin{equation}
\begin{aligned}
    \Coo(M)\,\simeq\,\Coo(\Sigma\times \bbR) \;&\xtwoheadrightarrow{\;\;\mathcal{O}(\iota_\Sigma)\;\;}\; \mathcal{O}(\Sigma\times \bbD_{\bbR,0}) \,\simeq\, \mathcal{O}(N^\infty_M\Sigma_{0}) \\[0.7ex]
    f(x,t) \;&\xmapsto{\;\;\phantom{\Coo(\iota_\Sigma)}\;\;}\; f(x,0) + \sum_{n>0}\left.\frac{\partial^{n\!} f(x,t)}{\partial t^n}\right|_{t=0}\!\!\!t^n
    \end{aligned}
\end{equation}
which sends a smooth function to its Taylor series at $t=0$.
\end{example}

Now we want to study what happens when we restrict sections of some fibre bundle to the \'etalification of a sub-stack of the base stack. Let us make this idea more precise.

\begin{remark}[Formal restriction of sections]\label{rem:formalrestriction}
Let $E\rightarrow X$ be a fibre bundle, as defined in Definition \ref{def:fibre_bundle}, and let $Y\xhookrightarrow{\;e\;}X$ be a formal derived smooth stack in $\mathbf{dFSmoothStack}$.
Recall the Definition \ref{def:sec_fibre_bundle} of $\infty$-groupoid of sections of a fibre bundle.
Then, we will call the $\infty$-groupoid
$\Gamma\!\left(\Im_XY,\, \iota_{Y\!}^\ast E\right)$, where $\iota_Y$ is the formal embedding map ${\Im_XY}\hookrightarrow X$, the \textit{$\infty$-groupoid of formal restricted sections} of $E$ on $Y$.
The embedding $\iota_Y:\Im_XY \hookrightarrow X$ of formal derived smooth set induces a morphism
\begin{equation}
    \Gamma(X,E) \;\xtwoheadrightarrow{\;\;\pi_Y\;\;}\; \Gamma\!\left(\Im_XY,\, \iota_{Y\!}^\ast E\right),
\end{equation}
which we will call \textit{formal restriction} of sections.
\end{remark}

Let us come back to the example of the thickened hyper-surface and let us concretely see how sections on the total smooth manifold restrict to the aforementioned thickened hyper-surface.

\begin{example}[Scalar field on thickened hypersurface]
Consider the situation of example \ref{ex:thicksurface}. Now, we introduce a trivial vector bundle $E\coloneqq M\times V \twoheadrightarrow M$, where $V$ is a vector space. The formal restriction of sections of such a bundle to the formal submanifold $\Sigma_{t_0}$ will be given by
\begin{equation}
\begin{aligned}
    \Gamma(M,E)\,\simeq\,\Gamma(\Sigma\times \bbR, E) \;&\xtwoheadrightarrow{\;\;\pi_\Sigma\;\;}\; \Gamma(\Sigma\times \bbD_{\bbR,0},\iota_\Sigma^\ast E) \,\simeq\, \Gamma(N^\infty_M\Sigma_{t_0},\iota_\Sigma^\ast E) \\[0.7ex]
    \phi^i(x,t) \;&\xmapsto{\;\;\phantom{\pi_\Sigma}\;\;}\; \phi^i(x,0) + \sum_{n>0}\left.\frac{\partial^{n\!} \phi^i(x,t)}{\partial t^n}\right|_{0}\!\!\!t^n.
    \end{aligned}
\end{equation}
In other words, the restriction sends a scalar field $\phi^i(x,t)$ to the collection of boundary conditions $\phi^i(x,0)$, $\dot{\phi}^i(x,0)$, $\ddot{\phi}^i(x,0)$, etc \dots, at a fixed $0\in \bbR$.
\end{example}

\begin{lemma}[Restriction of formal disk bundle]\label{lem:horverformal}
Consider a formal derived smooth stack $Y\xhookrightarrow{\;e\;}X$ in $\mathbf{dFSmoothStack}$ and let $T^\infty X|_Y \coloneqq T^\infty X{\times_X}Y$ be the restriction of the formal disk bundle of $X$ to $Y$. Then we have the equivalence of formal derived smooth stacks
\begin{equation}
     T^\infty X|_Y \;\simeq\; T^\infty Y \times_Y N^\infty_X Y.
\end{equation}
\end{lemma}

\begin{proof}
First, notice that the restriction of the formal disk bundle $T^\infty X|_Y \simeq  Y \times_{\Im(X)} X$ by the following pullback squares:
\begin{equation}
    \begin{tikzcd}[row sep=scriptsize, column sep=5ex, row sep=7.0ex]
    {T^\infty X|_Y}\arrow[d] \arrow[r] & T^\infty X \arrow[r] \arrow[d] & X \arrow[d]\\
    Y  \arrow[r, " "] & X \arrow[r] & \Im(X).
    \end{tikzcd}
\end{equation}
On the other hand, we also have the equivalence $T^\infty Y \times_Y N^\infty_X Y \simeq  Y \times_{\Im(X)} X$, which follows from the other pullback squares
\begin{equation}
    \begin{tikzcd}[row sep=scriptsize, column sep=4ex, row sep=7.0ex]
    T^\infty Y {\times_Y}N^\infty_X Y\arrow[d] \arrow[r] & N^\infty_XY \arrow[r] \arrow[d] & X \arrow[dd]\\
    T^\infty Y \arrow[d] \arrow[r, " "] & Y \arrow[d] & \\
    Y \arrow[r, " "] & \Im(Y)\arrow[r] & \Im(X).
    \end{tikzcd}
\end{equation}
Therefore, we have the conclusion of the lemma.
\end{proof}

%%%%%%%%%%%%%%%%%%%%%%%%%%%%%%%%%%%%%%%%%%%%%%%%%%%%%%%%%%%%%%%%%%%%%%%%%%%%%%%%%%%%%%%%%%%%%
\subsection{Formal moduli problems from derived infinitesimal cohesion}\label{subsec:FMP}

In this subsection we will briefly investigate the relation between formal derived smooth stacks, which we have defined in this paper, and formal moduli problems, which are the pivotal ingredient of the formalisation of BV-theory developed by \cite{FactI}.
This relation is summed up in \cref{fig:stackpicture}: a pointed formal moduli problem can be seen as a formal neighbourhood of a more general formal derived smooth stack.  

To begin with, let us consider the definition of formal moduli problem as it appears in \cite{FactII}. From now on, we will denote by $\mathsf{dgArt}_{\mathbbvar{k}}^{\leq 0}$ the category of \textit{local Artinian dg-algebras}. Recall that a local Artinian dg-algebra is a negatively graded dg-$\mathbbvar{k}$-algebra $\mathcal{A}$ concentrated in finitely many degrees, whose graded components are finite-dimensional and which comes equipped with a unique maximal differential ideal $\mathfrak{m}_\mathcal{A}\subset \mathcal{A}$ such that $\mathcal{A}/\mathfrak{m}_\mathcal{A}\cong \mathbbvar{k}$ and $\mathfrak{m}_\mathcal{A}^N$ for some $N\gg 0$. Equivalently, a local Artinian dg-algebra is a negatively graded dg-$\mathbbvar{k}$-algebra $\mathcal{A}$ concentrated in finitely many degrees, whose $0$th cohomology $\mathrm{H}^0(\mathcal{A})$ is a local Artinian algebra in the ordinary sense.
Then, the definition of formal moduli problem is the following.

\begin{definition}[Pointed formal moduli problem]
A \textit{pointed formal moduli problem} is a functor
\begin{equation}
    F \,:\; \mathsf{dgArt}_{\mathbbvar{k}}^{\leq 0} \,\longrightarrow\, \sSet,
\end{equation}
such that it satisfies the following properties:
\begin{itemize}
    \item $F(\mathbbvar{k})$ is contractible,
    \item $F$ maps surjective morphisms of Artinian dg-algebras to fibrations of simplicial sets,
    \item Let $\mathcal{A}\twoheadrightarrow \mathcal{C}$ and $\mathcal{B}\twoheadrightarrow \mathcal{C}$ be two surjective morphisms of dg-Artinian algebras. Then, the natural map $F(\mathcal{A}\times_\mathcal{C}\mathcal{B})\rightarrow F(\mathcal{A})\times_{F(\mathcal{C})}F(\mathcal{B})$ is a weak homotopy equivalence.
\end{itemize}
\end{definition}

In other words, we can see a pointed formal moduli problem as a derived stack on the $(\infty,1)$-site of dg-Artinian algebras, with the natural simplicial model structure induced by the usual $(\infty,1)$-site structure of commutative dg-algebras. 
A pivotal class of these objects will be provided by local $L_\infty$-algebras, whose definition from \cite{FactII} we now recall.

\begin{definition}[Local $L_\infty$-algebra]
A \textit{local $L_\infty$-algebra} $\mathfrak{L}(M)$ on a smooth manifold $M\in\Mfd$ is a $\mathbb{Z}$-graded vector bundle $L\twoheadrightarrow M$ whose space of sections $\mathfrak{L}(M)\coloneqq \Gamma(M,L)$ is equipped with a collection of poly-differential operators
\begin{equation}
    \ell_n \,:\; \mathfrak{L}(M)^{\otimes n} \,\longrightarrow\, \mathfrak{L}(M)
\end{equation}
of cohomological degree $2-n$ for any $n\geq 1$ such that $(\mathfrak{L}(M),\{\ell_n\}_{n\geq 1})$ is an $L_\infty$-algebra.
\end{definition}

The definition above is, then, a natural generalisation of the more familiar notion of $L_\infty$-algebra on a degree-wise finite-dimensional $\mathbb{Z}$-graded vector space to the case of a infinite-dimensional $\mathbb{Z}$-graded vector space of sections of a $\mathbb{Z}$-graded vector bundle.
As anticipated, an $L_\infty$-algebra, local or not, gives naturally rise to a formal moduli problem by the following construction.

\begin{definition}[Maurer-Cartan formal moduli problem]
Given an $L_\infty$-algebra $\mathfrak{g}$, the \textit{Maurer-Cartan formal moduli problem} $\mathbf{MC}(\mathfrak{g})$ can be defined by the functor
\begin{equation}
    \begin{aligned}
        \mathbf{MC}(\mathfrak{g}) \,:\;  \mathsf{dgArt}_{\mathbbvar{k}}^{\leq 0} \,&\longrightarrow\, \sSet \\
        \mathcal{A} \,&\longmapsto\, \mathrm{MC}(\mathfrak{g}\otimes_\mathbbvar{k} \mathfrak{m}_\mathcal{A}),
    \end{aligned}
\end{equation}
where $\mathfrak{m}_\mathcal{A}$ is the maximal differential ideal of $\mathcal{A}$ and $\mathrm{MC}(-)$ is the simplicial set of solutions to the Maurer-Cartan equation.
\end{definition}

Notice that the Maurer-Cartan formal moduli problem is a pointed formal moduli problem.

\begin{remark}[Any pointed formal moduli problem is equivalent to a Maurer-Cartan one]
Thanks to the results by \cite{HINICH2001209, pridham2010unifying}, we know that any pointed formal moduli problem $F$ is equivalent to a Maurer-Cartan formal moduli problem, i.e. there is an equivalence
\begin{equation}
    F \;\simeq\; \mathbf{MC}(\mathfrak{L}(M)),
\end{equation}
for some local $L_\infty$-algebra $\mathfrak{L}(M)$ on the smooth manifold $M$.
\end{remark}
Thus, without any loss of generality, we can focus on Maurer-Cartan moduli problems.

\begin{construction}[Artinian dg-algebras are finitely generated homotopy $\Coo$-algebras]$ $\vspace{-0.2cm}

\begin{itemize}
    \item Since Artinian dg-algebras $\mathsf{dgArt}_{\bbR}^{\leq 0}\subset \mathsf{dgcAlg}_\bbR^{\leq 0}$ naturally embed into the model category of dg-commutative algebras, then, by composing with the Dold-Kan correspondence functor $|-|_\mathrm{DK}: \mathsf{dgcAlg}_\bbR^{\leq 0}\longrightarrow\mathsf{scAlg}_\bbR$, see e.g. \cite{Jardine:1999}, we can embed Artinian dg-algebras into simplicial commutative algebras.
    \item Given any Artinian dg-algebra $\mathcal{A}\in\mathsf{dgArt}_{\bbR}^{\leq 0}$, its $0$-degree component $\mathcal{A}_0$ is an ordinary Artinian algebra and thus it is canonically a $\Coo$-algebra by the discussion in section \ref{sec:Coo}. Moreover, the $\{\mathcal{A}_{-i}\}_{i>0}$ are modules on $\mathcal{A}_0$. Therefore, we have a canonical dg-$\Coo$-algebra structure on $\mathcal{A}$ in the sense of \cite{carchedi2012homological}. Then, by \cite{carchedi2023derived}, its Dold-Kan simplicialisation is a homotopy $\Coo$-algebra, which we will denote by $|\mathcal{A}|_{\mathrm{DK}}^{\Coo}\in\mathsf{sC^\infty Alg}$.
    \item By definition, the $0$th cohomology $\mathrm{H}^0(\mathcal{A})\cong \pi_0|\mathcal{A}|_\mathrm{DK}$ of a local Artinian dg-algebra $\mathcal{A}$ is an ordinary local Artinian algebra and thus it is canonically a finitely presented $\Coo$-algebra and, in particular, a finitely generated $\Coo$-algebra. Recall that, for a simplicial $\Coo$-algebra $R$ to be finitely generated in the homotopical sense, it is sufficient that $\pi_0R$ is finitely generated in the ordinary sense. Therefore, $|\mathcal{A}|_{\mathrm{DK}}^{\Coo}\in\mathsf{sC^\infty Alg}_\mathrm{fg}$ is canonically an finitely generated $\Coo$-algebra.
\end{itemize}
Thus, by generalising the case of ordinary Artinian algebras, the Dold-Kan functor $|-|_\mathrm{DK}$ can be uniquely lifted as follows
\begin{equation}
    \begin{tikzcd}[row sep={15ex,between origins}, column sep={18ex,between origins}]
    & \mathsf{sC^\infty Alg}_\mathrm{fg} \arrow[d, "(-)^\mathrm{alg}"] \\
    \mathsf{dgArt}_\bbR^{\leq 0} \arrow[r, "|-|_\mathrm{DK}"]\arrow[ru,  "|-|_{\mathrm{DK}}^{\Coo}"] & \mathsf{scAlg}_\bbR
    \end{tikzcd}
\end{equation}
where $(-)^\mathrm{alg}:\mathsf{sC^\infty Alg}\rightarrow\mathsf{scAlg}_\bbR$ is, as usual, the forgetful functor which forgets the $\Coo$-algebra structure and leaves us with the underlying simplicial commutative algebra. Therefore we have an embedding
\begin{equation}
    |-|_{\mathrm{DK}}^{\Coo}\,:\, \mathsf{dgArt}_{\bbR}^{\leq 0} \,\longhookrightarrow\, \mathsf{sC^\infty Alg}_{\mathrm{fg}}.
\end{equation}
In other words, we can interpret an Artinian dg-algebra as the algebra of functions on a formal derived smooth manifold, which will be, in particular, a thickened point. 
\end{construction}

This means that we could see formal moduli problems as formal derived smooth stacks whose source category has been restricted to derived thickened points.
In this light, it is possible to see that we can always extract a formal moduli problem from a formal derived smooth stack $X$ by restricting the $(\infty,1)$-site of formal derived smooth manifolds to the $(\infty,1)$-site of thickened points and by sending such thickened points to some fixed point $x:\ast\rightarrow X$ of the original stack. 
Let us construct this operation step by step.

\begin{construction}[Formal moduli problems as formal completion of formal derived smooth stacks]\label{con:FMP_as_FC}
Let $X\in\mathbf{dFSmoothStack}$ be a formal derived smooth stack. As discussed above, we have the embedding $|-|_{\mathrm{DK}}^{\Coo}: \mathsf{dgArt}_{\bbR}^{\leq 0} \hookrightarrow \mathsf{sC^\infty Alg}_{\mathrm{fg}}$. This gives immediately rise to a formal moduli problem $X^\wedge$ which is defined by the pullback $X^\wedge \coloneqq (|-|_{\mathrm{DK}}^{\Coo})^\ast X$.
This is a functor
\begin{equation}\label{eq:stack_to_fmp}
\begin{aligned}
    X^\wedge \,:\;  \mathsf{dgArt}_{\bbR}^{\leq 0} \,&\longrightarrow\, \sSet \\
    \mathcal{A} \,&\longmapsto\, X\big( |\mathcal{A}|_{\mathrm{DK}}^{\Coo} \big),
\end{aligned}
\end{equation}
where $|\mathcal{A}|_{\mathrm{DK}}^{\Coo}\in\mathsf{sC^\infty Alg}_\mathrm{fg}$ is the finitely generated simplicial $\Coo$-algebra corresponding to the Artinian dg-algebra $\mathcal{A}\in\mathsf{dgArt}_\bbR^{\leq 0}$.
However, this functor does not encode a \textit{pointed} formal moduli problem, because the thickened points in the site are allowed to be sent to any point of the stack $X$ and not only to some fixed point $x\in X$. 
Let us then fix a point $x:\ast\rightarrow X$ and define the following pointed formal moduli problem:
\begin{equation}
\begin{aligned}
    X^\wedge_x \,:\;  \mathsf{dgArt}_{\bbR}^{\leq 0} \,&\longrightarrow\, \sSet \\
    \mathcal{A} \,&\longmapsto\, X\big( |\mathcal{A}|_{\mathrm{DK}}^{\Coo} \big)\times_{X(\ast)}\ast\,,
\end{aligned}
\end{equation}
which is the smooth version of the construction appearing in \cite[Section 4.2]{Toen:2014re} and \cite{Calaque:2017}, called \textit{formal completion} at $x$ of a derived stack.
\end{construction}

\begin{definition}[$(\infty,1)$-topos of formal moduli problems]
We define the $(\infty,1)$-category of \textit{formal moduli problems} by the $(\infty,1)$-category of pre-stacks
\begin{equation}
    \mathbf{FMP} \;\coloneqq\; \mathbf{N}_{hc}([\mathsf{dgArt}^{\leq 0}_\bbR, \sSet]_\mathrm{proj}^\circ),
\end{equation}
with its natural structure of $(\infty,1)$-topos of pre-stacks.
\end{definition}

\begin{proposition}[Infinitesimally cohesive $(\infty,1)$-topos of formal moduli problems]
The $(\infty,1)$-topos $\mathbf{FMP}$ of formal moduli problems has a natural infinitesimally cohesive structure as defined by \cite[Definition 4.1.21]{DCCTv2}.\end{proposition}

\begin{proof} By \cite[Proposition 4.1.24]{DCCTv2} the $(\infty,1)$-category of pre-stacks on an $(\infty,1)$-site containing a zero object (i.e. an object which is both initial and terminal) is an infinitesimal cohesive $(\infty,1)$-topos. The simplicial model category underlying $\mathbf{FMP}$ is precisely $[\mathsf{dgArt}^{\leq 0}_\bbR, \sSet]_\mathrm{proj}$, making $\mathbf{FMP}$ an $(\infty,1)$-category of pre-stacks. Now, we can see that the real line $\bbR$ is both a terminal and initial in $\mathsf{dgArt}^{\leq 0}_\bbR$. In fact, for any dg-Artinian algebra $\mathcal{A}$, there is not only a unique map $\bbR\rightarrow\mathcal{A}$, but crucially also a unique $\bbR$-point $\mathcal{A}\rightarrow \bbR$. Thus we have the conclusion.
\end{proof}

\begin{corollary}[Derived infinitesimal cohesion of formal moduli problems]
The immediate consequence of \cite[Proposition 4.1.24]{DCCTv2} is that, in particular, the $(\infty,1)$-topos $\mathbf{FMP}$ of formal moduli problems is naturally equipped with a cohesive structure of the form
\begin{equation}\label{eq:inf_cohesion}
    \begin{tikzcd}[row sep={15ex,between origins}, column sep={19ex}]
     \mathbf{FMP} &  \mathbf{\infty Grpd} \arrow[l, "\mathrm{Disc}^{\mathrm{inf}}"', hook', shift right=1.5ex] \arrow[l, "\mathit{\Gamma}^{\mathrm{inf}}"', shift right=-1.5ex, leftarrow] \arrow[l, "\mathit{\Pi}^{\mathrm{inf}}"', shift right=+4.5ex, leftarrow] \arrow[l, "\mathrm{coDisc}^{\mathrm{inf}}"', hook', shift left=4.5ex] .
    \end{tikzcd}
\end{equation}
\end{corollary}

Morally speaking, formal moduli problems in $\mathbf{FMP}$ can be thought of as infinitesimally thickened $\infty$-groupoids, in a formal derived sense.

Now, we will explore the relation between the $(\infty,1)$-topos of the formal derived smooth stacks with the $(\infty,1)$-topos of formal moduli problems.

\begin{lemma}[Derived relative base]
Formal derived smooth stacks are equipped with a relative base structure over the $(\infty,1)$-topos $\mathbf{FMP}$ of formal moduli problems, i.e. we have a quadruplet of adjoint $(\infty,1)$-functors
\begin{equation}\label{eq:rel_base}
    \begin{tikzcd}[row sep={15ex,between origins}, column sep={19ex}]
     \mathbf{dFSmoothStack} &  \mathbf{FMP} \arrow[l, "\mathrm{Disc}^{\mathrm{rel}}"', hook', shift right=1.5ex] \arrow[l, "\mathit{\Gamma}^{\mathrm{rel}}"', shift right=-1.5ex, leftarrow]
    \end{tikzcd}
\end{equation}
such that:
\begin{itemize}
    \item $\mathit{\Gamma}^{\mathrm{rel}} = (-)^\wedge$ is precisely the functor \eqref{eq:stack_to_fmp},
    \item $\mathrm{Disc}^{\mathrm{rel}}$ is fully faithful.
\end{itemize}
\end{lemma}

\begin{proof}
Recall that there is the following embedding of simplicial sites
\begin{equation}
    \begin{tikzcd}[row sep=scriptsize, column sep=12.0ex, row sep=18.0ex]
     \mathsf{sC^\infty Alg}_\mathrm{fg}   &  \arrow[l, "|-|_{\mathrm{DK}}^{\Coo}"', shift left=-0ex, hook'] \mathsf{dgArt}_{\bbR}^{\leq 0} .
    \end{tikzcd}
\end{equation}
This gives rise by left and right Kan extension to the following triplet of adjoint functors between the corresponding simplicial categories of pre-stacks:
\begin{equation}
    \begin{tikzcd}[row sep={15ex,between origins}, column sep={19ex}]
    {[ \mathsf{sC^\infty Alg}_\mathrm{fg}, \sSet]_\mathrm{proj}} \arrow[r, "{(|-|_{\mathrm{DK}}^{\Coo})_!}", shift right=-1.5ex, leftarrow] &  {[\mathsf{dgArt}^{\leq 0}_\bbR, \sSet]_\mathrm{proj}}  \arrow[l, "\,(|-|_{\mathrm{DK}}^{\Coo})^{\ast}"', shift right=-1.5ex, hook, leftarrow]  .
    \end{tikzcd}
\end{equation}
The functor $\mathit{\Gamma}^\mathrm{rel}\coloneqq(|-|_{\mathrm{DK}}^{\Coo})^{\ast}$ maps local fibrant objects into fibrant objects and, thus, it immediately preserves locally fibrant objects, since the simplicial category $[\mathsf{dgArt}^{\leq 0}_\bbR, \sSet]_\mathrm{proj}$ is equipped only with a global projective model structure. By \cite{ToenVezzo05}, the fact that the functor $\mathit{\Gamma}^\mathrm{rel}$ preserves locally fibrant objects implies that its left adjoint $\mathrm{Disc}^{\mathrm{rel}}\coloneqq(|-|_{\mathrm{DK}}^{\Coo})_!$ restricts to stacks.
Moreover, the fact that the functor $|-|_{\mathrm{DK}}^{\Coo}$ between the sites is fully faithful, implies that the functor $\mathrm{Disc}^{\mathrm{rel}}$ of stacks is fully faithful.
\end{proof}

By following \cite[Section 5.3.6]{DCCTv2}, we can also define the following modality.
\begin{definition}[Relative flat modality]
We define the \textit{relative flat modality} by the following endofunctor on formal derived smooth stacks:
\begin{equation}
   \flat^{\mathrm{rel}} \coloneqq \mathrm{Disc}^{\mathrm{rel}}\circ \mathit{\Gamma}^{\mathrm{rel}}\; :\,\mathbf{dFSmoothStack}\,\longrightarrow\,\mathbf{dFSmoothStack}.
\end{equation}
\end{definition}
Roughly speaking, the relative flat modality $\flat^{\mathrm{rel}}$ provides a formal derived thickened version of the flat modality $\flat$.

\begin{equation}
    \begin{tikzcd}[row sep={24ex,between origins}, column sep={10ex}]
     \mathbf{dFSmoothStack}\;\;\;   
        &  \;\;\;\mathbf{FMP}     \arrow[l, "\mathrm{Disc}^{\mathrm{rel}}"', hook', shift right=1.5ex] \arrow[l, "\mathit{\Gamma}^{\mathrm{rel}}"', shift right=-1.5ex, leftarrow]                        
       \arrow[d, "", shift right=-2.5ex,  shorten=3.5mm,sloped]   \\
    &  \;\;\;\mathbf{\infty Grpd},                             \arrow[u, "\mathrm{Disc}^\mathrm{inf}",sloped , shift right=-0.5ex,  shorten=3.5mm, hook] \arrow[u, "\mathit{\Gamma}^\mathrm{inf}",sloped , shift right=2.5ex, dash, shorten=3.5mm] \arrow[lu, "\mathrm{Disc}", hook', shift right=1.5ex,sloped,shorten=3.5mm] \arrow[lu, "\mathit{\Gamma}", shift right=-1.5ex, leftarrow,sloped,shorten=3.5mm]  
    \end{tikzcd} \vspace{0.5cm}
\end{equation}

\begin{lemma}[Relative flat modality as collection of formal disks]
For any given formal derived smooth stack $X\in\mathbf{dFSmoothStack}$, we have the following equivalence:
\begin{equation}
    \flat^\mathrm{rel}X \;\,\simeq\,\; \flat X \times_{\Im(X)}^h X.
\end{equation}
\end{lemma}

\begin{proof}
By unravelling the definition $\flat^\mathrm{rel} = \mathrm{Disc}^{\mathrm{rel}} \circ \mathit{\Gamma}^{\mathrm{rel}}$ of the relative flat modality, we see that a formal derived smooth stack $\flat^\mathrm{rel}X$ can be understood as the coproduct of formal disks $\bbD_{X,x}$ at all points $x:\ast \rightarrow X$.
In fact, notice that for any point $x\in X$, we must have the equivalence $\bfR\mathrm{Hom}(\mathrm{Disc}^\mathrm{rel}X^\wedge_x,Y)\simeq \bfR\mathrm{Hom}_\mathbf{FMP}(X^\wedge_x,Y^\wedge)$, which tells us $\mathrm{Disc}^\mathrm{rel}X^\wedge_x\simeq \bbD_{X,x}$.
Moreover, since we have $\mathit{\Gamma}^\mathrm{rel}(X)\simeq \coprod_{x:\ast\rightarrow X}X_x^\wedge$ and that $\mathrm{Disc}^\mathrm{rel}$ preserves colimits, we have the equivalence 
$\flat^\mathrm{rel}X \simeq \coprod_{x:\ast\rightarrow X}\mathrm{Disc}^\mathrm{rel}(X^\wedge_x)$, which immediately implies
\begin{equation}
    \flat^\mathrm{rel}X \;\simeq\, \coprod_{x:\ast\rightarrow X}\!\bbD_{X,x}. \vspace{-0.1cm}
\end{equation}
However, since we have the equivalence $\flat X \simeq \coprod_{x:\ast\rightarrow X}\ast$ and that the definition of infinitesimal disk  given by $\bbD_{X,x}\simeq\ast\times_{\Im(X)}^h X$, the equivalence above is precisely $\flat^\mathrm{rel}X \;\simeq\; \flat X \times_{\Im(X)}^h X$.
\end{proof}

\begin{corollary}[Relation with flat and infinitesimal shape modalities]
For any formal derived smooth stack $X$ we have the following equivalences:
\begin{equation}
    \flat(\flat^\mathrm{rel} X) \,\simeq\, \flat X, \qquad \Im(\flat^\mathrm{rel} X) \,\simeq\, \flat X.
\end{equation}
\end{corollary}

\begin{remark}[All the structures in the context of derived differential geometry]
By putting together all the structures we encountered in this section, we can write the following diagram of $(\infty,1)$-categories:
\vspace{-0.2cm}
\begin{equation}
    \begin{tikzcd}[row sep={30ex,between origins}, column sep={14ex}]
     \mathbf{dFSmoothStack}\;\;\;   
        &  \;\;\;\mathbf{FMP}     \arrow[l, "\mathrm{Disc}^{\mathrm{rel}}"', hook', shift right=1.5ex] \arrow[l, "\mathit{\Gamma}^{\mathrm{rel}}"', shift right=-1.5ex, leftarrow]                        
      \arrow[d, "", shift right=3.5ex,  shorten=3.5mm,sloped] \arrow[d, "", shift right=-2.5ex,  shorten=3.5mm,sloped]   \\
    \mathbf{SmoothStack}^{\pmb{+}}\;\;\; \arrow[u, "{\mathit{\Pi}^{\mathrm{dif}}}",sloped , shift right=-1.5ex, leftarrow, shorten=3.5mm] \arrow[u, "{\mathit{\Gamma}^{\mathrm{dif}}}",sloped , shift right=4.5ex, leftarrow, shorten=3.5mm]\arrow[u, "\hat{\imath}", shift right=-4.5ex,  shorten=3.5mm, hook, sloped] \arrow[u, "\mathrm{Disc}^{\mathrm{dif}}", shift right=1.5ex,  shorten=3.5mm, hook, sloped] &  \;\;\;\mathbf{\infty Grpd},      \arrow[l, "\mathrm{Disc}"', hook', shift right=1.5ex] \arrow[l, "\mathit{\Gamma}"', shift right=-1.5ex, leftarrow]                            \arrow[u, "", shift right=5.5ex,  shorten=3.5mm, hook, "\mathrm{coDisc}^\mathrm{inf}",sloped] \arrow[u, "\mathrm{Disc}^\mathrm{inf}",sloped , shift right=-0.5ex,  shorten=3.5mm, hook] \arrow[u, "{\mathit{\Pi}^\mathrm{inf}}",sloped , shift right=-3.5ex, dash, shorten=3.5mm] \arrow[u, "\mathit{\Gamma}^\mathrm{inf}",sloped , shift right=2.5ex, dash, shorten=3.5mm]
    \end{tikzcd} \vspace{0.5cm}
\end{equation}
\vspace{-0.7cm}

where, more in detailed, we have the following structures:
\begin{itemize}
    \item the left vertical quadruple $(\hat{\imath}\dashv\mathit{\Pi}^{\mathrm{dif}}\dashv\mathrm{Disc}^{\mathrm{dif}}\dashv\mathit{\Gamma}^{\mathrm{dif}})$ of $(\infty,1)$-functors presents a differential structure on the $(\infty,1)$-category $\mathbf{dFSmoothStack}$ of formal derived smooth stacks over the $(\infty,1)$-category $\mathbf{SmoothStack}^{\pmb{+}}$, from diagram \eqref{eq:diffcohesiondiag};
    \item the right vertical quadruple $(\mathit{\Pi}^{\mathrm{inf}}\dashv\mathrm{Disc}^{\mathrm{inf}}\dashv\mathit{\Gamma}^{\mathrm{inf}}\dashv\mathrm{coDisc}^{\mathrm{inf}})$ of $(\infty,1)$-functors presents an infinitesimal cohesive structure on the $(\infty,1)$-category $\mathbf{FMP}$ of formal moduli problems over the $(\infty,1)$-category $\mathbf{\infty Grpd}$ of $\infty$-groupoids, from diagram \eqref{eq:inf_cohesion};
    \item the upper horizontal pair $(\mathrm{Disc}^{\mathrm{rel}}\dashv\mathit{\Gamma}^{\mathrm{rel}})$ of $(\infty,1)$-functors presents the coreflective embedding of the $(\infty,1)$-category $\mathbf{FMP}$ of formal moduli problems into the $(\infty,1)$-category $\mathbf{dFSmoothStack}$ of formal derived smooth stacks, from diagram \eqref{eq:rel_base};
    \item the lower horizontal pair $(\mathrm{Disc}\dashv\mathit{\Gamma})$ of $(\infty,1)$-functors presents the terminal geometric morphism of the $(\infty,1)$-category $\mathbf{SmoothStack}^{\pmb{+}}$ over the $(\infty,1)$-category $\mathbf{\infty Grpd}$ of $\infty$-groupoids.
\end{itemize}
\end{remark}

Now, at the end of this subsection, we want to briefly explore the possibility of providing a step-by-step generalisation of \cite[Proposition 6.5.15]{DCCTv2} to derived smooth geometry. The short answer is that, do do so, the category of dg-Artinian algebras is not big enough and, thus, we must first introduce a slight extension of it, which is more natural from the perspective of formal derived smooth manifolds.

\begin{definition}[Pointed finitely generated simplicial $\Coo$-algebras]
We define the $(\infty,1)$-category $\mathbf{sC^{\infty\!} Alg}_{\mathrm{fg}}^\mathrm{pnt}$ of \textit{pointed finitely generated simplicial $\Coo$-algebras} by the pullback square
\begin{equation}\label{eq:square_sites}
    \begin{tikzcd}[row sep={14ex,between origins}, column sep={17ex,between origins}]
     \mathbf{sC^{\infty\!} Alg}_{\mathrm{fg}}^\mathrm{pnt} \arrow[r, hook] &  \mathbf{sC^{\infty\!} Alg}_{\mathrm{fg}} \\
     \ast \arrow[u, leftarrow] \arrow[r, "\bbR"] &  \mathbf{N}(\mathsf{C^\infty Alg}^{\mathrm{red}}_\mathrm{fg})\arrow[u, leftarrow, "(-)^\mathrm{red}"'] ,
    \end{tikzcd}
\end{equation}
where the functor $(-)^\mathrm{red}$ sends simplicial $\Coo$-algebras to their corresponding reduced ordinary $\Coo$-algebras.
\end{definition}

Dually, the opposite $(\infty,1)$-category $(\mathbf{sC^{\infty\!} Alg}_{\mathrm{fg}}^\mathrm{pnt})^\op\hookrightarrow \mathbf{dFMfd}$

\begin{remark}[Dg-Artinian algebras as pointed finitely generated simplicial $\Coo$-algebras]
Dg-Artinian $\mathbf{dgArt}_{\bbR} \hookrightarrow \mathbf{sC^{\infty\!} Alg}_{\mathrm{fg}}^\mathrm{pnt}$ naturally embed into pointed finitely generated simplicial $\Coo$-algebras. To see this, notice that any dg-Artinian algebra $\mathcal{A}$ has, by definition, a unique $\bbR$-point $\mathcal{A}\rightarrow\bbR$ and its $0$-th cohomology $\mathrm{H}^0(\mathcal{A})\cong \pi_0|\mathcal{A}|_{\mathrm{DK}}^{\Coo}$ is an ordinary finitely presented $\Coo$-algebra, which is finitely generated.
\end{remark}

\begin{proposition}[Finitely generated formal moduli problems]
Let $\mathbf{FMP}_{\!\mathrm{fg}}$ the $(\infty,1)$-category of pre-stacks on $(\mathbf{sC^{\infty\!} Alg}_{\mathrm{fg}}^\mathrm{pnt})^\op$, whose elements we will call \textit{finitely generated formal moduli problems}.
Then, there is an $(\infty,1)$-pushout square of $(\infty,1)$-topoi
\begin{equation}
    \begin{tikzcd}[row sep={19ex,between origins}, column sep={21ex,between origins}]
    \mathbf{dFSmoothStack} \arrow[r, "\mathit{\Gamma}^\mathrm{rel}_{\!\mathrm{fg}}"] &  \mathbf{FMP}_{\!\mathrm{fg}} \\
    \mathbf{SmoothStack}^{\pmb{+}}\arrow[u, hook, "\mathrm{Disc}^{\mathrm{dif}}"] \arrow[r, "\mathit{\Gamma}"] &  \mathbf{\infty Grpd}\arrow[u, hook, "\mathrm{Disc}^{\mathrm{inf}}_{\mathrm{fg}}"'] .
    \end{tikzcd}
\end{equation}
\end{proposition}

\begin{proof}
The result of \cite[Proposition 6.3.2.3]{topos} tells us that an $(\infty,1)$-pushout of $(\infty,1)$-topoi can be concretely computed as an $(\infty,1)$-pullback of the underlying $(\infty,1)$-categories, where the morphisms are the left adjoint $(\infty,1)$-functors in all pairs presenting the geometric morphisms. Then, we need to show that the square
\begin{equation}
    \begin{tikzcd}[row sep={19ex,between origins}, column sep={21ex,between origins}]
    \mathbf{dFSmoothStack} \arrow[r, leftarrow, "\mathrm{Disc}^{\mathrm{rel}}_{\mathrm{fg}}"] &  \mathbf{FMP}_{\!\mathrm{fg}} \\
    \mathbf{SmoothStack}^{\pmb{+}}\arrow[u, leftarrow, "\mathit{\Pi}^\mathrm{dif}"] \arrow[r, leftarrow, "\mathrm{Disc}"] &  \mathbf{\infty Grpd}\arrow[u, leftarrow, "\mathit{\Pi}^\mathrm{inf}_{\mathrm{fg}}"'] 
    \end{tikzcd}
\end{equation}
is an $(\infty,1)$-pullback of $(\infty,1)$-categories.
Since any stack can be written as the $(\infty,1)$-colimit of representables and the left adjoint preserves colimits, it is enough to check that the diagram of sites is a pullback square. But such a diagram is precisely the pullback square \eqref{eq:square_sites}.
\end{proof}

%%%%%%%%%%%%%%%%%%%%%%%%%%%%%%%%%%%%%%%%%%%%%%%%%%%%%%%%%%%%%%%%%%%%%%%%%%%%%%%%%%%%%%%%%%%%%
\subsection{$L_\infty$-algebroids as formal derived smooth stacks}\label{subsec:Loo_algebroids}

In this subsection we will develop a general picture of $L_\infty$-algebroids -- and thus of the geometric objects sometimes known as $NQ$-manifolds in the literature -- in the context of derived differential topos geometry. We will see an interesting interplay between the formal and the higher derived properties of formal derived smooth stacks, which is also related to the research by \cite{ArvFormal22}.

First, we write the appropriate definition of groupoid object internal in an $(\infty,1)$-category, as proposed in \cite{Principal1}, in our case of interest of formal derived smooth stacks.

\begin{definition}[Groupoid object]
A \textit{groupoid object} in the $(\infty,1)$-category $\mathbf{dFSmoothStack}$ of formal derived smooth stacks is a simplicial object $\mathcal{G}_\bullet:\Delta^\op\rightarrow \mathbf{dFSmoothStack}$ such that all the natural maps (also known as Segal maps)
\begin{equation}
    \mathcal{G}_n \;\longrightarrow\; \mathcal{G}_1 \times_{\mathcal{G}_0}^h \cdots \times_{\mathcal{G}_0}^h \mathcal{G}_1
\end{equation}
are equivalences of formal derived smooth stacks.
\end{definition}

As discussed by \cite{Principal1}, a groupoid object $\mathcal{G}_\bullet$ in an $(\infty,1)$-topos gives rise to the colimiting cocone $\mathcal{G}_0 \xtwoheadrightarrow{\;\;\;\;} \bfL\mathrm{colim}\,\mathcal{G}_\bullet$, which is an effective epimorphism. Conversely, any effective epimorphism $X \xtwoheadrightarrow{\;\;p\;\;} \mathcal{G}$  is equivalently a groupoid object $\mathcal{G}_\bullet$ with $\mathcal{G}_0\simeq X$.
Then, in particular, this must hold in the $(\infty,1)$-topos $\mathbf{dFSmoothStack}$ of formal derived smooth stacks.
To sum up, the relevant data of a groupoid object of formal derived smooth stacks can be packed in a diagram of the following form:
\begin{equation}\label{eq:groupoid_object}
    \begin{tikzcd}[row sep={15ex,between origins}, column sep={5ex}]
     \mathcal{G}_1 \arrow[r, "t"', shift left=-1.0ex]\arrow[r, "s", shift left=1.0ex] &  \mathcal{G}_0 \arrow[r, "p", two heads] & \bfL\mathrm{colim}\,\mathcal{G}_\bullet,
    \end{tikzcd}
\end{equation}
where $s$ plays the role of source map and $t$ the role of target map.

\begin{example}[Derived smooth group]
We call \textit{derived smooth group} $\B G\in\mathbf{dFSmoothStack}$ a groupoid object that is pointed, i.e. of the form $\ast \xtwoheadrightarrow{\;\ast\;} \mathbf{B}G$ with $\mathcal{G}_0=\ast$ and effective epimorphism given by the inclusion of the canonical point.
Here, diagram \eqref{eq:groupoid_object} reduces to
\begin{equation}
    \begin{tikzcd}[row sep={15ex,between origins}, column sep={5ex}]
     G \arrow[r, "\ast"', shift left=-1.0ex]\arrow[r, "\ast", shift left=1.0ex] &  \ast \arrow[r, "\ast", two heads] & \B G.
    \end{tikzcd}
\end{equation}
\end{example}

Now, in the formalism of derived differential structures, we are able to generalise an idea from \cite[Section 6.5.2.2]{DCCTv2} to derived geometry and, thus, provide a very general definition of what we may call derived smooth algebroid. 
Morally speaking, a derived smooth algebroid is going to be a groupoid object $\mathcal{G}_\bullet$ in the $(\infty,1)$-category of formal derived smooth stacks which is infinitesimally thickened over its base $\mathcal{G}_0$. 
We will show that such a notion generalises familiar $L_\infty$-algebroids.

\begin{definition}[Derived smooth algebroid]\label{def:der_algebroid}
We call \textit{derived smooth algebroid} a groupoid object $X \xtwoheadrightarrow{\;\;p\;\;} \mathcal{G}$ in $\mathbf{dFSmoothStack}$ such that the morphism $\Im(X) \xrightarrow{\;\Im(p)\;} \Im(\mathcal{G})$ is an equivalence.
We also call the map $p$ the \textit{anchor map} of the derived $L_\infty$-algebroid.
\end{definition}

%%%%%%% 

Now we will see that the usual notions of $L_\infty$-algebra and  $L_\infty$-algebroid fit into this wider definition of derived smooth algebroid.
First, let us see how $L_\infty$-algebras and $L_\infty$-algebroids are embedded into formal derived smooth stacks.

\begin{definition}[Delooping of an $L_\infty$-algebra and of an $L_\infty$-algebroid]$ $
\vspace{-0.2cm}\begin{itemize}
    \item The \textit{delooping} $\B\mathfrak{g}$ of an $L_\infty$-algebra $\mathfrak{g}$ can be defined by the formal derived smooth stack
\begin{equation}
    \begin{aligned}
        \B\mathfrak{g} \,:\;  \mathsf{dFMfd} \,&\longrightarrow\, \sSet \\
        U \,&\longmapsto\, \mathrm{MC}(\mathfrak{g}\otimes \mathfrak{m}_{\mathcal{O}(U)}),
    \end{aligned}
\end{equation}
where $\mathfrak{m}_{\mathcal{O}(U)}$ is the nilradical of the dg-commutative algebra $\mathrm{N}\mathcal{O}(U)^\mathrm{alg}$ and $\mathrm{MC}(-)$ is the simplicial set of solutions to the Maurer-Cartan equation.
    \item More generally, the \textit{delooping} $\B\mathfrak{L}(M)$ of a local $L_\infty$-algebra $\mathfrak{L}(M)$ can be defined by the formal derived smooth stack
\begin{equation}
    \begin{aligned}
        \B\mathfrak{L}(M) \,:\;  \mathsf{dFMfd} \,&\longrightarrow\, \sSet \\
        U \,&\longmapsto\, \mathrm{MC}(\mathfrak{L}(M)\,\widehat{\otimes}\, \mathfrak{m}_{\mathcal{O}(U)}),
    \end{aligned}
\end{equation}
    where we defined the pullback $\mathfrak{L}(M)\,\widehat{\otimes}\, \mathfrak{m}_{\mathcal{O}(U)} \coloneqq \mathfrak{L}(M)\,\widehat{\otimes}\, \mathrm{N}\mathcal{O}(U)_{\!} \times_{\mathfrak{L}(M)\widehat{\otimes}\mathrm{N}\mathcal{O}(U)^\mathrm{red}}\!\{0\}$.
    \item The \textit{delooping} $\B\mathfrak{a}$ of an $L_\infty$-algebroid $\mathfrak{a}\twoheadrightarrow M$ on an ordinary smooth manifold $M$ can be defined by the formal derived smooth stack
\begin{equation}
    \begin{aligned}
        \B\mathfrak{a} \,:\;  \mathsf{dFMfd} \,&\longrightarrow\, \sSet \\
        U \,&\longmapsto\, \coprod_{f:U^\mathrm{red}\rightarrow M}\!\!\!\!  \mathrm{MC} (\Gamma(U^\mathrm{red},f^\ast\mathfrak{a})\,\widehat{\otimes}_{\mathrm{N}\mathcal{O}(U)} \mathfrak{m}_{\mathcal{O}(U)}).
    \end{aligned}\vspace{-0.1cm}
\end{equation}
    where the $\Coo$-tensor product is given as above.
\end{itemize}
\end{definition}

We can now show that usual $L_\infty$-algebras and $L_\infty$-algebroids are examples of derived smooth algebroids as defined above.

\begin{example}[Usual $L_\infty$-algebras]
Let $\B\mathfrak{g}$ be the delooping of an $L_\infty$-algebra. The canonical map
$\ast \xtwoheadrightarrow{\;\,\ast\,\;} \mathbf{B}\mathfrak{g}$ gives rise to a map $\ast\simeq\Im(\ast)  \xtwoheadrightarrow{\;\;\;\;} \Im(\mathbf{B}\mathfrak{g})\simeq \ast$, which makes $\B \mathfrak{g}$ into a formal smooth algebroid on the point.
\end{example}

\begin{example}[Usual $L_\infty$-algebroids]
Let $\mathfrak{a}\twoheadrightarrow M$, where $M$ is an ordinary smooth manifold, be a $L_\infty$-algebroid in the usual sense.
Then the map $M \xtwoheadrightarrow{\;\;\rho\;\;} \B\mathfrak{a}$ presents the $L_\infty$-algebroid as a derived smooth algebroid in the sense above, since $M\simeq\Im(M)\twoheadrightarrow\Im(\B\mathfrak{a})\simeq M$.
\end{example}

\begin{remark}[The base of a derived smooth algebroid]
Notice that, in the definition of a derived $L_\infty$-algebroid $X \xtwoheadrightarrow{\;\;p\;\;} \mathcal{G}$, there is no requirement for $X$ to be an ordinary smooth manifold or even a formal derived smooth manifold. In fact, $X$ can be, in general, a formal derived smooth stack. 
In other words, derived smooth algebroids generalise $L_\infty$-algebroids by dropping the constraint that the base has to be an ordinary smooth manifold. 
Roughly speaking, a derived smooth algebroid is an infinitesimally thickened groupoid object where the base $X$ is generally a formal derived smooth stack.
\end{remark}

Let us now provide an archetypal example of such a generalised notion of derived smooth algebroid where the base is not just an ordinary smooth manifold, but fully fledged a formal derived smooth stack.

\begin{example}[Formal disk bundle as derived $L_\infty$-algebroid]
Let $X\in\mathbf{dFSmoothStack}$ be any formal derived smooth stack. Recall the definition of formal disk bundle $T^\infty X = X\times_{\Im(X)}^hX$ induced by the canonical morphism $\mathfrak{i}_X:X\longrightarrow \Im(X)$ to the de Rham space $\Im(X)$.
Thus, the formal disk bundle gives rise to a groupoid object of the form
\begin{equation}\label{eq:TX_as_algebroid}
    \begin{tikzcd}[row sep={15ex,between origins}, column sep={6ex}]
     T^\infty X \arrow[r, "\mathrm{ev}"', shift left=-1.0ex]\arrow[r, "\pi_X", shift left=1.0ex] &  X \arrow[r, "\mathfrak{i}_X", two heads] & \Im(X).
    \end{tikzcd}
\end{equation}
Moreover, notice that $\Re(\Im(X))\simeq \Re(X)$ since one has the equivalence $\iota^{\mathrm{red}\ast\!}\circ\iota^{\mathrm{red}}_{\ast}\circ \iota^{\mathrm{red}\ast}\simeq\mathrm{id}$. Thus, the diagram \eqref{eq:TX_as_algebroid} presents, in particular, a derived $L_\infty$-algebroid in the generalised sense of definition \ref{def:der_algebroid}.
\end{example}

Thus, the abstract definition of derived smooth algebroid above provides a generalisation of the usual definition of $L_\infty$-algebroid, which is based on the formalism of differential-graded manifolds. In section \ref{sec:global_aspects_of_bv_theory} we will explore some relevant examples motivated by physics.

\begin{remark}[Lie differentiation]
Finally, notice that we can use the infinitesimal flat modality to encompass Lie differentiation. In fact, by using the equivalences $\mathit{\Gamma}^{\mathrm{rel}}(\B G) \simeq \mathbf{MC}(\mathfrak{g})$ and $\mathrm{Disc}^{\mathrm{rel}} \big(\mathbf{MC}(\mathfrak{g})\big)\simeq\B\mathfrak{g}$.
From these two equivalences, we obtain an equivalence of formal derived smooth stacks 
\begin{equation}
 \flat^\mathrm{rel}\B G \;\simeq\; \B\mathfrak{g}.
\end{equation}
\end{remark}

%%%%%%%%%%%%%%%%%%%%%%%%%%%%%%%%%%%%%%%%%%%%%%%%%%%%%%%%%%%%%%%%%%%%%%%%%%%%%%%%%%%%%%%%%%%%%
\subsection{Derived jet bundles}

In this subsection we will provide a definition of jet bundles as formal derived smooth stacks, rooted in our differential structure, which we delineated above in this section. This will be an application of the framework developed by \cite{khavkine2017synthetic, DCCTv2}.

\begin{construction}
Let $M\in\mathbf{dFMfd}\hookrightarrow\mathbf{dFSmoothStack}$ be any fixed formal derived smooth manifold. A bundle $E\xrightarrow{\;p\;}M$ can be seen as an object of the slice $(\infty,1)$-category $\mathbf{dFSmoothStack}_{/M}$.
Recall that there is a morphism $\mathfrak{i}_M:M\rightarrow\Im(M)$, which is the infinitesimal shape unit of definition \ref{def:infinitesimalshape}, i.e. the canonical morphism from the derived formal smooth manifold $M$ to its de Rham space.
This induces a triplet of adjoint $(\infty,1)$-functors
\begin{equation}
   (\mathfrak{i}_M)_! \,\dashv\, (\mathfrak{i}_M)^\ast \,\dashv\,  (\mathfrak{i}_M)_\ast,
\end{equation}
which is the base change given by \cite[Proposition 6.3.5.1]{topos}, i.e. a triplet of $(\infty,1)$-functors the form
\begin{equation}
    \begin{tikzcd}[row sep=scriptsize, column sep=12.0ex, row sep=18.0ex]
     \mathbf{dFSmoothStack}_{/M} \arrow[r, "(\mathfrak{i}_M)_\ast", shift left=-3.1ex]\arrow[r, "(\mathfrak{i}_M)_{!\phantom{!}}", shift left=3.1ex] &  \mathbf{dFSmoothStack}_{/\Im(M)}, \arrow[l, "(\mathfrak{i}_M)^\ast"', shift left=0ex]
    \end{tikzcd}
\end{equation}
where $\mathbf{dFSmoothStack}_{/M}$ and $\mathbf{dFSmoothStack}_{/\Im(M)}$ are the slice $(\infty,1)$-categories of derived formal smooth sets respectively over $M$ and over its de Rham space $\Im(M)$.
\end{construction}

\begin{definition}[Derived jet bundle]
For a given fibre bundle $E\rightarrow M$, where $E$ is a formal derived smooth stack and $M$ is a formal derived smooth manifold, the \textit{jet bundle} $\mathrm{Jet}_M{E}\rightarrow M$ is a fibre bundle of formal derived smooth stacks which is defined by the image of the functor
\begin{equation}
\begin{aligned}
    \mathrm{Jet}_M \,:\; \mathbf{dFSmoothStack}_{/M} \;&\longrightarrow \mathbf{dFSmoothStack}_{/M} \\
    E \;&\longmapsto \;  \mathrm{Jet}_M E \,\coloneqq\, (\mathfrak{i}_M)^\ast (\mathfrak{i}_M)_\ast E.
\end{aligned}
\end{equation}
\end{definition}

That this generalises the usual definition of jet bundles becomes clearer after corollary \ref{co:jetfibre}.

\begin{remark}[Jet co-monad]
From the definition, similarly to the previously examined co-monad structures, one obtains that there exists a an equivalence of endofunctors
\begin{equation}\label{eq:jetnattrans}
    \itDelta\,:\;\mathrm{Jet}_M \;\simeq\;  \mathrm{Jet}_M\mathrm{Jet}_M.
\end{equation}
Thus we call the functor $\mathrm{Jet}_M$ \textit{jet co-monad} over $M$.
For any given bundle $(E\rightarrow M)\in\mathbf{dFSmoothStack}_{/M}$, the natural transformation \eqref{eq:jetnattrans} will give rise to a morphism
\begin{equation}\label{eq:jetcomonadcoproduct}
    \itDelta_E\,:\;\mathrm{Jet}_M E \;\longhookrightarrow\;  \mathrm{Jet}_M (\mathrm{Jet}_ME).
\end{equation}
This is the coproduct of the comonad structure associated to jet bundles, which was originally observed in the context of ordinary differential geometry by \cite{Marvan:1986}. 
In the rest of this subsection we will show that some essential results by \cite{khavkine2017synthetic, DCCTv2} follow through to the formal derived smooth case.
\end{remark}

\begin{lemma}[Adjunction with formal disk bundle]
There is a natural equivalence of functors
\begin{equation}
    (\mathfrak{i}_M)^\ast (\mathfrak{i}_M)_! \;\simeq\; T^\infty M\times_{M\!}(-)
\end{equation}
\end{lemma}

\begin{proof}
Consider any bundle $(E\xrightarrow{\,p\,}M)\in\mathbf{dFSmoothStack}_{/M}$. Then, the formal smooth set $T^\infty M\times_{M\!}E$ sits at the top-left corner of the following pullback squares:
\begin{equation}
    \begin{tikzcd}[row sep={12ex,between origins}, column sep={14ex,between origins}]
   T^\infty M\times_{M\!}E\arrow[r]\arrow[d] & E \arrow[d, "p"] \\
   T^\infty M \arrow[r]\arrow[d] & M \arrow[d, "\mathfrak{i}_M"] \\
   M \arrow[r, "\mathfrak{i}_M"] & \Im(M) ,
    \end{tikzcd}
\end{equation}
We can also see that there is an equivalence $T^\infty M\times_{M\!}E\simeq M\times_{\Im(M)\!}E$. Recall that, for a base change morphism, $(\mathfrak{i}_M)_!$ is the post-composition by $\mathfrak{i}_M$ and $(\mathfrak{i}_M)^\ast$ is the pullback along $\mathfrak{i}_M$.
Thus, the bundle $(\mathfrak{i}_M)_!E$ is nothing but the composition $\mathfrak{i}_M\circ p:E\rightarrow\Im(M)$ and the bundle $(\mathfrak{i}_M)^\ast (\mathfrak{i}_M)_!E$ is nothing but the pullback $T^\infty M\times_{M\!}E\rightarrow M$.
\end{proof}

\begin{theorem}[Formal disk bundle and jet bundle adjunction]
There is an adjunction
\begin{equation}
    T^\infty M\times_{M\!}(-) \;\dashv\;  \mathrm{Jet}_M.
\end{equation}
of endofunctors of the slice $(\infty,1)$-category $\mathbf{dFSmoothStack}_{/M}$.
\end{theorem}

\begin{proof}It is enough to notice that we have the following equivalences:
\begin{equation}\begin{aligned}
    \bfR\Hom_{{/M}}(E',\mathrm{Jet}_M{E}) \;&\simeq\; \bfR\Hom_{{/M}}(E',\, (\mathfrak{i}_M)^\ast (\mathfrak{i}_M)_\ast E) \\
     \;&\simeq\; \bfR\Hom_{{/M}}( (\mathfrak{i}_M)_! E', \,(\mathfrak{i}_M)_\ast E) \\
     \;&\simeq\; \bfR\Hom_{{/M}}((\mathfrak{i}_M)^\ast (\mathfrak{i}_M)_! E', \,E) \\
     \;&\simeq\; \bfR\Hom_{{/M}}(T^\infty M\times_{M\!}E', \,E)
\end{aligned}\end{equation}
where $\bfR\Hom_{{/M}}(-,-)$ is the hom-$\infty$-groupoid of the slice $(\infty,1)$-category $\mathbf{dFSmoothStack}_{/M}$. Therefore we have the wanted conclusion.
\end{proof}

\begin{corollary}[Mapping stack to jet bundles]\label{cor:adjunction}
We have the equivalence of formal derived smooth sets
\begin{equation}
    [E',\mathrm{Jet}_ME]_{/M} \;\simeq\; [T^\infty M\times_M E',E]_{/M}.
\end{equation}
\end{corollary}

\begin{corollary}[Sections of a jet bundle]
The $\infty$-groupoid of sections of a jet bundle $\mathrm{Jet}_M{E}$ is equivalent to the $\infty$-groupoid of bundle morphisms from $T^\infty M$ to $E$, i.e.
\begin{equation}
    \Gamma(M,\mathrm{Jet}_M{E}) \;\simeq\; \bfR\Hom_{{/M}}(T^\infty M,E).
\end{equation}
\end{corollary}

\begin{proof}
By setting in previous lemma $E'=M \xrightarrow{\;\mathrm{id}_M\;} M$ to be the tautological bundle, we obtain
\begin{equation*}\begin{aligned}
    \Gamma(M,\jet{E}) \;&\simeq\; \bfR\Hom_{{/M}}(M,\jet{E}) 
     \;\simeq\; \bfR\Hom_{{/M}}(T^\infty M\times_{M\!}M, \,E) 
      \;\simeq\;  \bfR\Hom_{{/M}}(T^\infty M,E),
\end{aligned}\end{equation*}
which is the result.
\end{proof}

By considering the special cases $E'=M$ and $E'=\ast \xhookrightarrow{\;x\;}M$ in corollary \ref{cor:adjunction} above, we obtain respectively the following two corollaries.

\begin{corollary}[Space of sections of a jet bundle]
We have the equivalence of formal derived smooth stacks
\begin{equation}
    \mathbf{\Gamma}(M,\mathrm{Jet}_ME) \;\simeq\; [T^\infty M,E]_{/M}
\end{equation}
\end{corollary}

\begin{corollary}[Fibre of a jet bundle]\label{co:jetfibre}
We have the equivalence of formal derived smooth stacks
\begin{equation}
    (\mathrm{Jet}_M{E})_x \;\simeq\; \mathbf{\Gamma}(\bbD_{M,x},E),
\end{equation}
where $(\mathrm{Jet}_M{E})_x$ is the fibre of $\jet{E}$ at any point $x\in M$ of the base manifold.
\end{corollary}

In other words, the jet bundle $\mathrm{Jet}_M{E}$ of a bundle $E$ is such that its fiber at any point $x\in M$ is the space of formal germs of sections of $E$ at $x$, as in the classical definition of jet bundle.

Notice that, for any fixed $M\in\mathbf{dFMfd}$, one has that $\mathrm{Jet}_M(-)$ is a functor on the slice category $\mathbf{dFSmoothStack}_{/M}$. In this light, we can define the jet prolongation of section as follows.

\begin{definition}[Jet prolongation of sections]
Given a section $\itPhi:M\rightarrow E$ of a bundle $E\xrightarrow{\;p\;}M$, its \textit{jet prolongation} can be defined by the composition 
\begin{equation}
    j(\itPhi)\,:\; M\xrightarrow{\;\;\simeq\;\;} \mathrm{Jet}_MM \xrightarrow{\;\,\mathrm{Jet}_M(\itPhi)\,\;} \mathrm{Jet}_ME,
\end{equation}
where $\mathrm{Jet}_MM$ is the jet bundle of the tautological bundle $\mathrm{id}_M:M\rightarrow M$.
\end{definition}

In other words, the jet prolongation provides a canonical map $j:\Gamma(M,E)\longtwoheadrightarrow\Gamma(M,\jet{E})$ which sends any section $\itPhi\in\Gamma(M,E)$ to its germs $j(\itPhi)\in\Gamma(M,\jet{E})$ at every point of the base manifold.
To sum up, we have a diagram of the following form:
\begin{equation}
\begin{tikzcd}[column sep={6.1em,between origins}, row sep={5.6em,between origins}]
    &\mathrm{Jet}_M{E} \arrow[d, two heads, "\pi^\infty_0"] \arrow[dd, bend left=50, "\pi_M", two heads]\\
    &E \arrow[d, "p", two heads] \\
    M\arrow[ur, "\itPhi"] \arrow[uur, bend left=30, "j(\itPhi)"] \arrow[r, "\mathrm{id}"] & M . 
\end{tikzcd}
\end{equation}

A paper in preparation \cite{AlfYouFuture} will be devolved to the exploitation of the features of derived jet bundles in the context of derived differential geometry.

%%%%%%%%%%%%%%%%%%%%%%%%%%%%%%%%%%%%%%%%%%%%%%%%%%%%%%%%%%%%%%%%%%%%%%%%%%%%%%%%%%%%%%%%%%%%%
\section{Global aspects of classical BV-theory}\label{sec:global_aspects_of_bv_theory}

In this section we finally get our hands dirty: we will use the new toolbox provided by derived differential geometry to investigate some global-geometric features of classical field theory. 
The point of this section is not to provide a systematic non-perturbative reformulation of BV-theory, but to show that the tools developed in this paper open at least the way to progress.

In subsection \ref{sec:review_BV}, we will provide a brief review of usual BV-theory via $L_\infty$-algebras -- as it is probably more familiar to the physically oriented reader -- and we will explain how this relates with the formal moduli problem picture. Moreover, we will provide the concrete examples of scalar field theory and Yang-Mills theory.
Such examples will be important for later comparison with the global-geometric picture which we are going to construct respectively in the second and the third subsection.
In fact, in subsection \ref{subsec:scalar}, we will study the global derived critical locus of an action functional on the smooth set $\bfGamma(M,E)$ of sections of a bundle of smooth manifolds $E\twoheadrightarrow M$, which should be seen as the global configuration space of a scalar field theory.
In subsection \ref{subsec:YMT}, we will study the global derived critical locus of the Yang-Mills action functional on the smooth stack $\Bun_G^\nabla(M)$ of principal $G$-bundles with connection on a spacetime manifold $M$, which should be seen as the global configuration space of a gauge theory.

%%%%%%%%%%%%%%%%%%%%%%%%%%%%%%%%%%%%%%%%%%%%%%%%%%%%%%%%%%%%%%%%%%%%%%%%%%%%%%%%%%%%%%%%%%%%%
\subsection{Review of BV-theory via $L_\infty$-algebras}\label{sec:review_BV}

In this subsection we will briefly review usual classical BV-theory, formulated in terms of $L_\infty$-algebras. For more details, we point at the references \cite{Pau14, Saem18bv, Saem19bv, Doubek_2019, Jurco:2019yfd, Jurco:2020yyu}.
Closely related applications of $L_\infty$-algebras to field theories have been explored by \cite{Hohm17, Hohm17x, Hohm19, Hohm19x}.

\begin{construction}[Usual BV-theory via $L_\infty$-algebras]
Let us consider an $L_\infty$-algebra $\mathfrak{L}$, which we can think as the algebra encoding the kinematics of a classical field theory: this will be the first input of BV-theory.
Such an $L_\infty$-structure can be dually given by its Chevalley-Eilenberg dg-algebra $\CE(\mathfrak{L})$, which is going to be of the form 
\begin{equation}
    \CE(\mathfrak{L}) \;=\; \big(\mathrm{Sym}\,\mathfrak{L}^\vee[-1],\;\di_{\CE(\mathfrak{L})}\big)
\end{equation}
and which is also known as BRST complex in physical contexts. The second ingredient to feed the machinery of BV-theory is the action functional for our field theory, which can be regarded as an element $S \in \CE(\mathfrak{L})$ of our Chevalley-Eilenberg dg-algebra.

Consider the graded vector space $\mathfrak{L}[1]$, which is the graded manifold with the property that $\Coo(\mathfrak{L}[1]) = \mathrm{Sym}\,\mathfrak{L}^\vee[-1]$.
Then, the machinery of BV-theory instructs us to take the $(-1)$-shifted cotangent bundle of such a graded vector space, namely
\begin{equation}
    T^\vee[-1]\mathfrak{L}[1] \;=\; (\mathfrak{L} \oplus \mathfrak{L}^\vee[-3])[1].
\end{equation}
Observe that, by generalising the case of ordinary cotangent bundles, a $(-1)$-shifted cotangent bundle comes equipped with a natural $(-1)$-shifted Poisson bracket $\{-,-\}$.
The objective of the machinery is to equip the new graded vector space $\mathfrak{L} \oplus \mathfrak{L}^\vee[-3]$ with the structure of an $L_\infty$-algebra which extends our starting $L_\infty$-algebra $\mathfrak{L}$ in a certain way. 
To do that, we can define the so-called classical BV-action $S_\BV\in\mathrm{Sym}(\mathfrak{L}\oplus \mathfrak{L}^\vee[-3])^\vee[-1]$ by the sum
\begin{equation}
    S_\BV \;=\; S + S_{\mathrm{BRST}},
\end{equation}
where $S \in \CE(\mathfrak{L})$ is the original action of the theory and $S_{\mathrm{BRST}}\coloneqq\widehat{\di_{\CE(\mathfrak{L})}}$ is the cotangent lift of the original Chevalley-Eilenberg differential $\di_{\CE(\mathfrak{L})}$, i.e. its image along the natural inclusion $\widehat{(-)}: \mathrm{Sym}(\mathfrak{L}^\vee[-1])\otimes \mathfrak{L}[1] \longhookrightarrow \mathrm{Sym}(\mathfrak{L}\oplus \mathfrak{L}^\vee[-3])^\vee[-1]$.
The classical BV-action satisfies the so-called \textit{classical master equation} $\{S_{\BV},S_{\BV}\} \,=\, 0$.
Then we can define the BV-differential by
\begin{equation}
    Q_\BV \;\coloneqq\; \{S_{\BV},-\},
\end{equation}
so that the classical master equation is indeed equivalent to $Q_\BV^2 = 0$. 
Moreover, notice that we have an isomorphism of graded vector spaces $\mathrm{Sym}(\mathfrak{L}\oplus \mathfrak{L}^\vee[-3])^\vee[-1] = \mathrm{Sym}(\mathfrak{L}^\vee[-1]\oplus \mathfrak{L}[2])$.
Thus we have all we need to define the following Chevalley-Eilenberg dg-algebra structure: 
\begin{equation}
    \CE\big(\mathfrak{Crit}(S)\big) \;\coloneqq\; \big(\mathrm{Sym}(\mathfrak{L}^\vee[-1]\oplus \mathfrak{L}[2]),\;\, Q_\BV = \{S_\BV,-\} \big).
\end{equation}
This can be dually interpreted  as an $L_\infty$-algebra $\mathfrak{Crit}(S)$ whose underlying graded vector space $T^\vee[-1]\mathfrak{L}[1]$, as we wanted. The Chevalley-Eilenberg dg-algebra $\CE\big(\mathfrak{Crit}(S)\big)$ is what is known as BV-complex in physical literature.
As noticed by \cite{FactI, FactII}, this discussion can be nicely refined by replacing, in the discussion above, $L_\infty$-algebras with local $L_\infty$-algebras on a fixed spacetime.
\end{construction}

\begin{remark}[Usual BV-theory via polyvectors]
Observe that, provided that we interpret $\mathrm{Sym}(\mathfrak{L}^\vee[-1]\oplus \mathfrak{L}[2])=  \mathrm{Pol}(\mathfrak{L}[1])$ as the dg-algebra of polyvector fields on the graded space $\mathfrak{L}[1]$, we can naturally see the BV-differential as
\begin{equation}
     Q_\BV \;=\; \iota_{(-)}\di_\dR S + \mathcal{L}_{\di_{\CE(\mathfrak{L})}},
\end{equation}
where $\iota_{(-)}\di_\dR S$ is the contraction of polyvectors with the de Rham differential of the starting action functional $S$ and $\mathcal{L}_{\di_{\CE(\mathfrak{L})}}$ is the Lie derivative of polyvectors along the Chevalley-Eilenberg differential of the BRST $L_\infty$-algebra $\mathfrak{L}$.
\end{remark}

\begin{construction}[Usual BV-theory via formal moduli problems]
In \cite{FactII}, a beautiful geometrical insight on BV-theory is provided.
The de Rham differential of the original action $S\in\CE(\mathfrak{L})$ can be seen as an element
$\di_\dR S\in\CE(\mathfrak{L},\mathfrak{L}^\vee[-1])$ of the Chevalley-Eilenberg dg-algebra of $\mathfrak{L}$ valued in the $\mathfrak{L}$-module $\mathfrak{L}^\vee[-1]$.
Remarkably, in \cite{FactII} it is shown that the classical BV-$L_\infty$-algebra $\mathfrak{Crit}(S)$ can be geometrically seen as the $L_\infty$-algebra associated to the pointed formal moduli problem which is the derived critical locus of the action $S$.
In other words, one has a notion of a cotangent pointed formal moduli problem $T^\vee\mathbf{MC}(\mathfrak{L})$, whose complex of sections is exactly $\CE(\mathfrak{L},\mathfrak{L}^\vee[-1])$. Then, the pointed formal moduli problem $\mathbf{MC}_{\!}\big(\mathfrak{Crit}(S)\big)$ can be identified with the homotopy pullback
\begin{equation}
    \begin{tikzcd}[row sep={15.5ex,between origins}, column sep={17ex,between origins}]
    \mathbf{MC}_{\!}\big(\mathfrak{Crit}(S)\big) \arrow[r] \arrow[d] & \mathbf{MC}(\mathfrak{L}) \arrow[d, "\di_\dR S"] \\
    \mathbf{MC}(\mathfrak{L})  \arrow[r, "0"] & T^\vee\mathbf{MC}(\mathfrak{L})
    \end{tikzcd}
\end{equation}
of formal moduli problems.
Thus, in principle, we can obtain the $L_\infty$-algebra $\mathfrak{Crit}(S)$ which encodes classical BV-theory from a purely geometric construction -- namely, a flavour of derived intersection -- which is not very manifest when we approach BV-theory by following the usual recipe based on constructing the classical BV-action.
\end{construction}

Let us now take some time to explore two fundamental classes of examples of BV-theories in terms of $L_\infty$-algebras and formal moduli problems: scalar fields and gauge theories.

\begin{example}[Klein-Gordon theory]
We start from the following graded vector space:
\begin{equation}
\begin{aligned}
    \mathfrak{L}[1] \;=\; \Coo(M),
\end{aligned}
\end{equation}
equipped with the trivial $L_\infty$-algebra structure.
The classical action of a Klein-Gordon field $\phi\in\Coo(M)$ with arbitrary interaction terms is given by
\begin{equation}
    S(\phi) \;=\; \int_M \bigg(\frac{1}{2}\phi\square\phi + \sum_{k>1}\frac{m_{k}}{k!}\phi^k\bigg)\mathrm{vol}_M.
\end{equation}
By following the aforementioned recipe, one can obtain an $L_\infty$-algebra on the complex
\begin{equation}
\begin{aligned}
    \mathfrak{Crit}(S)[1] \;=\; \,&\Big( \begin{tikzcd}[row sep={14.5ex,between origins}, column sep= 8ex]
    \Coo(M) \arrow[r, "\square+m_2"] & \Coo(M)
    \end{tikzcd}  \Big)\\
    {\scriptstyle\text{deg}\,=} &\quad \;\;\begin{tikzcd}[row sep={14.5ex,between origins}, column sep= 5ex]
    {\scriptstyle 0} && \quad \;\;{\scriptstyle 1}
    \end{tikzcd} 
\end{aligned}
\end{equation}
whose $L_\infty$-structure is given by 
\begin{equation}
\begin{aligned}
    \ell_1(\phi) \;&=\; (\square+m_2)\phi, \\
    \ell_k(\phi_1,\dots,\phi_k) \;&=\; m_{k+1}\phi_1\cdots\phi_k \quad\text{for $k>1$}
\end{aligned}
\end{equation}
for any $\phi_i\in\Coo(M)[0]$.
Informally speaking, it is suggestive to rewrite the $L_\infty$-algebra structure above, dually, in terms of its Chevalley-Eilenberg differential:
\begin{equation}
    \begin{aligned}
        Q_\BV &:\, \upphi \!\!\!&&\longmapsto\; 0,\\
        Q_\BV &:\, \upphi^+ \!\!\!&&\longmapsto\; \square\upphi \,+ \sum_{k>0}\frac{m_{k+1}}{k!}\upphi^k,
    \end{aligned}
\end{equation}
where $\upphi:\Coo(M)\rightarrow\bbR$ and $\upphi^+:\Coo(M)\rightarrow\bbR$ should be thought of as coordinate functions of the underlying graded vector space.
We can also explicitly write the Maurer-Cartan formal moduli problem $\mathbf{MC}(\mathfrak{Crit}(S))$ associated to the $L_\infty$-algebra above. Given a dg-Artinian algebra $\mathcal{R}$, the set of $0$-simplices of the simplicial set $\mathrm{MC}(\mathfrak{Crit}(S)\otimes \mathfrak{m}_{\mathcal{R}})$ is given by
\begin{equation*}
\begin{aligned}
        \mathrm{MC}(\mathfrak{Crit}(S)\otimes \mathfrak{m}_{\mathcal{R}})_0 \,=\, \left\{ \left. \begin{array}{ll}
         \phi &\!\!\!\!\in\mathcal{C}^{\infty\!}(M)\otimes\mathfrak{m}_{\mathcal{R},0} \\[1.0ex]
        \phi^+ &\!\!\!\!\in \mathcal{C}^{\infty\!}(M) \otimes\mathfrak{m}_{\mathcal{R},-1}
        \end{array}  \;\right|\;  
        \square\phi + \sum_{k>0}\frac{m_{k+1}}{k!}\phi^k \,=\, \di_{\mathcal{R}} \phi^+
        \right\},
\end{aligned}
\end{equation*}
and the set of $1$-simplices is
\begin{equation*}
\hspace{-0.4cm}\begin{aligned}
        \mathrm{MC}(\mathfrak{Crit}(S)\otimes \mathfrak{m}_{\mathcal{R}})_1 \,=\, \left\{\!\left. \begin{array}{ll}
          \phi_0 &\!\!\!\!\!\in\mathcal{C}^{\infty\!}(M)\otimes\mathfrak{m}_{\mathcal{R},0}\otimes\Omega^0([0,1]) \\[1.0ex]
          \phi^1\di t &\!\!\!\!\!\in\mathcal{C}^{\infty\!}(M)\otimes\mathfrak{m}_{\mathcal{R},-1}\otimes\Omega^1([0,1]) \\[1.0ex]
        \phi^+_0 &\!\!\!\!\!\in\mathcal{C}^{\infty\!}(M) \otimes\mathfrak{m}_{\mathcal{R},-1}\otimes\Omega^0([0,1]) \\[1.0ex]
        \phi^{+}_1\di t &\!\!\!\!\!\in \mathcal{C}^{\infty\!}(M) \otimes\mathfrak{m}_{\mathcal{R},-2}\otimes\Omega^1([0,1])
        \end{array} \! \right|\! \begin{array}{rl} 
        \square\phi_0 + \sum_{k\!}\frac{m_{k+1}}{k!}\phi^k_0  &\!\!\!=\, \di_{\mathcal{R}} \phi^+_0\\[1.0ex]
        \frac{\di}{\di t}\phi_0 &\!\!\!=\, \di_{\mathcal{R}} \phi_1  \\[1.0ex]
        \square\phi_1 + \sum_{k\!}\frac{m_{k+1}}{k!}\phi^{k-1}_0\phi_1 &\!\!\!=\, \di_{\mathcal{R}} \phi^{+}_1
        \end{array} \!\!\right\},
\end{aligned}
\end{equation*}
where $t$ is a coordinate on the unit interval $[0,1]\subset \bbR$.
And so on for the higher simplices.
\end{example}

\begin{example}[Yang-Mills theory]
Consider now the $L_\infty$-algebra $\mathfrak{L}$, whose underlying complex is the differential graded vector space 
\begin{equation}
\begin{aligned}
    \mathfrak{L}[1] \;=\; \,&\Big( \begin{tikzcd}[row sep={14.5ex,between origins}, column sep= 5ex]
    \Omega^0(M,\mathfrak{g}) \arrow[r, "\di"] & \Omega^1(M,\mathfrak{g})
    \end{tikzcd}  \Big)\\
    {\scriptstyle\text{deg}\,=} &\quad \;\;\begin{tikzcd}[row sep={14.5ex,between origins}, column sep= 5ex]
    {\scriptstyle -1} && \quad \;\;{\scriptstyle 0}
    \end{tikzcd} ,
\end{aligned}
\end{equation}
and whose $L_\infty$-bracket structure has only the following non-trivial brackets:
\begin{equation}
    \begin{aligned}
        \ell_1(c) \;&=\; \di c, \\
        \ell_2(c_1,c_2) \;&=\; [c_1,c_2]_\mathfrak{g}, \\
        \ell_2(c,A) \;&=\; [c,A]_\mathfrak{g},
    \end{aligned}
\end{equation}
for any elements $c_k\in\Omega^0(M,\mathfrak{g})$ and $A\in\Omega^1(M,\mathfrak{g})$.
Informally speaking, as it is often presented in the context of BRST-theory, this $L_\infty$-algebra us dually given by the Chevalley-Eilenberg differential
\begin{equation}
    \begin{aligned}
        \di_{\CE(\mathfrak{L})} &:\, \mathsf{c} \!\!\!&&\longmapsto\; -\frac{1}{2}[\mathsf{c},\mathsf{c}]_{\mathfrak{g}}, \\
        \di_{\CE(\mathfrak{L})} &:\, \mathsf{A} \!\!\!&&\longmapsto\; \di \mathsf{c} + [\mathsf{A},\mathsf{c}]_{\mathfrak{g}},
    \end{aligned}
\end{equation}
where $\mathsf{c}:\Omega^0(M,\mathfrak{g})\rightarrow\bbR$ and $\mathsf{A}:\Omega^1(M,\mathfrak{g})\rightarrow\bbR$ should be thought of as coordinate functions on the underlying graded vector space. 
Thus, the $L_\infty$-algebra $\mathfrak{L}$ is precisely the algebraic incarnation of the BRST complex of physics.
We want to consider the standard action functional of a Yang-Mills theory, which is given by
\begin{equation}
    S(A) \;=\; \frac{1}{2}\int_M \langle F_A, \star F_A \rangle_\mathfrak{g},
\end{equation}
where $F_A \coloneqq \nabla_{\!A}A = \di A+ [A\,\overset{\wedge}{,}\,A]_\mathfrak{g}$ is the field strength.
By exploiting the given pairing
\begin{equation}
    \langle -\,\overset{\wedge}{,}\,- \rangle_\mathfrak{g} \;:\; \Omega^{d-p}(M,\mathfrak{g})\times\Omega^{p}(M,\mathfrak{g}) \;\longrightarrow\; \Coo(M)
\end{equation}
we are led to an $L_\infty$-algebra $\mathfrak{Crit}(S)$ whose underlying differential graded vector space is
\begin{equation}
\begin{aligned}
    \mathfrak{Crit}(S)[1] \;=\; \,&\Big( \begin{tikzcd}[row sep={14.5ex,between origins}, column sep= 5ex]
    \Omega^0(M,\mathfrak{g}) \arrow[r, "\di"] & \Omega^1(M,\mathfrak{g}) \arrow[r, "\di\star\di"] & \Omega^{d-1}(M,\mathfrak{g}) \arrow[r, "\di"] & \Omega^{d}(M,\mathfrak{g})
    \end{tikzcd}  \Big)\\
    {\scriptstyle\text{deg}\,=} &\quad \;\;\begin{tikzcd}[row sep={14.5ex,between origins}, column sep= 5ex]
    {\scriptstyle -1} && \quad \;\;{\scriptstyle 0}  && \quad \;\;{\scriptstyle 1}  && \quad \;\;{\scriptstyle 2}
    \end{tikzcd} 
\end{aligned}
\end{equation}
and whose $L_\infty$-algebra structure has only the following non-trivial $L_\infty$-brackets:
\begin{equation}\label{eq:BV-yang-Mills}
    \begin{aligned}
        \ell_1(c) \;&=\; \di c, \\
        \ell_1(A) \;&=\; \di\star\di A, \qquad &\ell_1(A^+) \;&=\; \di A^+, \\
        \ell_2(c_1,c_2) \;&=\; [c_1,c_2]_\mathfrak{g}, \qquad &\ell_2(c,c^+) \;&=\; [c,c^+]_\mathfrak{g}, \\
        \ell_2(c,A) \;&=\; [c,A]_\mathfrak{g}, \qquad &\ell_2(c,A^+) \;&=\; [c,A^+]_\mathfrak{g}, \\
        & &\ell_2(A,A^+) \;&=\; [A \,\overset{\wedge}{,}\, A^+]_\mathfrak{g},
    \end{aligned}
\end{equation}\vspace{-0.4cm}
\begin{equation*}
    \begin{aligned}
        \ell_2(A_1,A_2) \;&=\; \di\star[A_1\,\overset{\wedge}{,}\,A_2]_\mathfrak{g} + [A_1 \,\overset{\wedge}{,}\, \star\di A_2]_\mathfrak{g} + [A_2 \,\overset{\wedge}{,}\, \star\di A_1]_\mathfrak{g}, \\
        \ell_3(A_1,A_2,A_3) \;&=\; \big[A_1\,\overset{\wedge}{,}\,\star [A_2\,\overset{\wedge}{,}\,A_3]_\mathfrak{g}\big]_\mathfrak{g} + \big[A_2\,\overset{\wedge}{,}\,\star [A_3\,\overset{\wedge}{,}\,A_1]_\mathfrak{g}\big]_\mathfrak{g} + \big[A_3\,\overset{\wedge}{,}\,\star [A_1\,\overset{\wedge}{,}\,A_2]_\mathfrak{g}\big]_\mathfrak{g}, \\
    \end{aligned}
\end{equation*}
for any $c_k\in\Omega^0(M,\mathfrak{g})$, $A_k\in\Omega^1(M,\mathfrak{g})$, $A^+_k\in\Omega^{d-1}(M,\mathfrak{g})$ and $c^+_k\in\Omega^d(M,\mathfrak{g})$ elements of the underlying graded vector space.
Informally speaking, we can think of this $L_\infty$-algebra as given, dually, by the following BV-differential:
\begin{equation}
    \begin{aligned}
        Q_\BV &:\, \mathsf{c} \!\!\!&&\longmapsto\; -\frac{1}{2}[\mathsf{c},\mathsf{c}]_{\mathfrak{g}} \\
        Q_\BV &:\, \mathsf{A} \!\!\!&&\longmapsto\; \di \mathsf{c} + [\mathsf{A},\mathsf{c}]_{\mathfrak{g}} \\[0.2ex]
        Q_\BV &:\, \mathsf{A}^+ \!\!\!&&\longmapsto\; -\nabla_{\!\mathsf{A}}\star\!F_{\mathsf{A}} - [\mathsf{c},\mathsf{A}^+]_{\mathfrak{g}} \\[0.3ex]
        Q_\BV &:\, \mathsf{c}^+ \!\!\!&&\longmapsto\; \nabla_{\!\mathsf{A}}\mathsf{A}^+ - [\mathsf{c},\mathsf{c}^+]_{\mathfrak{g}}
    \end{aligned}
\end{equation}
where $\mathsf{c}:\Omega^0(M,\mathfrak{g})\rightarrow\bbR$, $\mathsf{A}:\Omega^1(M,\mathfrak{g})\rightarrow\bbR$, $\mathsf{A}^+:\Omega^{d-1}(M,\mathfrak{g})\rightarrow\bbR$ and $\mathsf{c}^+:\Omega^{d}(M,\mathfrak{g})\rightarrow\bbR$ should be thought of as coordinate functions on the underlying graded vector space. Notice that this is precisely what is known as BV-BRST complex in physics.
Moreover, the classical BV-differential of Yang-Mills theory written above can be presented by a classical BV-action $S_\BV$, so that we have $Q_{\BV}=\{S_\BV,-\}$. Such a BV-action is the following familiar one:
\begin{equation}
    S_\BV(\mathsf{c},\mathsf{A},\mathsf{A}^+,\mathsf{c}^+) \;=\; \int_M \bigg( \underbrace{\frac{1}{2}\langle F_\mathsf{A}, \star F_\mathsf{A} \rangle_\mathfrak{g}}_{S} - \underbrace{\langle \mathsf{A}^+, \nabla_{\!\mathsf{A}}\mathsf{c} \rangle_\mathfrak{g} + \frac{1}{2} \langle \mathsf{c}^+,[\mathsf{c},\mathsf{c}]_\mathfrak{g} \rangle_\mathfrak{g}}_{S_\mathrm{BRST}} \bigg).
\end{equation}
\end{example}

Now, let us give a look to the formal moduli problem description of the example of Yang-Mills theory.
First, let us fix any ordinary Artinian algebra $\mathcal{R}\in\mathsf{Art}_\bbR$.
We will now explicitly construct the simplicial set $\mathrm{MC}(\mathfrak{L}\otimes \mathfrak{m}_{\mathcal{R}})$, where $\mathfrak{m}_{\mathcal{R}}$ is the maximal differential ideal of $\mathcal{R}$.
The set of $0$-simplices is just
\begin{equation*}
\begin{aligned}
        \mathrm{MC}(\mathfrak{L}\otimes \mathfrak{m}_{\mathcal{R}})_0 \,=\, \left\{  \begin{array}{ll}
         A &\!\!\!\!\in\Omega^1(M,\mathfrak{g})\otimes\mathfrak{m}_{\mathcal{R}}
        \end{array}\right\},
\end{aligned}
\end{equation*}
and the set of $1$-simplices is given by
\begin{equation*}
\begin{aligned}
        \mathrm{MC}(\mathfrak{L}\otimes \mathfrak{m}_{\mathcal{R}})_1 \,=\, \left\{\!\left. \begin{array}{ll}
          c_1 \di t &\!\!\!\!\in\Omega^0(M,\mathfrak{g})\otimes\mathfrak{m}_{\mathcal{R}}\otimes\Omega^1([0,1]) \\[1.0ex]
          A_0 &\!\!\!\!\in\Omega^1(M,\mathfrak{g})\otimes\mathfrak{m}_{\mathcal{R}}\otimes\Omega^0([0,1])
        \end{array} \! \right|\, \begin{array}{rl} 
        \frac{\di}{\di t}A_0 + \nabla_{\!A_0}c_1 &\!\!\!=\, 0 
        \end{array} \right\}.
\end{aligned}
\end{equation*}
and so on for higher simplices. 
This provided the formal moduli problem version of the BRST $L_\infty$-algebra $\mathfrak{L}$.
Now, we move on to the to the Maurer-Cartan formal moduli problem of the classical BV-BRST $L_\infty$-algebra $\mathfrak{Crit}(S)$, i.e. the functor
\begin{equation}
    \mathbf{MC}_{{}_{\!}}\big(\mathfrak{Crit}(S)\big) \,:\; \mathcal{R} \;\longmapsto\; \mathrm{MC}(\mathfrak{Crit}(S)\otimes \mathfrak{m}_{\mathcal{R}})
\end{equation}
where $\mathcal{R}$ is now allowed to be a dg-Artinian algebra.
For concreteness, let us write explicitly the sets of $0$- and $1$-simplices of this simplicial set for a fixed general dg-Artinian algebra $\mathcal{R}$.
So, the set of $0$-simplices is given by
\begin{equation*}
\begin{aligned}
        \mathrm{MC}(\mathfrak{Crit}(S)\otimes \mathfrak{m}_{\mathcal{R}})_0 \,=\, \left\{ \left. \begin{array}{ll}
         A &\!\!\!\!\in\Omega^1(M,\mathfrak{g})\otimes\mathfrak{m}_{\mathcal{R},0} \\[1.0ex]
        A^+ &\!\!\!\!\in \Omega^{d-1}(M,\mathfrak{g}) \otimes\mathfrak{m}_{\mathcal{R},-1} \\[1.0ex]
        c^+ &\!\!\!\!\in \Omega^{d}(M,\mathfrak{g}) \otimes\mathfrak{m}_{\mathcal{R},-2}
        \end{array}  \;\right|\; \begin{array}{rl} 
        \nabla_{\!A}\star \!F_{\!A} &\!\!\!=\, \di_{\mathcal{R}} A^+ \\[1.0ex]
        \nabla_{\!A}A^+ &\!\!\!=\, \di_{\mathcal{R}} c^+
        \end{array}\right\},
\end{aligned}
\end{equation*}
and the set of $1$-simplices is
\begin{equation*}
\hspace{-0.4cm}\begin{aligned}
        \mathrm{MC}(\mathfrak{Crit}(S)\otimes \mathfrak{m}_{\mathcal{R}})_1 \,=\, \left\{\!\left. \begin{array}{ll}
          c_1 \di t &\!\!\!\!\!\in\Omega^0(M,\mathfrak{g})\otimes\mathfrak{m}_{\mathcal{R},0}\otimes\Omega^1([0,1]) \\[1.0ex]
          A_0 &\!\!\!\!\!\in\Omega^1(M,\mathfrak{g})\otimes\mathfrak{m}_{\mathcal{R},0}\otimes\Omega^0([0,1]) \\[1.0ex]
          A^1\di t &\!\!\!\!\!\in\Omega^1(M,\mathfrak{g})\otimes\mathfrak{m}_{\mathcal{R},-1}\otimes\Omega^1([0,1]) \\[1.0ex]
        A^+_0 &\!\!\!\!\!\in \Omega^{d-1}(M,\mathfrak{g}) \otimes\mathfrak{m}_{\mathcal{R},-1}\otimes\Omega^0([0,1]) \\[1.0ex]
        A^{+}_1\di t &\!\!\!\!\!\in \Omega^{d-1}(M,\mathfrak{g}) \otimes\mathfrak{m}_{\mathcal{R},-2}\otimes\Omega^1([0,1]) \\[1.0ex]
        c^{+}_0 &\!\!\!\!\!\in \Omega^{d}(M,\mathfrak{g}) \otimes\mathfrak{m}_{\mathcal{R},-2}\otimes\Omega^0([0,1]) \\[1.0ex]
        c^{+}_1\di t &\!\!\!\!\!\in \Omega^{d}(M,\mathfrak{g}) \otimes\mathfrak{m}_{\mathcal{R},-3}\otimes\Omega^1([0,1])
        \end{array} \! \right| \begin{array}{rl} 
        \nabla_{\!A_0\!}\star \!F_{\!A_0} &\!\!\!=\, \di_{\mathcal{R}} A^+_0 \\[1.0ex]
        \nabla_{\!A_0}A^+_0 &\!\!\!=\, \di_{\mathcal{R}} c^{+}_0 \\[1.0ex]
        \frac{\di}{\di t}A_0 + \nabla_{\!A_0}c_1 &\!\!\!=\, \di_{\mathcal{R}} A_1  \\[1.0ex]
        \frac{\di}{\di t}A^+_0 +  \nabla_{\!A_0\!}\star\!F_{\!A_1} &\!\!\!\!\!\!+  \\[1.0ex]
        +\, [c_1,A^+_0] &\!\!\!=\, \di_{\mathcal{R}} A^{+}_1 \\[1.0ex]
        \frac{\di}{\di t}c^{+}_0 + \nabla_{\!A_0}A^{+}_1  &\!\!\!\!\!\!+  \\[1.0ex]
        +\, [c_1,c_0^+] &\!\!\!=\, \di_{\mathcal{R}} c^{+}_1
        \end{array} \!\!\right\},
\end{aligned}
\end{equation*}
where $t$ is a coordinate on the unit interval $[0,1]\subset \bbR$.
The elements of this set are $1$-simplices in the sense that each of them links a $0$-simplex
$(A,\, A^+,\, c^+) = (A_0(0),\, A^{+}_0(0),\, c^+_0(0))$ at $t=0$ to the $0$-simplex $(A^{\prime},\, A^{+\prime},\, c^{+\prime}) = (A_0(1),\, A^{+}_0(1),\, c^+_0(1))$ at $t=1$.
And so on for higher simplices. \vspace{0.5cm}

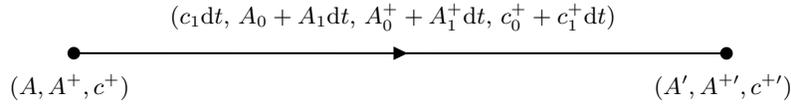
\begin{figure}[h]
    \centering
\tikzset{every picture/.style={line width=0.75pt}} %set default line width to 0.75pt        
\begin{tikzpicture}[x=0.75pt,y=0.75pt,yscale=-1,xscale=1]
%Straight Lines [id:da15515284919817396] 
\draw    (20.5,49.5) -- (349.67,49.58) ;
\draw [shift={(349.67,49.58)}, rotate = 0.01] [color={rgb, 255:red, 0; green, 0; blue, 0 }  ][fill={rgb, 255:red, 0; green, 0; blue, 0 }  ][line width=0.75]      (0, 0) circle [x radius= 2.68, y radius= 2.68]   ;
\draw [shift={(188.88,49.54)}, rotate = 180.01] [fill={rgb, 255:red, 0; green, 0; blue, 0 }  ][line width=0.08]  [draw opacity=0] (7.14,-3.43) -- (0,0) -- (7.14,3.43) -- cycle    ;
\draw [shift={(20.5,49.5)}, rotate = 0.01] [color={rgb, 255:red, 0; green, 0; blue, 0 }  ][fill={rgb, 255:red, 0; green, 0; blue, 0 }  ][line width=0.75]      (0, 0) circle [x radius= 2.68, y radius= 2.68]   ;
% Text Node
\draw (-13.5,58.4) node [anchor=north west][inner sep=0.75pt]  [font=\footnotesize]  {$( A,A^{+} ,c^{+})$};
% Text Node
\draw (311.5,58.4) node [anchor=north west][inner sep=0.75pt]  [font=\footnotesize]  {$( A',A^{+\prime} ,c^{+\prime})$};
% Text Node
\draw (67.5,22.9) node [anchor=north west][inner sep=0.75pt]  [font=\footnotesize]  {$( c_{1}\mathrm{d} t,\, A_{0} +A_{1}\mathrm{d} t,\, A_{0}^{+} +A_{1}^{+}\mathrm{d} t,\, c_{0}^{+} +c_{1}^{+}\mathrm{d} t)$};
\end{tikzpicture}
    \caption{$0$- and $1$-simplices of $\mathrm{MC}(\mathfrak{Crit}(S)\otimes \mathfrak{m}_{\mathcal{R}})$.}
    \label{fig:bv_fmp}
\end{figure}

The rest of this section will devoted to the construction of a global version of this formalism in the context of derived differential geometry.

%%%%%%%%%%%%%%%%%%%%%%%%%%%%%%%%%%%%%%%%%%%%%%%%%%%%%%%%%%%%%%%%%%%%%%%%%%%%%%%%%%%%%%%%%%%%%
\subsection{Global scalar field theory}\label{subsec:scalar}

In this subsection we will first illustrate the smooth set structure of the space of sections of a fibre bundle, which is the configuration space of a scalar field theory. Second, we will see how the derived critical locus of a smooth functional on such a space is defined and what is its formal derived smooth structure.
It is worth stressing that the fibre bundle $E\twoheadrightarrow M$ corresponding to a general classical  scalar field theory does not have to be a vector bundle; in fact, it can be a general fibre bundle of smooth manifolds. 

\begin{definition}[Smooth set of sections]
Let $M$ be a smooth manifold and $E\twoheadrightarrow M$ a fibre bundle of smooth manifolds. The \textit{smooth set of sections} $\mathbf{\Gamma}(M,E)\in\mathsf{SmoothSet}$ of $E$ is defined by the formal smooth sheaf
\begin{equation}
    \mathbf{\Gamma}(M,E)\,:\;U\,\longmapsto\,\Gamma(M\times U,\, \pi^\ast_M E)
\end{equation}
where $\pi_M:M\times U\rightarrow M$ is the natural projection and $U\in\Mfd$ is any smooth manifold.
\end{definition}

\begin{remark}[Diffeological space of sections]
Notice that the smooth set $\mathbf{\Gamma}(M,E)$ is a concrete sheaf and, thus, it is in particular a diffeological space.
\end{remark}

\begin{remark}[As sheaf on spacetime $M$]
Notice that the formal smooth set of sections $\mathbf{\Gamma}(M,E)$ is also a sheaf on the smooth manifold $M$.
This means that, for any good open cover $\{V_i\}_{i\in I}$ of the smooth manifold $M$, we have the limit
\begin{equation}
    \mathbf{\Gamma}(M,E) \;\simeq\; \lim_{} \left(   \begin{tikzcd}[row sep=scriptsize, column sep=4ex] \prod_{i} \mathbf{\Gamma}(V_i,E) & \arrow[l, yshift=0.7ex] \arrow[l, yshift=-0.7ex] \prod_{i,j} \mathbf{\Gamma}(V_i\cap V_j,E) 
    \end{tikzcd}   \right)
\end{equation}
in the category of formal smooth sets. In this sense, $\mathbf{\Gamma}(-,E)$ can be seen as a "sheaf of sheaves".
More precisely, we can see $\mathbf{\Gamma}(-,E)$ as a sheaf on the product site $\mathsf{Mfd}\times\mathsf{Op}(M)$, where $\mathsf{Mfd}$ is the site of ordinary smooth manifolds and $\mathsf{Op}(M)$ is the one of open subsets of the manifold $M$.
\end{remark}

The crucial reason why we promoted the bare set of sections $\Gamma(M,E)\in\mathsf{Set}$ to a smooth set $\mathbf{\Gamma}(M,E)\in\mathsf{SmoothSet}$ is that the latter comes with a smooth structure -- which is, in particular, the structure of a diffeological space. Therefore, as seen in section \ref{sec:smooth_set}, there is a well-defined notion of differential geometry on such a space.

\begin{example}[$\sigma$-models]
An interesting class of examples of such a configuration space is the one of $\sigma$-models, where the bundle is trivial and its total space is a product manifold $E \coloneqq M \times X$ for some smooth manifold $X$. This way, the configuration space $\mathbf{\Gamma}(M,E) \simeq [M,X]$ is given by the mapping space of the two manifolds, namely the smooth set of smooth maps from the manifold $M$ to the manifold $X$, which is usually called \textit{target space} of the theory. 
\end{example}

Next, let us extend our smooth set $\bfGamma(M,E)$ to a formal smooth set, by embedding it along the natural embedding $\mathsf{SmoothSet} \longhookrightarrow \mathsf{FSmoothSet}$ from section \ref{subsec:formalsmoothset}. For simplicity, we will keep denoting by $\bfGamma(M,E)$ the formal smooth set obtained by this embedding.

\begin{example}[Parameterised families of scalar fields]
Let us consider some basic examples of parametrised families of sections of a bundle of smooth manifolds $E\twoheadrightarrow M$.
\vspace{-0.2cm}%%%2
\begin{itemize}
    \item Let $U=\ast$ be the point. A $\ast$-parameterised family of sections $\itPhi:\ast\rightarrow\mathbf{\Gamma}(M,E)$ is nothing but an element of the bare set $\itPhi\in\Gamma(M,E)$.
    \item Now, let $U=\bbR^p$ with $p>0$. A $\bbR^p$-parameterised family of sections $\itPhi:\bbR^p\rightarrow\mathbf{\Gamma}(M,E)$ is nothing but a family of sections $\itPhi_u\in\Gamma(M,E)$ which smoothly varies by varying $u\in\bbR^p$.
    \item Now, let $U=\Spec(\bbR[\epsilon]/(\epsilon^2))$ be the formal smooth manifold whose $\Coo$-algebra of functions is given by the dual numbers (i.e. a thickened point). A $\Spec(\bbR[\epsilon]/(\epsilon^2))$-parameterised family of sections $\itPhi:\Spec(\bbR[\epsilon]/(\epsilon^2))\rightarrow\mathbf{\Gamma}(M,E)$ is equivalently a point $\itPhi:\ast\rightarrow T\mathbf{\Gamma}(M,E)$ in the tangent bundle of the original formal smooth set.
\end{itemize}
\end{example}

Now that we have our global-geometric configuration space $\bfGamma(M,E)$ of a scalar field theory, we need to introduce its dynamics. This can be done with an action functional for the scalar field theory. So, first, we need to construct the smooth set of compactly supported sections.

\begin{construction}[Smooth set of compactly supported sections]
We can construct the \textit{smooth set of compactly supported sections} $\bfGamma_{\! c}(M,E) \hookrightarrow \bfGamma(M,E)$ by the sheaf which sends any smooth manifold $U$ to the set of those sections $\itPhi_u\in\Gamma(M\times U,\, \pi^\ast_M E)$ whose support $\mathrm{supp}(\itPhi_u)\hookrightarrow M\times U \xtwoheadrightarrow{\;\pi_U\;} U$ maps properly.
\end{construction}

\begin{construction}[Variational calculus on spaces of sections]
$\bfGamma_{\!c}(M,E) \hookrightarrow \bfGamma(M,E)$.
As previously discussed, a smooth functional on sections of a bundle $E\twoheadrightarrow M$ is exactly a morphism of smooth sets
\begin{equation}
    S\,:\, \bfGamma_{\!c}(M,E) \,\longrightarrow\, \bbR,
\end{equation}
or, equivalently, a smooth function $S\in\mathcal{O}(\bfGamma_{\!c}(M,E))$ on the smooth set of sections.
On every element of the site $U\in\Mfd$, this is concretely given by a morphism of sets
\begin{equation}
    S(U)\,:\, \Gamma_{\!c}(M\times U, \pi^\ast_M E) \,\longrightarrow\, \Coo(U,\bbR)
\end{equation}
where $S(U)$ sends $U$-parametrised sections of the bundle $E\twoheadrightarrow M$ to smooth functions on $U$. Moreover, for any morphism $f:U\rightarrow U'$ in the site, we have compatibility conditions between $S(U)$ and $S(U')$.
The so-called first variation of this functional is nothing but the morphism of smooth sets
\begin{equation}
    \di_\dR S\,:\, \bfGamma_{\!c}(M,E) \,\xrightarrow{\;\;S\;\;}\, \bbR \,\xrightarrow{\;\;\di_\dR\;\;}\,\bfOmega^1,
\end{equation}
where $\bfOmega^1$ is the smooth set of differential $1$-forms and $\di_\dR\in\Hom(\bbR,\Omega^1)$ is the de Rham differential.
Such a morphism of smooth sets is a well-defined $1$-form $\di_\dR S\in\Omega^1(\bfGamma_{\!c}(M,E))$ on the smooth set of compactly supported sections $\bfGamma_{\!c}(M,E)$.
\end{construction}
Since $\di_\dR S$ is a differential form, notice that it maps vectors by
\begin{equation}
    (\di_\dR S)_{\itPhi}\,:\, T\bfGamma_{\!c}(M,E)_{\itPhi} \,\longrightarrow\,  \bbR
\end{equation}
at any point $\itPhi: \ast \rightarrow \bfGamma_{\!c}(M,E)$

\begin{construction}[Restricted cotangent bundle]
Now, let us consider the vertical tangent bundle $T_\mathrm{ver} E \coloneqq \mathrm{ker}(TE\twoheadrightarrow TM)$, which is a vector bundle on the base manifold $E$. Consider also its dual vector bundle $T_\mathrm{ver}^\vee E\twoheadrightarrow E$. 
These two bundles come equipped with the canonical pairing
$\langle-,-\rangle_E : T_\mathrm{ver}E \times_E T_\mathrm{ver}^\vee E \longrightarrow E\times\bbR$.
Since $T_\mathrm{ver} E$ and $T_\mathrm{ver}^\vee E$ are also bundles of smooth manifolds on the base manifold $M$ by post-composition with $E\twoheadrightarrow M$, we obtain a pairing
\begin{equation}
    \langle-,-\rangle_E \;:\; \bfGamma(M,T_\mathrm{ver}E) \times_{\bfGamma(M,E)} \bfGamma(M,T_\mathrm{ver}^\vee E) \,\longrightarrow\, \bfGamma(M,E)_{\!}\times[M,\bbR].
\end{equation}
Recall that there is a canonical equivalence $T\bfGamma(M,E) \simeq \bfGamma(M,T_\mathrm{ver}E)$.
Thus, it makes sense to define the \textit{restricted cotangent bundle} of the smooth set of sections $\mathbf{\Gamma}(M, E)$ by
\begin{equation}
    T_\mathrm{res}^\vee\mathbf{\Gamma}(M, E) \;\coloneqq\; \mathbf{\Gamma}(M,T_\mathrm{ver}^\vee E).
\end{equation}
If, as is usually the case in physics, the action functional $S\in\mathcal{O}(\bfGamma_{\! c}(M,E))$ of the considered field theory is a local functional\footnote{
The argument goes roughly as follows.
A local action functional can be expressed by $S(\phi)=\int_{M\!}j(\phi)^{\ast\!} L\mathrm{vol}_M$, where $j(\phi)$ is the jet prolongation of the section $\phi$ and $L$ is the Lagrangian, which is a function on the jet bundle.
It is possible to show \cite{Khavkine:2012jf, Khavkine:2014kya} that its first variation is given by $\di_\dR S(\phi)=\int_{M\!}j(\phi)^\ast \delta_{\mathrm{EL}}L\wedge\mathrm{vol}_M$,
where $\delta_{\mathrm{EL}}L$ is the so-called Euler-Lagrange form, which is a section 
$\delta_{\mathrm{EL}}L: \mathrm{Jet}_M\!E \rightarrow T^\vee_\mathrm{ver}E$.
Then, by defining $\delta S = j(-)^\ast\delta_{\mathrm{EL}}L$, one gets the functional derivative.
In \cite{AlfYouFuture} we will deal more systematically with these field-theoretic details.
}, then the de Rham differential of the action functional can be written in the form $\di_\dR S = \int_M\!\mathrm{vol}_M\langle \delta S,-\rangle_E$
for some morphism
\begin{equation}
    \delta S\,:\, \bfGamma(M,E) \,\longrightarrow\,  \mathbf{\Gamma}(M,T_\mathrm{ver}^\vee E),
\end{equation}
which we call \textit{variational derivative} of the action functional $S$, and some fixed volume form $\mathrm{vol}_M$.
In fact, this represents the notion of variational derivative familiar to physicists and the equation $\delta S =0$ is precisely the Euler-Lagrange equation.
\end{construction}

We can now introduce the derived critical locus of an action functional $S$ as the derived zero locus of its variational derivative $\delta S$.

\begin{definition}[Derived critical locus of an action functional]
Let $\mathbf{\Gamma}(M,E)\in\mathsf{SmoothSet}$ be the smooth set of sections of a bundle $E\twoheadrightarrow M$ of smooth manifolds and let $S:\mathbf{\Gamma}_{\!c}(M,E)\rightarrow \bbR$ be an action functional. We define the \textit{derived critical locus} $\bfR\crit{S}(M)\in\mathbf{dFSmoothSet}$ of the action functional $S$ by the homotopy pullback
\begin{equation}
    \begin{tikzcd}[row sep=scriptsize, column sep=5.5ex, row sep=14.5ex]
    \bfR\crit{S}(M) \arrow[r] \arrow[d] & \mathbf{\Gamma}(M,E) \arrow[d, "\delta S", hook] \\
    \mathbf{\Gamma}(M,E)  \arrow[r, "0"] & \mathbf{\Gamma}(M,T_\mathrm{ver}^\vee E),
    \end{tikzcd}
\end{equation}
in the $(\infty,1)$-category $\mathbf{dFSmoothSet}$, where $0$ is the zero-section and $\delta S$ is the de Rham differential of the action functional functional $S$.
\end{definition}

\begin{remark}[Derived critical locus is a formal derived diffeological space]
The ordinary critical locus $\crit{S}(M)\in\mathsf{SmoothSet}$ is given by the underived-truncation of the derived critical locus, i.e. by $\mathit{\Pi}^\mathrm{dif}\bfR\crit{S}(M) \simeq \crit{S}(M)$.
Notice that $\crit{S}(M)\hookrightarrow \bfGamma(M,E)$ is a diffeological space.
This implies that the derived critical locus $\bfR\crit{S}(M)\in\mathbf{dFDiffSp}$ is, in particular, a formal derived diffeological space.
\end{remark}

%\begin{remark}[Derived critical locus as stack on spacetime $M$]
%Notice that the derived critical locus $\bfR\crit{S}(M)$ is also a stack on the smooth manifold $M$.
%This means that, for every open cover $\{V_i\}_{i\in I}$ of the smooth manifold $M$, we have the homotopy limit
%\begin{equation}
%    \bfR\crit{S}(M) \;\simeq\; \bfR\!\lim_{} \left(   \begin{tikzcd}[row sep=scriptsize, column sep=4ex] \prod_{i} \bfR\crit{S}(V_i) & \arrow[l, yshift=0.7ex] \arrow[l, yshift=-0.7ex] \prod_{i,j} \bfR\crit{S}(V_i\cap V_j) & \,\cdots\, \arrow[l, yshift=1.4ex] \arrow[l] \arrow[l, yshift=-1.4ex]  
%    \end{tikzcd}   \right)
%\end{equation}
%in the $(\infty,1)$-category $\mathbf{dFSmoothSet}$ of formal derived smooth sets. In this sense, the derived critical locus $\bfR\crit{S}(-)$ can be seen as a "stack of derived stacks".
%More precisely, we can see $\bfR\crit{S}(-)$ as a stack on the product site $\mathbf{dFMfd}\times\mathsf{Op}(M)$, where $\mathbf{dFMfd}$ is the $(\infty,1)$-site of formal derived smooth manifolds and $\mathsf{Op}(M)$ is the $1$-site of open subsets of the manifold $M$.
%\end{remark}

\begin{remark}[Explicit expression for the $0$-simplices of the derived critical locus]\label{rem:zero_simplices}
Given a formal derived smooth manifold $U$, let us denote by $\bfR\Hom\big(U,\,\bfGamma(M,T_\mathrm{ver}^\vee E)\big)_{\!\itPhi}$ the fibre of the bundle of simplicial sets  $\bfR\Hom\big(U,\,\bfGamma(M,T_\mathrm{ver}^\vee E)\big)\longtwoheadrightarrow\bfR\Hom\big(U,\,\bfGamma(M, E)\big)$ at the point of the base  $\itPhi:U\rightarrow\bfGamma(M,E)$. 
The set of $0$-simplices of the $\infty$-groupoid $\bfR\Hom\big(U,\,\bfR\crit{S}(M)\big)$ of sections of the derived critical locus $\bfR\crit{S}(M)$ on a formal derived smooth manifold $U$ is
\begin{equation*}
\hspace{-0.7cm}\begin{aligned}
        \bfR\Hom\big(U,\,\bfR\crit{S}(M)\big)_0 \,=\, \Bigg\{  \begin{array}{l}\itPhi\in\bfR\Hom\big(U,\,\bfGamma(M,E)\big)_0\\[1.0ex] \itPhi^+\!\in \bfR\Hom\big(U,\,\bfGamma(M,T_\mathrm{ver}^\vee E)\big)_{\!\itPhi,1} \end{array}\;\Bigg| \!\begin{array}{ll}\delta S(\itPhi) &\!\!\!\!=\, \partial_0\itPhi^+ \\[1.0ex] \;\;\;\,0&\!\!\!\!=\, \partial_1\itPhi^+\end{array} \!\Bigg\},
\end{aligned}
\end{equation*}
where $\bfR\Hom\big(U,\,\bfGamma(M,E)\big)_0$ is the set of $0$-simplices of the simplicial set $\bfR\Hom\big(U,\,\bfGamma(M,E)\big)$ and $\bfR\Hom\big(U,\,\bfGamma(M,T_\mathrm{ver}^\vee E)\big)_{\!\itPhi,1}$ is the set of $1$-simplices of the simplicial set $\bfR\Hom\big(U,\,\bfGamma(M,T_\mathrm{ver}^\vee E)\big)_{\!\itPhi}$, which comes with face maps $\partial_{0,1}$.
\end{remark}

\begin{remark}[Global antifield]
Notice that a $0$-simplex of the simplicial set of sections $\bfR\Hom(U,\,\bfR\crit{S}(M))$ is a pair of the form
\begin{equation}
    (\itPhi,\itPhi^+)\,\in\,\bfR\Hom(U,\,\bfR\crit{S}(M)),
\end{equation}
where $\itPhi^+$ is a homotopy from the variational derivative $\delta S(\itPhi)$ of the action functional at the field configuration $\itPhi$ to zero, as written above.
Notice that $\itPhi$ is a scalar field and $\itPhi^+$ is the global-geometric version of what is known as its antifield in usual BV-theory.
However, it is clear that the two fields play a very different role in the global geometry of the scalar field theory: in fact, the antifield $\itPhi^+$ is not independent from the field $\itPhi$, but it lives in the fibre $\Gamma(M,T_\mathrm{ver}^\vee E)_\itPhi$.
\end{remark}

\begin{example}[The case of $E$ a vector bundle]
Let $E\twoheadrightarrow M$ be an vector bundle, so that the smooth set $\bfGamma(M,E)$ of its sections has a natural vector space structure. 
In this case, the restricted cotangent bundle reduces to $T^\vee_\mathrm{res}\bfGamma(M,E) \simeq \bfGamma(M,E\times_M E^\vee) \simeq \bfGamma(M,E)\oplus \bfGamma(M,E^\vee)$.
The set $\Gamma(M,E)$ of sections of a vector bundle comes also equipped with a $\Coo$-module structure on $\Coo(M)$, which allows the use of the $\Coo$-tensor product $\widehat{\otimes}$.
So, the set of $0$-simplices from remark \ref{rem:zero_simplices}, in case of $E\twoheadrightarrow M$ being a vector bundle reduces to the more familiar looking
\begin{equation*}
\hspace{-0.3cm}\begin{aligned}
        \bfR\Hom\big(U,\,\bfR\crit{S}(M)\big)_0 \,=\, \Bigg\{  \begin{array}{l}\itPhi\in\Gamma(M,E)\,\widehat{\otimes}\,\mathcal{O}(U)_0 \\[1.0ex]  \itPhi^+\!\in\Gamma(M,E^\vee)\,\widehat{\otimes}\,\mathcal{O}(U)_1\end{array}\;\Bigg| \begin{array}{ll}\delta S(\itPhi) &\!\!\!\!=\, \partial_0\itPhi^+ \\[1.0ex] \;\;\;\,0&\!\!\!\!=\, \partial_1\itPhi^+\end{array} \!\Bigg\},
\end{aligned}
\end{equation*}
where $\mathcal{O}(U)_0$ and $\mathcal{O}(U)_1$ are respectively the $\Coo$-algebras of $0$- and $1$-simplices of the simplicial $\Coo$-algebra $\mathcal{O}(U)$ and $\partial_{0,1}$ are the corresponding face maps.
\end{example}

Now, the pointed formal moduli problems of the form considered in subsection \ref{sec:review_BV} to study BV-theory can, in principle, be obtained by formal completion $\bfR\crit{S}(M)^\wedge_{\itPhi_0}$ at some fixed solution of the equations of motion $\itPhi_0\in \bfR\crit{S}$ as explained in construction \ref{con:FMP_as_FC}. 
Such an operation amounts to the construction of the pointed formal moduli problem $\bfR\crit{S}(M)^\wedge_{\itPhi_0}$ which infinitesimally approximates the formal derived smooth stack $\bfR\crit{S}(M)$ at $\itPhi_0\in \bfR\crit{S}$.
Let us now see this more in detail. 

\begin{remark}[Infinitesimal disk as formal moduli problem of Klein-Gordon theory]
As an example, let us consider Klein-Gordon theory, so let $S:[M,\bbR]_c\rightarrow \bbR$ be a  Klein-Gordon action of the form 
\begin{equation}
    S(\phi)\; =\; \int_M \!\Big( \phi\square\phi - V(\phi) \Big)\mathrm{vol}_M,
\end{equation}
where $V(\phi)$ is a function such that $V(0)=0$.
According to the machinery above, we can construct the derived critical locus $\bfR\crit{S}(M)$, which will be a formal derived smooth set.
The fact that the $0$-section is the trivial solution of the equations of motion, assures that there is a point $0:\ast\rightarrow\bfR\crit{S}(M)$, so we can consider the formal disk of the derived critical locus at such a point, according to definition \ref{def:formal_disk}.
It is possible to see that one has an equivalence
\begin{equation}
    \bbD_{\bfR\crit{S}(M),0} \;\simeq\; \B\mathfrak{L}(M),
\end{equation}
where $\mathfrak{L}(M)$ is the local $L_\infty$-algebra which has the underlying graded vector space given simply by $\mathfrak{L}(M) = \Coo(M)[-1] \oplus \Coo(M)[-2]$ and bracket structure given by
\begin{equation}
\begin{aligned}
    \ell_1(\phi) \;&=\; \square\phi - \frac{\partial V(\phi)}{\partial \phi }\bigg|_0\phi, \\
    \ell_k(\phi_1,\dots,\phi_k) \;&=\; -\frac{\partial^k V(\phi)}{\partial \phi^k }\bigg|_0\phi_1\cdots\phi_k \quad\text{for $k>1$}
\end{aligned}
\end{equation}
for any $\phi_i\in\Coo({M})$.
This is precisely the $L_\infty$-algebra which encodes the usual perturbative BV-theory of a Klein-Gordon scalar field.
We can formally complete our formal derived smooth stack $\bfR\crit{S}(M)$ at the trivial solution to obtain the pointed formal moduli problem
\begin{equation}
    \bfR\crit{S}(M)^\wedge_0 \;\,\simeq\,\; \mathit{\Gamma}^\mathrm{rel}\bbD_{\bfR\crit{S}(M),0} \;\,\simeq\;\, \mathbf{MC}\big(\mathfrak{L}(M)\big),
\end{equation}
where $\mathit{\Gamma}^\mathrm{rel}$ is the functor we introduced in section \ref{subsec:FMP}.
For a suitable choice of potential $V(\phi)$, this is nothing but the pointed formal moduli problem of Klein-Gordon theory appearing in \cite{FactI, FactII}.
Thus, this shows that the formal derived smooth stack $\bfR\crit{S}(M)$ provides a global-geometric version of the BV-theory of a Klein-Gordon scalar field. The usual perturbative formulation is given by the formal disk $\bbD_{\bfR\crit{S}(M),0}\simeq \bfR\crit{S}(M)\times_{\Im(\bfR\crit{S}(M))}\{0\}$ at the trivial solution, whose construction is made possible by the derived differential structure.
\end{remark}

Now, the usual perturbative BV-theory is most commonly dually stated in terms of dg-algebras of observables, also known as BV-complexes in physics.
To make contact with this perspective, we will now investigate what is the global-geometric version of the BV-complex of a scalar field.

\begin{remark}[Global BV-complex]
In what follows we will be deploying the compact notation $\mathbb{O}(X)\coloneqq\bfR\Gamma(X,\mathbb{O}_X)$ for the complex of global sections of the structure sheaf $\mathbb{O}_X\in\QCoh{X}$ of a formal derived smooth stack, as defined in subsection \ref{sec:Qcoh}.
As we have already noticed, the dual vector bundle of the vector bundle $\bfGamma(M,T_\mathrm{ver}^\vee E)\twoheadrightarrow\bfGamma(M,E)$ is precisely the tangent bundle $T\bfGamma(M,E)\simeq \bfGamma(M,T_\mathrm{ver} E)$ of the smooth set of sections.
By applying the machinery of derived zero loci, it is possible to see that the complex of global sections of $\bfR\mathrm{Crit}(S)(M)$ is given by 
\begin{equation*}
\begin{aligned}
    \mathbb{O}\big(\bfR\mathrm{Crit}(S)(M)\big) \;\simeq\; \Big( \cdots\xrightarrow{\;\,Q\;\,} \wedge^2\mathfrak{X}\big(\bfGamma(M,E)\big) \xrightarrow{\;\,Q\;\,} \mathfrak{X}\big(\bfGamma(M,E)\big) \xrightarrow{\;\,Q\;\,} \mathcal{O}\big(\bfGamma(M,E)\big)\Big)
\end{aligned}
\end{equation*}
where $\mathfrak{X}\big(\bfGamma(M,E)\big)$ is the set of vector fields on the ordinary smooth set $\bfGamma(M,E)$ and the differential $Q$ is given by the contraction $\iota_{(-)}\delta S$ of poly-vectors with the variational derivative $\delta S$ constructed above.
This is the picture that most directly generalises the BV-complex appearing in perturbative BV-theory.
Moreover, it generalises the functional approach to quantum mechanics of \cite{Chiaffrino:2021pob}.
To see that the complex $\mathbb{O}\big(\bfR\mathrm{Crit}(S)(M)\big)$ of global sections of the structure sheaf reduces to the usual BV-complex, it is enough to notice that, in the case of the formal disk $\bbD_{\bfR\crit{S}(M),0}\simeq \B \mathfrak{L}(M)$, we obtain the complex\footnote{It is a standard fact (see for example \cite{safronov2020shifted}) that the complex of global sections on a formal group stack of the form $\B\mathfrak{g}$, with $\mathfrak{g}$ an $L_\infty$-algebra, reduces to the Chevalley-Eilenberg algebra $\CE(\mathfrak{g})$ of $\mathfrak{g}$.}
\begin{equation*}
\begin{aligned}
    \mathbb{O}(\bbD_{\bfR\crit{S}(M),0}) \;\cong\; \CE\big(\mathfrak{L}(M)\big),
\end{aligned}
\end{equation*}
where $\mathrm{CE}\big(\mathfrak{L}(M)\big)$ is the Chevalley-Eilenberg algebra of the $L_\infty$-algebra $\mathfrak{L}(M)$ found above.
This tells us that the complex of sections $\mathbb{O}\big(\bfR\mathrm{Crit}(S)(M)\big)$ of the structure sheaf of the derived critical locus is a globally-defined generalisation of the usual BV-complex, which is recovered infinitesimally. Let us stress that the field bundle  $E\twoheadrightarrow M$ is a general fibre bundle of smooth manifolds and it does not have to be a vector bundle. 
\end{remark}

%%%%%%%%%%%%%%%%%%%%%%%%%%%%%%%%%%%%%%%%%%%%%%%%%%%%%%%%%%%%%%%%%%%%%%%%%%%%%%%%%%%%%%%%%%%%%
\subsection{Global BRST-BV formalism}\label{subsec:YMT}

In this subsection we will construct a global-geometric version of the BRST-BV formalism for Yang-Mills theory.
First, we will illustrate the smooth stack structure of the space of principal $G$-bundles with connection on a given smooth manifold, which is the configuration space of Yang-Mills theory. 
Second, we will see how the derived critical locus of the Yang-Mills action functional on such a smooth stack can be concretely constructed as formal derived smooth stack.
Finally, we will show that such a construction provides a global generalisation of usual the usual BV-formalism for Yang-Mills theory.

%%%%%%%%%%%%%%%%%%%%%%%%%%%%%%%%%%%%%%%%%%%%%%%%%%%%%%%%%%%%%%%%%%%%%%%%%%%%%%%%%%%%%%%%%%%%%
\subsubsection{Global BRST formalism}

Let us now temporarily take a step back and work in the $(\infty,1)$-topos $\mathbf{SmoothStack}$ of smooth stacks, i.e. stacks on the ordinary site of smooth manifolds. 
Our objective in this subsection is the construction of the smooth stack $\Bun^\nabla_G(M)$ of principal $G$-bundles on $M$ with connection. We will see such a stack as the global-geometric configuration space of a gauge theory on spacetime $M$ with gauge group $G$. This is because a field configuration of a gauge field is precisely the datum of a principal $G$-bundle on $M$ with a connection.

\begin{construction}[$\infty$-groupoid of principal $G$-bundles]
For a given ordinary Lie group $G$, the smooth stack $\B G = [\ast/G]$ is the moduli stack of principal $G$-bundles.
For a given manifold $M$, the $0$-simplices of the $\infty$-groupoid $\mathrm{Hom}(M,\mathbf{B}G)$ are all the non-abelian \v{C}ech $G$-cocycles $\{g_{\alpha\beta}\in\Coo(V_\alpha\cap V_\beta,G)\,|\,g_{\alpha\beta}\cdot g_{\beta\gamma}=g_{\alpha\gamma}\}$ on $M$ and the $1$-simplices are all the coboundaries $\{g_{\alpha\beta}\mapsto c_{\alpha} g_{\alpha\beta}c_\beta^{-1}\}$ between cocycles. Schematically, we have:
\begin{equation}\label{modulirem}
    \mathrm{Hom}(M,\mathbf{B}G)\,\simeq\,\left\{\begin{tikzcd}[row sep=scriptsize, column sep=12ex]
    M \arrow[r, bend left=50, ""{name=U, below}, "g_{\alpha\beta}"]
    \arrow[r, bend right=50, "c_{\alpha} g_{\alpha\beta}c_{\beta}^{-1}"', ""{name=D}]
    & \mathbf{B}G
    \arrow[Rightarrow, from=U, to=D, "c_{\alpha}"]
\end{tikzcd} \right\}.
\end{equation}
A principal $G$-bundle $P$ on an ordinary smooth manifold $M\in\Mfd$ is defined by its transition functions $g_{\alpha\beta}$, which are nothing but a \v{C}ech $G$-cocycle on $M$.
Thus, geometrically, the $0$-simplices are all the principal $G$-bundles over $M$, the $1$-simplices are all the isomorphisms (i.e. gauge transformations) between them and the higher simplices are given just by the composition of those. Thus, we can see a principal $G$-bundle as a point in the $\infty$-groupoid $\mathrm{Hom}(M,\mathbf{B}G)$.
Let us call 
\begin{equation}\nonumber
    \mathrm{Bun}_G(M ) \coloneqq \mathrm{Hom}(M,\B G)
\end{equation} the $\infty$-groupoid of principal $G$-bundles on a smooth manifold $M$.
\end{construction}

\begin{remark}[On a \v{C}ech cover]
More concretely, given a good open cover $\coprod_{\alpha\in I\!}V_\alpha\twoheadrightarrow M$ of our manifold, the simplicial set $\mathrm{Bun}_G(M)$ can be expressed as the homotopy limit
\begin{equation}
    \mathrm{Bun}_G(M) \;\simeq\; \bfR\!\lim_{} \Big(   \begin{tikzcd}[row sep=scriptsize, column sep=4ex] \displaystyle\prod_{\alpha} [\,\ast\,/\Coo(V_\alpha, G)] & \arrow[l, yshift=0.7ex, leftarrow] \arrow[l, yshift=-0.7ex, leftarrow] \displaystyle\prod_{\alpha,\beta} [\,\ast\,/\Coo(V_\alpha\cap V_\beta, G)] & \,\cdots\, \arrow[l, yshift=1.4ex, leftarrow] \arrow[l, leftarrow] \arrow[l, yshift=-1.4ex, leftarrow]  
    \end{tikzcd}   \Big),\vspace{-0.3cm}
\end{equation}
which glues explicitly the \v{C}ech local data of the $G$-bundles.
\end{remark}

\begin{remark}[Non-abelian cohomology]
To recover the more familiar topological picture one must look at the set of connected components of the $\infty$-groupoid of principal $G$-bundles, i.e.
\begin{equation}\label{eq:pathconnected}
    \mathrm{H}^1(M,G) \,=\, \pi_0\mathrm{Hom}(M,\mathbf{B}G).
\end{equation}
In other words, a morphism $M\rightarrow\mathbf{B}G$ in the homotopy category $\mathrm{Ho}(\mathbf{SmoothStack})$ of smooth stacks is equivalently a class in the cohomology $\mathrm{H}^1(M,G)$.
For example, for $G=U(1)$, we have by the isomorphism $\mathrm{H}^1(M,U(1)) \cong \mathrm{H}^2(M,\mathbb{Z})$ the first Chern class of circle bundles.    
\end{remark}

According to the general construction of $G$-bundles by \cite{Principal1, Principal2} in the context of a general $(\infty,1)$-topos, to any cocycle $M \rightarrow \mathbf{B}G$ is canonically associated a principal $G$-bundle $P\twoheadrightarrow M$ given by the pullback square
\begin{equation}\label{eq:firstbundle}
    \begin{tikzcd}[row sep=7.5ex, column sep=7.2ex]
    P \arrow[d, "\mathrm{hofib}(g)"', two heads]\arrow[r] & \ast \arrow[d] \\
    M \arrow[r, "g"] & \mathbf{B}G,
\end{tikzcd}
\end{equation}
where the homotopy fibre $\pi_M=\mathrm{hofib}(g)$ is the projection of the total space of the principal bundle to the base manifold.

However, as we have said, $\mathrm{Bun}_G(M)$ is just a bare $\infty$-groupoid (i.e. a Kan-fibrant simplicial set), lacking any smooth structure. What we want is to upgrade this object to a smooth stack.

\begin{definition}[Smooth stack of principal $G$-bundles]
The \textit{smooth stack of principal $G$-bundles} on a given smooth manifold $M$ is the mapping smooth stack
\begin{equation}
   \Bun_G(M) \;\coloneqq\; [M,\,\B G].
\end{equation}
\end{definition}

Notice that the underlying $\infty$-groupoid of this smooth stack, which we can extract by feeding it the point as $\Bun_G(M):\ast\mapsto \mathrm{Bun}_G(M )$, is precisely the one of principal $G$-bundles on $M$.

Now, we want to introduce the moduli stack $\mathbf{B}G_{\mathrm{conn}}$ of principal $G$-bundles with connection, which refines the moduli stack $\mathbf{B}G$ of principal bundles. We will have the following diagram:
\begin{equation}
    \begin{tikzcd}[row sep={11ex,between origins}, column sep={6ex}]
    & \mathbf{B}G_{\mathrm{conn}} \arrow[d, two heads, "\mathsf{F}"] \\
    M \arrow[r, "P"]\arrow[ur, "{(P,\nabla_{\!A})}"] & \mathbf{B}G.
    \end{tikzcd}
\end{equation}
Just as a cocycle $P:M\rightarrow \mathbf{B}G$ encodes the global geometric data of a principal bundle, a cocycle $(P,\nabla_{\!A}):M\rightarrow \mathbf{B}G_{\mathrm{conn}}$ will encode both the global geometric data of a principal bundle and the global differential data of a principal connection. 

\begin{construction}[$\infty$-groupoid of $G$-bundles with connection]
We can avoid the many technical subtleties and explicitly construct the stack $\mathbf{B}G_\mathrm{conn}\in\mathbf{SmoothStack}$ so that a cocycle $(A_{\alpha},g_{\alpha\beta})\in\mathrm{Hom}(M,\mathbf{B}G_{\mathrm{conn}})$ encodes precisely the global differential data of a principal $G$-bundle with connection on $M$ as follows (see, for instance, \cite{Benini_2018}): $A_{\alpha}\in\Omega^1(V_\alpha,\mathfrak{g})$ is a local $1$-form, which is glued on two-fold overlaps $V_\alpha\times_M V_\beta$ by
\begin{equation}
    A_{\beta} \;=\; g_{\beta\alpha}^{-1}(A_{\alpha}+\di)g_{\beta\alpha},
\end{equation}
where $g_{\alpha\beta}:M\rightarrow \mathbf{B}G$ is the \v{C}ech cocycle of a principal $G$-bundle, which is itself glued by
\begin{equation}
    g_{\alpha\beta}\cdot g_{\beta\gamma} \,=\, g_{\alpha\gamma}.
\end{equation}
on three-fold overlaps $V_\alpha\times_MV_\beta\times_M V_\gamma$.
Moreover, a coboundary $(A_{\alpha},g_{\alpha\beta}) \mapsto (A_{\alpha}',g'_{\alpha\beta})$ is given by the datum of a local $G$-valued scalar $c_{\alpha}\in\Coo(V_\alpha,G)$ such that
\begin{equation}
    \begin{aligned}
        g'_{\alpha\beta} \;&=\; c_{\beta}^{-1}g_{\alpha\beta}c_{\alpha},\\
         A_{\alpha}' \;&=\; c_{\alpha}^{-1}(A_{\alpha}+\di)c_{\alpha}.
    \end{aligned}
\end{equation}
Given a smooth manifold $M$ and a Lie group $G$, let us call
\begin{equation}
    \mathrm{Bun}_G^\nabla(M) \;:=\; \mathrm{Hom}(M,\, \mathbf{B}G_{\mathrm{conn}})
\end{equation}
the $\infty$-groupoid of $G$-bundles with connection on $M$.
\end{construction}

\begin{remark}[Underlying principal $G$-bundle]
In general, there is a forgetful morphism
\begin{equation}
    \mathbf{B}G_{\mathrm{conn}}\; \xrightarrow{\quad\mathsf{F}\quad}\; \mathbf{B}G,
\end{equation}
which forgets the connection of the $G$-bundles.
Thus, it is important that a cocycle $M\rightarrow \mathbf{B}G_{\mathrm{conn}}$ contains not only local connection data, but also the underlying bundle structure $M\rightarrow \mathbf{B}G$.
In our case, cocycles are mapped as
\begin{equation}
    \begin{aligned}
        \mathrm{Hom}(M,\mathsf{F})\,:\; \mathrm{Hom}(M,\mathbf{B}G_{\mathrm{conn}}) \;&\longrightarrow\; \mathrm{Hom}(M,\mathbf{B}G),   \\
       (g_{\alpha\beta},A_{\alpha}) \;&\longmapsto\; (g_{\alpha\beta})
    \end{aligned}
\end{equation}
so that the functor forgets the connection data, but retains the global geometric data.
\end{remark}

Now that we have the moduli stack $\mathbf{B}G_{\mathrm{conn}}$ of $G$-bundles with connection, we can formulate the following definition.

\begin{remark}[On a \v{C}ech cover]
More concretely, given a good open cover $\coprod_{\alpha\in I\!}V_\alpha\twoheadrightarrow M$ of our manifold, the simplicial set $\mathrm{Bun}^\nabla_G(M)$ can be expressed as the homotopy limit
\begin{equation*}
    \mathrm{Bun}^\nabla_G(M) \simeq \bfR\!\lim \! \bigg(  \! \!\begin{tikzcd}[row sep=scriptsize, column sep=3.3ex] \displaystyle\prod_{\alpha} [\Omega^1(V_\alpha,\mathfrak{g})/\Coo(V_\alpha, G)] & \arrow[l, yshift=0.7ex, leftarrow] \arrow[l, yshift=-0.7ex, leftarrow] \displaystyle\prod_{\alpha,\beta} [\Omega^1(V_\alpha\cap V_\beta,\mathfrak{g})/\Coo(V_\alpha\cap V_\beta, G)] & \cdots \arrow[l, yshift=1.4ex, leftarrow] \arrow[l, leftarrow] \arrow[l, yshift=-1.4ex, leftarrow]  
    \end{tikzcd} \! \! \bigg) \vspace{-0.3cm}
\end{equation*}
which glues explicitly the \v{C}ech local data of the $G$-bundles.
\end{remark}

\begin{remark}[Non-abelian differential cohomology]
In the homotopy category of smooth stacks $\mathrm{Ho}(\mathbf{SmoothStack})$, a morphism $M\rightarrow\mathbf{B}G_{\mathrm{conn}}$ is an element of
\begin{equation}
    \widehat{\mathrm{H}}^1(M,G) \;\coloneqq\; \pi_0\mathrm{Bun}_G^\nabla(M),
\end{equation}
which can be interpreted as (non-abelian) differential cohomology.
\end{remark}

Let us now fix once and for all a Lie group $G$, which we will think of as our gauge group, and an ordinary smooth manifold $M\in\Mfd$, which is going to play the role of spacetime.
What we want to do now is to update the bare $\infty$-groupoid $\mathrm{Bun}_G^\nabla(M)$ of principal $G$-bundles on $M$ with connection to some smooth stack $\Bun^\nabla_G(M)$, which we can see as the configuration space of a gauge theory on spacetime $M$ with gauge group $G$.

\begin{remark}[Technical subtleties]
For technical reasons \cite{Benini_2018}, the proper choice of definition for the smooth stack $\Bun^\nabla_G(M)$ of principal $G$-bundles on $M$ with connection cannot be, as one may naively think by comparison with the connection-less case, just the mapping smooth stack $[M, \B G_\mathrm{conn}]$. Such a choice would fail to have the desired properties.
As argued by \cite{Benini_2018}, the smooth stack $\Bun^\nabla_G(M)$ must be a certain concretification of the mapping stack $[M, \B G_\mathrm{conn}]$, which is constructed in the reference.
\end{remark}

\begin{construction}[Smooth stack of principal $G$-bundles with connection]
We construct the \textit{smooth stack }$\Bun_G^\nabla(M)$\textit{ of principal $G$-bundles with connection} as follows. First, let us fix a good open cover $\bigsqcup_\alpha \!V_\alpha \twoheadrightarrow M$ for the base manifold $M$. Then, for any smooth manifold $U\in\Mfd$ diffeomorphic to a Cartesian space $U\simeq \bbR^n$ we construct the following simplicial set of sections:
\begin{equation*}
    \Hom\big(U,\,\Bun_G^\nabla(M)\big) \,\simeq\, \mathrm{cosk}_{2\!}\left(\! \begin{tikzcd}[row sep={22.5ex,between origins}, column sep={10.5ex}]
    Z_2 \, \arrow[rr, yshift=3.4ex , "{\big(c_\alpha,\,\begin{smallmatrix}g_{\alpha\beta},A_\alpha\\g_{\alpha\beta}',A_\alpha'\end{smallmatrix}\big)}"] \arrow[rr, "{\big(c_\alpha',\,\begin{smallmatrix}g_{\alpha\beta}',A_\alpha'\\g_{\alpha\beta}'',A_\alpha''\end{smallmatrix}\big)}" description] \arrow[rr, yshift=-3.4ex, "{\big(c_\alpha'\cdot c_\alpha,\,\begin{smallmatrix}g_{\alpha\beta},A_\alpha\\g_{\alpha\beta}'',A_\alpha''\end{smallmatrix}\big)}"'] && \,Z_1\, \arrow[r, yshift=1.7ex , "{(g_{\alpha\beta},A_\alpha)}"] \arrow[r, yshift=-1.7ex, "{(g_{\alpha\beta}',A_\alpha')}"'] & \,Z_0
    \end{tikzcd}\!\right),
\end{equation*}
where the sets of $0$-, $1$- and $2$-simplices are respectively given by    
\begin{equation*}
\begin{aligned}
        \hspace{-1.8cm}Z_0 \,&=\, \left\{ \begin{array}{ll}
        g_{\alpha\beta} &\!\!\!\!\in\Coo(V_\alpha\cap V_\beta \times U ,G) \\[0.5ex]
         A_\alpha &\!\!\!\!\in\Omega^1_{\mathrm{ver}}(V_\alpha \times U,\mathfrak{g}) \end{array}  \;\left|\; \begin{array}{l} 
         g_{\alpha\beta}\cdot g_{\beta\gamma}\cdot g_{\gamma\alpha} = 1 \\[0.5ex]
        A_{\alpha}=g_{\beta\alpha}^{-1}(A_{\beta}+\di)g_{\beta\alpha} \end{array} \right.\right\},
        \end{aligned}
\end{equation*}
\begin{equation*}
\begin{aligned}
       \hspace{-0.8cm}Z_1 \,&=\, \left\{ \begin{array}{ll}
         c_\alpha &\!\!\!\!\in\Coo(V_\alpha \times U,G)  \\[1.5ex]
         g_{\alpha\beta},g'_{\alpha\beta}  &\!\!\!\!\in\Coo(V_\alpha\cap V_\beta \times U ,G) \\[0.5ex]
         A_\alpha, A_\alpha' &\!\!\!\!\in\Omega^1_{\mathrm{ver}}(V_\alpha \times U,\mathfrak{g})\end{array}  \;\left|\; \begin{array}{ll} 
         g_{\alpha\beta}&\!\!\!\!\cdot\, g_{\beta\gamma}\cdot g_{\gamma\alpha} = 1 \\[0.5ex]
        A_{\alpha}&\!\!\!\!=\,g_{\beta\alpha}^{-1}(A_{\beta}+\di)g_{\beta\alpha}  \\[2ex]
        g_{\alpha\beta}'&\!\!\!\!\cdot\, g_{\beta\gamma}'\cdot g_{\gamma\alpha}' = 1 \\[0.5ex]
        A_{\alpha}'&\!\!\!\!=\,g_{\beta\alpha}^{\prime -1}(A_{\beta}'+\di)g_{\beta\alpha}'  \\[2ex]
         g_{\alpha\beta}' &\!\!\!\!=\, c_\beta^{-1}g_{\alpha\beta}c_\alpha \\[0.5ex]
        A_{\alpha}' &\!\!\!\!=\, c_\alpha^{-1} (A_\alpha+\di) c_\alpha \end{array} \right.\right\}, \end{aligned}
\end{equation*}
\begin{equation*}
\begin{aligned}
        Z_2 \,&=\, \left\{ \begin{array}{ll}
         c_\alpha,c_\alpha' &\!\!\!\!\in\Coo(V_\alpha\times U,G)   \\[1.5ex]
         g_{\alpha\beta},g'_{\alpha\beta},g''_{\alpha\beta}  &\!\!\!\!\in\Coo(V_\alpha\cap V_\beta \times U,G) \\[0.5ex]
         A_\alpha, A_\alpha', A_\alpha'' &\!\!\!\!\in\Omega^1_\mathrm{ver}(V_\alpha\times U,\mathfrak{g})\end{array}  \;\left|\; \begin{array}{ll} 
         g_{\alpha\beta}&\!\!\!\!\cdot\, g_{\beta\gamma}\cdot g_{\gamma\alpha} = 1 \\[0.5ex]
        A_{\alpha}&\!\!\!\!=\,g_{\beta\alpha}^{-1}(A_{\beta}+\di)g_{\beta\alpha}  \\[2ex]
        g_{\alpha\beta}'&\!\!\!\!\cdot\, g_{\beta\gamma}'\cdot g_{\gamma\alpha}' = 1 \\[0.5ex]
        A_{\alpha}'&\!\!\!\!=\,g_{\beta\alpha}^{\prime -1}(A_{\beta}'+\di)g_{\beta\alpha}'  \\[2ex]
        g_{\alpha\beta}''&\!\!\!\!\cdot\, g_{\beta\gamma}''\cdot g_{\gamma\alpha}'' = 1 \\[0.5ex]
        A_{\alpha}''&\!\!\!\!=\,g_{\beta\alpha}^{\prime -1}(A_{\beta}''+\di)g_{\beta\alpha}''  \\[2ex]
         g_{\alpha\beta}' &\!\!\!\!=\, c_\beta^{-1}g_{\alpha\beta}c_\alpha \\[0.5ex]
        A_{\alpha}' &\!\!\!\!=\, c_\alpha^{-1} (A_\alpha+\di) c_\alpha  \\[0.5ex]
             g_{\alpha\beta}'' &\!\!\!\!=\, c_\beta^{\prime -1}g_{\alpha\beta}'c_\alpha' \\[0.5ex]
        A_{\alpha}'' &\!\!\!\!=\, c_\alpha^{\prime -1} (A_\alpha'+\di) c_\alpha' \end{array} \right.\right\},
\end{aligned}
\end{equation*}
where $\Omega^p_\mathrm{ver}(V_\alpha\times U,\mathfrak{g}_p)$ is the set of vertical differential $p$forms on the fibration $V_\alpha\times U\twoheadrightarrow U$.
Finally, for a general smooth manifold $U\in\Mfd$ we consider a good open cover $\bigsqcup_{i\in I\!}U_i\rightarrow U$ for it, so that all the overlaps $U_{i_1,\dots,i_n}$ are diffeomorphic to Cartesian spaces.
Thus, we define the simplicial set of sections at $U$ to be the homotopy limit
\begin{equation}
    \Hom\big(U,\,\Bun_G^\nabla(M)\big) \,\,\simeq\,\; \bfR\!\lim_{\!\!\!\!\!\!\![n]\in\Delta}\!\!\!\prod_{\;\;i_1,\dots,i_n\in I}\!\!\!\!\Hom\big(U_{i_1,\dots,i_n},\,\Bun_G^\nabla(M)\big).
\end{equation}
\end{construction}

\begin{remark}[Relation with bare groupoid of principal $G$-bundles with connection]
Notice that the underlying $\infty$-groupoid of the smooth stack defined above is precisely the $\infty$-groupoid of principal $G$-bundles with connection on the manifold $M$, i.e.
\begin{equation}
    \Bun^\nabla_G(M):\;\ast \,\longmapsto\, \mathrm{Bun}^\nabla_G(M).
\end{equation}
In this precise sense, $\Bun^\nabla_G(M)$ can be understood as the smooth stack version of the bare $\infty$-groupoid $\mathrm{Bun}^\nabla_G(M)$.
\end{remark}

Now, having introduced the smooth stack $\Bun_G(M)$ and its refinement $\Bun_G^\nabla(M)$, we will focus on their infinitesimal properties in the context of differential geometry.
To do that, we must embed both these smooth stacks into formal smooth stacks by exploiting the canonical embedding $\mathbf{SmoothStack}\longhookrightarrow \mathbf{FSmoothStack}$ from section \ref{sec:Coo}.
For simplicity, we will keep using the same symbols $\Bun_G(M)$ and $\Bun_G^\nabla(M)$ to indicate the two formal smooth stacks obtained by such an embedding.

\begin{proposition}[Formal disk of $\Bun_G(M)$]
The formal disk $\bbD_{\Bun_G(M),P}$ of the formal smooth stack $\Bun_G(M)$ of $G$-bundles on a fixed smooth manifold $M$, at a given $G$-bundle $P\twoheadrightarrow M$, is the formal smooth stack
\begin{equation}
     \bbD_{\Bun_G(M),P} \;\simeq\; \B\Omega^0(M,\mathfrak{g}_P),
\end{equation}
where $\mathfrak{g}_P \coloneqq P\times_G\mathfrak{g}$ is the adjoint bundle of $P\in\Bun_G(M)$ and $\Omega^0(M,\mathfrak{g}_P)$ is the local Lie algebra of $\mathfrak{g}_P$-valued $0$-forms on $M$.
\end{proposition}

\begin{proof}
Let $M\in\Mfd$ be a smooth manifold and $X$ any formal smooth stack. The formal disk of the mapping stack $[M,X]$ at the point $f:M\rightarrow X$ is defined by the pullback $\bbD_{[M,X],f}=\ast\times_{\Im[M,X]}[M,X]$.
Consider now the pullback $f^\ast T^\infty X \simeq M\times_{\Im(X)}X$ of the formal disk bundle of $X$ along the map $f$. Let $\bfGamma(M,E)$ denote the formal smooth stack of section of a bundle $E$ on $M$.
One can notice that we have an equivalence of formal smooth stacks $\bbD_{[M,X],f}\simeq\bfGamma(M,f^\ast T^\infty X)$.
Let us now consider our case of interest, $\Bun_G(M)=[M,\B G]$. Since the moduli stack of $G$-bundles is of the form $\B G \simeq [\,\ast\,/G]$, we have the formal disk bundle $T^\infty\B G \simeq [\,\ast\,/G\ltimes_{\mathrm{ad}}\mathfrak{g}]$.
Given $P:M\rightarrow\B G$, we have the pullback $P^\ast T^\infty\B G \simeq \B\mathfrak{g}_P$.
Therefore, we have the equivalence of formal smooth stacks
$\bbD_{\mathrm{Bun}_G(M),P} \simeq \bfGamma(M,P^\ast T^\infty \B G) \simeq \B\Omega^0(M,\mathfrak{g}_P)$.
\end{proof}

Recall that the infinitesimal automorphisms -- i.e. gauge transformations -- of a principal $G$-bundle $P\twoheadrightarrow M$ are indeed known to be given by sections $\Omega^0(M,\mathfrak{g}_P)$ of its adjoint bundle (see e.g. \cite{Alf19, Alfonsi:2021ymc} for a higher geometric point of view).

The next step will be to consider infinitesimal deformations of $\Bun_G^\nabla(M)$, which is the configuration space of a gauge theory with gauge group $G$ on spacetime $M$.
As we have seen in the derived case in definition \ref{def:formal_disk_bundle}, we can construct the formal disk bundle of the formal smooth stack $\Bun_G^\nabla(M)$ of $G$-bundles with connection on $M$ by the following pullback square
\begin{equation}
    \begin{tikzcd}[row sep={22.5ex,between origins}, column sep={25.0ex,between origins}]
    T^\infty\Bun_G^\nabla(M) \arrow[r]\arrow[d, ""'] & \Bun_G^\nabla(M) \arrow[d, "{\mathfrak{i}_{\Bun_G^\nabla(M)}}"] \\
    \Bun_G^\nabla(M) \arrow[r, "{\mathfrak{i}_{\Bun_G^{\nabla}(M)}}"] & \Im\big( \Bun_G^\nabla(M) \big) .
    \end{tikzcd}
\end{equation}
Recall that the fibre of the formal disk bundle of a formal smooth stack at a point is the formal disk at such a point.
Then, the fibre of the formal disk bundle $T^\infty\Bun_G^\nabla(M)$ at a fixed principal $G$-bundle with connection $(P,\nabla_{\!A}):\ast\rightarrow\Bun_G^\nabla(M)$ is given by the following formal smooth stack
\begin{equation}
    \bbD_{\Bun_G^\nabla(M), {(P,\nabla_{\!A})}} \;\simeq\; \B \big(\overrightarrow{\mathfrak{Brst}}(M)_{(P,\nabla_{\!A})}\big), 
\end{equation}
where $\overrightarrow{\mathfrak{Brst}}(M)_{(P,\nabla_{\!A})}$ is a local $L_\infty$-algebra whose underlying graded vector space is given by
\begin{equation*}
\begin{aligned}
    \overrightarrow{\mathfrak{Brst}}(M)_{(P,\nabla_{\!A})_{\!}}[1] \;\,=\,\; \,&\Big( \!\!\begin{tikzcd}[row sep={14.5ex,between origins}, column sep= 7ex]
    \Omega^0(M,\mathfrak{g}_P) \arrow[r, "\nabla_{\! A}"] & \Omega^1(M,\mathfrak{g}_P)
    \end{tikzcd} \!\! \Big).\\
    \!\!{\scriptstyle\text{deg}\,=} &\quad \;\;\begin{tikzcd}[row sep={14.5ex,between origins}, column sep= 6ex]
    \;{\scriptstyle -1} && \quad \;\;\;\;\;{\scriptstyle 0} 
    \end{tikzcd}  
\end{aligned} 
\end{equation*}
Notice that it depends on the point $(P,\nabla_{\!A})\in \Bun_G^\nabla(M)$.
Such an $L_\infty$-algebra controls the infinitesimal deformations $\nabla_{\!A}+\vec{A}$ of the fixed connection, together with infinitesimal gauge transformations for the deformed connection.
So, its $L_\infty$-bracket structure is given as follows:
\begin{equation}
    \begin{aligned}
        \ell_1(\vec{c}) \;&=\; \nabla_{\! A} \vec{c}, \\
        \ell_2(\vec{c}_1,\vec{c}_2) \;&=\; [\vec{c}_1,\vec{c}_2]_\mathfrak{g},\\
        \ell_2(\vec{c},\vec{A}) \;&=\; [\vec{c},\vec{A}]_\mathfrak{g},
    \end{aligned}
\end{equation}
for any $\vec{c}_k\in\Omega^0(M,\mathfrak{g}_P)$ and $\vec{A}\in\Omega^1(M,\mathfrak{g}_P)$ elements of the underlying graded vector space.

\begin{remark}[Formal disk bundle as $L_\infty$-algebroid]
Notice that, by construction, the formal disk $\bbD_{\Bun_G^\nabla(M), {(P,\nabla_{\!A})}}$ is indeed an infinitesimal object. In fact, we have that there is a natural equivalence $\Re\big(\bbD_{\Bun_G^\nabla(M), {(P,\nabla_{\!A})}}\big)\simeq \ast$ of its reduction to the point. More generally, we have a natural equivalence $\Re\big(T^\infty\Bun_G^\nabla(M)\big)\simeq \Bun_G^\nabla(M)$ of the reduction of the formal disk bundle of the smooth stack of $G$-bundles on $M$ to itself.
Let us stress the fact that the map $T^\infty\Bun_G^\nabla(M) \longrightarrow \Bun_G^\nabla(M)$ is a bundle of formal smooth stacks, whose base is not an ordinary manifold but the smooth stack $\Bun_G^\nabla(M)$ of principal $G$-bundles on $M$ with connection.
Moreover, as we have seen in subsection \ref{subsec:Loo_algebroids}, the formal smooth stack $T^\infty\Bun_G^\nabla(M)$ comes with a natural structure of smooth algebroid (i.e. of infinitesimal smooth groupoid) provided by the canonical effective epimorphism $\mathfrak{i}_{\Bun_G^\nabla(M)}:\Bun_G^\nabla(M)\longrightarrow \Im(\Bun_G^\nabla(M))$ to its de Rham space.
\end{remark}

\begin{remark}[Morphism forgetting the connection]
Recall from the beginning of this subsection that the formal disk in the formal smooth stack $\Bun_G(M)$ of principal $G$-bundles at a $P\in\Bun_G(M)$ is precisely given by the quotient stack
\begin{equation}
    \bbD_{\Bun_G(M),P} \;\simeq\; \B\Omega^0(M,\mathfrak{g}_P),
\end{equation}
where $\Omega^0(M,\mathfrak{g}_P)$ is the Lie algebra of $\mathfrak{g}_P$-valued $0$-forms.
Thus, we can notice that there exists a forgetful map of formal smooth stacks
\begin{equation}
    \bbD_{\Bun_G^\nabla(M), {(P,\nabla_{\!A})}} \;\xtwoheadrightarrow{\;\;\;\,\mathsf{F}\,\;\;\; } \bbD_{\Bun_G(M),P}
\end{equation}
which forgets the deformation of the connection data.
\end{remark}

\begin{remark}[Formal disk bundle in \v{C}ech data]
We can explicitly express the formal smooth stack $T^\infty\Bun_G^\nabla(M)$ in \v{C}ech data as follows. First, let us fix a good open cover $\bigsqcup_\alpha \!V_\alpha \twoheadrightarrow M$ for the base manifold $M$. Then, for any formal smooth manifold $U\in\mathsf{F}\Mfd$ equivalent to a formal Cartesian space $U\simeq \bbR^n\times \Spec W$ we can write by the $2$-coskeletal simplicial set of sections:
\begin{equation*}
    \Hom\big(U,\,T^\infty\Bun_G^\nabla(M)\big) \,\simeq\, \mathrm{cosk}_{2\!\!}\left(\! \begin{tikzcd}[row sep={22.5ex,between origins}, column sep={14.0ex}]
    Z_2 \, \arrow[rr, yshift=3.6ex , "{\big(c_\alpha,\vec{c}_\alpha,\,\begin{smallmatrix}g_{\alpha\beta},A_\alpha,\vec{A}_\alpha\\g_{\alpha\beta}',A_\alpha',\vec{A}_\alpha'\end{smallmatrix}\big)}"] \arrow[rr, "{\big(c_\alpha',\vec{c}_\alpha',\,\begin{smallmatrix}g_{\alpha\beta}',A_\alpha',\vec{A}_\alpha'\\g_{\alpha\beta}'',A_\alpha'',\vec{A}_\alpha''\end{smallmatrix}\big)}" description] \arrow[rr, yshift=-3.6ex, "{\big(c_\alpha'\cdot c_\alpha,\,\vec{c}_\alpha'+\vec{c}_\alpha,\,\begin{smallmatrix}g_{\alpha\beta},A_\alpha,\vec{A}_\alpha\\g_{\alpha\beta}'',A_\alpha'',\vec{A}''_\alpha\end{smallmatrix}\big)}"'] && \, Z_1 \, \arrow[r, yshift=1.8ex , "{(g_{\alpha\beta},A_\alpha,\vec{A}_\alpha)}"] \arrow[r, yshift=-1.8ex, "{(g_{\alpha\beta}',A_\alpha',\vec{A}'_\alpha)}"'] & \, Z_0
    \end{tikzcd}\!\right)\!,
\end{equation*}
where the sets of $0$-, $1$- and $2$-simplices are given respectively by
\begin{equation*}
\begin{aligned}
        \hspace{-1.75cm} Z_0 \,&=\, \left\{ \begin{array}{ll}
          g_{\alpha\beta} &\!\!\!\!\in\Coo(V_\alpha\cap V_\beta \times U ,G) \\[0.9ex]
         A_\alpha &\!\!\!\!\in\Omega^1_\mathrm{ver}(V_\alpha  \times U ,\mathfrak{g}) \\[1.0ex]
         \vec{A}_\alpha &\!\!\!\!\in\Omega^{1}_\mathrm{ver}(V_\alpha\times\bbR^n,\mathfrak{g})\otimes\mathfrak{m}_W
        \end{array}  \;\left|\; \begin{array}{l} 
        g_{\alpha\beta}\cdot g_{\beta\gamma}\cdot g_{\gamma\alpha} = 1 \\[0.9ex]
        A_{\alpha}=g_{\beta\alpha}^{-1}(A_{\beta}+\di)g_{\beta\alpha} \\[1.0ex]
        \vec{A}_\alpha = g_{\beta\alpha}^{-1}\vec{A}_\beta g_{\beta\alpha} \\[1.0ex] 
        \end{array} \right.\right\},
\end{aligned}
\end{equation*}
\begin{equation*}
\begin{aligned}
         \hspace{-0.75cm} Z_1 \,&=\, \left\{ \begin{array}{ll}
         c_\alpha &\!\!\!\!\in\Coo(V_\alpha  \times U ,G) \\[0.5ex]
         \vec{c}_\alpha &\!\!\!\!\in\Omega^{0}(V_\alpha\times\bbR^n,\mathfrak{g})\otimes\mathfrak{m}_W \\[1.5ex]
         g_{\alpha\beta},g'_{\alpha\beta}  &\!\!\!\!\in\Coo(V_\alpha\cap V_\beta  \times U ,G) \\[0.5ex]
         A_\alpha, A_\alpha' &\!\!\!\!\in\Omega^1_\mathrm{ver}(V_\alpha  \times U ,\mathfrak{g}) \\[0.5ex]
         \vec{A}_\alpha, \vec{A}_\alpha' &\!\!\!\!\in\Omega^{1}_\mathrm{ver}(V_\alpha\times\bbR^n,\mathfrak{g})\otimes\mathfrak{m}_W\end{array}  \;\left|\; \begin{array}{ll} 
         g_{\alpha\beta}&\!\!\!\!\!\!\cdot\, g_{\beta\gamma}\cdot g_{\gamma\alpha} = 1 \\[0.5ex]
        A_{\alpha}&\!\!\!\!\!\!=\,g_{\beta\alpha}^{-1}(A_{\beta}+\di)g_{\beta\alpha}  \\[0.5ex]
        \vec{A}_\alpha &\!\!\!\!\!\!=\, g_{\beta\alpha}^{-1}\vec{A}_\beta g_{\beta\alpha} \\[2.2ex]
        g_{\alpha\beta}'&\!\!\!\!\cdot\, g_{\beta\gamma}'\cdot g_{\gamma\alpha}' = 1 \\[0.5ex]
        A_{\alpha}'&\!\!\!\!=\,g_{\beta\alpha}^{\prime -1}(A_{\beta}'+\di)g_{\beta\alpha}'  \\[0.5ex]
        \vec{A}'_\alpha &\!\!\!\!\!\!=\, g_{\beta\alpha}^{\prime -1}\vec{A}'_\beta g_{\beta\alpha}' \\[2.2ex]
         g_{\alpha\beta}' &\!\!\!\!=\, c_\beta^{-1}g_{\alpha\beta}c_\alpha \\[0.5ex]
        A_{\alpha}' &\!\!\!\!=\, c_\alpha^{-1} (A_\alpha+\di) c_\alpha \\[0.5ex]
        \vec{A}_{\alpha}' &\!\!\!\!=\, \vec{A}_\alpha + \nabla_{\!A_\alpha\!}\vec{c}_\alpha \\[0.5ex]
        \vec{c}_{\alpha} &\!\!\!\!=\, g_{\beta\alpha}^{-1}\vec{c}_\beta g_{\beta\alpha}
        \end{array} \right.\right\},
\end{aligned}
\end{equation*}
\begin{equation*}
\begin{aligned}
        Z_2 \,&=\, \left\{ \begin{array}{ll}
         c_\alpha,c_\alpha' &\!\!\!\!\in\Coo(V_\alpha \times U ,G)  \\[0.5ex]
         \vec{c}_\alpha,\vec{c}_\alpha' &\!\!\!\!\in\Omega^{0}(V_\alpha\times\bbR^n,\mathfrak{g})\otimes\mathfrak{m}_W   \\[1.5ex]
         g_{\alpha\beta},g'_{\alpha\beta},g''_{\alpha\beta}  &\!\!\!\!\in\Coo(V_\alpha\cap V_\beta  \times U ,G) \\[0.5ex]
         A_\alpha, A_\alpha', A_\alpha'' &\!\!\!\!\in\Omega^1_{\mathrm{ver}}(V_\alpha \times U ,\mathfrak{g})\\[0.5ex]
         \vec{A}_\alpha, \vec{A}_\alpha', \vec{A}_\alpha'' &\!\!\!\!\in\Omega^{1}_\mathrm{ver}(V_\alpha\times\bbR^n,\mathfrak{g})\otimes\mathfrak{m}_W \end{array}  \;\left|\; \begin{array}{ll} 
         g_{\alpha\beta}&\!\!\!\!\cdot\, g_{\beta\gamma}\cdot g_{\gamma\alpha} = 1 \\[0.5ex]
        A_{\alpha}&\!\!\!\!=\,g_{\beta\alpha}^{-1}(A_{\beta}+\di)g_{\beta\alpha} \\[0.5ex]
        \vec{A}_\alpha &\!\!\!\!\!\!=\, g_{\beta\alpha}^{-1}\vec{A}_\beta g_{\beta\alpha} \\[2.2ex]
        g_{\alpha\beta}'&\!\!\!\!\cdot\, g_{\beta\gamma}'\cdot g_{\gamma\alpha}' = 1 \\[0.5ex]
        A_{\alpha}'&\!\!\!\!=\,g_{\beta\alpha}^{\prime -1}(A_{\beta}'+\di)g_{\beta\alpha}'  \\[0.5ex]
        \vec{A}'_\alpha &\!\!\!\!\!\!=\, g_{\beta\alpha}^{\prime -1}\vec{A}'_\beta g_{\beta\alpha}'\\[2.2ex]
        g_{\alpha\beta}''&\!\!\!\!\cdot\, g_{\beta\gamma}''\cdot g_{\gamma\alpha}'' = 1 \\[0.5ex]
        A_{\alpha}''&\!\!\!\!=\,g_{\beta\alpha}^{\prime -1}(A_{\beta}''+\di)g_{\beta\alpha}'' \\[0.5ex]
        \vec{A}''_\alpha &\!\!\!\!\!\!=\, g_{\beta\alpha}^{\prime\prime -1}\vec{A}''_\beta g_{\beta\alpha}'' \\[2.2ex]
         g_{\alpha\beta}' &\!\!\!\!=\, c_\beta^{-1}g_{\alpha\beta}c_\alpha \\[0.5ex]
        A_{\alpha}' &\!\!\!\!=\, c_\alpha^{-1} (A_\alpha+\di) c_\alpha  \\[0.5ex]
        \vec{A}_{\alpha}' &\!\!\!\!=\, \vec{A}_\alpha + \nabla_{\!A_\alpha\!}\vec{c}_\alpha      \\[0.5ex]
         \vec{c}_{\alpha} &\!\!\!\!=\, g_{\beta\alpha}^{-1}\vec{c}_\beta g_{\beta\alpha}    \\[0.5ex]
             g_{\alpha\beta}'' &\!\!\!\!=\, c_\beta^{\prime -1}g_{\alpha\beta}'c_\alpha' \\[0.5ex]
        A_{\alpha}'' &\!\!\!\!=\, c_\alpha^{\prime -1} (A_\alpha'+\di) c_\alpha'\\[0.5ex]
        \vec{A}_{\alpha}'' &\!\!\!\!=\, \vec{A}_\alpha' + \nabla_{\!A_\alpha'\!}\vec{c}_\alpha'   \\[0.5ex]
        \vec{c}_{\alpha}' &\!\!\!\!=\, g_{\beta\alpha}^{\prime-1}\vec{c}_\beta' g_{\prime\beta\alpha}
        \end{array} \right.\right\}.
\end{aligned}
\end{equation*}
Finally, for any general formal smooth manifold $U\in\mathsf{F}\Mfd$ we consider a good open cover $\bigsqcup_{i\in I\!}U_i\rightarrow U$ for it, so that all the overlaps $U_{i_1,\dots,i_n}$ are isomorphic to thickened Cartesian spaces.
Thus, we define the simplicial set of sections at $U$ to be the homotopy limit
\begin{equation}
    \Hom\big(U,\,T^\infty\Bun_G^\nabla(M)\big) \,\,\simeq\,\; \bfR\!\lim_{\!\!\!\!\!\!\![n]\in\Delta}\!\!\!\prod_{\;\;i_1,\dots,i_n\in I}\!\!\!\!\Hom\big(U_{i_1,\dots,i_n},\,T^\infty\Bun_G^\nabla(M)\big).
\end{equation}
\end{remark}

%%%%%%%%%%%%%%%%%%%%%%%%%%%%%%%%%%%%%%%%%%%%%%%%%%%%%%%%%%%%%%%%%%%%%%%%%%%%%%%%%%%%%%%%%
\subsubsection{Global Yang-Mills theory}

In the previous subsection, we constructed the smooth stack $\Bun_G^\nabla(M)$ which provides a global-geometric formulation of the configuration space of a gauge field with gauge Lie group $G$ on a spacetime $M$. 
In this subsection, we will proceed with the construction of the derived critical locus $\bfR\crit{S}(M)\rightarrow \Bun_G^\nabla(M)$ of the Yang-Mills action $S$ as a formal derived smooth stack in the context of derived differential geometry. Finally, we will show that such a geometric object provides a global-geometric version of usual BV-BRST theory.

\begin{construction}[Stack of densities]
We take spacetime to be an oriented $d$-dimensional smooth manifold $M$ equipped with a (pseudo-)Riemannian metric.
We construct the quotient stack $\mathbf{Dens}_M \coloneqq[\bfOmega^{d}(M)/\bfOmega^{d-1}(M)]$ of top forms $\mu$ on $M$, with the action $\mu\mapsto \mu + \di_\dR\lambda$ of $(d-1)$-forms $\lambda$. 
Notice that the connected components $\pi_0\mathbf{Dens}_M$ are classes of top forms up to total derivative.
\end{construction}

\begin{construction}[Yang-Mills action functional]
The datum of the Yang-Mills action functional is equivalently a morphism of formal smooth stacks given by
\begin{equation}\label{eq:YMAction}
\begin{aligned}
       \breve{S} \,:\; \Bun_G^\nabla(M) \;&\longrightarrow\; \mathbf{Dens}_M \\[0.2ex]
        (g_{\alpha\beta},\, A_\alpha) \;&\longmapsto\;  \frac{1}{2}\langle F_{A} \,\overset{\wedge}{,}\, \star_{\!\!} F_{A} \rangle_\mathfrak{g}
\end{aligned}
\end{equation}
where we called $F_A = \nabla_{\!A_\alpha}A_\alpha$ the curvature of the principal $G$-bundle with connection $(g_{\alpha\beta},\, A_\alpha)$.
\end{construction}

Now, the question becomes how can we encode the variational derivative of the Yang-Mills action functional or, in other words, the Euler-Lagrange equations of motion.
In analogy with the case of scalar field theory, we will construct a restricted cotangent bundle $T^\vee_\mathrm{res} \Bun_G^\nabla(M)$ such that the variational derivative can be formalised as its section
$\delta S:  \Bun_G^\nabla(M) \longrightarrow T^\vee_\mathrm{res} \Bun_G^\nabla(M)$.\vspace{0.15cm}

\begin{construction}[Fibre of the restricted cotangent bundle]
Notice that the Killing form $\langle -,- \rangle_\mathfrak{g}$ on the Lie algebra $\mathfrak{g}$ induces a natural pairing between $\mathfrak{g}_P$-valued differential forms $\langle -\,\overset{\wedge}{,}\,- \rangle_\mathfrak{g} : \Omega^{d-p}(M,\mathfrak{g}_P)\times\Omega^{p}(M,\mathfrak{g}_P)  \longrightarrow \Omega^{d}(M)$, where $d\coloneqq\mathrm{dim} M$ is the dimension of the base manifold and $\mathfrak{g}_P$ is the adjoint bundle of a principal $G$-bundle $P\twoheadrightarrow M$.
We want to use this fact to induce a well-defined morphism of formal derived smooth stacks of the form
\begin{equation}\label{eq:pairing_stacks}
    \begin{aligned}
    \langle -\,\overset{\wedge}{,}\,- \rangle_\mathfrak{g} \;:\;\, \mathcal{F}_{(P,\nabla_{\!A})} \times \bbD_{\Bun_G^\nabla(M),(P,\nabla_{\!A})} \;&\longrightarrow\;\mathbf{Dens}_M,
    \end{aligned}
\end{equation}
where $\mathcal{F}_{(P,\nabla_{\!A})}$ is a suitable formal derived smooth stack which we must construct. Let us define a formal derived smooth set by the derived kernel
\begin{equation}
    \mathcal{F}_{(P,\nabla_{\!A})} \;\coloneqq\; \bfR\!\ker\!\Big(\nabla_{\!A}:\bfOmega^{d-1}(M,\mathfrak{g}_P)\rightarrow\bfOmega^{d}(M,\mathfrak{g}_P)\Big),
\end{equation}
for any fixed principal $G$-bundle with connection $(P,\nabla_{\!A})\in\Bun_G^\nabla(M)$. A section is given by a $(d-1)$-form $\tilde{A}$ together with a homotopy $\tilde{c}$ from $\nabla_{\!A}\tilde{A}$ to $0$. 
The natural morphism \eqref{eq:pairing_stacks} is the constructed by sending $0$-simplices $(\tilde{A},\vec{A})$ to the density $\langle \tilde{A}\,\overset{\wedge}{,}\,\vec{A} \rangle_\mathfrak{g}$. This assignment is invariant up to total derivative, in fact an infinitesimal gauge transformation $\vec{A}\mapsto \vec{A}  + \nabla_{\!A}\vec{c}$ is sent to the $1$-simplex $\langle \tilde{A}\,\overset{\wedge}{,}\,\vec{A} \rangle_\mathfrak{g} \mapsto \langle \tilde{A}\,\overset{\wedge}{,}\,\vec{A} \rangle_\mathfrak{g} + \di_\dR\langle\tilde{A}\,\overset{\wedge}{,}\,\vec{c}\, \rangle_\mathfrak{g}$ in $\mathbf{Dens}_M$ for any $\vec{A}\in\bbD_{\Bun_G^\nabla(M),(P,\nabla_{\!A})}(U)$ and $\tilde{A}\in\mathcal{F}_{(P,\nabla_{\!A})}(U)$.
This is because the term $\langle \nabla_{\!A}\tilde{A}\,\,\overset{\wedge}{,}\,\vec{c}\, \rangle_\mathfrak{g}$ is homotopic to $0$.
Thus, the natural morphism \eqref{eq:pairing_stacks} is well-defined.
A reasonable definition of restricted cotangent bundle must be such that its fibre at the point $(P,\nabla_A)\in\Bun_G^\nabla(M)$ is the formal derived smooth set $\mathcal{F}_{(P,\nabla_{\!A})}$.
\end{construction}\vspace{-0.05cm}

We have a fibre-wise construction of a formal derived smooth stack $T_\mathrm{res}^\vee \Bun_G^\nabla(M)$, which we will call \textit{restricted cotangent bundle}, in analogy with scalar field theory above.
Therefore, by construction, there will be the natural pairing
\begin{equation*}
    \begin{aligned}
    \langle -\,\overset{\wedge}{,}\,- \rangle_\mathfrak{g} \;:\; T_\mathrm{res}^\vee \Bun_G^\nabla(M) \times_{\Bun_G^\nabla(M)} T^\infty\Bun_G^\nabla(M) \,&\longrightarrow\,\Bun_G^\nabla(M)\times\mathbf{Dens}_M.
    \end{aligned}
\end{equation*}

In the rest of this section, we will deploy the compact notation $f'\xleftarrow{\,f_1\,}f$ to denote a $1$-simplex $f_1$ whose boundaries are $\partial_0f_1= f$ and $\partial_1f_1=f'$, and similarly for higher simplices.
(This is the notation we used in Example \ref{derived stack warm up}, which the reader may find helpful to recall at this point.) \vspace{0.15cm}

\begin{construction}[Restricted cotangent bundle]\label{rem:res_cot_bundle}
Let us provide a concrete construction of the restricted cotangent bundle $T^\vee_\mathrm{res}{\Bun_G^\nabla(M)}$ in terms of \v{C}ech data. 
Such a construction is not the easiest, so our strategy will be the following:
first, we will define a pre-stack $T^\vee_\mathrm{res}{\Bun_G^\nabla(M)}^\mathrm{pre}$ which -- roughly speaking -- approximates the wanted formal derived smooth stack by encoding its local sections; then, we will stackify it. This means gluing local sections of the pre-stack in a way that is compatible with the descent condition on formal derived smooth manifolds.
To keep our notation consistent with the ordinary case, given an ordinary smooth manifold $V$, let us define the following simplicial sets:
\begin{equation*}
    \Coo(V\times U,\, G) \,\coloneqq\, \bfR\Hom(U,[V,G]), \qquad \Omega^p_\mathrm{ver}(V\times U,\mathfrak{g}) \,\coloneqq\, \bfR\Hom(U,\,\bfOmega^p(V,\mathfrak{g})),
\end{equation*}
for any formal derived smooth manifold $U$.\newpage
Now, the simplicial set of sections of $T^\vee_\mathrm{res}{\Bun_G^\nabla(M)}^\mathrm{pre}$ on any formal derived smooth manifold $U\in\mathbf{dFMfd}$ in our $(\infty,1)$-site is of the following form:
\begin{equation*}
\begin{aligned}
    &\bfR\Hom\big(U,\,T^\vee_\mathrm{res}{\Bun_G^\nabla(M)}^\mathrm{pre}\big) \;\simeq\; \\[1.0ex]
    &\simeq\;\left(\! \begin{tikzcd}[row sep={25.5ex,between origins}, column sep={15.5ex}]
    \cdots\, \arrow[r, yshift=1.8ex]\arrow[r, yshift=5.4ex]\arrow[r, yshift=-5.4ex]\arrow[r, yshift=-1.8ex] &\, Z_2 \, \arrow[rr, yshift=3.6ex , "{\big(c_\alpha,\,g_{1\!,\!\!\;\alpha\beta},h_{1\!,\!\!\;\alpha},\,\begin{smallmatrix}g_{\alpha\beta},A_\alpha,\tilde{A}_\alpha,\tilde{c}_\alpha\\g_{\alpha\beta}',A_\alpha',\tilde{A}_\alpha',\tilde{c}_\alpha'\end{smallmatrix}\big)}"] \arrow[rr, "{\big(c_\alpha',\,g_{1\!,\!\!\;\alpha\beta}',h_{1\!,\!\!\;\alpha}',\,\begin{smallmatrix}g_{\alpha\beta}',A_\alpha',\tilde{A}_\alpha',\tilde{c}_\alpha'\\g_{\alpha\beta}'',A_\alpha'',\tilde{A}_\alpha'',\tilde{c}_\alpha''\end{smallmatrix}\big)}" description] \arrow[rr, yshift=-3.6ex, "{\big(c_\alpha'',\,g_{1\!,\!\!\;\alpha\beta}'',h_{1\!,\!\!\;\alpha}'',\,\begin{smallmatrix}g_{\alpha\beta},A_\alpha,\tilde{A}_\alpha,\tilde{c}_\alpha\\g_{\alpha\beta}'',A_\alpha'',\tilde{A}''_\alpha,\tilde{c}_\alpha''\end{smallmatrix}\big)}"'] && \, Z_1 \, \arrow[r, yshift=1.8ex , "{(g_{\alpha\beta},A_\alpha,\tilde{A}_\alpha,\tilde{c}_\alpha)}"] \arrow[r, yshift=-1.8ex, "{(g_{\alpha\beta}',A_\alpha',\tilde{A}'_\alpha,\tilde{c}_\alpha')}"'] & \, Z_0
    \end{tikzcd}\!\right),
\end{aligned}
\end{equation*}
where, for simplicity, we packed together the data $h_{1\!,\!\!\;\alpha} \coloneqq (A_{1\!,\!\!\;\alpha},\,\tilde{A}_{1\!,\!\!\;\alpha},\,\tilde{c}_{1\!,\!\!\;\alpha})$ and where the sets of $0$- and $1$-simplices are respectively given by
\begin{equation*}
\begin{aligned}
        \hspace{-1.60cm}Z_0 \,&=\, \left\{ \begin{array}{ll}
          g_{\alpha\beta} &\!\!\!\!\in\Coo(V_\alpha\cap V_\beta \times U ,G)_0 \\[1.0ex]
         A_\alpha &\!\!\!\!\in\Omega^1_\mathrm{ver}(V_\alpha \times U,\mathfrak{g})_0 \\[1.0ex]
        \tilde{A}_\alpha &\!\!\!\!\in \Omega^{d-1}_\mathrm{ver}(V_\alpha \times U,\mathfrak{g})_0 \\[1.0ex]
        \tilde{c}_\alpha &\!\!\!\!\in \Omega^{d}_\mathrm{ver}(V_\alpha \times U,\mathfrak{g})_1
        \end{array}  \;\left|\; \begin{array}{l} 
        g_{\alpha\beta}\cdot g_{\beta\gamma}\cdot g_{\gamma\alpha} = 1 \\[0.9ex]
        A_{\alpha}=g_{\beta\alpha}^{-1}(A_{\beta}+\di)g_{\beta\alpha} \\[1.0ex]
        \tilde{A}_\alpha = g_{\beta\alpha}^{-1}\tilde{A}_\beta g_{\beta\alpha} \\[1.0ex] 
        \tilde{c}_\alpha = g_{\beta\alpha}^{-1}\tilde{c}_\beta g_{\beta\alpha} \\[1.0ex] 
        0\xleftarrow{\,\tilde{c}_\alpha\,}\nabla_{\!A_\alpha\!}\tilde{A}_\alpha\end{array} \right.\right\},
\end{aligned}
\end{equation*}
\begin{equation*}
\begin{aligned}
       Z_1 &=\,\left\{ \begin{array}{ll}
          c_{\alpha} &\!\!\!\!\!\in\Coo(V_\alpha \times U,G)_0 \\[2.2ex]
        g_{\alpha\beta},  g_{\alpha\beta}' &\!\!\!\!\!\in\Coo(V_\alpha\cap V_\beta \times U,G)_0 \\[1.0ex]
        A_{\alpha}, A_{\alpha}' &\!\!\!\!\!\in\Omega^1_\mathrm{ver}(V_\alpha  \times U,\mathfrak{g})_0 \\[1.0ex]
        \tilde{A}_{\alpha},  \tilde{A}^{\prime}_{\alpha} &\!\!\!\!\!\in \Omega^{d-1}_\mathrm{ver}(V_\alpha \times U,\mathfrak{g})_0 \\[1.0ex]
        \tilde{c}_{\alpha}, \tilde{c}^{\prime}_{\alpha} &\!\!\!\!\!\in \Omega^{d}_\mathrm{ver}(V_\alpha \times U,\mathfrak{g})_1  \\[2.2ex]
        g_{1,\alpha\beta} &\!\!\!\!\!\in\Coo(V_\alpha\cap V_\beta \times U,G)_1 \\[1.0ex]
        A_{1,\alpha} &\!\!\!\!\!\in\Omega^1_\mathrm{ver}(V_\alpha  \times U,\mathfrak{g})_1 \\[1.0ex]
        \tilde{A}_{1,\alpha} &\!\!\!\!\!\in \Omega^{d-1}_\mathrm{ver}(V_\alpha \times U,\mathfrak{g})_1 \\[1.0ex]
        \tilde{c}_{1,\alpha} &\!\!\!\!\!\in \Omega^{d}_\mathrm{ver}(V_\alpha \times U,\mathfrak{g})_2
        \end{array}  \left|\, \begin{array}{ll}
        g_{\alpha\beta}&\!\!\!\!\!\!\cdot\, g_{\beta\gamma}\cdot g_{\gamma\alpha} = 1 \\[0.8ex]
        A_{\alpha} &\!\!\!\!\!\!=g_{\beta\alpha}^{-1}(A_{\beta}+\di)g_{\beta\alpha} \\[1.0ex]
        \tilde{A}_\alpha &\!\!\!\!\!\!= g_{\beta\alpha}^{-1}\tilde{A}_\beta g_{\beta\alpha} \\[1.0ex] 
        \tilde{c}_\alpha &\!\!\!\!\!\!= g_{\beta\alpha}^{-1}\tilde{c}_\beta g_{\beta\alpha}  \\[0.6ex]
        0&\!\!\!\!\!\!\!\!\!\!\xleftarrow{\,\tilde{c}_\alpha\,}\nabla_{\!A_\alpha\!}\tilde{A}_\alpha \\[3.2ex]
        g_{\alpha\beta}'&\!\!\!\!\!\!\cdot\, g_{\beta\gamma}'\cdot g_{\gamma\alpha}' = 1 \\[0.8ex]
        A_{\alpha}' &\!\!\!\!\!\!=g_{\beta\alpha}^{\prime -1}(A_{\beta}'+\di)g_{\beta\alpha}' \\[1.0ex]
        \tilde{A}^{\prime}_\alpha &\!\!\!\!\!\!= g_{\beta\alpha}^{\prime -1}\tilde{A}^{\prime}_\beta g_{\beta\alpha}' \\[1.0ex] 
        \tilde{c}^{\prime}_\alpha &\!\!\!\!\!\!= g_{\beta\alpha}^{\prime -1}\tilde{c}^{\prime}_\beta g_{\beta\alpha}'  \\[0.6ex]
        0&\!\!\!\!\!\!\!\!\!\!\xleftarrow{\,\tilde{c}_\alpha^\prime\,}\nabla_{\!A_\alpha^\prime\!}\tilde{A}_\alpha^\prime \\[3.2ex]
        g_{\alpha\beta}' &\!\!\!\!\!\xleftarrow{\,g_{1\!,\alpha\beta}\,} c_\beta^{-1}g_{\alpha\beta}c_\alpha  \\[0.2ex]
        A_{\alpha}' &\!\!\!\!\!\xleftarrow{\,A_{1\!,\alpha}\,} c_\alpha^{-1} (A_\alpha+\di) c_\alpha  \\[0.4ex]
        \tilde{A}^{\prime}_\alpha &\!\!\!\!\!\xleftarrow{\,\tilde{A}_{1\!,\alpha}\,} c_{\alpha}^{-1\!}\tilde{A}_\alpha c_{\alpha} \\[0.8ex]
        &\!\!\!\!\!\!\!\!\!\!\!\!\!\!\!\!  

\tikzset{every picture/.style={line width=0.75pt}} %set default line width to 0.75pt        
\begin{tikzpicture}[x=0.75pt,y=0.75pt,yscale=-1,xscale=1]
%Straight Lines [id:da6301223776402531] 
\draw    (100.5,80.08) -- (127.56,38.25) ;
\draw [shift={(127.56,38.25)}, rotate = 302.89] [color={rgb, 255:red, 0; green, 0; blue, 0 }  ][fill={rgb, 255:red, 0; green, 0; blue, 0 }  ][line width=0.75]      (0, 0) circle [x radius= 2.01, y radius= 2.01]   ;
\draw [shift={(115.98,56.14)}, rotate = 122.89] [color={rgb, 255:red, 0; green, 0; blue, 0 }  ][line width=0.75]    (6.56,-2.94) .. controls (4.17,-1.38) and (1.99,-0.4) .. (0,0) .. controls (1.99,0.4) and (4.17,1.38) .. (6.56,2.94)   ;
\draw [shift={(100.5,80.08)}, rotate = 302.89] [color={rgb, 255:red, 0; green, 0; blue, 0 }  ][fill={rgb, 255:red, 0; green, 0; blue, 0 }  ][line width=0.75]      (0, 0) circle [x radius= 2.01, y radius= 2.01]   ;
%Straight Lines [id:da9192652598182032] 
\draw    (156.28,80.08) -- (127.56,38.25) ;
\draw [shift={(127.56,38.25)}, rotate = 235.53] [color={rgb, 255:red, 0; green, 0; blue, 0 }  ][fill={rgb, 255:red, 0; green, 0; blue, 0 }  ][line width=0.75]      (0, 0) circle [x radius= 2.01, y radius= 2.01]   ;
\draw [shift={(139.88,56.2)}, rotate = 55.53] [color={rgb, 255:red, 0; green, 0; blue, 0 }  ][line width=0.75]    (6.56,-2.94) .. controls (4.17,-1.38) and (1.99,-0.4) .. (0,0) .. controls (1.99,0.4) and (4.17,1.38) .. (6.56,2.94)   ;
\draw [shift={(156.28,80.08)}, rotate = 235.53] [color={rgb, 255:red, 0; green, 0; blue, 0 }  ][fill={rgb, 255:red, 0; green, 0; blue, 0 }  ][line width=0.75]      (0, 0) circle [x radius= 2.01, y radius= 2.01]   ;
%Straight Lines [id:da21090259638959208] 
\draw    (156.28,80.08) -- (100.5,80.08) ;
\draw [shift={(100.5,80.08)}, rotate = 180] [color={rgb, 255:red, 0; green, 0; blue, 0 }  ][fill={rgb, 255:red, 0; green, 0; blue, 0 }  ][line width=0.75]      (0, 0) circle [x radius= 2.01, y radius= 2.01]   ;
\draw [shift={(132.99,80.08)}, rotate = 180] [color={rgb, 255:red, 0; green, 0; blue, 0 }  ][line width=0.75]    (6.56,-2.94) .. controls (4.17,-1.38) and (1.99,-0.4) .. (0,0) .. controls (1.99,0.4) and (4.17,1.38) .. (6.56,2.94)   ;
\draw [shift={(156.28,80.08)}, rotate = 180] [color={rgb, 255:red, 0; green, 0; blue, 0 }  ][fill={rgb, 255:red, 0; green, 0; blue, 0 }  ][line width=0.75]      (0, 0) circle [x radius= 2.01, y radius= 2.01]   ;
% Text Node
\draw (122.5,23.4) node [anchor=north west][inner sep=0.75pt]  [font=\scriptsize]  {$0$};
% Text Node
\draw (28.5,79.4) node [anchor=north west][inner sep=0.75pt]  [font=\tiny]  {$c_{\alpha }^{-1} (\nabla _{\!A_{\scaleto{\alpha}{1.5pt} }\!}\tilde{A}_{\alpha }) c_{\alpha }$};
% Text Node
\draw (160.5,79.4) node [anchor=north west][inner sep=0.75pt]  [font=\tiny]  {$\nabla _{\!A'_{\scaleto{\alpha}{1.5pt} }\!}\tilde{A} '_{\alpha }$};
% Text Node
\draw (64,44.9) node [anchor=north west][inner sep=0.75pt]  [font=\scriptsize]  {$c_{\alpha }^{-1}\tilde{c}_{\alpha } c_{\alpha }$};
% Text Node
\draw (146,45.4) node [anchor=north west][inner sep=0.75pt]  [font=\scriptsize]  {$\tilde{c} '_{\alpha }$};
% Text Node
\draw (119,58.4) node [anchor=north west][inner sep=0.75pt]  [font=\tiny]  {$\tilde{c}_{1\!,\alpha }$};
% Text Node
\draw (105.5,83.9) node [anchor=north west][inner sep=0.75pt]  [font=\scriptsize]  {$\nabla _{\!A_{\scaleto{1\!,\alpha}{3pt} }\!}\tilde{A}_{1\!,\alpha }$};
\end{tikzpicture}

        \end{array} \right.\right\},
\end{aligned}
\end{equation*}
and where the higher simplices $\{Z_k\}_{k\geq 2}$ are given by compositions of gauge transformations, as before, but up to homotopies.
Then, we must stackify our prestack $T^\vee_\mathrm{res}{\Bun_G^\nabla(M)}^\mathrm{pre}$ to obtain a fully fledged formal derived smooth stack. \newpage
For any formal derived smooth manifold $U\in\mathbf{dFMfd}$ in our $(\infty,1)$-site, the simplicial set of sections of the restricted cotangent complex $T^\vee_\mathrm{res}{\Bun_G^\nabla(M)}$ on $U$ is given by the homotopy colimit
\begin{equation*}
    \bfR\Hom\big(U,\,T^\vee_\mathrm{res}\Bun_G^\nabla(M)\big) \;\,\simeq\,\; \bfL\mathrm{co}\!\lim_{\!\!\!\!\!\!\!\!\!\!\!H(U)} \,\bfR\!\lim_{\!\!\!\!\!\![n]\in\Delta}\, \bfR\Hom\big(H(U)_n,\,T^\vee_\mathrm{res}\Bun_G^\nabla(M)^\mathrm{pre}\big),
\end{equation*}
where the colimit is taken over all hypercovers $H(U)$ -- cf. Definition \ref{etale hypercover} -- which cover $U$. This stackification procedure is explained, for instance, in \cite[Section 6.5.3]{topos}
\end{construction}

Notice that a $0$-simplex in the space of sections above (which we can also call a $U$-point) is given, first, by the \v{C}ech-Deligne cocycle $(g_{\alpha\beta},A_{\alpha})$ of a $U$-parametrised family of principal $G$-bundles $(P,\nabla_{\!A})$ with connection and, second, by a $U$-parametrised family of differential forms $\tilde{A}\in\Omega^{d-1}_\mathrm{ver}(M\times U,\mathfrak{g}_P)_0$ and $\tilde{c}\in\Omega^{d}_\mathrm{ver}(M\times U,\mathfrak{g}_P)_1$ which are valued in the adjoint bundle of the aforementioned family of bundles. \vspace{0.1cm}

\begin{remark}[Derived-extension of the group action]
In the construction above we exploited the following facts.
As we said, given a smooth manifold $V$, there is a morphism of smooth sets
$\rho:[V,G]\times \bfOmega^1(V,\mathfrak{g})\rightarrow\bfOmega^1(V,\mathfrak{g})$ defined by $(c,A)\mapsto c^{-1}(A+\di)c$.
Then, we can embed such a morphism of smooth sets into a morphism $i\rho$ of formal derived smooth stacks by derived-extension from definition \ref{def:der_ext_fun}.
On a given formal derived smooth manifold $U$, we then have the morphism of simplicial sets $i\rho(U): \Coo(V\times U)\times \Omega^1_\mathrm{ver}(V\times U,\mathfrak{g})\rightarrow \Omega^1_\mathrm{ver}(V\times U,\mathfrak{g})$, with the same notation as above.
In complete analogy, we can derived-extend the group multiplication morphism $\cdot:[V,G]\times [V,G]\rightarrow[V,G]$, given by
$(c,c^\prime)\mapsto c\cdot c^\prime$, and the morphism encoding the adjoint action on differential forms
$\mathrm{Ad}:[V,G]\times \bfOmega^p(V,\mathfrak{g})\rightarrow\bfOmega^p(V,\mathfrak{g})$, given by $(c,\tilde{A})\mapsto c^{-1}\tilde{A}c$.
\end{remark}

Now that we have constructed restricted cotangent bundle, we can show that the Yang-Mills action functional $S$ induces a section $\delta S: \Bun_G^\nabla(M) \longrightarrow T^\vee_\mathrm{res}\Bun_G^\nabla(M)$ of it, which is going to encode its equations of motion.
In fact, as shown for example in \cite{JCMBLecGauge}, the first variation of the Yang-Mills action functional can be expressed in the form $\di_\dR S = \int_M\langle \delta S\,\overset{\wedge}{,}\,- \rangle_\mathfrak{g}$ where the variational derivative, which encodes the Yang-Mills equations, must be of the form $\delta S(P,\nabla_{\! A}) = \nabla_{\! A}\!\star\! F_{A}\in\Omega^{d-1}(M,\mathfrak{g}_P)$ at any bundle $(P,\nabla_{\! A})$. Let us now see that this can be indeed interpreted as a section of the restricted cotangent bundle. \vspace{0.15cm}

\begin{construction}[Variational derivative of the action functional]
The de Rham differential $\di_\dR S$ of the action functional gives rise to a morphism of formal derived smooth stacks, which we call variational derivative, given by
\begin{equation}\label{eq:DYMAction}
\begin{aligned}
     \delta S \,:\; \Bun_G^\nabla(M) \;&\longrightarrow\; T^\vee_\mathrm{res}\Bun_G^\nabla(M) \\[0.1ex]
     (g_{\alpha\beta},\, A_\alpha) \;&\longmapsto\; (g_{\alpha\beta},\, A_\alpha,\, \nabla_{\! A_\alpha\!}\!\star\! F_{A_\alpha},\, 0),
     \end{aligned}
\end{equation}
and the higher simplices are naturally embedded.
\end{construction}

Now, since we have a good definition of the variational derivative, we have all the ingredients we need to define the derived critical locus $\bfR\crit{S}(M)$ of the Yang-Mills action functional.
\newpage
\begin{definition}[Derived critical locus of Yang-Mills action functional]
We construct the \textit{derived critical locus of Yang-Mills action functional} by the formal derived smooth stack given by the following homotopy pullback square:
\begin{equation}
    \begin{tikzcd}[row sep={20ex,between origins}, column sep={22ex,between origins}]
    \bfR\crit{S}(M) \arrow[r]\arrow[d, ""'] & \Bun_G^\nabla(M) \arrow[d, "\delta S"] \\
    \Bun_G^\nabla(M) \arrow[r, "0"] & T^\vee_\mathrm{res}{\Bun_G^\nabla(M)} ,
\end{tikzcd}
\end{equation}
\end{definition}
where $\delta S$ is the morphism \eqref{eq:YMAction} constructed above and $0$ is the zero-section. \vspace{0.1cm}

\begin{remark}[Derived critical locus in \v{C}ech data]
Let us unravel the definition of the derived critical locus $\bfR\crit{S}(M)$ of the Yang-Mills action functional in terms of \v{C}ech data.
As in the previous example, our strategy to present the derived critical locus will be the following:
first, we will explicitly write a pre-stack $\bfR\crit{S}(M)^\mathrm{pre}$ which -- roughly speaking -- approximates the derived critical locus by encoding its local sections; then, we will stackify it. 
So, the simplicial set of sections of $\bfR\crit{S}(M)^\mathrm{pre}$ on any formal derived smooth manifold $U\in\mathbf{dFMfd}$ in our $(\infty,1)$-site is of the following form:\vspace{-0.2cm}

\begin{equation*}
\begin{aligned}
    &\bfR\Hom\big(U,\,\bfR\crit{S}(M)^\mathrm{pre}\big) \;\simeq\; \\[1.0ex]
    &\simeq\;\left(\! \begin{tikzcd}[row sep={25.5ex,between origins}, column sep={16.5ex}]
    \cdots\, \arrow[r, yshift=1.8ex]\arrow[r, yshift=5.4ex]\arrow[r, yshift=-5.4ex]\arrow[r, yshift=-1.8ex] &\, Z_2 \, \arrow[rr, yshift=3.6ex , "{\big(c_\alpha,\,g_{1\!,\!\!\;\alpha\beta},h_{1\!,\!\!\;\alpha},\,\begin{smallmatrix}g_{\alpha\beta},A_\alpha,{A}^+_\alpha,{c}^+_\alpha\\g_{\alpha\beta}',A_\alpha',{A}^{+\prime}_\alpha,{c}^{+\prime}_\alpha\end{smallmatrix}\big)}"] \arrow[rr, "{\big(c_\alpha',\,g_{1\!,\!\!\;\alpha\beta}',h_{1\!,\!\!\;\alpha}',\,\begin{smallmatrix}g_{\alpha\beta}',A_\alpha',{A}^{+\prime}_\alpha,c^{+\prime}_\alpha\\g_{\alpha\beta}'',A_\alpha'',{A}^{+\prime\prime}_\alpha,{c}^{+\prime\prime}_\alpha\end{smallmatrix}\big)}" description] \arrow[rr, yshift=-3.6ex, "{\big(c_\alpha'',\,g_{1\!,\!\!\;\alpha\beta}'',h_{1\!,\!\!\;\alpha}'',\,\begin{smallmatrix}g_{\alpha\beta},A_\alpha,{A}^+_\alpha,{c}^+_\alpha\\g_{\alpha\beta}'',A_\alpha'',A^{+\prime\prime}_\alpha,{c}^{+\prime\prime}_\alpha\end{smallmatrix}\big)}"'] && \, Z_1 \, \arrow[r, yshift=1.8ex , "{(g_{\alpha\beta},A_\alpha,{A}^+_\alpha,{c}^+_\alpha)}"] \arrow[r, yshift=-1.8ex, "{(g_{\alpha\beta}',A_\alpha',{A}^{+\prime}_\alpha,{c}^{+\prime}_\alpha)}"'] & \, Z_0
    \end{tikzcd}\!\right),
\end{aligned}
\end{equation*}
where, for simplicity, we packed again together $h_{1\!,\!\!\;\alpha}\coloneqq(A_{1\!,\!\!\;\alpha},\,A^+_{1\!,\!\!\;\alpha},\,c^+_{1\!,\!\!\;\alpha})$ and where the sets of $0$- and $1$-simplices are respectively given by:
\begin{equation*}
\begin{aligned}
       \hspace{-3.25cm}  Z_0 \,&=\, \left\{ \begin{array}{ll}
          g_{\alpha\beta} &\!\!\!\!\in\Coo(V_\alpha\cap V_\beta  \times U ,G)_0 \\[1.0ex]
         A_\alpha &\!\!\!\!\in\Omega^1_\mathrm{ver}(V_\alpha  \times U,\mathfrak{g})_0 \\[1.0ex]
        A^+_\alpha &\!\!\!\!\in \Omega^{d-1}_\mathrm{ver}(V_\alpha  \times U,\mathfrak{g})_1 \\[1.0ex]
        c^+_\alpha &\!\!\!\!\in \Omega^{d}_\mathrm{ver}(V_\alpha \times U,\mathfrak{g})_2
        \end{array}  \;\left|\; \begin{array}{l} 
        g_{\alpha\beta}\cdot g_{\beta\gamma}\cdot g_{\gamma\alpha} = 1 \\[0.9ex]
        A_{\alpha}=g_{\beta\alpha}^{-1}(A_{\beta}+\di)g_{\beta\alpha} \\[1.0ex]
        A^+_\alpha = g_{\beta\alpha}^{-1}A^+_\beta g_{\beta\alpha} \\[1.0ex] 
        c^+_\alpha = g_{\beta\alpha}^{-1}c^+_\beta g_{\beta\alpha}  \\[0.3ex]
        0\xleftarrow{\,A^+_\alpha\,}\nabla_{\!A_\alpha\!}\star \!F_{\!A_\alpha} \\[0.5ex]
        
\tikzset{every picture/.style={line width=0.75pt}} %set default line width to 0.75pt        
\begin{tikzpicture}[x=0.75pt,y=0.75pt,yscale=-1,xscale=1]
%Straight Lines [id:da6301223776402531] 
\draw    (100.5,80.08) -- (127.56,38.25) ;
\draw [shift={(127.56,38.25)}, rotate = 302.89] [color={rgb, 255:red, 0; green, 0; blue, 0 }  ][fill={rgb, 255:red, 0; green, 0; blue, 0 }  ][line width=0.75]      (0, 0) circle [x radius= 2.01, y radius= 2.01]   ;
\draw [shift={(115.98,56.14)}, rotate = 122.89] [color={rgb, 255:red, 0; green, 0; blue, 0 }  ][line width=0.75]    (6.56,-2.94) .. controls (4.17,-1.38) and (1.99,-0.4) .. (0,0) .. controls (1.99,0.4) and (4.17,1.38) .. (6.56,2.94)   ;
\draw [shift={(100.5,80.08)}, rotate = 302.89] [color={rgb, 255:red, 0; green, 0; blue, 0 }  ][fill={rgb, 255:red, 0; green, 0; blue, 0 }  ][line width=0.75]      (0, 0) circle [x radius= 2.01, y radius= 2.01]   ;
%Straight Lines [id:da9192652598182032] 
\draw    (156.28,80.08) -- (127.56,38.25) ;
\draw [shift={(127.56,38.25)}, rotate = 235.53] [color={rgb, 255:red, 0; green, 0; blue, 0 }  ][fill={rgb, 255:red, 0; green, 0; blue, 0 }  ][line width=0.75]      (0, 0) circle [x radius= 2.01, y radius= 2.01]   ;
\draw [shift={(139.88,56.2)}, rotate = 55.53] [color={rgb, 255:red, 0; green, 0; blue, 0 }  ][line width=0.75]    (6.56,-2.94) .. controls (4.17,-1.38) and (1.99,-0.4) .. (0,0) .. controls (1.99,0.4) and (4.17,1.38) .. (6.56,2.94)   ;
\draw [shift={(156.28,80.08)}, rotate = 235.53] [color={rgb, 255:red, 0; green, 0; blue, 0 }  ][fill={rgb, 255:red, 0; green, 0; blue, 0 }  ][line width=0.75]      (0, 0) circle [x radius= 2.01, y radius= 2.01]   ;
%Straight Lines [id:da21090259638959208] 
\draw    (156.28,80.08) -- (100.5,80.08) ;
\draw [shift={(100.5,80.08)}, rotate = 180] [color={rgb, 255:red, 0; green, 0; blue, 0 }  ][fill={rgb, 255:red, 0; green, 0; blue, 0 }  ][line width=0.75]      (0, 0) circle [x radius= 2.01, y radius= 2.01]   ;
\draw [shift={(132.99,80.08)}, rotate = 180] [color={rgb, 255:red, 0; green, 0; blue, 0 }  ][line width=0.75]    (6.56,-2.94) .. controls (4.17,-1.38) and (1.99,-0.4) .. (0,0) .. controls (1.99,0.4) and (4.17,1.38) .. (6.56,2.94)   ;
\draw [shift={(156.28,80.08)}, rotate = 180] [color={rgb, 255:red, 0; green, 0; blue, 0 }  ][fill={rgb, 255:red, 0; green, 0; blue, 0 }  ][line width=0.75]      (0, 0) circle [x radius= 2.01, y radius= 2.01]   ;
% Text Node
\draw (123,26.9) node [anchor=north west][inner sep=0.75pt]  [font=\tiny]  {$0$};
% Text Node
\draw (87.5,79.4) node [anchor=north west][inner sep=0.75pt]  [font=\tiny]  {$0$};
% Text Node
\draw (160.5,79.9) node [anchor=north west][inner sep=0.75pt]  [font=\tiny]  {$0$};
% Text Node
\draw (99.5,49.4) node [anchor=north west][inner sep=0.75pt]  [font=\scriptsize]  {$0$};
% Text Node
\draw (120,58.4) node [anchor=north west][inner sep=0.75pt]  [font=\tiny]  {$c_{\alpha }^{+}$};
% Text Node
\draw (146.5,49.4) node [anchor=north west][inner sep=0.75pt]  [font=\scriptsize]  {$0$};
% Text Node
\draw (109.5,84.9) node [anchor=north west][inner sep=0.75pt]  [font=\scriptsize]  {$\nabla _{\!A_{\scaleto{\alpha}{1.5pt} }\!} A_{\alpha }^{+}$};
\end{tikzpicture}

        \end{array} \right.\right\},
\end{aligned}
\end{equation*}
\begin{equation*}
\begin{aligned}
        Z_1 &=\, \left\{ \begin{array}{ll}
          c_{\alpha} &\!\!\!\!\!\in\Coo(V_\alpha \times U,G)_0 \\[2.2ex]
        g_{\alpha\beta},  g_{\alpha\beta}' &\!\!\!\!\!\in\Coo(V_\alpha\cap V_\beta \times U,G)_0 \\[1.0ex]
        A_{\alpha}, A_{\alpha}' &\!\!\!\!\!\in\Omega^1_\mathrm{ver}(V_\alpha \times U,\mathfrak{g})_0 \\[1.0ex]
        A^{+}_{\alpha},  A^{+\prime}_{\alpha} &\!\!\!\!\!\in \Omega^{d-1}_\mathrm{ver}(V_\alpha \times U,\mathfrak{g})_1 \\[1.0ex]
        c^{+}_{\alpha}, c^{+\prime}_{\alpha} &\!\!\!\!\!\in \Omega^{d}_\mathrm{ver}(V_\alpha \times U,\mathfrak{g})_2  \\[2.2ex]
        g_{1,\alpha\beta} &\!\!\!\!\!\in\Coo(V_\alpha\cap V_\beta \times U,G)_1 \\[1.0ex]
        A_{1,\alpha} &\!\!\!\!\!\in\Omega^1_\mathrm{ver}(V_\alpha \times U,\mathfrak{g})_1 \\[1.0ex]
        A^{+}_{1,\alpha} &\!\!\!\!\!\in \Omega^{d-1}_\mathrm{ver}(V_\alpha \times U,\mathfrak{g})_2 \\[1.0ex]
        c^{+}_{1,\alpha} &\!\!\!\!\!\in \Omega^{d}_\mathrm{ver}(V_\alpha \times U,\mathfrak{g})_3
        \end{array}  \left|\,  \begin{array}{ll}
        \qquad\qquad g_{\alpha\beta}&\!\!\!\!\!\!\cdot\, g_{\beta\gamma}\cdot g_{\gamma\alpha} = 1 \\[0.8ex]
        \qquad\qquad A_{\alpha} &\!\!\!\!\!\!=g_{\beta\alpha}^{-1}(A_{\beta}+\di)g_{\beta\alpha} \\[1.0ex]
        \qquad\qquad A^+_\alpha &\!\!\!\!\!\!= g_{\beta\alpha}^{-1}A^+_\beta g_{\beta\alpha} \\[1.0ex] 
        \qquad\qquad c^+_\alpha &\!\!\!\!\!\!\!\!= g_{\beta\alpha}^{-1}c^+_\beta g_{\beta\alpha}  \\[0.4ex]
        \qquad\qquad 0 &\!\!\!\!\!\!\!\!\!\! \xleftarrow{\,A^+_\alpha\,}\nabla_{\!A_\alpha\!}\star \!F_{\!A_\alpha} \\[0.4ex]
        &\!\!\!\!\!\!\!\!\!\!\!\!\!\!
        \tikzset{every picture/.style={line width=0.75pt}} %set default line width to 0.75pt        
\begin{tikzpicture}[x=0.75pt,y=0.75pt,yscale=-1,xscale=1]
%Straight Lines [id:da6301223776402531] 
\draw    (100.5,80.08) -- (127.56,38.25) ;
\draw [shift={(127.56,38.25)}, rotate = 302.89] [color={rgb, 255:red, 0; green, 0; blue, 0 }  ][fill={rgb, 255:red, 0; green, 0; blue, 0 }  ][line width=0.75]      (0, 0) circle [x radius= 2.01, y radius= 2.01]   ;
\draw [shift={(115.98,56.14)}, rotate = 122.89] [color={rgb, 255:red, 0; green, 0; blue, 0 }  ][line width=0.75]    (6.56,-2.94) .. controls (4.17,-1.38) and (1.99,-0.4) .. (0,0) .. controls (1.99,0.4) and (4.17,1.38) .. (6.56,2.94)   ;
\draw [shift={(100.5,80.08)}, rotate = 302.89] [color={rgb, 255:red, 0; green, 0; blue, 0 }  ][fill={rgb, 255:red, 0; green, 0; blue, 0 }  ][line width=0.75]      (0, 0) circle [x radius= 2.01, y radius= 2.01]   ;
%Straight Lines [id:da9192652598182032] 
\draw    (156.28,80.08) -- (127.56,38.25) ;
\draw [shift={(127.56,38.25)}, rotate = 235.53] [color={rgb, 255:red, 0; green, 0; blue, 0 }  ][fill={rgb, 255:red, 0; green, 0; blue, 0 }  ][line width=0.75]      (0, 0) circle [x radius= 2.01, y radius= 2.01]   ;
\draw [shift={(139.88,56.2)}, rotate = 55.53] [color={rgb, 255:red, 0; green, 0; blue, 0 }  ][line width=0.75]    (6.56,-2.94) .. controls (4.17,-1.38) and (1.99,-0.4) .. (0,0) .. controls (1.99,0.4) and (4.17,1.38) .. (6.56,2.94)   ;
\draw [shift={(156.28,80.08)}, rotate = 235.53] [color={rgb, 255:red, 0; green, 0; blue, 0 }  ][fill={rgb, 255:red, 0; green, 0; blue, 0 }  ][line width=0.75]      (0, 0) circle [x radius= 2.01, y radius= 2.01]   ;
%Straight Lines [id:da21090259638959208] 
\draw    (156.28,80.08) -- (100.5,80.08) ;
\draw [shift={(100.5,80.08)}, rotate = 180] [color={rgb, 255:red, 0; green, 0; blue, 0 }  ][fill={rgb, 255:red, 0; green, 0; blue, 0 }  ][line width=0.75]      (0, 0) circle [x radius= 2.01, y radius= 2.01]   ;
\draw [shift={(132.99,80.08)}, rotate = 180] [color={rgb, 255:red, 0; green, 0; blue, 0 }  ][line width=0.75]    (6.56,-2.94) .. controls (4.17,-1.38) and (1.99,-0.4) .. (0,0) .. controls (1.99,0.4) and (4.17,1.38) .. (6.56,2.94)   ;
\draw [shift={(156.28,80.08)}, rotate = 180] [color={rgb, 255:red, 0; green, 0; blue, 0 }  ][fill={rgb, 255:red, 0; green, 0; blue, 0 }  ][line width=0.75]      (0, 0) circle [x radius= 2.01, y radius= 2.01]   ;
% Text Node
\draw (123,26.9) node [anchor=north west][inner sep=0.75pt]  [font=\tiny]  {$0$};
% Text Node
\draw (87.5,79.4) node [anchor=north west][inner sep=0.75pt]  [font=\tiny]  {$0$};
% Text Node
\draw (160.5,79.9) node [anchor=north west][inner sep=0.75pt]  [font=\tiny]  {$0$};
% Text Node
\draw (99.5,49.4) node [anchor=north west][inner sep=0.75pt]  [font=\scriptsize]  {$0$};
% Text Node
\draw (120,58.4) node [anchor=north west][inner sep=0.75pt]  [font=\tiny]  {$c_{\alpha }^{+}$};
% Text Node
\draw (146.5,49.4) node [anchor=north west][inner sep=0.75pt]  [font=\scriptsize]  {$0$};
% Text Node
\draw (109.5,84.9) node [anchor=north west][inner sep=0.75pt]  [font=\scriptsize]  {$\nabla _{\!A_{\scaleto{\alpha}{1.5pt} }\!} A_{\alpha }^{+}$};
\end{tikzpicture}
           
           \\[4.2ex]
        \qquad\qquad g_{\alpha\beta}'&\!\!\!\!\!\!\cdot\, g_{\beta\gamma}'\cdot g_{\gamma\alpha}' = 1 \\[0.8ex]
        \qquad\qquad A_{\alpha}' &\!\!\!\!\!\!=g_{\beta\alpha}^{\prime -1}(A_{\beta}'+\di)g_{\beta\alpha}' \\[1.0ex]
        \qquad\qquad A^{+\prime}_\alpha &\!\!\!\!\!\!= g_{\beta\alpha}^{\prime -1}A^{+\prime}_\beta g_{\beta\alpha}' \\[1.0ex] 
        \qquad\qquad c^{+\prime}_\alpha &\!\!\!\!\!\!= g_{\beta\alpha}^{\prime -1}c^{+\prime}_\beta g_{\beta\alpha}'  \\[0.4ex]
        \qquad\qquad 0 &\!\!\!\!\!\!\!\!\!\! \xleftarrow{\,A^{+\prime}_\alpha\,}\nabla_{\!A_\alpha^\prime\!}\star \!F_{\!A_\alpha^\prime} \\[0.4ex]
         &\!\!\!\!\!\!\!\!\!\!\!\!\!
        \tikzset{every picture/.style={line width=0.75pt}} %set default line width to 0.75pt        
\begin{tikzpicture}[x=0.75pt,y=0.75pt,yscale=-1,xscale=1]
%Straight Lines [id:da6301223776402531] 
\draw    (100.5,80.08) -- (127.56,38.25) ;
\draw [shift={(127.56,38.25)}, rotate = 302.89] [color={rgb, 255:red, 0; green, 0; blue, 0 }  ][fill={rgb, 255:red, 0; green, 0; blue, 0 }  ][line width=0.75]      (0, 0) circle [x radius= 2.01, y radius= 2.01]   ;
\draw [shift={(115.98,56.14)}, rotate = 122.89] [color={rgb, 255:red, 0; green, 0; blue, 0 }  ][line width=0.75]    (6.56,-2.94) .. controls (4.17,-1.38) and (1.99,-0.4) .. (0,0) .. controls (1.99,0.4) and (4.17,1.38) .. (6.56,2.94)   ;
\draw [shift={(100.5,80.08)}, rotate = 302.89] [color={rgb, 255:red, 0; green, 0; blue, 0 }  ][fill={rgb, 255:red, 0; green, 0; blue, 0 }  ][line width=0.75]      (0, 0) circle [x radius= 2.01, y radius= 2.01]   ;
%Straight Lines [id:da9192652598182032] 
\draw    (156.28,80.08) -- (127.56,38.25) ;
\draw [shift={(127.56,38.25)}, rotate = 235.53] [color={rgb, 255:red, 0; green, 0; blue, 0 }  ][fill={rgb, 255:red, 0; green, 0; blue, 0 }  ][line width=0.75]      (0, 0) circle [x radius= 2.01, y radius= 2.01]   ;
\draw [shift={(139.88,56.2)}, rotate = 55.53] [color={rgb, 255:red, 0; green, 0; blue, 0 }  ][line width=0.75]    (6.56,-2.94) .. controls (4.17,-1.38) and (1.99,-0.4) .. (0,0) .. controls (1.99,0.4) and (4.17,1.38) .. (6.56,2.94)   ;
\draw [shift={(156.28,80.08)}, rotate = 235.53] [color={rgb, 255:red, 0; green, 0; blue, 0 }  ][fill={rgb, 255:red, 0; green, 0; blue, 0 }  ][line width=0.75]      (0, 0) circle [x radius= 2.01, y radius= 2.01]   ;
%Straight Lines [id:da21090259638959208] 
\draw    (156.28,80.08) -- (100.5,80.08) ;
\draw [shift={(100.5,80.08)}, rotate = 180] [color={rgb, 255:red, 0; green, 0; blue, 0 }  ][fill={rgb, 255:red, 0; green, 0; blue, 0 }  ][line width=0.75]      (0, 0) circle [x radius= 2.01, y radius= 2.01]   ;
\draw [shift={(132.99,80.08)}, rotate = 180] [color={rgb, 255:red, 0; green, 0; blue, 0 }  ][line width=0.75]    (6.56,-2.94) .. controls (4.17,-1.38) and (1.99,-0.4) .. (0,0) .. controls (1.99,0.4) and (4.17,1.38) .. (6.56,2.94)   ;
\draw [shift={(156.28,80.08)}, rotate = 180] [color={rgb, 255:red, 0; green, 0; blue, 0 }  ][fill={rgb, 255:red, 0; green, 0; blue, 0 }  ][line width=0.75]      (0, 0) circle [x radius= 2.01, y radius= 2.01]   ;
% Text Node
\draw (123,26.9) node [anchor=north west][inner sep=0.75pt]  [font=\tiny]  {$0$};
% Text Node
\draw (87.5,79.4) node [anchor=north west][inner sep=0.75pt]  [font=\tiny]  {$0$};
% Text Node
\draw (160.5,79.9) node [anchor=north west][inner sep=0.75pt]  [font=\tiny]  {$0$};
% Text Node
\draw (99.5,49.4) node [anchor=north west][inner sep=0.75pt]  [font=\scriptsize]  {$0$};
% Text Node
\draw (120,58.4) node [anchor=north west][inner sep=0.75pt]  [font=\tiny]  {$c_{\alpha }^{+\prime}$};
% Text Node
\draw (146.5,49.4) node [anchor=north west][inner sep=0.75pt]  [font=\scriptsize]  {$0$};
% Text Node
\draw (109.5,84.9) node [anchor=north west][inner sep=0.75pt]  [font=\scriptsize]  {$\nabla _{\!A_{\scaleto{\alpha}{1.5pt} }^\prime\!} A_{\alpha }^{+\prime}$};
\end{tikzpicture}
           
             \\[4.2ex]
        \qquad\qquad g_{\alpha\beta}' &\!\!\!\!\!\xleftarrow{\,g_{1\!,\alpha\beta}\,} c_\beta^{-1}g_{\alpha\beta}c_\alpha  \\[0.3ex]
        \qquad\qquad A_{\alpha}' &\!\!\!\!\!\xleftarrow{\,A_{1\!,\alpha}\,} c_\alpha^{-1} (A_\alpha+\di) c_\alpha   \\[0.6ex]
          &\!\!\!\!\!\!\!\!\!\!\!\!\!\!\!\!\!\!\!\!\!\!\!\!\!\!\!\!\!\!\!\!\!\!\!\!\!\!\!\!\!

\tikzset{every picture/.style={line width=0.75pt}} %set default line width to 0.75pt        
\begin{tikzpicture}[x=0.75pt,y=0.75pt,yscale=-1,xscale=1]
%Straight Lines [id:da6301223776402531] 
\draw    (100.5,80.08) -- (127.56,38.25) ;
\draw [shift={(127.56,38.25)}, rotate = 302.89] [color={rgb, 255:red, 0; green, 0; blue, 0 }  ][fill={rgb, 255:red, 0; green, 0; blue, 0 }  ][line width=0.75]      (0, 0) circle [x radius= 2.01, y radius= 2.01]   ;
\draw [shift={(115.98,56.14)}, rotate = 122.89] [color={rgb, 255:red, 0; green, 0; blue, 0 }  ][line width=0.75]    (6.56,-2.94) .. controls (4.17,-1.38) and (1.99,-0.4) .. (0,0) .. controls (1.99,0.4) and (4.17,1.38) .. (6.56,2.94)   ;
\draw [shift={(100.5,80.08)}, rotate = 302.89] [color={rgb, 255:red, 0; green, 0; blue, 0 }  ][fill={rgb, 255:red, 0; green, 0; blue, 0 }  ][line width=0.75]      (0, 0) circle [x radius= 2.01, y radius= 2.01]   ;
%Straight Lines [id:da9192652598182032] 
\draw    (156.28,80.08) -- (127.56,38.25) ;
\draw [shift={(127.56,38.25)}, rotate = 235.53] [color={rgb, 255:red, 0; green, 0; blue, 0 }  ][fill={rgb, 255:red, 0; green, 0; blue, 0 }  ][line width=0.75]      (0, 0) circle [x radius= 2.01, y radius= 2.01]   ;
\draw [shift={(139.88,56.2)}, rotate = 55.53] [color={rgb, 255:red, 0; green, 0; blue, 0 }  ][line width=0.75]    (6.56,-2.94) .. controls (4.17,-1.38) and (1.99,-0.4) .. (0,0) .. controls (1.99,0.4) and (4.17,1.38) .. (6.56,2.94)   ;
\draw [shift={(156.28,80.08)}, rotate = 235.53] [color={rgb, 255:red, 0; green, 0; blue, 0 }  ][fill={rgb, 255:red, 0; green, 0; blue, 0 }  ][line width=0.75]      (0, 0) circle [x radius= 2.01, y radius= 2.01]   ;
%Straight Lines [id:da21090259638959208] 
\draw    (156.28,80.08) -- (100.5,80.08) ;
\draw [shift={(100.5,80.08)}, rotate = 180] [color={rgb, 255:red, 0; green, 0; blue, 0 }  ][fill={rgb, 255:red, 0; green, 0; blue, 0 }  ][line width=0.75]      (0, 0) circle [x radius= 2.01, y radius= 2.01]   ;
\draw [shift={(132.99,80.08)}, rotate = 180] [color={rgb, 255:red, 0; green, 0; blue, 0 }  ][line width=0.75]    (6.56,-2.94) .. controls (4.17,-1.38) and (1.99,-0.4) .. (0,0) .. controls (1.99,0.4) and (4.17,1.38) .. (6.56,2.94)   ;
\draw [shift={(156.28,80.08)}, rotate = 180] [color={rgb, 255:red, 0; green, 0; blue, 0 }  ][fill={rgb, 255:red, 0; green, 0; blue, 0 }  ][line width=0.75]      (0, 0) circle [x radius= 2.01, y radius= 2.01]   ;
% Text Node
\draw (123.5,25.4) node [anchor=north west][inner sep=0.75pt]  [font=\tiny]  {$0$};
% Text Node
\draw (23.0,74.4) node [anchor=north west][inner sep=0.75pt]  [font=\tiny]  {$c_{\alpha }^{-1\!} (\nabla _{\!A_{\scaleto{\alpha}{1.5pt} }\!\!\!\!} \star_{\!}\! F_{\!A_{\scaleto{\alpha\!}{1.5pt} }}) c_{\alpha }$};
% Text Node
\draw (161.5,74.4) node [anchor=north west][inner sep=0.75pt]  [font=\tiny]  {$\nabla _{\!A'_{\scaleto{\alpha}{1.5pt} }\!\!\!} \star \!F_{\!A'_{\scaleto{\alpha}{1.5pt} }}$};
% Text Node
\draw (57,42.4) node [anchor=north west][inner sep=0.75pt]  [font=\scriptsize]  {$c_{\alpha }^{-1}{A}_{\alpha }^{+} c_{\alpha }$};
% Text Node
\draw (145,42.4) node [anchor=north west][inner sep=0.75pt]  [font=\scriptsize]  {${A}_{\alpha }^{+\prime }$};
% Text Node
\draw (118,57.0) node [anchor=north west][inner sep=0.75pt]  [font=\tiny]  {${A}_{1\!,\alpha }^{+}$};
% Text Node
\draw (92.5,90.0) node [anchor=north west][inner sep=0.75pt]  [font=\scriptsize]  {$\nabla _{\!A_{\scaleto{1\!,\alpha }{3pt}}\!\!\!} \star_{\!}\! \nabla _{\!A_{\scaleto{1\!,\alpha }{3pt}}\!} A_{1\!,\alpha }$};
\end{tikzpicture} \\[0.2ex] 

 &\!\!\!\!\!\!\!\!\!\!\!\!\!\!\!\!\!\!\!\!\!\!

\tikzset{every picture/.style={line width=0.75pt}} %set default line width to 0.75pt        
\begin{tikzpicture}[x=0.75pt,y=0.75pt,yscale=-1,xscale=1]
%Shape: Triangle [id:dp6163584919321383] 
\draw  [draw opacity=0][fill={rgb, 255:red, 155; green, 155; blue, 155 }  ,fill opacity=0.11 ] (170.34,115.96) -- (100.95,79.09) -- (199.86,72.78) -- cycle ;
%Straight Lines [id:da8685437513255596] 
\draw [color={rgb, 255:red, 155; green, 155; blue, 155 }  ,draw opacity=1 ] [dash pattern={on 0.84pt off 2.51pt}]  (179.89,100.28) -- (176.67,85.2) ;
%Straight Lines [id:da46316567360339445] 
\draw [color={rgb, 255:red, 179; green, 179; blue, 179 }  ,draw opacity=1 ] [dash pattern={on 0.84pt off 2.51pt}]  (143.47,52.4) -- (141.87,38) ;
%Straight Lines [id:da13333955161140731] 
\draw [color={rgb, 255:red, 155; green, 155; blue, 155 }  ,draw opacity=1 ]   (185.07,125.2) -- (179.89,100.28) ;
%Straight Lines [id:da9400745626202063] 
\draw [color={rgb, 255:red, 179; green, 179; blue, 179 }  ,draw opacity=1 ]   (141.87,38) -- (139.07,16.8) ;
%Straight Lines [id:da6301223776402531] 
\draw    (100.5,80.08) -- (154.27,26) ;
\draw [shift={(154.27,26)}, rotate = 314.83] [color={rgb, 255:red, 0; green, 0; blue, 0 }  ][fill={rgb, 255:red, 0; green, 0; blue, 0 }  ][line width=0.75]      (0, 0) circle [x radius= 2.01, y radius= 2.01]   ;
\draw [shift={(129.92,50.49)}, rotate = 134.83] [color={rgb, 255:red, 0; green, 0; blue, 0 }  ][line width=0.75]    (6.56,-2.94) .. controls (4.17,-1.38) and (1.99,-0.4) .. (0,0) .. controls (1.99,0.4) and (4.17,1.38) .. (6.56,2.94)   ;
\draw [shift={(100.5,80.08)}, rotate = 314.83] [color={rgb, 255:red, 0; green, 0; blue, 0 }  ][fill={rgb, 255:red, 0; green, 0; blue, 0 }  ][line width=0.75]      (0, 0) circle [x radius= 2.01, y radius= 2.01]   ;
%Straight Lines [id:da9192652598182032] 
\draw    (169.47,115.2) -- (154.27,26) ;
\draw [shift={(154.27,26)}, rotate = 260.33] [color={rgb, 255:red, 0; green, 0; blue, 0 }  ][fill={rgb, 255:red, 0; green, 0; blue, 0 }  ][line width=0.75]      (0, 0) circle [x radius= 2.01, y radius= 2.01]   ;
\draw [shift={(161.26,67.05)}, rotate = 80.33] [color={rgb, 255:red, 0; green, 0; blue, 0 }  ][line width=0.75]    (6.56,-2.94) .. controls (4.17,-1.38) and (1.99,-0.4) .. (0,0) .. controls (1.99,0.4) and (4.17,1.38) .. (6.56,2.94)   ;
\draw [shift={(169.47,115.2)}, rotate = 260.33] [color={rgb, 255:red, 0; green, 0; blue, 0 }  ][fill={rgb, 255:red, 0; green, 0; blue, 0 }  ][line width=0.75]      (0, 0) circle [x radius= 2.01, y radius= 2.01]   ;
%Straight Lines [id:da21090259638959208] 
\draw    (169.47,115.2) -- (100.5,80.08) ;
\draw [shift={(100.5,80.08)}, rotate = 206.98] [color={rgb, 255:red, 0; green, 0; blue, 0 }  ][fill={rgb, 255:red, 0; green, 0; blue, 0 }  ][line width=0.75]      (0, 0) circle [x radius= 2.01, y radius= 2.01]   ;
\draw [shift={(139.08,99.73)}, rotate = 206.98] [color={rgb, 255:red, 0; green, 0; blue, 0 }  ][line width=0.75]    (6.56,-2.94) .. controls (4.17,-1.38) and (1.99,-0.4) .. (0,0) .. controls (1.99,0.4) and (4.17,1.38) .. (6.56,2.94)   ;
\draw [shift={(169.47,115.2)}, rotate = 206.98] [color={rgb, 255:red, 0; green, 0; blue, 0 }  ][fill={rgb, 255:red, 0; green, 0; blue, 0 }  ][line width=0.75]      (0, 0) circle [x radius= 2.01, y radius= 2.01]   ;
%Straight Lines [id:da5941432442624102] 
\draw    (198.27,73.2) -- (169.47,115.2) ;
\draw [shift={(169.47,115.2)}, rotate = 124.44] [color={rgb, 255:red, 0; green, 0; blue, 0 }  ][fill={rgb, 255:red, 0; green, 0; blue, 0 }  ][line width=0.75]      (0, 0) circle [x radius= 2.01, y radius= 2.01]   ;
\draw [shift={(186.47,90.41)}, rotate = 124.44] [color={rgb, 255:red, 0; green, 0; blue, 0 }  ][line width=0.75]    (6.56,-2.94) .. controls (4.17,-1.38) and (1.99,-0.4) .. (0,0) .. controls (1.99,0.4) and (4.17,1.38) .. (6.56,2.94)   ;
\draw [shift={(198.27,73.2)}, rotate = 124.44] [color={rgb, 255:red, 0; green, 0; blue, 0 }  ][fill={rgb, 255:red, 0; green, 0; blue, 0 }  ][line width=0.75]      (0, 0) circle [x radius= 2.01, y radius= 2.01]   ;
%Straight Lines [id:da6341102385967854] 
\draw    (198.27,73.2) -- (154.27,26) ;
\draw [shift={(154.27,26)}, rotate = 227.01] [color={rgb, 255:red, 0; green, 0; blue, 0 }  ][fill={rgb, 255:red, 0; green, 0; blue, 0 }  ][line width=0.75]      (0, 0) circle [x radius= 2.01, y radius= 2.01]   ;
\draw [shift={(173.81,46.97)}, rotate = 47.01] [color={rgb, 255:red, 0; green, 0; blue, 0 }  ][line width=0.75]    (6.56,-2.94) .. controls (4.17,-1.38) and (1.99,-0.4) .. (0,0) .. controls (1.99,0.4) and (4.17,1.38) .. (6.56,2.94)   ;
\draw [shift={(198.27,73.2)}, rotate = 227.01] [color={rgb, 255:red, 0; green, 0; blue, 0 }  ][fill={rgb, 255:red, 0; green, 0; blue, 0 }  ][line width=0.75]      (0, 0) circle [x radius= 2.01, y radius= 2.01]   ;
%Straight Lines [id:da3084010225260736] 
\draw  [dash pattern={on 0.84pt off 2.51pt}]  (198.27,73.2) -- (100.5,80.08) ;
\draw [shift={(100.5,80.08)}, rotate = 175.97] [color={rgb, 255:red, 0; green, 0; blue, 0 }  ][fill={rgb, 255:red, 0; green, 0; blue, 0 }  ][line width=0.75]      (0, 0) circle [x radius= 2.01, y radius= 2.01]   ;
\draw [shift={(153.97,76.32)}, rotate = 175.97] [color={rgb, 255:red, 0; green, 0; blue, 0 }  ][line width=0.75]    (6.56,-2.94) .. controls (4.17,-1.38) and (1.99,-0.4) .. (0,0) .. controls (1.99,0.4) and (4.17,1.38) .. (6.56,2.94)   ;
\draw [shift={(198.27,73.2)}, rotate = 175.97] [color={rgb, 255:red, 0; green, 0; blue, 0 }  ][fill={rgb, 255:red, 0; green, 0; blue, 0 }  ][line width=0.75]      (0, 0) circle [x radius= 2.01, y radius= 2.01]   ;
%Straight Lines [id:da36356908586193715] 
\draw [color={rgb, 255:red, 179; green, 179; blue, 179 }  ,draw opacity=1 ]   (118.67,108.4) -- (140.67,71.6) ;
%Straight Lines [id:da9570030365134448] 
\draw [color={rgb, 255:red, 179; green, 179; blue, 179 }  ,draw opacity=1 ]   (171.87,64.4) -- (191.89,40.61) ;
% Text Node
\draw (149.6,13.9) node [anchor=north west][inner sep=0.75pt]  [font=\tiny]  {$0$};
% Text Node
\draw (89.3,78) node [anchor=north west][inner sep=0.75pt]  [font=\tiny]  {$0$};
% Text Node
\draw (165.5,120.1) node [anchor=north west][inner sep=0.75pt]  [font=\tiny]  {$0$};
% Text Node
\draw (181.43,126.13) node [anchor=north west][inner sep=0.75pt]  [font=\tiny,color={rgb, 255:red, 0; green, 0; blue, 0 }  ,opacity=1 ]  {$\nabla _{\!A_{\scaleto{1\!,\alpha }{3pt}}\! } A_{1\!,\!\alpha }^{+}$};
% Text Node
\draw (192.1,33.37) node [anchor=north west][inner sep=0.75pt]  [font=\tiny]  {$0$};
% Text Node
\draw (202.6,71.8) node [anchor=north west][inner sep=0.75pt]  [font=\tiny]  {$0$};
% Text Node
\draw (141.0,57.0) node [anchor=north west][inner sep=0.75pt]  [font=\tiny]  {$c_{1\!,\!\alpha }^{+}$};
% Text Node
\draw (110.5,110.0) node [anchor=north west][inner sep=0.75pt]  [font=\tiny]  {$c^{+\prime }_\alpha$};
% Text Node
\draw (103.9,3.0) node [anchor=north west][inner sep=0.75pt]  [font=\tiny]  {$c^{-1}_\alpha c^{+ }_\alpha c_\alpha$};
\end{tikzpicture}

        \end{array} \qquad\right.\right\} ;
\end{aligned}
\end{equation*}
and where the set $Z_2$ of $2$-simplices is given by composition of gauge transformations up to the datum of a homotopy; and so on for higher simplices.
To obtain the simplicial set of sections of the derived critical locus $\bfR\crit{S}(M)$ on a general formal derived smooth manifold $U$, we must stackify the pre-stack above in the same sense as in remark \ref{rem:res_cot_bundle}.
\end{remark}

\begin{remark}[Intuitive meaning of the physical fields]
The intuitive picture of the \v{C}ech data of the derived critical locus $\bfR\crit{S}(M)$ above can be given as follows.
\begin{itemize}
    \item $0$-simplices:
    \begin{itemize}
        \item $g_{\alpha\beta}$ transition functions,
        \item $A_\alpha$ connection,
        \item $A^+_\alpha$ equations of motion,
        \item $c^+_\alpha$ Noether identities,
    \end{itemize}
    \item $1$-simplices:
    \begin{itemize}
        \item $c_\alpha$ gauge transformations,
        \item $g_{1,\alpha\beta}$ homotopies of transition functions,
        \item $A_{1,\alpha}$ homotopies of connections,
        \item $A^+_{1,\alpha}$ homotopies of equations of motions,
        \item $c^+_{1,\alpha}$ homotopies of Noether identities,
    \end{itemize}
    \item $(n\geq 2)$-simplices: compositions of gauge transformations and homotopies of homotopies.
\end{itemize}
\end{remark}

From a physical standpoint, we can interpret $A^+$ and $c^+$ as antifield and antighost, respectively.

\begin{remark}[Global antifields and antighosts]
Notice that a section of our formal derived smooth stack $\bfR\crit{S}(M)$ on a formal derived smooth manifold $U\in\mathbf{dFMfd}$ will be of the form $(P,\nabla_{\!A},A^+,c^+)$, where we have the following:
\begin{enumerate}[label=(\textit{\roman*})]
    \item $(P,\nabla_{\!A})$ is a $U$-parametrised family of $G$-bundles on $M$ with connection,
    \item $A^+\in\Omega^{d-1}_\mathrm{ver}(M \times U,\mathfrak{g}_P)_1$ is a $U$-parametrised family of so-called \textit{antifields},
    \item $c^+\in\Omega^{d}_\mathrm{ver}(M \times U,\mathfrak{g}_P)_2$ is a $U$-parametrised family of so-called \textit{antighosts}.
\end{enumerate}
Moreover, notice that the antifields and the antighosts appearing here have a global-geometric structure and, in fact, they are differential forms valued in the adjoint bundle $\mathfrak{g}_P=P\times_{G}\mathfrak{g}$ of the underlying principal $G$-bundle $P$. 
\end{remark}

\begin{remark}[Infinitesimal disk of derived critical locus]
In the special case where $U\simeq\ast$, a section is a point $(P,\nabla_A)\in\bfR\crit{S}(M)$ in the derived critical locus, i.e. a principal $G$-bundle on $M$ with connection which satisfies the Yang-Mills equations of motion.
Recall from section \ref{sec:derdiffcohe} that, in the context of derived differential geometry, we can consider a formal disk $\bbD_{\bfR\crit{S}(M),(P,\nabla_A)}$ of the formal derived smooth stack $\bfR\crit{S}(M)$ at the point $(P,\nabla_A)\in\bfR\crit{S}(M)$, as in definition \ref{def:formal_disk}. Such an formal disk describes the behaviour of the formal derived smooth stack in an infinitesimal neighborhood of the chosen point, where the latter is a global solution of the Yang-Mills equation. This is defined by the pullback square
\begin{equation}
    \begin{tikzcd}[row sep={21ex,between origins}, column sep={24ex,between origins}]
    \bbD_{\bfR\crit{S}(M),(P,\nabla_A)} \arrow[r]\arrow[d, ""'] & \bfR\crit{S}(M) \arrow[d, "\mathfrak{i}_{\bfR\crit{S}(M)}"] \\
    \ast \arrow[r, "{(P,\nabla_A)}"] & \Im\big(\bfR\crit{S}(M)\big) .
\end{tikzcd}
\end{equation}
Since our infinitesimal disk is in fact an infinitesimal object, as we saw above in section \ref{subsec:Loo_algebroids}, it is of the form
\begin{equation}
    \bbD_{\bfR\crit{S}(M),(P,\nabla_A)} \;\simeq\; \mathbf{B}\big(\overrightarrow{\mathfrak{Crit}}(S)_{(P,\nabla_{\!A})}\big), 
\end{equation}
for some  $L_\infty$-algebra $\overrightarrow{\mathfrak{Crit}}(S)_{(P,\nabla_{\!A})}$ which encodes the infinitesimal deformations of the derived critical locus around the fixed point $(P,\nabla_{\!A})\in \bfR\crit{S}(M)$.
By unravelling this $L_\infty$-algebra, we see that its underlying differential graded vector space is given by the cochain complex
\begin{equation*}
\begin{aligned}
    \overrightarrow{\mathfrak{Crit}}(S)_{(P,\nabla_{\!A})_{\!}}[1] \;=\; \,&\Big( \!\!\begin{tikzcd}[row sep={14.5ex,between origins}, column sep= 7ex]
    \Omega^0(M,\mathfrak{g}_P) \arrow[r, "\nabla_{\! A}"] & \Omega^1(M,\mathfrak{g}_P) \arrow[r, "\nabla_{\! A\!}\star_{\!}\nabla_{\! A}"] & \Omega^{d-1}(M,\mathfrak{g}_P) \arrow[r, "\nabla_{\! A}"] & \Omega^{d}(M,\mathfrak{g}_P)
    \end{tikzcd} \!\! \Big)\\
    \!\!{\scriptstyle\text{deg}\,=} &\quad \;\;\begin{tikzcd}[row sep={14.5ex,between origins}, column sep= 6ex]
    \;{\scriptstyle -1} && \quad \;\;\;\;\;{\scriptstyle 0}  && \quad \;\;\;\,\;\;{\scriptstyle 1}  && \quad \;\;\;\,\;\;{\scriptstyle 2} 
    \end{tikzcd}  
\end{aligned} \vspace{-0.2cm}
\end{equation*}
which depends on the point $(P,\nabla_{\!A})\in \bfR\crit{S}(M)$.
Such an $L_\infty$-algebra controls the infinitesimal deformations $\nabla_{\!A}+\vec{A}$ of the fixed connection, together with infinitesimal gauge transformations and equations of motion for the deformed connection.
Thus, not too surprisingly, the $L_\infty$-bracket structure is given as follows:
\begin{equation}
    \begin{aligned}
        \ell_1(\vec{c}) \;&=\; \nabla_{\! A} \vec{c}, \\
        \ell_1(\vec{A}) \;&=\; \nabla_{\! A\!}\star_{\!}\nabla_{\! A} \vec{A}, \qquad &\ell_1(\vec{A}^+) \;&=\; \nabla_{\! A} \vec{A}^+, \\
        \ell_2(\vec{c}_1,\vec{c}_2) \;&=\; [\vec{c}_1,\vec{c}_2]_\mathfrak{g}, \qquad &\ell_2(\vec{c},\vec{c}^+) \;&=\; [\vec{c},\vec{c}^+]_\mathfrak{g}, \\
        \ell_2(\vec{c},\vec{A}) \;&=\; [\vec{c},\vec{A}]_\mathfrak{g}, \qquad &\ell_2(\vec{c},\vec{A}^+) \;&=\; [\vec{c},\vec{A}^+]_\mathfrak{g}, \\
        & &\ell_2(\vec{A},\vec{A}^+) \;&=\; [\vec{A} \,\overset{\wedge}{,}\, \vec{A}^+]_\mathfrak{g},
    \end{aligned}
\end{equation}\vspace{-0.30cm}
\begin{equation*}
    \begin{aligned}
        \ell_2(\vec{A}_1,\vec{A}_2) \;&=\; \nabla_{\!A}\star[\vec{A}_1\,\overset{\wedge}{,}\,\vec{A}_2]_\mathfrak{g} + [\vec{A}_1 \,\overset{\wedge}{,}\, \star_{\!}\nabla_{\!A} \vec{A}_2]_\mathfrak{g} + [\vec{A}_2 \,\overset{\wedge}{,}\, \star_{\!}\nabla_{\!A} \vec{A}_1]_\mathfrak{g}, \\[0.4ex]
        \ell_3(\vec{A}_1,\vec{A}_2,\vec{A}_3) \;&=\; \big[\vec{A}_1\,\overset{\wedge}{,}\,\star [\vec{A}_2\,\overset{\wedge}{,}\,\vec{A}_3]_\mathfrak{g}\big]_\mathfrak{g} + \big[\vec{A}_2\,\overset{\wedge}{,}\,\star [\vec{A}_3\,\overset{\wedge}{,}\,\vec{A}_1]_\mathfrak{g}\big]_\mathfrak{g} + \big[\vec{A}_3\,\overset{\wedge}{,}\,\star [\vec{A}_1\,\overset{\wedge}{,}\,\vec{A}_2]_\mathfrak{g}\big]_\mathfrak{g}, 
    \end{aligned}
\end{equation*}
for any $\vec{c}_k\in\Omega^0(M,\mathfrak{g}_P)$, $\vec{A}_k\in\Omega^1(M,\mathfrak{g}_P)$, $\vec{A}^+_k\in\Omega^{d-1}(M,\mathfrak{g}_P)$ and $\vec{c}^+_k\in\Omega^d(M,\mathfrak{g}_P)$ elements of the underlying graded vector space.
Notice that, if we pick a $G$-bundle $(P,\nabla_{\!A})\in \bfR\crit{S}(M)$ which is topologically trivial $P\simeq M\times G$ and has flat connection $\nabla_{\!A} = \di$, we recover the $L_\infty$-algebra structure from equation \eqref{eq:BV-yang-Mills}. Thus, usual BV-BRST theory can be understood as the infinitesimal disk $\bbD_{\bfR\crit{S}(M),(M\times G,\di)}$ at the trivial $G$-bundle with flat connection, which is in fact a solution of the Yang-Mills equations.
\end{remark}

To conclude this section, we will examine the smooth stack of solutions of Yang-Mills theory, i.e. the underived critical locus $\crit{S}(M)\in\mathbf{SmoothStack}$, seen as a smooth stack that can be obtained by underived truncation of the derived critical locus $\bfR\crit{S}(M)$.

\begin{remark}[Underived critical locus]
Let the underived critical locus be the smooth stack given by the underived truncation $\crit{S}(M) \coloneqq {t}_0\bfR\crit{S}(M)$.
Such a smooth stack will come equipped with a canonical morphism $\crit{S}(M) \longhookrightarrow \Bun_G^\nabla(M)$ of smooth stacks and, roughly speaking, $\crit{S}(M)$ will include only those principal $G$-bundles on $M$ with connection such that they satisfy the Yang-Mills equations of motion.
Thus, any principal $G$-bundle with connection $(P,\nabla_A)\in \crit{S}(M)$ will satisfy by construction both the Bianchi identities and the Yang-Mills equations of motion:
\begin{equation*}
\begin{aligned}
    \nabla_{\!A}F_A \;&=\; 0 && \text{(Bianchi identity)},\\[0.5ex]
    \nabla_{\!A}\star\! F_A \;&=\; 0 && \text{(Equations of motion)},
\end{aligned}
\end{equation*}
where $F_A\in\Omega^2(M,\mathfrak{g}_P)$ is the curvature of the bundle $(P,\nabla_A)$.
A subtlety is that, in $\crit{S}(M)$, Noether identities are not anymore simplicially unravelled, but they are imposed on the nose.
More concretely, if we pick an ordinary smooth manifold $U\in\Mfd$ diffeomorphic to a Cartesian space, we can concretely write the smooth stack $\crit{S}(M)$ by the $2$-coskeletal simplicial set 
\begin{equation*}
    \Hom\big(U,\,\crit{S}(M)\big) \,\simeq\, \mathrm{cosk}_{2\!}\left(\! \begin{tikzcd}[row sep={22.5ex,between origins}, column sep={10.5ex}]
    Z_2 \, \arrow[rr, yshift=3.4ex , "{\big(c_\alpha,\,\begin{smallmatrix}g_{\alpha\beta},A_\alpha\\g_{\alpha\beta}',A_\alpha'\end{smallmatrix}\big)}"] \arrow[rr, "{\big(c_\alpha',\,\begin{smallmatrix}g_{\alpha\beta}',A_\alpha'\\g_{\alpha\beta}'',A_\alpha''\end{smallmatrix}\big)}" description] \arrow[rr, yshift=-3.4ex, "{\big(c_\alpha'\cdot c_\alpha,\,\begin{smallmatrix}g_{\alpha\beta},A_\alpha\\g_{\alpha\beta}'',A_\alpha''\end{smallmatrix}\big)}"'] && \,Z_1\, \arrow[r, yshift=1.7ex , "{(g_{\alpha\beta},A_\alpha)}"] \arrow[r, yshift=-1.7ex, "{(g_{\alpha\beta}',A_\alpha')}"'] & \,Z_0
    \end{tikzcd}\!\right),
\end{equation*}
where the sets of $0$- and $1$-simplices are, respectively, given by
\begin{equation*}
\begin{aligned}
       \hspace{-1.05cm} Z_0 \,&=\, \left\{ \begin{array}{ll}
        g_{\alpha\beta} &\!\!\!\!\in\Coo(V_\alpha\cap V_\beta \times U,G) \\[0.8ex]
         A_\alpha &\!\!\!\!\in\Omega^1_\mathrm{ver}(V_\alpha \times U,\mathfrak{g}) \end{array}  \;\left|\; \begin{array}{l} 
         g_{\alpha\beta}\cdot g_{\beta\gamma}\cdot g_{\gamma\alpha} = 1 \\[0.8ex]
        A_{\alpha}=g_{\beta\alpha}^{-1}(A_{\beta}+\di)g_{\beta\alpha}\\[0.8ex]
        \nabla_{\!A_\alpha\!}\star \!F_{\!A_\alpha}  = \,0 \end{array} \right.\right\},
\end{aligned}
\end{equation*}
\begin{equation*}
\begin{aligned}
        Z_1 \,&=\, \left\{ \begin{array}{ll}
         c_\alpha &\!\!\!\!\in\Coo(V_\alpha \times U,G)  \\[1.5ex]
         g_{\alpha\beta},g'_{\alpha\beta}  &\!\!\!\!\in\Coo(V_\alpha\cap V_\beta \times U ,G) \\[0.5ex]
         A_\alpha, A_\alpha' &\!\!\!\!\in\Omega^1_\mathrm{ver}(V_\alpha \times U,\mathfrak{g})\end{array}  \;\left|\; \begin{array}{ll} 
         g_{\alpha\beta}&\!\!\!\!\!\!\cdot\, g_{\beta\gamma}\cdot g_{\gamma\alpha} = 1 \\[0.5ex]
        A_{\alpha}&\!\!\!\!\!\!=\,g_{\beta\alpha}^{-1}(A_{\beta}+\di)g_{\beta\alpha}  \\[0.5ex]
        \nabla_{\!A_\alpha\!}&\!\!\!\!\!\!\star\, \!F_{\!A_\alpha}  = \,0 \\[2ex]
        g_{\alpha\beta}'&\!\!\!\!\cdot\, g_{\beta\gamma}'\cdot g_{\gamma\alpha}' = 1 \\[0.5ex]
        A_{\alpha}'&\!\!\!\!=\,g_{\beta\alpha}^{\prime -1}(A_{\beta}'+\di)g_{\beta\alpha}'  \\[0.5ex]
        \nabla_{\!A_\alpha'\!}&\!\!\!\!\!\!\star\, \!F_{\!A_\alpha'}  = \,0 \\[2ex]
         g_{\alpha\beta}' &\!\!\!\!=\, c_\beta^{-1}g_{\alpha\beta}c_\alpha \\[0.5ex]
        A_{\alpha}' &\!\!\!\!=\, c_\alpha^{-1} (A_\alpha+\di) c_\alpha \end{array} \right.\right\},
\end{aligned}
\end{equation*}
and where the set of $2$-simplices $Z_2$ is simply given by composition of gauge transformations, in analogy with the smooth stack $\Bun_G^\nabla(M)$.
As before, to obtain the $\infty$-groupoid of sections on a generic smooth manifold $U$, we only have to take the homotopy limit over the \v{C}ech nerve $\check{C}(U)_\bullet\rightarrow U$ provided by a good open cover $\coprod_{i\in I\!}U_i \twoheadrightarrow U$.
\end{remark}

%Notice that we can fix a $G$-bundle with connection $(P,\nabla_A)\in\crit{S}(M)$ which satisfies the Yang-Mills equations of motion and consider the fibre $\bfR\crit{S}(M)_{(P,\nabla_A)}$ of the derived critical locus at such a fixed solution. Then, the fibre is defined by pullback square
%\begin{equation}
%    \begin{tikzcd}[row sep={21ex,between origins}, column sep={24ex,between origins}]
%    \bfR\crit{S}(M)_{(P,\nabla_A)} \arrow[r]\arrow[d, ""'] & \bfR\crit{S}(M) \arrow[d, ""] \\
%    \ast \arrow[r, "{(P,\nabla_A)}"] & \hat{\imath}\crit{S}(M) .
%\end{tikzcd}
%\end{equation}
%This is a derived enhancement of the point, i.e. one has $\mathit{\Pi}^\mathrm{dif}\bfR\crit{S}(M)_{(P,\nabla_A)}\simeq \ast$.

%%%%%%%%%%%%%%%%%%%%%%%%%%%%%%%%%%%%%%%%%%%%%%%%%%%%%%%%%%%%%%%%%%%%%%%%%%%%%%%%%%%%%%%%%%%%%
\section{Outlook}

The authors hope that the derived differential topos geometry exhibited in the present paper may prove a useful language for addressing various open problems in QFT. In this final section we will point to some of them.

\paragraph{Non-perturbative BV-quantisation as higher geometric quantisation.}
In the $L_\infty$-algebra formulation of BV-theory, one quantises a field theory by lifting its classical BV-action $S_{\BV}\in\mathcal{O}_{{}_{\!}}\big(T^\vee[-1]X\big)$ to a quantum BV-action $S_{\BV}^\hbar\in\mathcal{O}_{{}_{\!}}\big(T^\vee[-1]X\big)[[\hbar]]$ satisfying the quantum master equation
\begin{equation}
    i\hbar\triangle S_{\BV}^\hbar + \frac{1}{2}\{S_{\BV}^\hbar,S_{\BV}^\hbar\} \;=\; 0,
\end{equation}
where $\triangle$ is the BV-Laplacian.
In fact, see e.g. \cite{FactII}, the introduction of the quantum BV-differential
\begin{equation}
    Q^\hbar_\BV \;\coloneqq\; i\hbar\triangle + \{ S_{\BV}^\hbar ,-\}
\end{equation}
makes the $\mathbb{P}_0$-algebra of observables into a $\mathbb{BD}_0$-algebra (i.e. a Beilinson-Drinfeld algebra), whose structure provides a quantisation of the algebra of observables.

In \cite{FactI} it was also observed that the dg-algebra of quantum observables has an interesting geometric origin. In fact, one can define the Heisenberg algebra  
\begin{equation}
\begin{tikzcd}[column sep={2.2cm,between origins}]
0 \arrow[r] & i\hbar\mathbb{R}[-1] \arrow[r] & \mathfrak{Heis}(X) \arrow[r] & T^\vee[-1]X \arrow[r] & 0,
\end{tikzcd}
\end{equation} 
where the extended bracket is given by the canonical pairing on the $(-1)$-shifted cotangent bundle $T^\vee[-1]X$, i.e. we have $[\upalpha,\upbeta] \coloneqq i\hbar\{\upalpha,\upbeta \}$ for any $\upalpha,\upbeta\in T^\vee[-1]X$.
This is nothing but a degree-shifted version of the ordinary Heisenberg algebra.
Thus, one has that the dg-algebra of functions is $\mathcal{O}_{{}_{\!}}\big(\mathfrak{Heis}(X)\big) \,\simeq\, \mathcal{O}_{{}_{\!}}\big(T^\vee[-1]X\big)[[\hbar]]$, which means that the observables on the Heisenberg algebra are the quantum observables. 

In a certain sense, an ordinary Heisenberg algebra can be thought of as a Lie algebra version of a prequantum $U(1)$-bundle. This suggests an intriguing relation between geometric quantisation and BV-quantisation. Possibly, it suggests that non-perturbative BV-theory may be thought of as a kind of higher geometric quantisation. In an algebraic-geometric context, aspects of such a relation have been investigated by \cite{safronov2020shifted}.
The formalism proposed in the present paper combines global smooth geometry with derived geometry and thus provides a toolbox to study BV-theory as a derived geometric quantisation in a truly non-perturbative sense.
Schematically, one would aim to define a derived prequantum bundle as a lift of the form
\begin{equation}
\begin{tikzcd}[column sep={3.2cm,between origins},row sep={2.2cm,between origins}]
&  \arrow[d, "\mathrm{curv}"] \B U(1)_\mathrm{conn}(-1) \\
\bfR\crit{S}(M) \arrow[r, "{\omega}"]\arrow[ru] & \pmb{\mathcal{A}}^{2}_\mathrm{cl}(-1),
\end{tikzcd}
\end{equation} 
where $\bfR\crit{S}(M)$ is the derived critical locus of our chosen classical field theory on spacetime, $\pmb{\mathcal{A}}^{2}_\mathrm{cl}(-1)$ is the moduli stack of ($-1$)-shifted closed $2$-forms and the stack $\B U(1)_\mathrm{conn}(-1)$ is a well-defined $(-1)$-shifted version of the moduli stack $\B U(1)_\mathrm{conn}$ of $U(1)$-bundles with connection.

\paragraph{Derived $n$-plectic geometry.} 
Interestingly, as explored by \cite{Rog11, SaSza11x, SaSza11, SaSza13, Rog13, FSS13, Sch16, FSS16, BSS16, BS16, Sza19}, the language of $n$-plectic manifolds is a natural setting for higher geometric (pre)quantisation, just as that of ordinary symplectic manifolds is natural for ordinary geometric quantisation.
In higher geometric quantisation of $n$-plectic manifolds, the prequantum bundle of ordinary geometric prequantisation is typically generalised to a bundle $(n-1)$-gerbe \cite{SaSza11}. This procedure can be naturally applied to an $n$-plectic manifold, by finding the bundle $(n-1)$-gerbe whose curvature coincides with the $n$-plectic form. Recent work in this area includes \cite{BSS16, BS16, Sza19, BMS20, Bunk:2021quu}.

It is interesting to consider whether higher geometric quantisation of $n$-plectic manifolds can be generalised to derived smooth geometry. In a paper in preparation, \cite{AlfYouFuture}, we will give a notion of derived $n$-plectic geometry and propose its application to BV-BFV theory.

\begin{figure}[h!]
    \centering
\tikzset{every picture/.style={line width=0.75pt}} %set default line width to 0.75pt        
\begin{tikzpicture}[x=0.75pt,y=0.75pt,yscale=-1,xscale=1]
%Straight Lines [id:da19008437333525485] 
\draw    (80.55,169.25) -- (80.55,55.56) ;
\draw [shift={(80.56,53.56)}, rotate = 90.11] [color={rgb, 255:red, 0; green, 0; blue, 0 }  ][line width=0.75]    (10.93,-3.29) .. controls (6.95,-1.4) and (3.31,-0.3) .. (0,0) .. controls (3.31,0.3) and (6.95,1.4) .. (10.93,3.29)   ;
%Straight Lines [id:da24959390322590602] 
\draw  [dash pattern={on 3.5pt off 2.5pt}]  (370.0,169.75) -- (370.0,54.56) ;
\draw [shift={(370.06,52.56)}, rotate = 90.11] [color={rgb, 255:red, 0; green, 0; blue, 0 }  ][line width=0.75]    (10.93,-3.29) .. controls (6.95,-1.4) and (3.31,-0.3) .. (0,0) .. controls (3.31,0.3) and (6.95,1.4) .. (10.93,3.29)   ;
%Straight Lines [id:da19410687558708206] 
\draw    (289.5,25.58) -- (164,25.58) ;
\draw [shift={(162,25.58)}, rotate = 360] [color={rgb, 255:red, 0; green, 0; blue, 0 }  ][line width=0.75]    (10.93,-3.29) .. controls (6.95,-1.4) and (3.31,-0.3) .. (0,0) .. controls (3.31,0.3) and (6.95,1.4) .. (10.93,3.29)   ;
%Straight Lines [id:da9382203417497463] 
\draw  [dash pattern={on 3.5pt off 2.5pt}]  (289,195.08) -- (163.5,195.08) ;
\draw [shift={(161.5,195.08)}, rotate = 360] [color={rgb, 255:red, 0; green, 0; blue, 0 }  ][line width=0.75]    (10.93,-3.29) .. controls (6.95,-1.4) and (3.31,-0.3) .. (0,0) .. controls (3.31,0.3) and (6.95,1.4) .. (10.93,3.29)   ;
% Text Node
\draw    (0.5,0.46) -- (161.5,0.46) -- (161.5,52.46) -- (0.5,52.46) -- cycle  ;
\draw (81,26.46) node   [align=left] {\begin{minipage}[lt]{106.76pt}\setlength\topsep{0pt}
\begin{center}
Lagrangian classical field theory
\end{center}
\end{minipage}};
% Text Node
\draw    (-0.5,168.96) -- (160.5,168.96) -- (160.5,220.96) -- (-0.5,220.96) -- cycle  ;
\draw (80,194.96) node   [align=left] {\begin{minipage}[lt]{106.76pt}\setlength\topsep{0pt}
\begin{center}
classical BV-theory
\end{center}
\end{minipage}};
% Text Node
\draw    (290.5,-0.54) -- (451.5,-0.54) -- (451.5,51.46) -- (290.5,51.46) -- cycle  ;
\draw (371,25.46) node   [align=left] {\begin{minipage}[lt]{106.76pt}\setlength\topsep{0pt}
\begin{center}
$\displaystyle n$-plectic geometry
\end{center}
\end{minipage}};
% Text Node
\draw  [dash pattern={on 4.5pt off 4.5pt}]  (290,168.96) -- (451,168.96) -- (451,220.96) -- (290,220.96) -- cycle  ;
\draw (370.5,194.96) node   [align=left] {\begin{minipage}[lt]{106.76pt}\setlength\topsep{0pt}
\begin{center}
\textbf{derived }\\$\displaystyle n$\textbf{-plectic geometry}
\end{center}
\end{minipage}};
% Text Node
\draw (65.83,146.5) node [anchor=north west][inner sep=0.75pt]  [font=\footnotesize,rotate=-270] [align=left] {truncation};
% Text Node
\draw (355.83,146) node [anchor=north west][inner sep=0.75pt]  [font=\footnotesize,rotate=-270] [align=left] {\phantom{0}truncation};
% Text Node
\draw (191,11.5) node [anchor=north west][inner sep=0.75pt]  [font=\footnotesize] [align=left] {\phantom{0}transgression};
% Text Node
\draw (191,181) node [anchor=north west][inner sep=0.75pt]  [font=\footnotesize] [align=left] {transgression};
\end{tikzpicture}
    \caption{Derived $n$-plectic geometry would complete this diagram of formalisms. Just like by transgressing ordinary $n$-plectic geometry one obtains Lagrangian classical field theory, by transgressing derived $n$-plectic geometry one recovers classical BV-theory. By underived-truncation, one gets ordinary $n$-plectic geometry.}
    \label{fig:outro}
\end{figure}
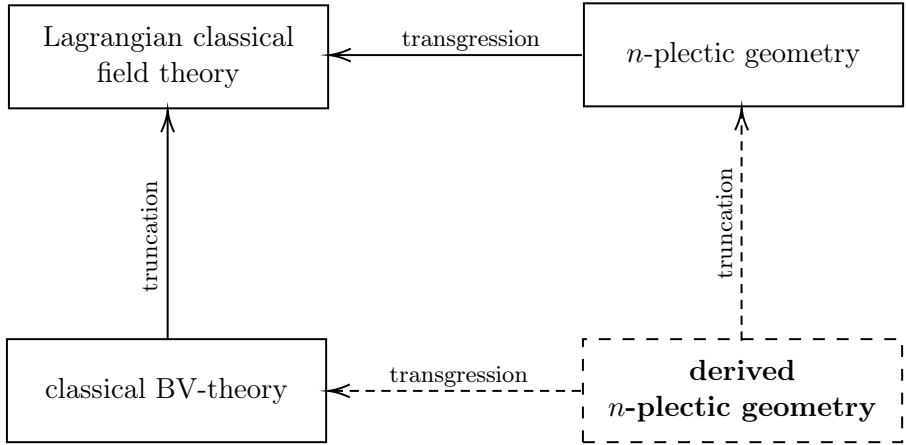

\paragraph{Non-perturbative aspects of string dualities.}
In recent years, in String Theory, there has been an increasing understanding of string dualities in terms of higher principal bundles \cite{DesSae18, DesSae18x, NikWal18, DesSae19, Alf19, Alfonsi:2021ymc, Alfonsi:2021uwh, Alfonsi:2021bot, Kim:2022opr}. This line of research is rooted in the seminal work by \cite{Hull06, Hul06x, BelHulMin07} on duality-covariant string theories.
A geometric T-duality is, roughly speaking, as follows.
First, consider two $T^n$-bundles ${\pi}:M\rightarrow M_0$ and ${\widetilde{\pi}}:\widetilde{M}\rightarrow M_0$ over a common base manifold $M_0$. Then, consider a couple of bundle gerbes $\Pi:\mathscr{G}\rightarrow{}M$ and $\widetilde{\Pi}:\widetilde{\mathscr{G}}\rightarrow{}\widetilde{M}$ respectively on manifolds $M$ and $\widetilde{M}$.
Then, these two bundle gerbes are geometric T-dual if there exists an equivalence
\begin{equation}
    \begin{tikzcd}[row sep={9ex,between origins}, column sep={9.5ex,between origins}]
    & \mathscr{G}\times_{M_0} \widetilde{M}\arrow[rr, "\simeq"]\arrow[dr, "\Pi"']\arrow[dl, "\widetilde{\pi}"] & & M\times_{M_0}\widetilde{\mathscr{G}}\arrow[dr, "\pi"']\arrow[dl, "\widetilde{\Pi}"] \\
    \mathscr{G}\arrow[dr, "\Pi"'] & & M\times_{M_0}\widetilde{M}\arrow[dr, "\pi"']\arrow[dl, "\widetilde{\pi}"] & & [-2.5em]\widetilde{\mathscr{G}}\arrow[dl, "\widetilde{\Pi}"] \\
    & M\arrow[dr, "\pi"'] & & \widetilde{M}\arrow[dl, "\widetilde{\pi}"] & \\
    & & M_0 & &
    \end{tikzcd}
\end{equation}
such that it satisfies a certain condition, known as the Poincar\'e condition.
We immediately see that such a formalisation requires the geometry of higher smooth stacks.
However, the notion of geometric T-duality sketched above is only part of the story, because it does not take into account the dynamics of the string.
A seminal characterisation of full T-duality was provided by \cite{AAL94}, who describe it as a canonical transformation (namely, a symplectomorphism) of the phase space $T^\vee_\mathrm{res}[S^1\!,\,M] \;\coloneqq\; [S^1\!,\,T^\vee M]$ of the classical string preserving its Hamiltonian (see \cite{Alfonsi:2021bot} for a geometric discussion).
This fact combined with the discussion of global aspects of BV-theory in section \ref{sec:global_aspects_of_bv_theory} suggests that a formalisation of full T-duality could be achieved by completing the higher geometric picture of the kinematics of the string with the derived geometric picture of its dynamics.
This way, we also open the door for a non-perturbative BV-quantisation of T-duality: this could provide new valuable insights into the quantum behaviour of global string dualities, which is still generally not well understood. \vspace{-0.3cm}

\paragraph{Non-commutative and non-associative string backgrounds.}
A story that is intimately related to string dualities is the appearance of non-associative geometry in the context of open String Theory. This feature of stringy geometry is understood to be linked not only to the non-geometric fluxes typically produced by T-duality \cite{SN1, SN2, SN4}, but also to higher differential geometry \cite{SN3, SN5, SN6, SN7, Sza18}. 
Adding formal derived smooth stacks to the mix to encode the dynamics of the field theories involved, would provide an intriguing overlap between these exotic backgrounds and derived differential geometry.
Moreover, it could bring some global-geometric insight into the rising field of braided QFT \cite{DimitrijevicCiric:2021jea, CiricDimitrijevic:2022eei, DimitrijevicCiric:2023hua}, which is based on BV-formalism via $L_\infty$-algebras. \vspace{-0.3cm}

\paragraph{Global aspects of double copy.}
Double copy is a theory stating that gravitational scattering amplitudes can be obtained from the ones of gauge theory essentially by replacing the colour factor with an extra kinematical factor.
Over the last few years this phenomenon has been understood in the context of BV-BRST theory via $L_\infty$-algebras by
\cite{Borsten:2020zgj, Borsten:2021hua, Macrelli_2022, Borsten:2021ljb, borsten2023treelevel, Borsten_2023}.
Orthogonal to this, a potential global-geometric story for the double copy of classical solutions of the field equations have been investigated by \cite{Alfonsi:2020lub}, but the features and the limits of such a formulation are still not completely clear. 
Derived differential geometry provides a formalism which may allow a theoretical interpolation of these approaches and, thus, a global-geometric BV-BRST treatment of double copy. \vspace{-0.3cm}

\paragraph{M-theory and Hypothesis H.}
Higher geometry and, more specifically, higher geometric quantisation has been used to investigate the underlying geometry of M-theory by \cite{FSS12, FSS15x, FSS19x, FSS19xxx, BSS19, HSS19, FSS19coho}. In these references, \textit{Hypothesis H} was proposed as candidate mathematical formulation of M-theory, whose core statement is that the charge quantisation of the theory is controlled by a non-abelian cohomology theory known as twisted cohomotopy. This idea was then further explored by \cite{BSS18,SS19,FSS19xx, Fiorenza:2020iax, Sati:2020nob, Sati:2021uhj}.
The proposal collected a large number of theoretical achievements including the derivation of a variety of expected anomaly cancellations and, remarkably, a formal description of a multitude of quantum phenomena expected to emerge on high-energy intersecting D-branes.
Given these intriguing results, there has been some recent discussion about what precise role the dynamics should play in the theory, respect to the kinematics. A possible new way to address this question could be the implementation of the dynamical side of the theory by using formal derived smooth stacks representing the derived critical loci of the action functionals, in a way that would generalise and systematise the preliminary aspects discussed in \ref{sec:global_aspects_of_bv_theory}.

%%%%%%%%%%%%%%%%%%%%%%%%%%%%%%%%%%%%%%%%%%%%%%%%%%%%%%%%%%%%%%%%%%%%%%%%%%%%%%%%%%%%%%%%%%%%%
\section*{Acknowledgments}

The authors would like to thank Christian S\"amann, Alexander Schenkel, John Pridham, Severin Bunk and James Waldron for helpful discussions. 

The authors would also like to thank David Carchedi and Pelle Steffens for essential discussion which led to significant improvement of this paper.

LA would like to thank Hisham Sati and Urs Schreiber for helpful discussions. 

The authors gratefully acknowledge the financial support of the Leverhulme Trust, Research Project Grant number RPG-2021-092.

%%%%%%%%%%%%%%%%%%%%%%%%%%%%%%%%%%%%%%%%%%%%%%%%%%%%%%%%%%%%%%%%%%%%%%%%%%%%%%%%%%%%%%%%%%%%%%%%%%%%%%%%%%%%%%%%%%
%%%%%%%%%%%%%%%%%%%%%%%%%%%%%%%%%%%%%%%%%%%%%%%%%%%%%%%%%%%%%%%%%%%%%%%%%%%%%%%%%%%%%%%%%%%%%%%%%%%%%%%%%%%%%%%%%%
%\medskip
\setlength{\baselineskip}{0pt}
\renewcommand*{\bibfont}{\scriptsize}
{\renewcommand*\MakeUppercase[1]{#1}%
\printbibliography[heading=bibintoc,title={\bibtitle}]}

@article{Alf19,
author = {Alfonsi, Luigi},
title = "Global Double Field Theory is Higher Kaluza-Klein Theory",
journal = {Fortschritte der Physik},
volume = {68},
number = {3-4},
pages = {2000010},
doi = {10.1002/prop.202000010},
ISSN={1521-3978},
eprint = {1912.07089},
archivePrefix  = "arXiv",
primaryClass   = "hep-th",
year = {2020},
month={Feb},
publisher={Wiley},
}

@article{Alfonsi:2021uwh,
    author = "Alfonsi, Luigi",
    title = "{Towards an extended/higher correspondence -- Generalised geometry, bundle gerbes and global Double Field Theory}",
    eprint = "2102.10970",
    pages = {302--328},
    volume = {8},
    number = {1},
    journal = {Complex Manifolds},
    doi = {doi:10.1515/coma-2020-0121},
    year = {2021},
    archivePrefix = "arXiv",
    primaryClass = "hep-th",
    reportNumber = "QMUL-PH-21-11",
    doi = "10.1515/coma-2020-0121"
}

@article{Hohm17,
      author         = "Hohm, Olaf and Zwiebach, Barton",
      title          = "{$L_{\infty}$ Algebras and Field Theory}",
      journal        = "Fortsch. Phys.",
      volume         = "65",
      year           = "2017",
      number         = "3-4",
      pages          = "1700014",
      doi            = "10.1002/prop.201700014",
      eprint         = "1701.08824",
      archivePrefix  = "arXiv",
      primaryClass   = "hep-th",
      reportNumber   = "MIT-CTP-4875",
      SLACcitation   = "%%CITATION = ARXIV:1701.08824;%%"
}

@article{Hohm17x,
      author         = "Hohm, Olaf and Kupriyanov, Vladislav and Lust, Dieter and
                        Traube, Matthias",
      title          = "{Constructions of $L_{\infty}$ algebras and their field
                        theory realizations}",
      journal        = "Adv. Math. Phys.",
      volume         = "2018",
      year           = "2018",
      pages          = "9282905",
      doi            = "10.1155/2018/9282905",
      eprint         = "math-ph/1709.10004",
      archivePrefix  = "arXiv",
      primaryClass   = "math-ph",
      reportNumber   = "MPP-2017-209, LMU-ASC 57/17, LMU-ASC-57-17",
      SLACcitation   = "%%CITATION = ARXIV:1709.10004;%%"
}

@misc{Hohm19,
      author         = "Bonezzi, Roberto and Hohm, Olaf",
      title          = "{Duality Hierarchies and Differential Graded Lie
                        Algebras}",
      year           = "2019",
      eprint         = "1910.10399",
      archivePrefix  = "arXiv",
      primaryClass   = "hep-th",
      SLACcitation   = "%%CITATION = ARXIV:1910.10399;%%"
}

@misc{Hohm19x,
      author         = "Bonezzi, Roberto and Hohm, Olaf",
      title          = "{Leibniz Gauge Theories and Infinity Structures}",
      year           = "2019",
      eprint         = "1904.11036",
      archivePrefix  = "arXiv",
      primaryClass   = "hep-th",
      reportNumber   = "HU-EP-19-08",
      SLACcitation   = "%%CITATION = ARXIV:1904.11036;%%"
}

@unpublished{DCCTv2,
      label          = {DCCT},
      author         = "Schreiber, Urs",
      title          = "{Differential cohomology in a cohesive $\infty$-topos}",
      note           = "v2. \newline \href{https://ncatlab.org/schreiber/files/dcct170811.pdf}{{\tt ncatlab.org/schreiber/files/dcct170811.pdf}}"
}

@article{Principal1,
      author         = "Thomas Nikolaus and Urs Schreiber and Danny Stevenson",
      title          = "{Principal $\infty$-bundles - General theory}",
      journal        = "Journal of Homotopy and Related Structures",
      volume         = "10",
      year           = "2015",
      pages          = "749-801",
      doi            = "10.1007/s40062-014-0083-6",
      eprint         = "math.AT/1207.0248",
      archivePrefix  = "arXiv",
      primaryClass   = "math.AT",
      SLACcitation   = "%%CITATION = ARXIV:1207.0248;%%"
}

@article{Principal2,
   title="{Principal $\infty$-bundles: presentations}",
   volume= "10",
   ISSN= "1512-2891",
   url= "http://dx.doi.org/10.1007/s40062-014-0077-4",
   DOI= "10.1007/s40062-014-0077-4",
   number= "3",
   journal= "{ Journal of Homotopy and Related Structures}",
   publisher= "{Springer Science and Business Media LLC}",
   author="Nikolaus, Thomas and Schreiber, Urs and Stevenson, Danny",
   year= "2014",
   month= "02",
   pages= "565–622"
}

@article{Rog13,
      author         = "Rogers, Christopher L.",
      title          = "{2-plectic geometry, Courant algebroids, and categorified
                        prequantization}",
      journal        = "J. Sympl. Geom.",
      volume         = "11",
      year           = "2013",
      pages          = "53-91",
      doi            = "10.4310/JSG.2013.v11.n1.a4",
      eprint         = "1009.2975",
      archivePrefix  = "arXiv",
      primaryClass   = "math-ph",
      SLACcitation   = "%%CITATION = ARXIV:1009.2975;%%",
}

@article{NikWal18,
      author         = "Nikolaus, Thomas and Waldorf, Konrad",
        ISSN={1432-0916},
     url={http://dx.doi.org/10.1007/s00220-019-03496-3},
      DOI={10.1007/s00220-019-03496-3},
     journal={Communications in Mathematical Physics},
     publisher={Springer Science and Business Media LLC},
     year={2019},
     month={Jun},
      title          = "{Higher geometry for non-geometric T-duals}",
      eprint         = "math.AT/1804.00677",
      archivePrefix  = "arXiv",
      primaryClass   = "math.AT",
      SLACcitation   = "%%CITATION = ARXIV:1804.00677;%%"
}

@inproceedings{Rog11,
    title={Higher Symplectic Geometry},
    author={Christopher L. Rogers},
    year={2011},
    eprint={1106.4068},
    archivePrefix={arXiv},
    primaryClass={math-ph}
}

@book{Pau14,
  author    = "Paugam, Frederic", 
  title     = "{Towards the Mathematics of Quantum Field Theory}",
  publisher = "{Springer}",
  year      = "2014",
}

@inproceedings{SaSza11,
      author         = "Saemann, Christian and Szabo, Richard J.",
      title          = "{Quantization of 2-Plectic Manifolds}",
      booktitle      = "{Proceedings, 4th Annual Meeting of the European Non
                        Commutative Geometry Network (EU-NCG): Progress in
                        Operator Algebras, Noncommutative Geometry, and their
                        Applications: Bucharest, Romania, April 25-30, 2011}",
      year           = "2011",
      eprint         = "1106.1890",
      archivePrefix  = "arXiv",
      primaryClass   = "hep-th",
      reportNumber   = "HWM-11-13, EMPG-11-17",
      SLACcitation   = "%%CITATION = ARXIV:1106.1890;%%"
}

@article{SaSza11x,
      author         = "Saemann, Christian and Szabo, Richard J.",
      title          = "{Groupoid Quantization of Loop Spaces}",
      booktitle      = "{Proceedings, 11th Hellenic School and Workshops on
                        Elementary Particle Physics and Gravity (CORFU2011):
                        Corfu, Greece, September 4-18, 2011}",
      journal        = "PoS",
      volume         = "CORFU2011",
      year           = "2011",
      pages          = "046",
      doi            = "10.22323/1.155.0046",
      eprint         = "1203.5921",
      archivePrefix  = "arXiv",
      primaryClass   = "hep-th",
      reportNumber   = "HWM-12-02, EMPG-12-03",
      SLACcitation   = "%%CITATION = ARXIV:1203.5921;%%"
}

@article{SaSza13,
      author         = "Saemann, Christian and Szabo, Richard J.",
      title          = "{Groupoids, Loop Spaces and Quantization of 2-Plectic
                        Manifolds}",
      journal        = "Rev. Math. Phys.",
      volume         = "25",
      year           = "2013",
      pages          = "1330005",
      doi            = "10.1142/S0129055X13300057",
      eprint         = "1211.0395",
      archivePrefix  = "arXiv",
      primaryClass   = "hep-th",
      reportNumber   = "EMPG-12-24, NI-12034-BSM",
      SLACcitation   = "%%CITATION = ARXIV:1211.0395;%%"
}

@article{BSS16,
      author         = "Bunk, Severin and Saemann, Christian and Szabo, Richard
                        J.",
      title          = "{The 2-Hilbert Space of a Prequantum Bundle Gerbe}",
      journal        = "Rev. Math. Phys.",
      volume         = "30",
      year           = "2017",
      number         = "01",
      pages          = "1850001",
      doi            = "10.1142/S0129055X18500010",
      eprint         = "1608.08455",
      archivePrefix  = "arXiv",
      primaryClass   = "math-ph",
      reportNumber   = "EMPG-16-16",
      SLACcitation   = "%%CITATION = ARXIV:1608.08455;%%"
}

@article{BS16,
      author         = "Bunk, Severin and Szabo, Richard J.",
      title          = "{Fluxes, bundle gerbes and 2-Hilbert spaces}",
      journal        = "Lett. Math. Phys.",
      volume         = "107",
      year           = "2017",
      number         = "10",
      pages          = "1877-1918",
      doi            = "10.1007/s11005-017-0971-x",
      eprint         = "1612.01878",
      archivePrefix  = "arXiv",
      primaryClass   = "hep-th",
      reportNumber   = "EMPG-16-19",
      SLACcitation   = "%%CITATION = ARXIV:1612.01878;%%"
}

@article{Sza18,
    author = "Szabo, Richard J.",
    title = "{Higher Quantum Geometry and Non-Geometric String Theory}",
    eprint = "1803.08861",
    archivePrefix = "arXiv",
    primaryClass = "hep-th",
    reportNumber = "EMPG-18-08",
    doi = "10.22323/1.318.0151",
    journal = "PoS",
    volume = "CORFU2017",
    pages = "151",
    year = "2018"
}

@article{Sza19,
      author         = "Bunk, Severin and Müller , Lukas and Szabo, Richard J.",
      title          = "{Geometry and 2-Hilbert Space for Nonassociative Magnetic
                        Translations}",
      journal        = "Lett. Math. Phys.",
      volume         = "109",
      year           = "2019",
      number         = "8",
      pages          = "1827-1866",
      doi            = "10.1007/s11005-019-01160-4",
      eprint         = "1804.08953",
      archivePrefix  = "arXiv",
      primaryClass   = "hep-th",
      reportNumber   = "Hamburger Beitrage zur Mathematik Nr. 728, ZMP-HH/18-9,
                        EMPG-18-09, HAMBURGER-BEITRAGE-ZUR-MATHEMATIK-NR.-728,
                        ZMP-HH-18-9",
      SLACcitation   = "%%CITATION = ARXIV:1804.08953;%%"
}

@article{SN1,
      author         = "Mylonas, Dionysios and Schupp, Peter and Szabo, Richard
                        J.",
      title          = "{Non-Geometric Fluxes, Quasi-Hopf Twist Deformations and
                        Nonassociative Quantum Mechanics}",
      journal        = "J. Math. Phys.",
      volume         = "55",
      year           = "2014",
      pages          = "122301",
      doi            = "10.1063/1.4902378",
      eprint         = "1312.1621",
      archivePrefix  = "arXiv",
      primaryClass   = "hep-th",
      reportNumber   = "EMPG-13-20",
      SLACcitation   = "%%CITATION = ARXIV:1312.1621;%%"
}

@article{SN2,
      author         = "Mylonas, Dionysios and Schupp, Peter and Szabo, Richard
                        J.",
      title          = "{Nonassociative geometry and twist deformations in
                        non-geometric string theory}",
      booktitle      = "{Proceedings, 3rd International Satellite Conference on
                        Mathematical Methods in Physics (ICMP13): Londrina,
                        Brazil, October 21-26, 2013}",
      journal        = "PoS",
      volume         = "ICMP2013",
      year           = "2013",
      pages          = "007",
      doi            = "10.22323/1.200.0007",
      eprint         = "1402.7306",
      archivePrefix  = "arXiv",
      primaryClass   = "hep-th",
      reportNumber   = "EMPG-14-5",
      SLACcitation   = "%%CITATION = ARXIV:1402.7306;%%"
}

@article{SN3,
      author         = "Barnes, Gwendolyn E. and Schenkel, Alexander and Szabo,
                        Richard J.",
      title          = "{Nonassociative geometry in quasi-Hopf representation
                        categories I: Bimodules and their internal homomorphisms}",
      journal        = "J. Geom. Phys.",
      volume         = "89",
      year           = "2014",
      pages          = "111-152",
      doi            = "10.1016/j.geomphys.2014.12.005",
      eprint         = "1409.6331",
      archivePrefix  = "arXiv",
      primaryClass   = "math.QA",
      reportNumber   = "EMPG-14-17",
      SLACcitation   = "%%CITATION = ARXIV:1409.6331;%%"
}

@article{SN4,
      author         = "Aschieri, Paolo and Szabo, Richard J.",
      title          = "{Triproducts, nonassociative star products and geometry
                        of R-flux string compactifications}",
      booktitle      = "{Proceedings, Conceptual and Technical Challenges for
                        Quantum Gravity 2014 – Parallel session: Noncommutative
                        Geometry and Quantum Gravity: Rome, Italy, September
                        8-12,2014}",
      journal        = "J. Phys. Conf. Ser.",
      volume         = "634",
      year           = "2015",
      number         = "1",
      pages          = "012004",
      doi            = "10.1088/1742-6596/634/1/012004",
      eprint         = "1504.03915",
      archivePrefix  = "arXiv",
      primaryClass   = "hep-th",
      reportNumber   = "EMPG-15-06",
      SLACcitation   = "%%CITATION = ARXIV:1504.03915;%%"
}

@article{SN5,
      author         = "Barnes, Gwendolyn E. and Schenkel, Alexander and Szabo,
                        Richard J.",
      title          = "{Nonassociative geometry in quasi-Hopf representation
                        categories II: Connections and curvature}",
      journal        = "J. Geom. Phys.",
      volume         = "106",
      year           = "2016",
      pages          = "234-255",
      doi            = "10.1016/j.geomphys.2016.04.005",
      eprint         = "1507.02792",
      archivePrefix  = "arXiv",
      primaryClass   = "math.QA",
      reportNumber   = "EMPG-15-10",
      SLACcitation   = "%%CITATION = ARXIV:1507.02792;%%"
}

@article{SN6,
      author         = "Barnes, Gwendolyn E. and Schenkel, Alexander and Szabo,
                        Richard J.",
      title          = "{Working with Nonassociative Geometry and Field Theory}",
      booktitle      = "{Proceedings, 15th Hellenic School and Workshops on
                        Elementary Particle Physics and Gravity (CORFU2015):
                        Corfu, Greece, September 1-25, 2015}",
      journal        = "PoS",
      volume         = "CORFU2015",
      year           = "2016",
      pages          = "081",
      doi            = "10.22323/1.263.0081",
      eprint         = "1601.07353",
      archivePrefix  = "arXiv",
      primaryClass   = "hep-th",
      reportNumber   = "EMPG-16-02",
      SLACcitation   = "%%CITATION = ARXIV:1601.07353;%%"
}

@article{SN7,
      author         = "Aschieri, Paolo and Dimitrijevic Ciric, Marija and
                        Szabo, Richard J.",
      title          = "{Nonassociative differential geometry and gravity with
                        non-geometric fluxes}",
      journal        = "JHEP",
      volume         = "02",
      year           = "2018",
      pages          = "036",
      doi            = "10.1007/JHEP02(2018)036",
      eprint         = "1710.11467",
      archivePrefix  = "arXiv",
      primaryClass   = "hep-th",
      reportNumber   = "EMPG-17-16",
      SLACcitation   = "%%CITATION = ARXIV:1710.11467;%%"
}

@incollection{Sch16,
      author         = "Schreiber, Urs",
      title          = "{Higher prequantum geometry}",
      year           = "2016",
      eprint         = "1601.05956",
      archivePrefix  = "arXiv",
      primaryClass   = "math-ph",
      SLACcitation   = "%%CITATION = ARXIV:1601.05956;%%"
}

@inproceedings{FSS13,
      author         = "Fiorenza, Domenico and Sati, Hisham and Schreiber, Urs",
      title          = "{A higher stacky perspective on Chern-Simons theory}",
      booktitle      = "{Proceedings, Winter School in Mathematical Physics:
                        Mathematical Aspects of Quantum Field Theory: Les Houches,
                        France, January 29-February 3, 2012}",
      organization   = "Springer",
      publisher      = "Springer",
      year           = "2015",
      pages          = "153-211",
      doi            = "10.1007/978-3-319-09949-1_6",
      eprint         = "1301.2580",
      archivePrefix  = "arXiv",
      primaryClass   = "hep-th",
      SLACcitation   = "%%CITATION = ARXIV:1301.2580;%%"
}

@article{FSS16,
      author         = "Fiorenza, Domenico and Rogers, Christopher L. and
                        Schreiber, Urs",
      title          = "{Higher $U(1)$-gerbe connections in geometric
                        prequantization}",
      journal        = "Rev. Math. Phys.",
      volume         = "28",
      year           = "2016",
      number         = "06",
      pages          = "1650012",
      doi            = "10.1142/S0129055X16500124",
      eprint         = "1304.0236",
      archivePrefix  = "arXiv",
      primaryClass   = "math-ph",
      SLACcitation   = "%%CITATION = ARXIV:1304.0236;%%"
}

@article{FSS19x,
      author         = "Fiorenza, Domenico and Sati, Hisham and Schreiber, Urs",
      title          = "{Super-exceptional geometry: origin of heterotic M-theory
                        and super-exceptional embedding construction of M5}",
      year           = "2019",
      eprint         = "1908.00042",
      archivePrefix  = "arXiv",
      primaryClass   = "hep-th",
      SLACcitation   = "%%CITATION = ARXIV:1908.00042;%%"
}

@article{FSS15x,
      author         = "Fiorenza, Domenico and Sati, Hisham and Schreiber, Urs",
      title          = "{The Wess-Zumino-Witten term of the M5-brane and
                        differential cohomotopy}",
      journal        = "J. Math. Phys.",
      volume         = "56",
      year           = "2015",
      number         = "10",
      pages          = "102301",
      doi            = "10.1063/1.4932618",
      eprint         = "1506.07557",
      archivePrefix  = "arXiv",
      primaryClass   = "math-ph",
      SLACcitation   = "%%CITATION = ARXIV:1506.07557;%%"
}

@article{FSS12,
      author         = "Fiorenza, Domenico and Sati, Hisham and Schreiber, Urs",
      title          = "{Multiple M5-branes, String 2-connections, and 7d
                        nonabelian Chern-Simons theory}",
      journal        = "Adv. Theor. Math. Phys.",
      volume         = "18",
      year           = "2014",
      number         = "2",
      pages          = "229-321",
      doi            = "10.4310/ATMP.2014.v18.n2.a1",
      eprint         = "1201.5277",
      archivePrefix  = "arXiv",
      primaryClass   = "hep-th",
      SLACcitation   = "%%CITATION = ARXIV:1201.5277;%%"
}

@misc{SS19,
      author         = "Sati, Hisham and Schreiber, Urs",
      title          = "{Equivariant Cohomotopy implies orientifold tadpole
                        cancellation}",
      year           = "2019",
      eprint         = "1909.12277",
      archivePrefix  = "arXiv",
      primaryClass   = "hep-th",
      SLACcitation   = "%%CITATION = ARXIV:1909.12277;%%"
}

@article{FSS19xx,
      author         = "Fiorenza, Domenico and Sati, Hisham and Schreiber, Urs",
      title          = "{Twisted Cohomotopy implies level quantization of the
                        full 6d Wess-Zumino term of the M5-brane}",
      year           = "2019",
      eprint         = "1906.07417",
      archivePrefix  = "arXiv",
      primaryClass   = "hep-th",
      SLACcitation   = "%%CITATION = ARXIV:1906.07417;%%"
}

@article{FSS19xxx,
      author         = "Fiorenza, Domenico and Sati, Hisham and Schreiber, Urs",
      title          = "{The Rational Higher Structure of M‐theory}",
      booktitle      = "{Durham Symposium, Higher Structures in M-Theory Durham,
                        UK, August 12-18, 2018}",
      journal        = "Fortsch. Phys.",
      volume         = "67",
      year           = "2019",
      number         = "8-9",
      pages          = "1910017",
      doi            = "10.1002/prop.201910017",
      eprint         = "1903.02834",
      archivePrefix  = "arXiv",
      primaryClass   = "hep-th",
      SLACcitation   = "%%CITATION = ARXIV:1903.02834;%%"
}

@article{BSS19,
      author         = "Braunack-Mayer, Vincent and Sati, Hisham and Schreiber,
                        Urs",
      title          = "{Gauge enhancement of super M-branes via parametrized
                        stable homotopy theory}",
      journal        = "Commun. Math. Phys.",
      volume         = "371",
      year           = "2019",
      number         = "1",
      pages          = "197-265",
      doi            = "10.1007/s00220-019-03441-4",
      eprint         = "1806.01115",
      archivePrefix  = "arXiv",
      primaryClass   = "hep-th",
      SLACcitation   = "%%CITATION = ARXIV:1806.01115;%%"
}

@article{HSS19,
      author         = "Huerta, John and Sati, Hisham and Schreiber, Urs",
      title          = "{Real ADE-equivariant (co)homotopy and Super M-branes}",
      journal        = "Commun. Math. Phys.",
      volume         = "371",
      year           = "2019",
      number         = "2",
      pages          = "425-524",
      doi            = "10.1007/s00220-019-03442-3",
      eprint         = "1805.05987",
      archivePrefix  = "arXiv",
      primaryClass   = "hep-th",
      SLACcitation   = "%%CITATION = ARXIV:1805.05987;%%"
}

@article{FSS19coho,
      author         = "Fiorenza, Domenico and Sati, Hisham and Schreiber, Urs",
      title          = "{Twisted Cohomotopy implies M-Theory anomaly
                        cancellation}",
      year           = "2019",
      eprint         = "1904.10207",
      archivePrefix  = "arXiv",
      primaryClass   = "hep-th",
      SLACcitation   = "%%CITATION = ARXIV:1904.10207;%%"
}

@misc{BSS18,
    title= "{Lift of fractional D-brane charge to equivariant Cohomotopy theory}",
    author= "Simon Burton and Hisham Sati and Urs Schreiber",
    year= "2018",
    eprint= "1812.09679",
    archivePrefix= "arXiv",
    primaryClass= "math.RT"
}

@article{Hul06x,
      author         = "Hull, C. M.",
      title          = "{Global aspects of T-duality, gauged sigma models and
                        T-folds}",
      journal        = "JHEP",
      volume         = "10",
      year           = "2007",
      pages          = "057",
      doi            = "10.1088/1126-6708/2007/10/057",
      eprint         = "hep-th/0604178",
      archivePrefix  = "arXiv",
      primaryClass   = "hep-th",
      reportNumber   = "IMPERIAL-TP-06-CH-01",
      SLACcitation   = "%%CITATION = HEP-TH/0604178;%%"
}

@article{Hull06,
      author         = "Hull, C. M.",
      title          = "{Doubled Geometry and T-Folds}",
      journal        = "JHEP",
      volume         = "07",
      year           = "2007",
      pages          = "080",
      doi            = "10.1088/1126-6708/2007/07/080",
      eprint         = "hep-th/0605149",
      archivePrefix  = "arXiv",
      primaryClass   = "hep-th",
      reportNumber   = "IMPERIAL-TP-06-CH-02",
      SLACcitation   = "%%CITATION = HEP-TH/0605149;%%"
}

@misc{BelHulMin07,
      author         = "Belov, Dmitriy M. and Hull, Chris M. and Minasian, Ruben",
      title          = "{T-duality, gerbes and loop spaces}",
      year           = "2007",
      eprint         = "hep-th/0710.5151",
      archivePrefix  = "arXiv",
      primaryClass   = "hep-th",
      reportNumber   = "IMPERIAL-TP-07-DMB-01, SPHT-T07-134",
      SLACcitation   = "%%CITATION = ARXIV:0710.5151;%%"
}

@article{Saem18bv,
      author         = "Jurco, Branislav and Raspollini, Lorenzo and Saemann,
                        Christian and Wolf, Martin",
      title          = "{$L_\infty$-Algebras of Classical Field Theories and the
                        Batalin-Vilkovisky Formalism}",
      journal        = "Fortsch. Phys.",
      volume         = "67",
      year           = "2019",
      number         = "7",
      pages          = "1900025",
      doi            = "10.1002/prop.201900025",
      eprint         = "1809.09899",
      archivePrefix  = "arXiv",
      primaryClass   = "hep-th",
      reportNumber   = "EMPG-18-19, DMUS-MP-18/05",
      SLACcitation   = "%%CITATION = ARXIV:1809.09899;%%"
}

@article{Saem19bv,
      author         = "Jurco, Branislav and Macrelli, Tommaso and Raspollini,
                        Lorenzo and Saemann, Christian and Wolf, Martin",
      title          = "{$L_\infty$-Algebras, the BV Formalism, and Classical
                        Fields}",
      booktitle      = "{Durham Symposium, Higher Structures in M-Theory Durham,
                        UK, August 12-18, 2018}",
      journal        = "Fortsch. Phys.",
      volume         = "67",
      year           = "2019",
      number         = "8-9",
      pages          = "1910025",
      doi            = "10.1002/prop.201910025",
      eprint         = "1903.02887",
      archivePrefix  = "arXiv",
      primaryClass   = "hep-th",
      SLACcitation   = "%%CITATION = ARXIV:1903.02887;%%"
}

@article{Jurco:2020yyu,
    author = "Jurco, Branislav and Kim, Hyungrok and Macrelli, Tommaso and Saemann, Christian and Wolf, Martin",
    title = "{Perturbative Quantum Field Theory and Homotopy Algebras}",
    eprint = "2002.11168",
    archivePrefix = "arXiv",
    primaryClass = "hep-th",
    reportNumber = "DMUS-MP-20/01, EMPG-20-06",
    doi = "10.22323/1.376.0199",
    journal = "PoS",
    volume = "CORFU2019",
    pages = "199",
    year = "2020"
}

@article{Jurco:2019yfd,
    author = {Jur\v{c}o, Branislav and Macrelli, Tommaso and S\"amann, Christian and Wolf, Martin},
    title = "{Loop Amplitudes and Quantum Homotopy Algebras}",
    eprint = "1912.06695",
    archivePrefix = "arXiv",
    primaryClass = "hep-th",
    reportNumber = "EMPG-19-26, DMUS-MP-19/11",
    doi = "10.1007/JHEP07(2020)003",
    journal = "JHEP",
    volume = "07",
    pages = "003",
    year = "2020"
}

@article{DesSae18,
      author         = "Deser, Andreas and Saemann, Christian",
      title          = "{Extended Riemannian Geometry I: Local Double Field
                        Theory}",
      volume         = "19",
      year           = "2018",
      number         = "8",
      pages          = "2297-2346",
      doi            = "10.1007/s00023-018-0694-2",
      eprint         = "1611.02772",
      archivePrefix  = "arXiv",
      primaryClass   = "hep-th",
      reportNumber   = "ITP-UH-22-16, EMPG-16-18",
      SLACcitation   = "%%CITATION = ARXIV:1611.02772;%%"
}

@article{DesSae18x,
      author         = "Deser, Andreas and Heller, Marc Andre and Saemann,
                        Christian",
      title          = "{Extended Riemannian Geometry II: Local Heterotic Double
                        Field Theory}",
      journal        = "JHEP",
      volume         = "04",
      year           = "2018",
      pages          = "106",
      doi            = "10.1007/JHEP04(2018)106",
      eprint         = "1711.03308",
      archivePrefix  = "arXiv",
      primaryClass   = "hep-th",
      reportNumber   = "EMPG-17-18",
      SLACcitation   = "%%CITATION = ARXIV:1711.03308;%%"
}

@article{DesSae19,
      author         = "Deser, Andreas and Saemann, Christian",
      title          = "{Extended Riemannian Geometry III: Global Double Field
                        Theory with Nilmanifolds}",
      journal        = "JHEP",
      volume         = "05",
      year           = "2019",
      pages          = "209",
      doi            = "10.1007/JHEP05(2019)209",
      eprint         = "1812.00026",
      archivePrefix  = "arXiv",
      primaryClass   = "hep-th",
      reportNumber   = "EMPG-18-24",
      SLACcitation   = "%%CITATION = ARXIV:1812.00026;%%"
}

@book{topos,
    title={Higher Topos Theory},
     ISBN = {9780691140490},
     publisher = {Princeton University Press},
    author={Jacob Lurie},
    year={2006},
    eprint={math/0608040},
    archivePrefix={arXiv},
    primaryClass={math.CT}
}

@article{BMS20,
    author = "{Bunk, Severin and M\"{u}ller, Lukas and Szabo, Richard J.}",
    title = "{Smooth 2-Group Extensions and Symmetries of Bundle Gerbes}",
    eprint = "2004.13395",
    archivePrefix = "arXiv",
    primaryClass = "math.DG",
    reportNumber = "Hamburger Beitrage zur Mathematik Nr. 834, ZMP-HH/20-10, EMPG-20-08",
    doi = "10.1007/s00220-021-04099-7",
    journal = "Commun. Math. Phys.",
    volume = "384",
    number = "3",
    pages = "1829--1911",
    year = "2021"
}

@article{FRS18,
    author = "Fiorenza, Domenico and Rogers, Christopher L. and Schreiber, Urs",
    title = "{A Higher Chern-Weil derivation of AKSZ $\sigma$-models}",
    eprint = "1108.4378",
    archivePrefix = "arXiv",
    primaryClass = "math-ph",
    doi = "10.1142/S0219887812500788",
    journal = "Int. J. Geom. Meth. Mod. Phys.",
    volume = "10",
    pages = "1250078",
    year = "2013"
}

@article{Alfonsi:2021bot,
    author = "Alfonsi, Luigi and Berman, David S.",
    title = "{Double field theory and geometric quantisation}",
    eprint = "2101.12155",
    archivePrefix = "arXiv",
    primaryClass = "hep-th",
    doi = "10.1007/JHEP06(2021)059",
    journal = "JHEP",
    volume = "06",
    pages = "059",
    year = "2021"
}

@article{Alfonsi:2020lub,
    author = "Alfonsi, Luigi and White, Chris D. and Wikeley, Sam",
    title = "{Topology and Wilson lines: global aspects of the double copy}",
    eprint = "2004.07181",
    archivePrefix = "arXiv",
    primaryClass = "hep-th",
    reportNumber = "QMUL-PH-20-08",
    doi = "10.1007/JHEP07(2020)091",
    journal = "JHEP",
    volume = "07",
    pages = "091",
    year = "2020"
}

@misc{Bunk:2021quu,
    author = "Bunk, Severin",
    title = "{Gerbes in Geometry, Field Theory, and Quantisation}",
    eprint = "2102.10406",
    archivePrefix = "arXiv",
    primaryClass = "math.DG",
    reportNumber = "Hamburger Beitraege Nr. 886, ZMP-HH/21-1",
    month = "2",
    year = "2021"
}

@article{AAL94,
   title={A canonical approach to duality transformations},
   volume={336},
   ISSN={0370-2693},
   url={http://dx.doi.org/10.1016/0370-2693(94)00982-1},
   DOI={10.1016/0370-2693(94)00982-1},
   number={2},
   journal={Physics Letters B},
   publisher={Elsevier BV},
   author={Álvarez, Enrique and Álvarez-Gaumé, Luis and Lozano, Yolanda},
   year={1994},
   month={Sep},
   pages={183–189}
}

@inproceedings{Fiorenza:2013jz,
    author = "Fiorenza, Domenico and Sati, Hisham and Schreiber, Urs",
    title = "{A higher stacky perspective on Chern-Simons theory}",
    booktitle = "{Winter School in Mathematical Physics}: {Mathematical Aspects of Quantum Field Theory}",
    eprint = "1301.2580",
    archivePrefix = "arXiv",
    primaryClass = "hep-th",
    doi = "10.1007/978-3-319-09949-1_6",
    publisher = "Springer",
    month = "1",
    year = "2013"
}

@misc{Sev01,
    title={Some title containing the words "homotopy" and "symplectic", e.g. this one},
    author={Pavol Severa},
    year={2001},
    eprint={math/0105080},
    archivePrefix={arXiv},
    primaryClass={math.SG}
}

@article{Doubek_2019,
   title={Quantum ${L_\infty}$-Algebras and the Homological Perturbation Lemma},
   volume={367},
   ISSN={1432-0916},
   url={http://dx.doi.org/10.1007/s00220-019-03375-x},
   DOI={10.1007/s00220-019-03375-x},
   number={1},
   journal={Communications in Mathematical Physics},
   publisher={Springer Science and Business Media LLC},
   author={Doubek, Martin and Jurčo, Branislav and Pulmann, Ján},
   year={2019},
   month={Feb},
   pages={215–240}
}

@misc{Fiorenza:2020iax,
    author = "Fiorenza, Domenico and Sati, Hisham and Schreiber, Urs",
    title = "{Twistorial Cohomotopy implies Green-Schwarz anomaly cancellation}",
    eprint = "2008.08544",
    archivePrefix = "arXiv",
    primaryClass = "hep-th",
    month = "8",
    year = "2020"
}

@article{Sati:2020nob,
    author = "Sati, Hisham and Schreiber, Urs",
    title = "{The character map in equivariant twistorial Cohomotopy implies the Green-Schwarz mechanism with heterotic M5-branes}",
    eprint = "2011.06533",
    archivePrefix = "arXiv",
    primaryClass = "hep-th",
    month = "11",
    year = "2020"
}

@misc{Sati:2021uhj,
    author = "Sati, Hisham and Schreiber, Urs",
    title = "{M/F-Theory as $Mf$-Theory}",
    eprint = "2103.01877",
    archivePrefix = "arXiv",
    primaryClass = "hep-th",
    month = "3",
    year = "2021"
}

@article{Benini:2019uge,
    author = "Benini, Marco and Schenkel, Alexander",
    title = "{Higher Structures in Algebraic Quantum Field Theory}",
    eprint = "1903.02878",
    archivePrefix = "arXiv",
    primaryClass = "hep-th",
    reportNumber = "ZMP-HH/18-27, Hamburger Beitraege zur Mathematik Nr. 782",
    doi = "10.1002/prop.201910015",
    journal = "Fortsch. Phys.",
    volume = "67",
    number = "8-9",
    pages = "1910015",
    year = "2019"
}

@article{Benini_2019,
	doi = {10.1007/s00220-019-03640-z},
	url = {https://doi.org/10.1007%2Fs00220-019-03640-z},
	year = 2019,
	month = {dec},
	publisher = {Springer Science and Business Media {LLC}
},
	volume = {378},
	number = {1},
	pages = {185--218},
	author = {Marco Benini and Simen Bruinsma and Alexander Schenkel},
	title = {Linear Yang{\textendash}Mills Theory as a Homotopy {AQFT}},
	journal = {Communications in Mathematical Physics}
}

@article{Benini:2018oeh,
    author = "Benini, Marco and Schenkel, Alexander and Woike, Lukas",
    title = "{Homotopy theory of algebraic quantum field theories}",
    eprint = "1805.08795",
    archivePrefix = "arXiv",
    primaryClass = "math-ph",
    reportNumber = "ZMP-HH/18-11, Hamburger Beitraege zur Mathematik Nr. 738, ZMP-HH-18-11, HAMBURGER-BEITRAEGE-ZUR-MATHEMATIK-NR.-738",
    doi = "10.1007/s11005-018-01151-x",
    journal = "Lett. Math. Phys.",
    volume = "109",
    number = "7",
    pages = "1487--1532",
    year = "2019"
}

@article{Fredenhagen_2012,
	doi = {10.1007/s00220-012-1601-1},
	url = {https://doi.org/10.1007%2Fs00220-012-1601-1},
	year = 2012,
	month = {nov},
	publisher = {Springer Science and Business Media {LLC}
},
	volume = {317},
	number = {3},
	pages = {697--725},
	author = {Klaus Fredenhagen and Katarzyna Rejzner},
	title = {Batalin-Vilkovisky Formalism in Perturbative Algebraic Quantum Field Theory},
	journal = {Communications in Mathematical Physics}
}

@misc{Borsten:2021ljb,
    author = "Borsten, Leron and Kim, Hyungrok and Saemann, Christian",
    title = "{$EL_\infty$-algebras, Generalized Geometry, and Tensor Hierarchies}",
    eprint = "2106.00108",
    archivePrefix = "arXiv",
    primaryClass = "hep-th",
    reportNumber = "EMPG-21-07",
    month = "5",
    year = "2021"
}

@article{khavkine2017synthetic,
      title={Synthetic geometry of differential equations: I. Jets and comonad structure}, 
      author={Igor Khavkine and Urs Schreiber},
      year={2017},
      eprint={1701.06238},
      archivePrefix={arXiv},
      primaryClass={math.DG}
}

@book{Quillen:1967ha,
    author = "Quillen, Daniel",
    title = "{Homotopical Algebra}",
    publisher = "Springer",
    volume = "43",
    isbn = "978-3-540-03914-3",
    year = "1967",
}

@article{Carchedi2019OnTU,
  title={On the Universal Property of Derived Manifolds},
  author={David R. Carchedi and Pelle Steffens},
  journal={arXiv: Algebraic Topology},
  year={2019}
}

@book{kock_2006, place={Cambridge}, edition={2}, series={London Mathematical Society Lecture Note Series}, title={Synthetic Differential Geometry}, DOI={10.1017/CBO9780511550812}, publisher={Cambridge University Press}, author={Kock, Anders}, year={2006}, collection={London Mathematical Society Lecture Note Series}}

@book{Moerdijk:1991,
    author = "Moerdijk, Ieke and Reyes, Gonzalo E.",
    title = "{Models for Smooth Infinitesimal Analysis}",
      publisher = "{Springer}",
          isbn = "978-0-387-97489-7",
             DOI={10.1007/978-1-4757-4143-8},
    year = "1991"
}

@book{Jardine:1999,
    author = "Goerss, Paul G. and  Jardine, John F.",
    title = "{Simplicial Homotopy Theory}",
      publisher = "{Springer}",
          isbn = "978-3-0348-9737-2",
             DOI={10.1007/978-3-0348-8707-6},
                    series = "{Progress in Mathematics (Birkh\"{a}user)}",
    year = "1999"
}

@misc{Vezzosi:2011,
  doi = {10.48550/ARXIV.1109.5213},
    eprint = "1109.5213",
    archivePrefix = "arXiv",
    primaryClass = "math.AG",
  author = {Vezzosi, Gabriele},
  title = {Derived critical loci I - Basics},
  publisher = {arXiv},
  year = {2011},
  copyright = {arXiv.org perpetual, non-exclusive license}
}

@inbook{shulman_2021,
place={Cambridge}, title={Homotopy Type Theory: The Logic of Space}, volume={1}, DOI={10.1017/9781108854429.009}, booktitle={New Spaces in Mathematics: Formal and Conceptual Reflections}, publisher={Cambridge University Press}, author={Shulman, Michael}, editor={Anel, Mathieu and Catren, GabrielEditors}, year={2021}, pages={322–404}}

@InProceedings{Soriau79,
author="Souriau, J. M.",
editor="Garc{\'i}a, P. L.
and P{\'e}rez-Rend{\'o}n, A.
and Souriau, J. M.",
title="Groupes differentiels",
booktitle="Differential Geometrical Methods in Mathematical Physics",
year="1980",
publisher="Springer Berlin Heidelberg",
address="Berlin, Heidelberg",
pages="91--128",
isbn="978-3-540-38405-2"
}

@InProceedings{Soriau84,
author="Souriau, Jean-Marie",
editor="Denardo, G.
and Ghirardi, G.
and Weber, T.",
title="Groupes diff{\'e}rentiels et physique math{\'e}matique",
booktitle="Group Theoretical Methods in Physics",
year="1984",
publisher="Springer Berlin Heidelberg",
address="Berlin, Heidelberg",
pages="511--513",
isbn="978-3-540-38859-3"
}

@book{iglesiaszemmour,
  TITLE = {{Diffeology}},
  AUTHOR = {Iglesias-Zemmour, Patrick},
  URL = {https://hal.archives-ouvertes.fr/hal-01288504},
  PUBLISHER = {{American Mathematical Society}},
  SERIES = {Mathematical Surveys and Monographs},
  VOLUME = {185},
  PAGES = {439 pages},
  YEAR = {2013},
  DOI = {10.1090/surv/185},
  HAL_ID = {hal-01288504},
  HAL_VERSION = {v1},
}

@misc{Toen14,
  doi = {10.48550/ARXIV.1403.6995},
  url = {https://arxiv.org/abs/1403.6995},
  author = {Toen, Bertrand},
  title = {Derived Algebraic Geometry and Deformation Quantization},
  publisher = {arXiv},
  year = {2014},
  copyright = {arXiv.org perpetual, non-exclusive license}
}

@book{dubuc2006kan,
  title={Kan Extensions in Enriched Category Theory},
  author={Dubuc, E.J.},
  isbn={9783540363071},
  lccn={77131542},
  series={Lecture Notes in Mathematics},
  url={https://books.google.co.uk/books?id=Zy17CwAAQBAJ},
  year={2006},
  publisher={Springer Berlin Heidelberg}
}

@book{ToenVezzo08,
    label = {HAG-II},
  doi = {10.48550/ARXIV.MATH/0404373},
  url = {https://arxiv.org/abs/math/0404373},
  author = {Toen, Bertrand and Vezzosi, Gabriele},
  title = {Homotopical Algebraic Geometry II: geometric stacks and applications},
    PUBLISHER = {{American Mathematical Society}},
  SERIES = {Memoirs of the American Mathematical Society},
  isbn= {978-1-4704-0508-3},
  year = {2008}
}

@article{ToenVezzo05,
    label = {HAG-I},
   title={Homotopical algebraic geometry I: topos theory},
   volume={193},
   ISSN={0001-8708},
   url={http://dx.doi.org/10.1016/J.AIM.2004.05.004},
   DOI={10.1016/j.aim.2004.05.004},
   number={2},
   journal={Advances in Mathematics},
   publisher={Elsevier BV},
   author={Toën, Bertrand and Vezzosi, Gabriele},
   year={2005},
   month={Jun},
   pages={257–372} 
}

@book{Rejzner:2016hdj,
    author = "Rejzner, Kasia",
    title = "{Perturbative Algebraic Quantum Field Theory}: {An Introduction for Mathematicians}",
    doi = "10.1007/978-3-319-25901-7",
    isbn = "978-3-319-25899-7, 978-3-319-25901-7",
    publisher = "Springer",
    address = "New York",
    series = "Mathematical Physics Studies",
    year = "2016"
}

@incollection{Marvan:1986,
    author = "Marvan, Michal",
    title = "{A note on the category of partial differential equations}",
    booktitle = "{Differential geometry and its applications, Proceedings of the Conference}",
    year = "1987",
    editor      = "D. Krupka and A. Svec",
    pages       = "235-244",
    publisher   = "Univ. J. E. Purkyn\'e, Brno"
}

@article{Gwilliam:2017ses,
    author = "Gwilliam, Owen and Rejzner, Kasia",
    title = "{Relating Nets and Factorization Algebras of Observables: Free Field Theories}",
    eprint = "1711.06674",
    archivePrefix = "arXiv",
    primaryClass = "math-ph",
    doi = "10.1007/s00220-019-03652-9",
    journal = "Commun. Math. Phys.",
    volume = "373",
    number = "1",
    pages = "107--174",
    year = "2020"
}

@article{Benini:2019ujs,
    author = "Benini, Marco and Perin, Marco and Schenkel, Alexander",
    title = "{Model-independent comparison between factorization algebras and algebraic quantum field theory on Lorentzian manifolds}",
    eprint = "1903.03396",
    archivePrefix = "arXiv",
    primaryClass = "math-ph",
    reportNumber = "ZMP-HH/19-6, Hamburger Beitraege zur Mathematik Nr. 781",
    doi = "10.1007/s00220-019-03561-x",
    journal = "Commun. Math. Phys.",
    volume = "377",
    number = "2",
    pages = "971--997",
    year = "2019"
}

@article{Rejzner:2020xid,
    author = "Rejzner, Kasia and Schiavina, Michele",
    title = "{Asymptotic Symmetries in the BV-BFV Formalism}",
    eprint = "2002.09957",
    archivePrefix = "arXiv",
    primaryClass = "math-ph",
    doi = "10.1007/s00220-021-04061-7",
    journal = "Commun. Math. Phys.",
    volume = "385",
    number = "2",
    pages = "1083--1132",
    year = "2021"
}

@misc{Rejzner:2020bsc,
    author = "Rejzner, Kasia",
    title = "{BV quantization in perturbative algebraic QFT: Fundamental concepts and perspectives}",
    eprint = "2004.14272",
    archivePrefix = "arXiv",
    primaryClass = "math-ph",
    month = "4",
    year = "2020"
}

@phdthesis{Rejzner:2011jsa,
    author = "Rejzner, Katarzyna Anna",
    title = "{Batalin-Vilkovisky formalism in locally covariant field theory}",
    eprint = "1111.5130",
    archivePrefix = "arXiv",
    primaryClass = "math-ph",
    reportNumber = "DESY-THESIS-2011-041",
    school = "Hamburg U.",
    year = "2011"
}

@article{Fredenhagen:2011an,
    author = "Fredenhagen, Klaus and Rejzner, Katarzyna",
    title = "{Batalin-Vilkovisky formalism in the functional approach to classical field theory}",
    eprint = "1101.5112",
    archivePrefix = "arXiv",
    primaryClass = "math-ph",
    doi = "10.1007/s00220-012-1487-y",
    journal = "Commun. Math. Phys.",
    volume = "314",
    pages = "93--127",
    year = "2012"
}

@misc{Joyce:2009,
  doi = {10.48550/ARXIV.1001.0023},
  url = {https://arxiv.org/abs/1001.0023},
  author = {Joyce, Dominic},
  title = {Algebraic Geometry over $C^\infty$-rings},
  publisher = {arXiv},
  year = {2010},
  copyright = {arXiv.org perpetual, non-exclusive license},
  note = {Memoirs of the American Mathematical Society}
}

@misc{Pridham:2021,
  doi = {10.48550/ARXIV.2109.14594},
  url = {https://arxiv.org/abs/2109.14594},
  author = {Eugster, J. and Pridham, J. P.},
  title = {An introduction to derived (algebraic) geometry},
  publisher = {arXiv},
  year = {2021},
  copyright = {arXiv.org perpetual, non-exclusive license}
}

@misc{Pridham:2018,
  doi = {10.48550/ARXIV.1804.07622},
  url = {https://arxiv.org/abs/1804.07622},
  author = {Pridham, J. P.},
  title = {An outline of shifted Poisson structures and deformation quantisation in derived differential geometry},
  publisher = {arXiv},
  year = {2018},
  copyright = {arXiv.org perpetual, non-exclusive license}
}

@ARTICLE{Borisov:2011,
       author = {{Borisov}, Dennis and {Noel}, Justin},
        title = "{Simplicial approach to derived differential manifolds}",
      journal = {arXiv e-prints},
         year = 2011,
        month = nov,
          eid = {arXiv:1112.0033},
        pages = {arXiv:1112.0033},
archivePrefix = {arXiv},
       eprint = {1112.0033},
 primaryClass = {math.DG}
}

@article{Spivak:2010,
	doi = {10.1215/00127094-2010-021},
	url = {https://doi.org/10.1215%2F00127094-2010-021},
	year = 2010,
	month = {may},
	publisher = {Duke University Press},
	volume = {153},
	number = {1},
	author = {David I. Spivak},
	title = {Derived smooth manifolds},
	journal = {Duke Mathematical Journal}
}

@inbook{Joyce:2014, place={Cambridge}, series={London Mathematical Society Lecture Note Series}, title={An introduction to d-manifolds and derived differential geometry}, DOI={10.1017/CBO9781107279544.006}, booktitle={Moduli Spaces}, publisher={Cambridge University Press}, author={Joyce, Dominic}, editor={Brambila-Paz, Leticia and Newstead, Peter and Thomas, Richard P. and García-Prada, OscarEditors}, year={2014}, pages={230–281}, collection={London Mathematical Society Lecture Note Series}}

@misc{Borisov:2012derived,
    title={Derived manifolds and Kuranishi models},
    author={Dennis Borisov},
    year={2012},
    eprint={1212.1153},
    archivePrefix={arXiv},
    primaryClass={math.DG}
}

@misc{Joyce:2012dmanifolds,
    title={D-manifolds, d-orbifolds and derived differential geometry: a detailed summary},
    author={Dominic Joyce},
    year={2012},
    eprint={1208.4948},
    archivePrefix={arXiv},
    primaryClass={math.DG}
}

@article{Toen:2014re,
   title={Derived algebraic geometry},
   volume={1},
   ISSN={2308-2151},
   url={http://dx.doi.org/10.4171/emss/4},
   DOI={10.4171/emss/4},
   number={2},
   journal={EMS Surveys in Mathematical Sciences},
   publisher={European Mathematical Society - EMS - Publishing House GmbH},
   author={Toën, Bertrand},
   year={2014},
   pages={153–245} }

@article{Calaque:2017,
	doi = {10.1112/topo.12012},
  
	url = {https://doi.org/10.1112%2Ftopo.12012},
  
	year = 2017,
	month = {apr},
  
	publisher = {Wiley},
  
	volume = {10},
  
	number = {2},
  
	pages = {483--584},
  
	author = {Damien Calaque and Tony Pantev and Bertrand Toën and Michel Vaqui{\'{e}
} and Gabriele Vezzosi},
  
	title = {Shifted Poisson structures and deformation quantization},
	journal = {Journal of Topology}
}

@misc{Wallbridge:2016,
  doi = {10.48550/ARXIV.1610.00441},
  url = {https://arxiv.org/abs/1610.00441},
  author = {Wallbridge, James},
  title = {Derived smooth stacks and prequantum categories},
  publisher = {arXiv},
  year = {2016},
  copyright = {arXiv.org perpetual, non-exclusive license}
}

@misc{Kerodon,
      label          = {Kerodon},
  author = "Jacob Lurie",
  title  = "{Kerodon}",
    note   = "URL: \url{https://kerodon.net}.",
}

@misc{Barwick:2016,
  doi = {10.48550/ARXIV.1607.04343},
  url = {https://arxiv.org/abs/1607.04343},
  author = {Barwick, Clark and Shah, Jay},
  title = {Fibrations in $\infty$-category theory},
  publisher = {arXiv},
  year = {2016},
  copyright = {arXiv.org perpetual, non-exclusive license}
}

@misc{Benini:2021tyt,
    author = "Benini, Marco and Safronov, Pavel and Schenkel, Alexander",
    title = "{Classical BV formalism for group actions}",
    eprint = "2104.14886",
    archivePrefix = "arXiv",
    primaryClass = "math-ph",
    month = "4",
    year = "2021"
}

@article{BATALIN198127,
title = {Gauge algebra and quantization},
journal = {Physics Letters B},
volume = {102},
number = {1},
pages = {27-31},
year = {1981},
issn = {0370-2693},
doi = {https://doi.org/10.1016/0370-2693(81)90205-7},
author = {I.A. Batalin and G.A. Vilkovisky}
}

@misc{ArvFormal22,
  doi = {10.48550/ARXIV.2204.12613},
  url = {https://arxiv.org/abs/2204.12613},
  author = {Arvanitakis, Alex S.},
  title = {Formal exponentials and linearisations of QP-manifolds},
  publisher = {arXiv},
  year = {2022},
  copyright = {arXiv.org perpetual, non-exclusive license}
}

@misc{Moerdijk2010SimplicialMF,
  title={Simplicial Methods for Operads and Algebraic Geometry},
  author={Ieke Moerdijk and Bertrand To{\"e}n},
  year={2010}
}

@phdthesis{Alfonsi:2021ymc,
    author = "Alfonsi, Luigi",
    title = "{Extended Field Theories as higher Kaluza-Klein theories}",
    eprint = "2108.10297",
    archivePrefix = "arXiv",
    primaryClass = "hep-th",
    school = "Queen Mary, University of London",
    year = "2021"
}

@article{Grady:2016LIEAA,
  title="{Lie algebroids as $L_{\infty}$ spaces}",
  author={Ryan E. Grady and Owen Gwilliam},
  journal={Journal of the Institute of Mathematics of Jussieu},
  year={2016},
  volume={19},
  pages={487 - 535}
}

@misc{Grady:2014oqa,
    author = "Grady, Ryan and Gwilliam, Owen",
    title = "{$L_\infty$ spaces and derived loop spaces}",
    eprint = "1404.5426",
    archivePrefix = "arXiv",
    primaryClass = "math.AG",
    month = "4",
    year = "2014"
}

@article{Grady:2020,
    author = "Grady, Ryan E.",
    title = "{Quantizing Derived Mapping Stacks}",
    eprint = "2009.04064",
    archivePrefix = "arXiv",
    primaryClass = "math-ph",
    doi = "10.1142/S0217751X20300173",
    journal = "Int. J. Mod. Phys. A",
    volume = "35",
    number = "30",
    pages = "2030017",
    year = "2020"
}

@article{Dubuc1979,
author = {Dubuc, Eduardo J.},
journal = {Cahiers de Topologie et Géométrie Différentielle Catégoriques},
language = {fre},
number = {3},
pages = {231-279},
publisher = {Dunod éditeur, publié avec le concours du CNRS},
title = {Sur les modèles de la géométrie différentielle synthétique},
url = {http://eudml.org/doc/91216},
volume = {20},
year = {1979},
}

@article{Benini_2018,
	doi = {10.1007/s00220-018-3120-1},
	url = {https://doi.org/10.1007%2Fs00220-018-3120-1},
	year = 2018,
	month = {mar},
	publisher = {Springer Science and Business Media {LLC}
},
	volume = {359},
	number = {2},
	pages = {765--820},
	author = {Marco Benini and Alexander Schenkel and Urs Schreiber},
	title = {The Stack of Yang-Mills Fields on Lorentzian Manifolds},
	journal = {Communications in Mathematical Physics}
}

@book{FactI, place={Cambridge}, series={New Mathematical Monographs}, title={Factorization Algebras in Quantum Field Theory}, volume={1}, DOI={10.1017/9781316678626}, publisher={Cambridge University Press}, author={Costello, Kevin and Gwilliam, Owen}, year={2016}, collection={New Mathematical Monographs}}

@book{FactII, place={Cambridge}, series={New Mathematical Monographs}, title={Factorization Algebras in Quantum Field Theory}, volume={2}, DOI={10.1017/9781107163157}, publisher={Cambridge University Press}, author={Costello, Kevin and Gwilliam, Owen}, year={2021}, collection={New Mathematical Monographs}}

@misc{safronov2020shifted,
      title={Shifted geometric quantization}, 
      author={Pavel Safronov},
      year={2020},
      eprint={2011.05730},
      archivePrefix={arXiv},
      primaryClass={math.SG}
}

@misc{AlfYouFuture,
    author = "Alfonsi, Luigi and Charles, Young A.",
    title = "{Classical BV-BFV theory as derived $n$-plectic geometry}",
    year={2023},
    note = "in preparation"
}

@misc{steffens2023derived,
      title={Derived $C^{\infty}$-Geometry I: Foundations}, 
      author={Pelle Steffens},
      year={2023},
      eprint={2304.08671},
      archivePrefix={arXiv},
      primaryClass={math.AG}
}

@misc{myers2022orbifolds,
      title={Orbifolds as microlinear types in synthetic differential cohesive homotopy type theory}, 
      author={David Jaz Myers},
      year={2022},
      eprint={2205.15887},
      archivePrefix={arXiv},
      primaryClass={math.AT}
}

@book{vogler2013derived,
  title={Derived Manifolds from Functors of Points},
  author={Vogler, Franz},
  volume={22},
  year={2013},
  publisher={Logos Verlag Berlin GmbH}
}

@misc{zeng2022derived,
      title={Derived Lie $\infty$-groupoids and algebroids in higher differential geometry}, 
      author={Andy Zeng},
      year={2022},
      eprint={2210.05856},
      archivePrefix={arXiv},
      primaryClass={math.DG}
}

@misc{joyce2017kuranishi,
  title={Kuranishi spaces and Symplectic Geometry},
  author={Joyce, Dominic},
  note="{Volume II}",
  year={2017}
}

@misc{borsten2023treelevel,
      title={Tree-Level Color-Kinematics Duality from Pure Spinor Actions}, 
      author={Leron Borsten and Branislav Jurco and Hyungrok Kim and Tommaso Macrelli and Christian Saemann and Martin Wolf},
      year={2023},
      eprint={2303.13596},
      archivePrefix={arXiv},
      primaryClass={hep-th}
}

@inproceedings{Macrelli_2022,
	doi = {10.22323/1.414.0426},
	year = {2022},
	month = {nov},
	publisher = {Sissa Medialab},
	author = {Tommaso Macrelli and Leron Borsten and Hyungrok Kim and Branislav Jurco and Christian Saemann and Martin Wolf},
	title = {Colour-kinematics duality, double copy, and homotopy algebras},
	booktitle = {Proceedings of 41st International Conference on High Energy physics {\textemdash} {PoS}({ICHEP}2022)}
}

@article{Borsten_2023,
	doi = {10.1016/j.nuclphysb.2023.116144},
	year = 2023,
	month = {apr},
	publisher = {Elsevier {BV}},
	volume = {989},
	pages = {116144},
	author = {Leron Borsten and Hyungrok Kim and Branislav Jur{\v{c}}o and Tommaso Macrelli and Christian Saemann and Martin Wolf},
	title = {Tree-level color{\textendash}kinematics duality implies loop-level color{\textendash}kinematics duality up to counterterms},
	journal = {Nuclear Physics B}
}

@article{Borsten:2021hua,
    author = "Borsten, Leron and Kim, Hyungrok and Jur\v{c}o, Branislav and Macrelli, Tommaso and Saemann, Christian and Wolf, Martin",
    title = "{Double Copy from Homotopy Algebras}",
    eprint = "2102.11390",
    archivePrefix = "arXiv",
    primaryClass = "hep-th",
    reportNumber = "DMUS-MP-21/04, EMPG-21-03",
    doi = "10.1002/prop.202100075",
    journal = "Fortsch. Phys.",
    volume = "69",
    number = "8-9",
    pages = "2100075",
    year = "2021"
}

@article{Borsten:2020zgj,
    author = "Borsten, Leron and Jur\v{c}o, Branislav and Kim, Hyungrok and Macrelli, Tommaso and Saemann, Christian and Wolf, Martin",
    title = "{Becchi-Rouet-Stora-Tyutin-Lagrangian Double Copy of Yang-Mills Theory}",
    eprint = "2007.13803",
    archivePrefix = "arXiv",
    primaryClass = "hep-th",
    reportNumber = "DMUS-MP-20/06, EMPG-20-13",
    doi = "10.1103/PhysRevLett.126.191601",
    journal = "Phys. Rev. Lett.",
    volume = "126",
    number = "19",
    pages = "191601",
    year = "2021"
}

@misc{Kim:2022opr,
    author = "Kim, Hyungrok and Saemann, Christian",
    title = "{Non-Geometric T-Duality as Higher Groupoid Bundles with Connections}",
    eprint = "2204.01783",
    archivePrefix = "arXiv",
    primaryClass = "hep-th",
    reportNumber = "EMPG-22-05",
    month = "4",
    year = "2022"
}

@article{DimitrijevicCiric:2021jea,
    author = "Dimitrijevi\'c \'Ciri\'c, Marija and Giotopoulos, Grigorios and Radovanovi\'c, Voja and Szabo, Richard J.",
    title = "{Braided $L_\infty$-algebras, braided field theory and noncommutative gravity}",
    eprint = "2103.08939",
    archivePrefix = "arXiv",
    primaryClass = "hep-th",
    reportNumber = "EMPG-21-04",
    doi = "10.1007/s11005-021-01487-x",
    journal = "Lett. Math. Phys.",
    volume = "111",
    number = "6",
    pages = "148",
    year = "2021"
}

@inproceedings{CiricDimitrijevic:2022eei,
    author = "\'Ciri\'c Dimitrijevi\'c, Marija and Konjik, Nikola and Radovanovi\'c, Voja and Szabo, Richard J. and Toman, Mi\v{s}a",
    title = "{$L_\infty$-algebra of braided electrodynamics}",
    booktitle = "{21st Hellenic School and Workshops on Elementary Particle Physics and Gravity}",
    eprint = "2204.06448",
    archivePrefix = "arXiv",
    primaryClass = "hep-th",
    reportNumber = "EMPG-22-07",
    month = "4",
    year = "2022"
}

@misc{DimitrijevicCiric:2023hua,
    author = "Dimitrijevi\'c \'Ciri\'c, Marija and Konjik, Nikola and Radovanovi\'c, Voja and Szabo, Richard J.",
    title = "{Braided Quantum Electrodynamics}",
    eprint = "2302.10713",
    archivePrefix = "arXiv",
    primaryClass = "hep-th",
    reportNumber = "EMPG-23-02",
    month = "2",
    year = "2023"
}

@article{Chiaffrino:2021pob,
    author = "Chiaffrino, Christoph and Hohm, Olaf and Pinto, Allison F.",
    title = "{Homological Quantum Mechanics}",
    eprint = "2112.11495",
    archivePrefix = "arXiv",
    primaryClass = "hep-th",
    reportNumber = "HU-EP-21/56-RTG",
    month = "12",
    year = "2021"
}

@article{Khavkine:2014kya,
    author = "Khavkine, Igor",
    title = "{Covariant phase space, constraints, gauge and the Peierls formula}",
    eprint = "1402.1282",
    archivePrefix = "arXiv",
    primaryClass = "math-ph",
    doi = "10.1142/S0217751X14300099",
    journal = "Int. J. Mod. Phys. A",
    volume = "29",
    number = "5",
    pages = "1430009",
    year = "2014"
}

@article{Khavkine:2012jf,
    author = "Khavkine, Igor",
    title = "{Characteristics, Conal Geometry and Causality in Locally Covariant Field Theory}",
    eprint = "1211.1914",
    archivePrefix = "arXiv",
    primaryClass = "gr-qc",
    month = "11",
    year = "2012"
}

@article{Blander2001LocalPM,
  title={Local Projective Model Structures on Simplicial Presheaves},
  author={Benjamin A. Blander},
  journal={K-theory},
  year={2001},
  volume={24},
  pages={283-301}
}

@misc{Dugger2009SHEAVESAH,
  title={Sheaves and homotopy theory},
  author={Daniel Dugger},
    howpublished = "\url{https://math.mit.edu/~dspivak/files/cech.pdf}",
  year={2009}
}

@misc{JCMBLecGauge,
  author        = {Jos\'e Figueroa-O'Farrill},
  title         = {Lecture notes in Gauge Theory -- Lecture 3: The Yang–Mills equations},
  year          = {2006},
  howpublished = "\url{https://empg.maths.ed.ac.uk/Activities/GT/Lect3.pdf}",
  publisher={The University of Edinburgh}
}

@article{hawkins2020star,
  title={The star product in interacting quantum field theory},
  author={Hawkins, Eli and Rejzner, Kasia},
  journal={Letters in Mathematical Physics},
  volume={110},
  number={6},
  pages={1257--1313},
  year={2020},
  publisher={Springer}
}

@book{Costello2011RenormalizationAE,
  title={Renormalization and Effective Field Theory},
  author={Kevin J. Costello},
    volume={170},
  publisher={American Mathematical Society},
  collection = {Mathematical Surveys and Monographs},
  year={2011}
}

@inproceedings{toen2007moduli,
  title={Moduli of objects in dg-categories},
  author={To{\"e}n, Bertrand and Vaqui{\'e}, Michel},
  booktitle={Annales scientifiques de l'Ecole normale sup{\'e}rieure},
  volume={40},
  number={3},
  pages={387--444},
  year={2007}
}

@article{badzioch2002algebraic,
  title={Algebraic theories in homotopy theory},
  author={Badzioch, Bernard},
  journal={Annals of mathematics},
  volume={155},
  number={3},
  pages={895--913},
  year={2002},
  publisher={JSTOR}
}

@article{Cattaneo:2023hxv,
    author = "Cattaneo, Alberto S. and Mnev, Pavel and Schiavina, Michele",
    title = "{BV Quantization}",
    eprint = "2307.07761",
    archivePrefix = "arXiv",
    primaryClass = "math-ph",
    month = "7",
    year = "2023"
}

@book{adámek_rosický_vitale_lawvere_2010, place={Cambridge}, series={Cambridge Tracts in Mathematics}, title={Algebraic Theories: A Categorical Introduction to General Algebra}, DOI={10.1017/CBO9780511760754}, publisher={Cambridge University Press}, author={Adámek, J. and Rosický, J. and Vitale, E. M. and Lawvere, F. W.}, year={2010}, collection={Cambridge Tracts in Mathematics}}

@article{pridham2010unifying,
  title={Unifying derived deformation theories},
  author={Pridham, Jon P},
  journal={Advances in Mathematics},
  volume={224},
  number={3},
  pages={772--826},
  year={2010},
  publisher={Elsevier}
}

@article{HINICH2001209,
title = {DG coalgebras as formal stacks},
journal = {Journal of Pure and Applied Algebra},
volume = {162},
number = {2},
pages = {209-250},
year = {2001},
issn = {0022-4049},
doi = {https://doi.org/10.1016/S0022-4049(00)00121-3},
url = {https://www.sciencedirect.com/science/article/pii/S0022404900001213},
author = {Vladimir Hinich}
}

@misc{carchedi2023derived,
      title={Derived Manifolds as Differential Graded Manifolds}, 
      author={David Carchedi},
      year={2023},
      eprint={2303.11140},
      archivePrefix={arXiv},
      primaryClass={math.DG}
}

@article{carchedi2012homological,
  title={Homological algebra for superalgebras of differentiable functions},
  author={Carchedi, David and Roytenberg, Dmitry},
  journal={arXiv preprint arXiv:1212.3745},
  year={2012}
}

@misc{kontsevich1998noncommutative,
      title={Noncommutative smooth spaces}, 
      author={Maxim Kontsevich and Alexander Rosenberg},
      year={1998},
      eprint={math/9812158},
      archivePrefix={arXiv},
      primaryClass={math.AG}
}

@misc{Simpson:1997,
title = {De Rham's theorem for $\infty$-stacks},
year={1997},
url = {https://math.berkeley.edu/~teleman/math/simpson.pdf},
author = {Carlos Simpson and Constantin Teleman}
}

\end{document}